\documentclass[11pt]{article}
\usepackage[margin=1in]{geometry}
\usepackage{fullpage}
\usepackage{tgtermes}
\usepackage[T1]{fontenc}
\usepackage[colorlinks,citecolor=blue,linkcolor=blue,urlcolor=red,pagebackref]{hyperref}
\usepackage{graphicx}
\usepackage{amsfonts,amsmath,amsthm,amssymb,dsfont,mathtools}
\usepackage{xfrac,nicefrac}
\usepackage{mathdots}
\usepackage{bm,bbm}
\usepackage{url}
\usepackage{paralist}
\usepackage{enumerate}
\usepackage[normalem]{ulem}
\usepackage{enumitem}
\usepackage{xspace}
\xspaceaddexceptions{]\}}
\usepackage[capitalise]{cleveref}
\usepackage{comment}
\usepackage{tabu}
\usepackage{framed}
\usepackage{float,wrapfig}
\usepackage{subfigure}
\usepackage{tikz}
\usepackage{bm}
\usepackage{soul}
\usepackage{enumitem}
\usepackage[usenames,dvipsnames]{pstricks}
\usepackage[linesnumbered,boxed,ruled,vlined]{algorithm2e}

\SetCommentSty{mycommfont}
\usetikzlibrary{decorations.pathreplacing}
\theoremstyle{plain}

\newtheorem{theorem}{Theorem}[section]
\newtheorem{lemma}[theorem]{Lemma}
\newtheorem{fact}[theorem]{Fact}
\newtheorem{observation}[theorem]{Observation}
\newtheorem{prop}[theorem]{Proposition}
\newtheorem{cor}[theorem]{Corollary}
\theoremstyle{definition}
\newtheorem{claim}{Claim}
\newtheorem{definition}[theorem]{Definition}
\newtheorem{remark}[theorem]{Remark}

\newcommand{\highlight}[1]{\medskip\noindent\textbf{#1}}

\newenvironment{proofof}[1]{\begin{proof}[Proof of #1]}{\end{proof}}

\newenvironment{reminder}[1]{\bigskip
	\noindent {\bf Reminder of #1.  }\em}{\smallskip}

\def\ShowAuthNotes{1}
\ifnum\ShowAuthNotes=1
\newcommand{\authnote}[2]{\ \\ \textcolor{red}{\parbox{0.9\linewidth}{[{\footnotesize {\bf #1:} { {#2}}}]}}\newline}
\else
\newcommand{\authnote}[2]{}
\fi

\newcommand{\cnote}[1]{\authnote{Ce}{#1}}
\newcommand{\wnote}[1]{\authnote{Hongxun}{#1}}
\newcommand{\lnote}[1]{\authnote{Lijie}{#1}}
\newcommand{\rnote}[1]{\authnote{Ryan}{#1}}

\renewcommand{\epsilon}{\varepsilon}
\newcommand{\eps}{\varepsilon}
\renewcommand{\Pr}{\operatorname*{\mathbf{Pr}}}
\newcommand{\E}{\operatorname*{\mathbf{E}}}

\newcommand{\poly}{\operatorname{\mathrm{poly}}}
\newcommand{\polylog}{\poly\log}

\newcommand{\N}{\mathbbm{N}}

\newcommand{\caH}{\mathcal{H}}
\newcommand{\vc}{\vec{c}}

\newcommand{\calX}{\mathcal{X}}
\newcommand{\calY}{\mathcal{Y}}
\newcommand{\calZ}{\mathcal{Z}}

\newcommand{\getsR}{\in_{\sf R}}
\newcommand{\calB}{\mathcal{B}}
\newcommand{\calA}{\mathcal{A}}
\newcommand{\calG}{\mathcal{G}}
\newcommand{\calF}{\mathcal{F}}
\newcommand{\calK}{\mathcal{K}}
\newcommand{\Kshort}{\calK^{\sf short}}

\newcommand{\elong}{\event_{\sf long}}

\newcommand{\bh}{\bm{h}}
\newcommand{\bw}{\bm{w}}
\newcommand{\bs}{\bm{s}}
\newcommand{\level}{\mathsf{level}}

\newcommand{\vk}{\vec{k}}
\newcommand{\vb}{\vec{b}}
\newcommand{\event}{\mathcal{E}}

\newcommand{\bT}{\bm{T}}

\newcommand{\br}{\bm{r}}
\newcommand{\bg}{\bm{g}}

\newcommand{\parent}{\mathsf{par}}
\newcommand{\walk}{\mathsf{walk}}
\newcommand{\nex}{\mathsf{next}}
\newcommand{\x}{\mathsf{x}}
\renewcommand{\a}{\mathsf{a}}
\newcommand{\extwalk}{\mathsf{Walk}}
\newcommand{\rig}{\mathsf{right}}

\newcommand{\vidx}{\mu}

\DeclarePairedDelimiter{\tvdbasic}{\lVert}{\rVert}
\makeatletter
\newcommand{\@tvdstar}[2]{\tvdbasic*{#1 - #2}_{\mathrm{TV}}}
\newcommand{\@tvdnostar}[3][]{\tvdbasic[#1]{#2 - #3}_{\mathrm{TV}}}
\newcommand{\tvd}{\@ifstar\@tvdstar\@tvdnostar}
\makeatother

\usepackage[breakable]{tcolorbox}
\newtcolorbox{construction}[2][]
{
	breakable,
	colframe = gray!50,
	colback  = gray!10,
	coltitle = gray!10!black,
	before skip = 10pt,
	after skip = 10pt,
	title    = \textbf{#2},
	#1,
}
\newtcolorbox{graphview}[2][]
{
	breakable,
	colframe = black!30,
	colback  = black!0,
	coltitle = gray!10!black,
	before skip = 10pt,
	after skip = 10pt,
	title    = \textbf{#2},
	#1,
}
\newtcolorbox{notation}[2][]
{
	breakable,
	colframe = black!30,
	colback  = black!0,
	coltitle = gray!10!black,
	before skip = 10pt,
	after skip = 10pt,
	title    = \textbf{#2},
	#1,
}

\newcommand{\ie}{\textit{i}.\textit{e}.\@\xspace}
\newcommand{\eg}{\textit{e}.\textit{g}.\@\xspace}

\newcommand{\pa}{\operatorname*{\mathsf{par}}}
\newcommand{\vK}{\vec{K}}

\newcommand{\vidxb}[1]{\vidx\left[#1\right]}
\newcommand{\vstart}{\vspace{-1mm}}
\newcommand{\vadj}{\vspace{-2mm}}
\newcommand{\vmid}{\vspace{-8mm}}

\newcommand{\twopropsshort}[4]{
	\vstart
	\begin{flalign}\label{#2}
	\bullet~~\text{#1}&&
	\end{flalign}
	\vmid
	\begin{flalign}\label{#4}
	\bullet~~\text{#3}&&
	\end{flalign}
}

\newcommand{\twoprops}[4]{
	\vstart
	\begin{flalign}\label{#2}
	\qquad\bullet\qquad\text{#1}&&
	\end{flalign}
	\vmid
	\begin{flalign}\label{#4}
	\qquad\bullet\qquad\text{#3}&&
	\end{flalign}
}
\newcommand{\threeprops}[6]{
	\vstart
	\begin{flalign}\label{#2}
	\qquad\bullet\qquad\text{#1}&&
	\end{flalign}
	\vmid
	\begin{flalign}\label{#4}
	\qquad\bullet\qquad\text{#3}&&
	\end{flalign}
	\vmid
	\begin{flalign}\label{#6}
	\qquad\bullet\qquad\text{#5}&&
	\end{flalign}
}
\newcommand{\fourprops}[8]{
	\vstart
	\begin{flalign}\label{#2}
	\qquad\bullet\qquad\text{#1}&&
	\end{flalign}
	\vmid
	\begin{flalign}\label{#4}
	\qquad\bullet\qquad\text{#3}&&
	\end{flalign}
	\vmid
	\begin{flalign}\label{#6}
	\qquad\bullet\qquad\text{#5}&&
	\end{flalign}
	\vmid
	\begin{flalign}\label{#8}
	\qquad\bullet\qquad\text{#7}&&
	\end{flalign}
}
\begin{document}

\author{Lijie Chen\\MIT\\ \texttt{lijieche@mit.edu} \and Ce Jin\\MIT \\ \texttt{cejin@mit.edu}\and R. Ryan Williams\\MIT \\ \texttt{rrw@mit.edu}\and Hongxun Wu \\ Tsinghua University \\ \texttt{wuhx18@mails.tsinghua.edu.cn} }

\title{Truly Low-Space Element Distinctness and Subset Sum\\ via Pseudorandom Hash Functions\thanks{Supported by NSF CCF-1909429 and NSF CCF-2127597. Lijie Chen is also supported by an IBM Fellowship.}}
	\setcounter{page}{0} \clearpage
	\maketitle
	\thispagestyle{empty}
	\begin{abstract}
	We consider low-space algorithms for the classic {\sc Element Distinctness} problem: given an array of $n$ input integers with $O(\log n)$ bit-length, decide whether or not all elements are pairwise distinct. 
	Beame, Clifford, and Machmouchi [FOCS 2013] gave an $\tilde O(n^{1.5})$-time randomized algorithm for {\sc Element Distinctness} using only $O(\log n)$ bits of working space. However, their algorithm assumes a random oracle (in particular, read-only random access to polynomially many random bits), and it was asked as an open question whether this assumption can be removed.

In this paper, we positively answer this question by giving an $\tilde O(n^{1.5})$-time randomized algorithm using $O(\log ^3 n\log \log n)$ bits of space, \emph{with one-way access to random bits}. As a corollary, we also obtain a $\poly(n)$-space $O^*(2^{0.86n})$-time randomized algorithm for the \emph{Subset Sum} problem, removing the random oracles required in the algorithm of Bansal, Garg, Nederlof, and Vyas [STOC 2017].    

The main technique underlying our results is a pseudorandom hash family based on iterative restrictions, which can fool the cycle-finding procedure in the algorithms of Beame et al.\ and Bansal et al.
\end{abstract} \newpage

\maketitle

\renewcommand{\tilde}{\widetilde}

\setcounter{page}{0} \clearpage
\thispagestyle{empty}
\setcounter{tocdepth}{2}
\tableofcontents
\newpage

\section{Introduction}
What problems can be solved simultaneously in low time and low space? When we restrict the space usage for solving a problem, how does this affect the possible running time of algorithms? The area of time-space tradeoffs has studied such questions for decades, beginning with Cobham~\cite{conf/focs/Cobham66}. 
A central problem studied in time-space tradeoffs is {\sc Element Distinctness}:
\begin{quote}
    {\sc Element Distinctness:} Given an array of $n$ positive integers $a_1,a_2,\dots,a_n$ with $a_i \leq \poly(n)$ for all $i$, decide whether all $a_i$'s are distinct.
\end{quote}
The problem is extremely basic and useful: thinking of the array as describing a function from $[n]$ to $[\poly(n)]$, we are asking if the function is injective. The obvious algorithm that checks all pairs of elements takes $O(n^2)$ time and uses $O(\log n)$ bits of workspace. If we allow $\tilde{O}(n)$ bits of workspace, {\sc Element Distinctness} can be solved in near-linear time by sorting the input array. 
Applying low-space sorting algorithms directly~\cite{DBLP:journals/tcs/MunroP80,DBLP:conf/focs/PagterR98}, one can interpolate between these two algorithms and solve {\sc Element Distinctness} in time $T(n)$ and space $S(n)$ for all $T(n),S(n)$ such that $T(n) \cdot S(n) \leq \tilde{O}(n^2)$. 
For \emph{comparison-based} algorithms, in which the only operation on elements allowed are pairwise comparisons, this time-space tradeoff was shown to be near-optimal in the 1980s~\cite{DBLP:journals/siamcomp/BorodinFHUW87,DBLP:conf/focs/Yao88}.

In 2013, Beame, Clifford, and Machmouchi~\cite{beame2013element} surprisingly bypassed this longstanding lower bound, by giving a non-comparison-based algorithm for {\sc Element Distinctness} with the time-space tradeoff $T(n)\leq \tilde O(n^{3/2}/S(n)^{1/2})$. In particular, their algorithm can run in $\tilde O(n^{1.5})$ time using only $O(\log n)$ bits of space. For brevity, we call this the \emph{BCM algorithm}. A major disadvantage of the BCM algorithm is that it requires a \emph{random oracle}: read-only random access to polynomially many uniform random bits (which do not count towards the space complexity). In the BCM algorithm, these random bits are used to specify the outgoing edges of a random 1-out digraph, on which Floyd's cycle-finding algorithm~\cite{donald1969art} is performed to look for a pair of equal elements. Due to complicated dependencies on the paths in this random digraph, it looks difficult to reduce the number of random bits using pseudorandomness. It was asked as an open question \cite{beame2013element,twonikhils} whether the BCM algorithm can be modified to work with only ``one-way access'' to random bits, where we may toss up to $O(t)$ coins in time $t$, but cannot randomly access arbitrary coins tossed in the past. In particular, \cite{beame2013element} stated it ``seems plausible'' that the random oracle in the BCM algorithm could be replaced by some family of $\polylog(n)$-wise independent hash functions in the analysis.

\subsection{Our Results}

Our main result in this paper proves that one-way access to randomness is sufficient for implementing the BCM algorithm. 
We design a pseudorandom hash family with $O(\log^3 n\log \log n)$-bit seed length based on iterative restrictions of $O(\log n\log \log n)$-wise independent generators, and show that the analysis of the 
BCM algorithm still works when the random oracle is replaced by our pseudorandom generator. In fact, our proofs use a careful coupling-based analysis of an infinite tree generated from our pseudorandom generator. Hence we have the following result.

\begin{theorem}
\label{thm:main-ed}
{\sc Element Distinctness} can be decided by a Monte Carlo algorithm in $\tilde{O}(n^{1.5})$ time, with $O(\log^3 n\log \log n)$ bits of workspace and no random oracle. Moreover, when there is a colliding pair, the algorithm reports one.
\end{theorem}

A closely related problem is the {\sc List Disjointness} problem (which is equivalent to the {\sc 2-Sum} problem).
\begin{quote}
{\sc List Disjointness:} Given two integer arrays $(a_1,a_2,\dots,a_n)$ and $(b_1,b_2,\dots,b_n)$ with entries in $[\poly(n)]$, decide whether there are $i,j \in [n]$ such that $a_i=b_j$.
\end{quote}
This problem is harder than {\sc Element Distinctness}, since the latter problem can be easily reduced to the former with only $O(\log n)$-factor overhead. The BCM algorithm for {\sc Element Distinctness} does not straightforwardly extend to {\sc List Disjointness}, and it is still open whether {\sc List Disjointness} can be solved in $n^{o(1)}$-space and $n^{2-\Omega(1)}$ time, even allowing random oracles.  Recently, Bansal, Garg, Nederlof, and Vyas~\cite{twonikhils} showed that a variant of the BCM algorithm can be applied to solve {\sc List Disjointness} with an improved running time, provided that the input arrays have small second frequency moment (\ie, there are few collision pairs within each arrays).  Formally, define \[F_2(a) = \sum_{i=1}^n\sum_{j=1}^n\mathbf{1}[a_i=a_j],\] and assume an upper bound $p$ on $F_2(a)+F_2(b)$ is known. Then their algorithm solves the {\sc List Disjointness} problem in $\tilde O(n\sqrt{p/s})$ time and $O(s\log n)$ space (with random oracle), for any $s\le n^2/p$.  In this paper, we show that our pseudorandom family designed for the BCM algorithm also applies to this setting for $s=1$. 
\begin{theorem} 
\label{thm:listdisj} There is a Monte Carlo algorithm for {\sc List Disjointness} such that, given input arrays $a=(a_1,\dots,a_n),b=(b_1,\dots,b_n)$ and an upper bound $p\ge F_2(a)+F_2(b)$, runs in $\tilde O(n\sqrt{p})$ time and uses $O(\log^3 n\log \log n)$ bits of workspace and no random oracle.
\end{theorem}

Combining the above {\sc List Disjointness} algorithm with additive-combinatorial techniques, Bansal et al.\ gave a $\poly(n)$-space $O^*(2^{0.86n})$-time algorithm for \emph{Subset Sum}: Given positive input integers $a_1,a_2\dots,a_n$ and a target integer $t$, find a subset of the input integers that sums to exactly $t$. They also solved the harder \emph{Knapsack} problem with essentially the same time and space complexity.
Replacing their {\sc List Disjointness} subroutine with our \Cref{thm:listdisj}, we immediately remove the assumption of random oracles in these algorithms as well. 

\begin{theorem}[{Follows from \Cref{thm:listdisj} and \cite{twonikhils}}]
\label{thm:subsetsum}
Subset Sum and Knapsack can be solved by a Monte Carlo algorithm in $O^*(2^{0.86 n})$ time, with $O(\poly(n))$ working space and no random oracle.
\end{theorem}

In our {\sc Element Distinctness} algorithm (\Cref{thm:main-ed}), the $1.5$ exponent in the time complexity seems hard to improve using current techniques. However, it is also difficult to prove a matching lower bound for such a decision problem. Hence we are motivated to look at a closely related multi-output problem for which our techniques still apply, and for which stronger time-space lower bounds are known. We consider  the \emph{Set Intersection} problem:
\begin{quote}
    {\sc Set Intersection:} Given two integer sets $A,B$ represented as two (not necessarily sorted) input arrays $(a_1,\dots,a_n),(b_1,\dots,b_n)$ which are promised to not contain duplicates, print all the elements in their intersection $A\cap B$. 
\end{quote}
Patt-Shamir and Peleg \cite{DBLP:journals/tcs/Patt-ShamirP93} showed that any $\polylog(n)$-space algorithm for this problem must have time complexity $\tilde \Omega(n^{1.5})$,
even if the printed elements can be in any order, and each element in $A\cap B$ is allowed to be printed multiple times. (The recent work of Dinur  \cite{DBLP:conf/eurocrypt/Dinur20} also implies the same lower bound.) 
We observe that our techniques imply a \emph{nearly-matching time upper bound} for this problem, up to polylogarithmic factors.

\begin{theorem}[Set Intersection]
\label{thm:setintersection}
There is a randomized algorithm that, given input arrays $A=(a_1,\dots,a_n),B=(b_1,\dots,b_n)$ where $A$ and $B$ are both YES instances of {\sc Element Distinctness}, prints all elements in $\{a_1,\dots,a_n\}\cap \{b_1,\dots,b_n\}$ in $\tilde O(n^{1.5})$ time, with $O(\log^3 n\log \log n)$ bits of workspace and no random oracle. The algorithm prints elements in no particular order, and the same element may be printed multiple times.
\end{theorem}

\subsection{Related Work}
In the following, we discuss several related works from various areas.

\paragraph*{Element Distinctness and Collision Finding.} In cryptography there has been intensive study on finding collisions in \emph{random-like} functions using attacks based on the birthday paradox. Floyd's cycle-finding algorithm \cite{donald1969art,1975-pollard} has been used in memoryless birthday attacks \cite{DBLP:journals/joc/OorschotW99}, which can be seen as low-space algorithms for {\sc Element Distinctness} (or {\sc List Disjointness}) with \emph{random-like} input. In contrast, we consider worst-case input and do not rely on any heuristic assumptions.

Ambainis \cite{ambainis2007quantum} gave a quantum algorithm for {\sc Element Distinctness} (as well as {\sc List Disjointness}) with optimal $O(n^{2/3})$ query complexity \cite{DBLP:journals/jacm/AaronsonS04}. The space complexity of Ambainis' algorithm is $\tilde O(n^{2/3})$. In the $\polylog(n)$-space setting, there are no known quantum algorithms that can significantly beat the simple $O(n)$-query algorithm obtainable from Grover Search~\cite{tqc}.

\paragraph*{Time-Space Tradeoff Lower Bounds.}

Borodin and Cook \cite{DBLP:journals/siamcomp/BorodinC82} proved nearly-optimal time-space tradeoff lower bounds for the sorting problem against (multi-way) branching programs. Their techniques were extended to prove time-space lower bounds for many other multi-output functions 
\cite{DBLP:journals/jcss/Yesha84,DBLP:journals/siamcomp/Abrahamson87,DBLP:journals/jcss/Abrahamson91,DBLP:journals/siamcomp/Beame91,DBLP:journals/tcs/MansourNT93,DBLP:journals/tcs/Patt-ShamirP93}. Recently, McKay and Williams   \cite{DBLP:conf/innovations/McKayW19} generalized techniques of Beame~\cite{DBLP:journals/siamcomp/Beame91} to show quadratic time-space product lower bounds against branching programs armed with random oracles. However, these techniques cannot prove nontrivial time-space lower bounds for \emph{decision} problems such as {\sc Element Distinctness}. For decision problems, the current best known time-space lower bound states that $\mathsf{SAT}$ cannot be solved in $n^{1.801}$ time and $n^{o(1)}$ space (\cite{DBLP:journals/cc/Williams08, DBLP:journals/cc/BussW15}, building on \cite{DBLP:journals/jacm/FortnowLMV05}). For {\sc Element Distinctness}, Ajtai~\cite{Ajtai05} proved that for every $k \ge 1$, there exists an $\eps > 0$ such that it cannot be solved by $kn$-time $\eps n$-space algorithms in the RAM model. Other time-space tradeoff lower bounds for decision problems are proved in~\cite{DBLP:journals/tcs/Karchmer86,Ajtai02,BeameV02,BeameSSV03}.

\paragraph*{Random oracles.}
In the usual notion of randomized space-bounded computation, the outcomes of previous coin tosses cannot be recalled unless they are stored in working memory: this is typically called \emph{one-way} access to randomness. 
The stronger model where all previous coin tosses can be recalled (\ie, \emph{two-way} access to randomness) has also been studied in the computational complexity literature. For example, Nisan \cite{DBLP:journals/tcs/Nisan93} showed that bounded two-sided error log-space machines with one-way access to randomness can be simulated by \emph{zero-error} randomized log-space machines with two-way access to randomness ($\mathsf{BPL}\subseteq \mathsf{2wayZPL}$). 

In the streaming literature, it is common to first design an streaming algorithm assuming access to a random oracle, then to use pseudorandom generators to remove this assumption, sometimes incurring a blowup in space complexity. Nisan's pseudorandom generator \cite{DBLP:journals/combinatorica/Nisan92} offers a generic way to derandomize many streaming algorithms (\eg, \cite{DBLP:journals/jacm/Indyk06}). In our case, it is entirely unclear whether any off-the-shelf pseudorandom generators (such as \cite{DBLP:journals/combinatorica/Nisan92} or \cite{DBLP:conf/focs/ForbesK18}) can be directly applied to replace the random oracle, since the queries made to the random oracle by the cycle detection algorithm are highly adaptive, dependent on the outcomes of previous queries. 

\paragraph*{(Pseudo-)random graphs}
The {\sc Element Distinctness} algorithm of Beame et al. \cite{beame2013element} (and related work) uses versions of the following basic fact about random mappings (a.k.a.\ random 1-out digraphs): starting from any vertex, the expected number of reachable vertices is $\Theta(\sqrt{n})$.  The statistical properties (such as cycle lengths and component sizes) of random mappings have been extensively studied, see \eg, \cite[Chapter 16]{frieze2016introduction} and the references therein. However,  most of these studies crucially assume the random graphs are generated with full independence, and generally do not imply useful results about pseudorandomly generated graphs. One exception is the work of Alon and Nussboim \cite{DBLP:conf/focs/AlonN08} on $k$-wise independent Erd\H{o}s-R\'enyi graphs, but it is very different from our setting of 1-out digraphs.

\paragraph*{Subset Sum and Related Problems.}
The best known time complexity for Subset Sum is $O^*(2^{n/2})$ based on a meet-in-middle approach, first given by Horowitz and Sahni \cite{horowitz1974computing} in 1974. The space complexity of this algorithm was later improved from $O^*(2^{n/2})$ to $O^*(2^{n/4})$ by Schroeppel and Shamir \cite{schroeppel1981t}.
Very recently, Nederlof and W\k{e}grzycki gave an $O^*(2^{n/2})$-time $O^*(2^{0.249999n})$-space algorithm \cite{nederlof2020improving}.
This algorithm (as well as the $O^*(2^{0.86n})$-time $\poly(n)$-space algorithm \cite{twonikhils}) used the techniques developed in  \cite{DBLP:conf/stacs/AustrinKKN15,AKKN16}, which were inspired by advances on \emph{average-case} Subset Sum algorithms \cite{hgj10}.

The low-space {\sc List Disjointness} algorithm of \cite{twonikhils} also has implications for average-case {\sc $k$-Sum} algorithms in low space \cite{twonikhils,DBLP:conf/esa/GoldsteinLP18}. See also \cite{DBLP:conf/esa/Wang14,DBLP:conf/icalp/LincolnWWW16}.

There is also a long line of research on low-space \emph{pseudopolynomial-time} algorithms (\ie, with running time $\poly(n,t)$) for Subset Sum \cite{LokshtanovN10,ElberfeldJT10,Kane10,karl,JinVW21}, culminating in an $\tilde O(nt)$-time $O(\log n\log \log n +\log t)$-space algorithm \cite{JinVW21}.

\subsection{Open Questions} \label{sec:conclusion}
We conclude by discussing several interesting questions left open by our work.
\paragraph*{Time-space Tradeoffs?}
Beame et al. \cite{beame2013element} (and Bansal et al. \cite{twonikhils}) not only gave efficient log-space algorithms for {\sc Element Distinctness} (and {\sc List Disjointness}), but also provided a smooth time-space trade-off interpolating between the log-space algorithms and the linear-space algorithms.
These algorithms, when given $S$ memory, perform the cycle-finding procedure from $S$ starting vertices, and use a redirection idea (which requires $S$ space to store the redirected edges) to nicely handle the collisions among all these $S$ walks. Our analysis of the pseudorandom family only considers the case with a single starting vertex, corresponding to the $\polylog(n)$-space algorithm.
It would be interesting to see whether the analysis can be generalized to the case of multiple starting vertices, and hence remove the random oracle assumption for these time-space trade-off algorithms as well.

\paragraph*{Shorter Seed Length?}
Our algorithm needs $O(\log^3 n\log \log n)$ bits of space to store the ``seed'': the description of the pseudorandom mapping. An interesting question is whether we can reduce this seed length to $O(\log n)$. It seems plausible that our $k$-wise generators could be replaced by almost $k$-wise generators (\eg, \cite{AlonGHP90}) which have shorter seed length. However, to get $O(\log n)$ seed length, one might need to significantly modify our $O(\log n)$-level iterative restriction approach, which already incurs an $O(\log n)$ multiplicative factor.

\paragraph*{Faster List Disjointness Algorithm?}
We reiterate the question raised by Bansal et al. \cite{twonikhils}: can {\sc List Disjointness} be decided in $n^{2-\Omega(1)}$ time and $n^{o(1)}$ space, even allowing random oracles? In hard instances for the current algorithms, there is only one ``real'' collision between the two arrays, but many ``pseudo-collisions'' coming from the same array, and it is not clear how to filter these pseudo-collisions without affecting the real collision. As to the question of whether {\sc List Disjointness} does not have such an algorithm, the current lower bound techniques do not seem to distinguish between {\sc Element Distinctness} and {\sc List Disjointness}, and 
it is entirely unclear how to prove an $n^{1.5+\Omega(1)}$-time lower bound for $n^{o(1)}$-space algorithms solving such decision problems (for example, the best known time lower bound for {\sc Element Distinctness} in the small-space setting is barely superlinear~\cite{Ajtai05}).

\subsection{Organization}
In~\Cref{sec:tech-overview}, we provide an overview of the intuitions behind the proof of Theorem~\ref{thm:main-ed}. In \Cref{sec:prelim} we give useful definitions and notations. In \Cref{sec:property} we give the construction of our pseudorandom family, formally state the properties satisfied by the pseudorandom family, and show how to use them to obtain algorithms for {\sc Element Distinctness} and {\sc List Disjointness}. Then in \Cref{sec:construction} we define the extended random walk and dependency tree. Finally, in \Cref{sec:one,sec:two} we prove that our pseudorandom family satisfies the desired properties.

\section{Overview of Techniques}\label{sec:tech-overview}

Now we give an informal overview of the techniques behind the proof of~\Cref{thm:main-ed}.

\highlight{Notation.} Let $a=(a_1,\dots,a_n) \in [m]^n$ be the input array to {\sc Element Distinctness}.
Throughout this overview, we will assume our instances are NO instances (note the YES case is simply the absence of a collision pair), and for simplicity we assume our NO instances have at most one collision pair $a_u=a_v$ where $u\neq v$. (It turns out that the hardest NO instances are those with exactly one collision pair.) We will always use $(u,v)$ to denote the unique collision pair in the NO instance that our algorithm needs to find.

Let $\caH_{\sf full}$ be the collection of all functions from $[m]$ to $[n]$, and $h \getsR \caH_{\sf full}$ be a truly random  function (implemented using a random oracle in the BCM algorithm). We define a 1-out digraph (\ie, each node has at most one outgoing edge) $G_{a,h}$ on the vertex set $[n]$ with the edge set $\{(x, h(a_x)) \mid x \in [n]\} \subseteq [n] \times [n]$. 
For a collision pair $(u,v)$, note that vertices $u, v \in [n]$ point to the same vertex $h(a_u) = h(a_v)$ since $a_u = a_v$. We also use $f^*_{a,h}(s)$ to denote the set of all vertices reachable from $s$ in the digraph $G_{a,h}$.

In the following, we often use bold letters (\eg, $\mathbf{x}$ and $\mathbf{y}$) to denote random variables.

\subsection{Review of the BCM Algorithm}

It is instructive to first review the $O(\log n)$-space $\tilde O(n^{1.5})$-time BCM algorithm for {\sc Element Distinctness}, and understand why it requires a random oracle. 

\highlight{The BCM algorithm.} The BCM algorithm first chooses a random vertex $s\in [n]$ and performs Floyd's cycle-finding algorithm on digraph $G_{a,h}$ starting from $s$. This will successfully detect $u,v$ if both $u$ and $v$ are reachable from $s$, since $u$ and $v$ point to the same vertex. To bound the running time, the following two properties are established, using a birthday-paradox-style argument.
\begin{align}
    &\E_{\bh \getsR \caH_{\sf full},\bs \getsR [n]}[|f^*_{a,\bh}(\bs)|] \le O(\sqrt{n}), \label{exp-upb}\\
    &\Pr_{\bh \getsR \caH_{\sf full}, \bs \getsR [n]} [u,v \in f^*_{a,\bh}(\bs)] \ge \Omega(1/n).\label{hit-lowb}
\end{align}
Condition~\eqref{hit-lowb} says the probability that both $u$ and $v$ are reachable from $s$ is at least $\Omega(1/n)$. Thus, running $\tilde O(n)$ independent trials of cycle detection (each using a different $h$) will lead to at least one trial with $u$ and $v$ reachable, with high probability. Condition~\eqref{exp-upb} says  we expect $O(\sqrt{n})$ vertices to be reachable from $s$. Together, these imply the running time can be bounded by $\widetilde{O}(n^{1.5})$. See \Cref{sec:property} for a more formal description. 

\highlight{Why the BCM algorithm needs a high degree of independence.} Let us see why the birthday argument mentioned above apparently needs the values of $h$ to be fully independent (or close to that). For simplicity, we consider how one proves that the probability of reaching $v$ from a random starting vertex $\bs$ is $\Theta(1/\sqrt{n})$ (the probability of reaching both $u$ and $v$ can be analyzed similarly). Let $\bs_0=\bs,\bs_1,\bs_2,\dots$ be the vertices on the walk starting from $\bs$. Conditioning on $\bs_0=s_0,\bs_1=s_1,\dots,\bs_k=s_k$, where $a_{s_0},a_{s_1},\dots,a_{s_k}$ are distinct, the distribution of the next vertex $\bs_{k+1}$ is uniform over $[n]$, due to the full independence of $h$. Once the elements are not distinct (a collision has occurred), the walk will follow the formed cycle (which is completely determined by the walk history) and no new vertices will be reached. 
From there, a standard birthday argument can be applied, yielding the desired $\Theta(1/\sqrt{n})$ probability bound of reaching $v$, and $\Theta(1/n)$ of reaching both $u$ and $v$.

Note in the argument above we have to condition on \emph{all} previous $(k+1)$ random choices, because determining the value of $\bs_{k+1}$ involves the $(k+1)$ compositions of the $h$ function. Since $k$ is typically as large as $\sqrt{n}$, it appears that one needs at least $\Omega(\sqrt{n})$-wise independence of the values of $h$.

\subsection{Overcoming the $\Omega(\sqrt{n})$-wise Independence Barrier}

\newcommand{\Htoy}{\caH_{\sf toy}}

We first show how to overcome the need of $\Omega(\sqrt{n})$-wise independence with a toy pseudorandom hash function family $\Htoy$ based on a simple two-level iterative restriction. In particular, $\Htoy$ is constructed from three $\tilde{\Theta}(n^{1/4})$-wise independent hash functions, so that $\bh \getsR \Htoy$ can be sampled using $\tilde{\Theta}(n^{1/4})$ random bits.

\begin{construction}{Drawing a sample $\bh$ from the toy pseudorandom hash function family $\Htoy$}
\begin{itemize}
	\item Set a parameter $\tau = \Theta(n^{1/4} \log n)$, and independently draw two $\tau$-wise independent uniform hash functions $\br_1, \br_2\colon [m] \to [n]$. Independently draw a $\tau$-wise independent hash function $\bg_1\colon [m] \to \{0,1\}$ such that for every $x \in [m]$, $\bg_1(x) = \begin{cases} 0 & \text{with probability $n^{-1/4}$,} \\ 1 & \text{otherwise.}\end{cases}$ 
	
	\item Finally, $\bh\colon [m] \to [n]$ is defined by $
	\bh(x) = \begin{cases} \br_1(x) & \text{when } \bg_1(x) = 1, \\ \br_2(x) & \text{otherwise.} \end{cases}
	$	
\end{itemize}
\end{construction}

Now we instantiate the BCM algorithm with the hash function $\bh \getsR \Htoy$. We can also view $f^*_{a, \bh}(\bs)$ as the following random walk on the vertex set $[n]$. 

\begin{construction}{The random walk corresponding to $f^*_{a, \bh}(\bs)$ for $\bh \getsR \Htoy$}
$f^*_{a, \bh}(\bs)$ contains the vertices on the following random walk:
\begin{compactitem}
	\item $\bw_1 = \bs$.
	\item For each integer $j \ge 2$, set $\bw_{j} = \bh(a_{\bw_{j-1}})$ if there is no $ k \in \{2,\dotsc,j-1\}$ satisfying $a_{\bw_{k-1}} = a_{\bw_{j-1}}$; otherwise the walk is terminated. \\
\end{compactitem}

Since $\bh \getsR \Htoy$, the following is an equivalent view of the walk above, in terms of $\bg_1, \br_1$ and $\br_2$:
\begin{compactitem}
	\item Initially, $\bs_0 = \bs$ and $\bw$ is empty. 
	\item For each integer $i \geq 1$:
	\begin{compactenum}
		\item We start the $i$-th \emph{subwalk} from $\bs_{i - 1}$ following the edges defined by $x \mapsto \br_1(a_x)$. 
		
		\item Each time we visit a new vertex $x$ (including $\bs_{i - 1}$), suppose there are already $j - 1$ vertices in $\bw$. We set $\bw_j = x$ if there is no $ k \in \{2,\dotsc,j-1\}$ satisfying $a_{\bw_{k-1}} = a_{\bw_{j-1}}$; otherwise we terminate the whole walk.
		
		\item Then we check whether $\bg_1(a_x) = 0$. If this happens (with probability $n^{-1/4}$), we stop this subwalk, and let $\bs_{i} = \br_2(a_x)$. Namely, we follow the edge $x\mapsto \br_2(a_x)$ for one step. Then we move to Step (1) to continue with the $(i+1)$-th subwalk, starting from $\bs_i$.
	\end{compactenum}
\end{compactitem}
\end{construction}

Roughly speaking, when $\bh \getsR \Htoy$, the random walk generated above alternates between subwalks of typical length $O(n^{1/4})$ defined by $\br_1$, and single steps defined by $\br_2$. In the below, we provide some intuition about why such a random walk suffices for analyzing the BCM algorithm. For simplicity, we will make a unrealistic assumption, which we will mark by \ul{underlining} it. Later, we will explain how to remove the assumption.

\paragraph*{Intuition.} We first argue that each subwalk has length less than $\tau / 2$ with high probability. Fix an integer $i \ge 1$, and $\bs_{i - 1} = s_{i - 1}$ for some $s_{i - 1} \in [n]$. From $s_{i - 1}$, suppose the subwalk has visited $t + 1$ vertices $x_1,x_2,\dotsc,x_{t+1}$ before termination. From the definition of our subwalk, we have $\bg_1(a_{x_k}) = 1$ for every $k \in [t]$. \ul{Assuming the walk does not stop before $\bs_{i}$}, the elements $a_{x_{1}},\dotsc,a_{x_{t+1}}$ must be distinct. By the $\tau$-wise independence of $\bg_1$, such an event happens with probability at most $(1 - n^{-1/4})^{-\min(t,\tau)}$. Applying a union bound over all possible $s_{i - 1} \in [n]$, we can conclude that all subwalks have length at most $\tau / 2$, with at probability at least \[1 - n (1 - n^{-1/4})^{-\tau/2} = 1 - n^{-\Theta(1)}.\] 

From now on, we will condition on the event that all subwalks have length at most $\tau /2$.

In each subwalk, we follow the edges defined by $\br_1$ for at most $\tau / 2$ steps. By the $\tau$-wise independence of $\br_1$, each subwalk has the same distribution as a truly random walk with the same length, as long as its starting point $\bs_{i - 1}$ is independent of $\br_1$. However, we also note that different subwalks are \emph{not} independent. Therefore, our analysis has to overcome the following two challenges:
\begin{enumerate}[label=(\roman*)]
    \item Remove the dependency of $\bs_{i - 1}$ on $\br_1$. \label{challenge-1}
    \item Handle correlations between subwalks. \label{challenge-2}
\end{enumerate}

First, we show how to handle challenge \ref{challenge-1}.
If $\bs_{i-1}$ is the random starting point $\bs$, it is independent of $\br_1$. Otherwise, $\bs_{i - 1}$ is the vertex reached by a subwalk started from $\bs_{i - 2}$ (which depends on $\br_1$) together with a single step defined by $\br_2$.  We wish to remove this dependency on $\br_1$ using the single step following $\br_2$. 

The key observation is the following. A truly random walk has typical length $\Theta(\sqrt{n})$, while each subwalk has a typical length $\Theta(n^{1/4})$. So to mimic a truly random walk, our analysis only needs to handle $O(n^{1/4})$ queries to the hash function $\br_2$ (each query represents one step following $\br_2$). \ul{Assuming that the walk does not stop before $\bs_{i}$ for all $i \in [n^{1/4}]$}, these $O(n^{1/4})$ queries are distinct. Then by the $\tau$-wise independence of $\br_2$ and the fact that $n^{1/4} \ll \tau$, each $\bs_i$ can indeed be replaced by a truly uniformly random variable over $[n]$ without changing the distribution of the generated random walk. Therefore, $\bs_{i-1}$ is independent of $\br_1$ and $\bg_1$ as desired. 

To handle challenge \ref{challenge-2} (\ie, the correlation across subwalks), the key idea is that in a standard birthday paradox argument, we do not require complete independence of all items; in fact, pairwise independence already suffices. Since each subwalk has length at most $\tau / 2$ and $\br_1$ is $\tau$-wise independent, such subwalks are also pairwise independent, which enables us to perform a birthday-paradox-style analysis. Of course, this is an oversimplification, and our actual analysis framework will be clarified in~\Cref{sec:tech-alternative-analysis}. 

Here we made the (unrealistic) assumption that the walk does not stop before reaching each $\bs_i$. (In reality, the walk \emph{has to} stop during some subwalk.) Note that whether the walk stops at the $j$-th step is equivalent to whether $j$ is no greater than the length of the walk $|\bw|$. Since $|\bw|$ is a random variable depending on all of $\br_1, \br_2, \bg_1$, we have to carefully ensure that our analysis does not involve $|\bw|$, to keep $\bs_{i-1}$ and $\br_1$ independent. We will explain how we overcome such difficulty in \Cref{sec:tech-alternative-analysis}, and in \Cref{dependency-tree} we will extend the two-level structure above into a $O(\log n)$-level tree (using signficantly less randomness in our hash functions).

\subsection{An Alternative Analysis of the BCM Algorithm}\label{sec:tech-alternative-analysis}

The starting point of our work is a coupling-based proof of Condition~\eqref{hit-lowb}, based on what we call \emph{extended} random walks.\footnote{Condition \eqref{exp-upb} is easier to establish. We will focus on Condition~\eqref{hit-lowb} since it is more difficult.} 
This proof will introduce the key strategy of our later analysis, when we replace the random oracle by a pseudorandom hash function.

\highlight{The random walk corresponding to $f^*_{a,\bh}(\bs)$.} Note for $\bh \getsR \caH_{\sf full}$ and $\bs \getsR [n]$,  $f^*_{a,\bh}(\bs)$ can be seen as a random walk on the vertex set $[n]$ in a straightforward way. 
\begin{construction}{The random walk $\bw$ corresponding to $f^*_{a,\bh}(\bs)$}
	\begin{compactitem}
		\item $\bw_1 = \bs$.
		\item For each integer $j \ge 2$, set $\bw_{j} = \bh(a_{\bw_{j-1}})$ if there is no $2 \le k \le j - 1$ such that $a_{\bw_{j-1}} = a_{\bw_{k-1}}$; otherwise stop the the walk.
	\end{compactitem}
\end{construction}
Since the walk stops immediately after a collision occurs, one can see that $f^*_{a,\bh}(\bs)$ is exactly the set of all vertices in the walk $\bw=(\bw_1,\ldots,\bw_{|\bw|})$.\footnote{Note it is possible that for some $j \ge 2$, $a_{w_{j-1}}$ is distinct from all $a_{w_{k-1}}$ for $k < j$, but $\bh(a_{w_{j-1}})$ has a collision with a previous $\bh(a_{w_{k-1}})$. In this case, the walk moves to $\bw_j = \bw_k$ (which was already visited before) and stops at step $j + 1$.}

Recall $(u,v)$ is the unique collision pair.
In order to prove Condition \eqref{hit-lowb}, our goal now is to lower bound the probability 
\begin{align}
  &\Pr[(\exists~ (i,j) \in [|\bw|]^2)[\bw_i = u \wedge \bw_j = v]] \label{eq:prob-exist} \\
= &\sum_{(i,j) \in \N^2} \Pr[ (i \le |\bw| \wedge \bw_i = u) \wedge (j \le |\bw| \wedge \bw_j = v)]. \label{eq:exp-first}
\end{align}
The equality of \eqref{eq:prob-exist} and \eqref{eq:exp-first} holds
since, by definition, if there is an $(i,j)$ such that $(\bw_i = u) \wedge (\bw_j = v)$, then the walk would immediately stop at step $\max(i,j) + 1$ (\ie, $|\bw| = \max(i,j)$). So $\bw$ contains at most one pair $(i,j)$ such that $(\bw_i,\bw_j) = (u,v)$, and hence we can decompose~\eqref{eq:prob-exist} into~\eqref{eq:exp-first}.

Our initial hope is that~\eqref{eq:exp-first} may be simpler to analyze, as it is a sum of many simpler terms, each of which only depends on two entries $\bw_i$ and $\bw_j$. However, the condition $(i \le |\bw| \wedge \bw_i = u)$ is still difficult to analyze, as it depends on the length $|\bw|$.

\newcommand{\bbw}{\bar{\bw}}

\highlight{Coupling with the basic extended walk.}  To move forward, we wish to find a way to lower bound~\eqref{eq:prob-exist} by a sum of many simpler probabilities that do not involve $|\bw|$. The first idea is to extend the random walk $\bw$ to an \emph{infinite} extended random walk $\bbw$. We stress that the walk $\bbw$ defined below is only used in the analysis, and not in the algorithm.

\begin{construction}{Basic extended walk $\bbw$}
\begin{compactitem}
	\item Extend the domain of $\bh$ from $[m]$ to $[m] \cup \{\star_0, \star_1, \dots\}$ as follows: for each $t \in \N$, sample $\bh(\star_t) \getsR [n]$, where all samples are independent.
	\item Perform the random walk $\bw$. After $\bw$ ends, set $\bbw = \bw$ and for every $t \in \N$ append $\bh(\star_t)$ to the end of $\bbw$. 
\end{compactitem}
\end{construction}
Note that $\bbw$ and $\bw$ are both defined over the joint probability space $(\bh,\bs)$ (for the extended $\bh$), and $\bw$ is always a prefix of $\bbw$. From the definition of $\bbw$, we have the following nice properties:
\twoprops{All entries of $\bbw$ are i.i.d. samples from $[n]$.}{bbw-all-ind}{For all $i$, if $a_{\bbw_{j}} \ne a_{\bbw_{k}}$ for all $1 \le j < k < i$, then $\bw_i = \bbw_i$.}{bbw-same-cond}

\highlight{Proof strategy: subtracting the overcount.}
By~\eqref{bbw-same-cond}, we know that $u,v \in f^*_{a,\bh}(\bs)$ if there are $i,j \in \N$ such that (1) $\bbw_i = u$ and $\bbw_j = v$, and (2) for all $1 \le t < q < \max(i,j)$, $a_{\bbw_{t}} \ne a_{\bbw_{q}}$. In this way, we have reformulated the success condition $u,v \in f^*_{a,\bh}(\bs)$ as a statement that does {\bf not} involve the length $|\bw|$ of the original random walk $\bw$, and can be analyzed more easily.
Fixing a length parameter $L =  c \sqrt{n}$ for some small constant $c > 0$ to be determined later, we have
\begin{align}
    \Pr[u,v \in f^*_{a,\bh}(\bs)] &\ge \Pr[\text{$\exists i,j \in [L]$ s.t. $(\bbw_i,\bbw_j) = (u,v)$ and for all $1 \le t < q \le L$, $a_{\bbw_{t}} \ne a_{\bbw_{q}}$}] \notag\\
&= \sum_{i,j \in [L]} \Pr[\text{$(\bbw_i,\bbw_j) = (u,v)$ and for all $1 \le t < q \le L$, $a_{\bbw_{t}} \ne a_{\bbw_{q}}$}]. \label{eq:exp-second}
\end{align}

The last equality above holds because if for all $1 \le t < q \le L$, we have $a_{\bbw_{t}} \ne a_{\bbw_{q}}$, then there can only be one pair $(i,j) \in [L]^2$ satisfying $(\bbw_i,\bbw_j) = (u,v)$.

To further lower bound~\eqref{eq:exp-second}, we define the following two quantities:
\[
E_{\sf total} = \sum_{(i,j) \in [L]^2}\Pr[\bbw_i = u \wedge \bbw_j = v]\quad\text{and}\quad E_{\sf bad} = \sum_{\substack{(i,j) \in [L]^2\\1 \le t < q \le L} } \Pr[\bbw_i = u \wedge \bbw_j = v \wedge a_{\bbw_{t}} = a_{\bbw_{q}}].
\]
We claim that $\Pr[u,v \in f^*_{a,\bh}(\bs)] \ge E_{\sf total} - E_{\sf bad}$: note that $E_{\sf total}$ counts the total expected number of pairs $(i,j)$ with $(\bbw_i,\bbw_j) = (u,v)$, and $E_{\sf total} - E_{\sf bad}$ subtracts all the ``bad pairs'' from the total count.\footnote{We call such an $(i,j)$ a ``bad pair'' because it should not be counted in~\eqref{eq:exp-second}, and has to be subtracted from the total count. Also, we remark that is possible that a bad pair is subtracted more than once in $E_{\sf bad}$. This is not an issue for us, as we are trying to lower bound $\Pr[u,v \in f^*_{a,\bh}(\bs)]$.} 

The rest of the analysis is a straightforward calculation using the property~\eqref{bbw-all-ind}. We can see that $E_{\sf total}=\Theta(L^2/n^2) = \Theta(c^2/n)$, and $E_{\sf bad} = \Theta ( L^4/n^3) = \Theta(c^4/n)$.
Setting $c$ to be small enough, we have $E_{\sf total} - E_{\sf bad} \ge \Omega(1/n)$, which concludes the proof.

\highlight{Remark.} Setting
\begin{equation}
E_{\sf total} = \sum_{i \in [L]}\Pr[\bbw_i = u]\quad\text{and}\quad E_{\sf bad} = \sum_{\substack{j \in [L]\\1 \le t < q \le L} } \Pr[\bbw_i = u \wedge a_{\bbw_{t-1}} = a_{\bbw_{q-1}}], \label{eq:Etotal-Ebad-simple}
\end{equation}
one can also show $\Pr[v \in f^*_{a,\bh}(\bs)] \ge \Omega(1/\sqrt{n})$ for all possible vertices $v$, by showing $E_{\sf total} - E_{\sf bad} \ge \Omega(1/\sqrt{n})$ (for $L = c \sqrt{n}$ and appropriately small $c > 0$). Later in this overview, we will explain how to get an $\Omega(1/\sqrt{n})$ lower bound for this single-vertex case when we replace the random oracle $h$ by a pseudorandom function, and discuss additional challenges that arise for the two-vertex case (with $u,v$).

\subsection{Pseudorandom Hash Functions, the Dependency Tree, and the Indexing Scheme} \label{dependency-tree}

Next we describe our construction of pseudorandom hash functions $\bh$ based on \emph{iterative restrictions}. In particular, we use a small number of independent partial functions defined by random restrictions to form a full hash function. 
By considering how the hash values of vertices on the random walk 
are determined by the iterative restriction, we can naturally organize these vertices into a  hierarchical structure we call the \emph{dependency tree}, which will play an crucial role in our later analysis.

\highlight{Pseudorandom hashing by iterative restrictions.} Instead of using full randomness, we will implement the hash function $h\colon [m]\to [n]$ by the following iterative pseudorandom restriction process, using only $\polylog(n)$ seed length.
Initially, all values of $h(x)$ are undefined. The values are defined over $\ell \le \log n$ iterations. In the $i$-th iteration, we sample $O(\log n\log \log n)$-wise random functions $g_i\colon [m] \to \{0,1\}, r_i\colon [m] \to [n]$, and for every $x\in[m]$ such that $g_i(x)=1$ and $h(x)$ is still undefined, we define $h(x)$ to be $r_i(x)$. See \Cref{sec:construct-family} for details. Informally, in each iteration we independently use $O(\log n\log \log n)$-wise generators to fix about half of the remaining undefined values in $h$: the $g_i$ selects which half, and the $r_i$ selects the values. (It is possible that a tiny number of hash values $h(x)$ may still be undefined after $\log(n)$ iterations, but this is not a significant issue for us and we ignore it in this overview.)

Let $\caH$ denote the above family of pseudorandom functions. In the following, $\bh$ will denote the random variable for a function randomly drawn from $\caH$. Analogously to Section~\ref{sec:tech-alternative-analysis}, one can define a random walk $\bw$ on the random graph $G_{a,\bh}$.

\highlight{Tree structure of pseudorandom walks.} We now describe a \emph{dependency tree} $\bT$ for a walk $\bw$ on $G_{a,\bh}$. We use non-negative integers to denote the nodes of $\bT$: node $0$ is a ``dummy'' node representing the root, and for $\mu \ge 1$, node $\mu$ corresponds to the $\mu$-th node of walk $\bw$ if it exists (\ie, node $\mu$ is associated with vertex $\bw_\mu$). 
We will use Greek letters $\alpha,\beta,\mu,\dotsc$ to refer to nodes in the dependency tree $\bT$. 

The tree $\bT$ has one ``level'' for each iteration $1,\ldots,\ell$ of the process defining $\bh$. For each node $\mu$ of $\bT$, we define $\level(\mu)$ (the ``level of $\mu$'') to be the smallest integer $j$ such that $g_j(a_{\bw_\mu}) = 1$ (note that this $j$ corresponds to the iteration in which the hash value of $a_{\bw_\mu}$ is defined). If no such $j$ exists, then we set $\level(\mu) = \ell + 1$. We also set $\level(0) = \ell + 1$, and define $\nex(\mu) = \bh(a_{\bw_{\mu}}) = r_{\level(\mu)}(a_{\bw_\mu})$. Informally, $\nex(\mu)$ corresponds to the ``next'' vertex on the walk after $\bw_\mu$.

\begin{construction}{Dependency tree $\bT$ based on $\bw$}
	\begin{compactitem}
		\item Node $0$ is the root of $\bT$.
		\item For each node $\mu$ of $\bT$, its parent $\pa(\mu)$ is defined as the largest node $\nu < \mu$ with level at least $\level(\mu)$. 
	\end{compactitem}
\end{construction}

Observe that the walk $\bw$ is simply the pre-order traversal of $\bT$. Also, observe that every root-to-node path of $\bT$ has non-increasing node levels.

\highlight{Indexing a tree node.} Recall $\ell \le \log n$ is the number of iterations, which bounds the number of levels of $\bT$. Each node $x$ of $\bT$ can be assigned a unique ``index'' in a natural way, via a sequence $\vk=(k_1,k_2,\dots,k_{\ell})$ of non-negative integers, where $k_i$ specifies the number of level-$i$ nodes on the path from the root to the node $x$. See Figure~\ref{fig:dependency-intro} for an illustration of a tree and the index scheme. We will explain why such indexing scheme helps our analysis at the end of the next subsection.

\begin{figure}
	\centering
	\scalebox{0.9}{
\begin{tikzpicture}[node distance={14mm}, thick, vertex/.style = {draw, thick, circle,inner sep=0pt,minimum size=15pt,outer sep=3pt}] 
  \node (l1) {level $1$};
  \node [above of = l1, yshift = -0.4cm] (l2) {level $2$};
  \node [above of = l2, yshift = -0.4cm] (l3) {level $3$};
  \node [above of = l3, yshift = -0.4cm] (l4) {level $4$};
  \node [above of = l4, yshift = -0.4cm] (l5) {level $5$};
  \node [above of = l4, yshift = -0.8cm] (ll) {($\ell = 4$)};
  \node [vertex, right of = l5] (n0) {$0$};
  \node [vertex, right of = l1, xshift = 1cm] (n1) {$1$} edge [<-] (n0);
  \node [vertex, right of = l4, xshift = 2cm] (n2) {$2$} edge [<-] (n0);
  \node [vertex, right of = l2, xshift = 3cm] (n3) {$3$} edge [<-] (n2);
  \node [vertex, right of = l3, xshift = 4cm] (n4) {$4$} edge [<-] (n2);
  \node [vertex, right of = l1, xshift = 5cm] (n5) {$5$} edge [<-] (n4);
  \node [vertex, right of = l2, xshift = 6cm] (n6) {$6$} edge [<-] (n4);
  \node [vertex, right of = l3, xshift = 7cm] (n7) {$7$} edge [<-] (n4);
  \node [vertex, right of = l2, xshift = 8cm] (n8) {$8$} edge [<-] (n7);
  \node [vertex, right of = l4, xshift = 9cm] (n9) {$9$} edge [<-] (n2);
  \node [vertex, right of = l3, xshift = 10cm] (n10) {$10$} edge [<-] (n9);
  \node [vertex, right of = l3, xshift = 11cm] (n11) {$11$} edge [<-] (n10);
  \node [vertex, right of = l5, xshift = 12cm] (n12) {$12$} edge [<-] (n0);

\end{tikzpicture}
}
	\caption{An example of a dependency tree $\bT$. For example, the index of $7$ is $(0, 0, 2, 1)$, since the path $0\gets 2 \gets 4\gets 7$ has two level-$3$ nodes (node $4$ and node $7$), and one level-$4$ node (node $2$). }
	\label{fig:dependency-intro}
\end{figure}
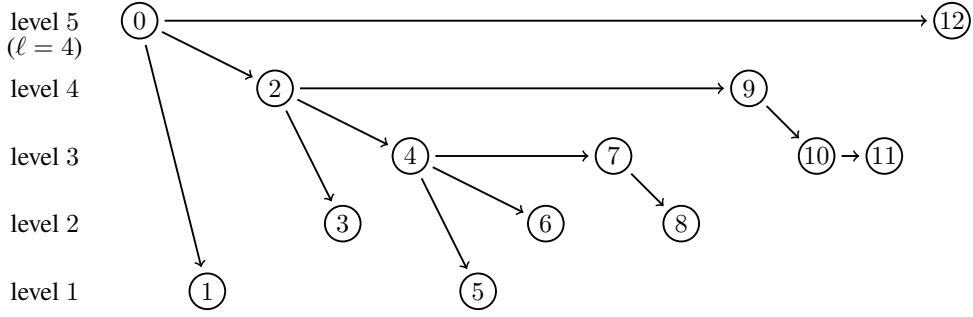

\subsection{A Coupling-based Approach Based on the Dependency Tree}
\label{sec:2.5}

We wish to mimic the strategy of the coupling-based proof in Section~\ref{sec:tech-alternative-analysis}. Instead of proving an $\Omega(1/n)$ lower bound for $\Pr[u,v \in f^*_{a,\bh}(\bs)]$, we will first consider how to prove an $\Omega(1/\sqrt{n})$ lower bound for $\Pr[u \in f^*_{a,\bh}(\bs)]$, which already contains all the important ideas. Then, we will briefly discuss additional technical challenges that arise for the analysis of the two-vertex case (computing $\Pr[u,v \in f^*_{a,\bh}(\bs)]$).

\newcommand{\bbT}{\bar{\bT}}

As in \Cref{sec:tech-alternative-analysis}, our strategy is again to carefully design an extended random walk $\bbw$ which is coupled with $\bw$, so that $\bw$ is always a prefix of $\bbw$. We will also build a corresponding extended dependency tree (``extended tree'' for short) $\bbT$ on $\bbw$. Note that $\bT$ would be a subtree of $\bbT$ as $\bw$ is a prefix of $\bbw$. We will similarly define $\nex$ and $\level$ values for nodes on extended tree $\bbT$, and these values would be consistent with $\bT$ on the corresponding subtree. We will sometimes use $\nex_{\bT}$ or $\nex_{\bbT}$ when there is a chance of confusion on which tree $\nex$ is referring to.

We hope to define an extended walk $\bbw$ that  maintains Condition~\eqref{bbw-same-cond} as before. For notational convenience, we slightly change Condition~\eqref{bbw-same-cond} to
\begin{equation}\label{eq:bbw-same-cond}
\text{For all $i$, if $a_{\nex_{\bbT}(\alpha)} \ne a_{\nex_{\bbT}(\beta)}$ for all $0 \le \alpha < \beta < i-1$, then $\bw_i = \bbw_i$.} 
\end{equation}
Note that since $\nex_{\bbT}(\alpha) = h(a_{\bbw_\alpha}) = \bbw_{\alpha + 1}$, the above is equivalent to~\eqref{bbw-same-cond}.

For an index $\vk \in \N^\ell$, we also let $\vidx^{\vk}$ denote the node indexed by $\vk$ in the dependency tree $\bbT$. 
Note that such a node may not exist in the tree; we use $\calF^{\vk}$ to denote the event that $\vidx^{\vk}$ exists in $\bbT$.
To lower bound $\Pr[u \in f^*_{a,\bh}(\bs)]$, we define the following two quantities analogous to~\eqref{eq:Etotal-Ebad-simple}:
\begin{equation}
E_{\sf total} = \sum_{\vk \in \N^\ell} \Pr\left[\calF^{\vk} \wedge \nex(\vidx^{\vk}) = u \right],
\label{eqn:etotalnew}
\end{equation}
and
\begin{equation}
E_{\sf bad} = \sum_{\substack{\vk \in \N^\ell\\ \vk^1 < \vk^2 \in \N^\ell}} \Pr[\calF^{\vk} \wedge \nex(\vidx^{\vk}) = u \wedge \calF^{\vk^1} \wedge \calF^{\vk^2} \wedge a_{\nex(\vidx^{\vk^1})} = a_{\nex(\vidx^{\vk^2})}]. \label{eqn:ebadnew}
\end{equation}
Note that our choice of $E_{\sf bad}$ in \eqref{eqn:ebadnew} is a bit different from that in 
Section~\ref{sec:tech-alternative-analysis},
as we consider a ``bad occurrence'' to happen whenever there is a collision in $\bbw$ (while in~\eqref{eq:Etotal-Ebad-simple} we restricted $t,q$ to the interval $[\ell]$). This will not be a problem if we choose $\ell$ carefully.

By an argument similar to that of Section~\ref{sec:tech-alternative-analysis}, we have that $\Pr[u \in f^*_{a,\bh}(\bs)] \ge E_{\sf total} - E_{\sf bad}$. Hence, the goal is to design $\bbw$ and $\bbT$ such that~\eqref{eq:bbw-same-cond} holds and the summands in $E_{\sf total}$ and $E_{\sf bad}$ can be bounded.

\highlight{Quick estimate: a sanity check.} To better understand the summands in $E_{\sf total}$ and $E_{\sf bad}$, let us first calculate these summands under the unrealistic assumption that all involved events are independent. 
Note that $\calF^{\vk}$ asserts the existence of node $\mu^{\vk}$ in the tree $\bbT$, which requires that there is a tree path starting from the root, and extending down the levels in a way that is consistent with the vector $\vk$, which specifies the number of level-$i$ nodes on this path for every $i\in [\ell]$.  
Observe that,  for every node $\beta$ of level $i$ on this path, we must have $g_i(a_{\bw_\beta})=1$, since otherwise $\beta$ would not have been on level $i$, and the path would not extend to reach $\beta$. 
Hence, the event $(\calF^{\vk} \wedge \nex(\vidx^{\vk}) = u)$ is equivalent to the conjunction of the two conditions:
\begin{compactitem}
\item[(1)] Let $\alpha = \vidx^{\vk}$. For all the $k_i$ level-$i$ nodes $\beta$ on the path from root to node $\alpha$, we have $g_i(a_{\bw_\beta}) = 1$, and 
\item[(2)] $r_{\level(\alpha)}(a_{\bw_{\alpha}}) = u$, 
\end{compactitem}
where Item (2) directly follows from our definition of $\nex(\cdot)$. Observe that the event in Item (2) 
happens with $1/n$ probability, and for each $\beta$ the event in 
Item (1) happens with $1/2$ probability.
Pretending that all these events are independent, we would have 
\begin{equation}\label{eq:total-count}
\Pr\left[\calF^{\vk} \wedge \nex(\vidx^{\vk}) = u \right] = \left(2^{|\vk|_1} \cdot n\right)^{-1},
\end{equation}
where $|\vk|_1$ is the $\ell_1$-norm of $\vk$. Similarly, pretending all events are independent, we would have 
\begin{equation}\label{eq:bad-count}
\Pr\left[\calF^{\vk} \wedge \nex(\vidx^{\vk}) = u \wedge \calF^{\vk^1} \wedge \calF^{\vk^2} \wedge a_{\nex(\vidx^{\vk^1})} = a_{\nex(\vidx^{\vk^2})}\right] = \left(2^{|\vk|_1 + |\vk^1|_1 + |\vk^2|_1} \cdot n^2\right)^{-1}.
\end{equation}
Observe that $\sum_{\vk \in \N^\ell} 2^{-|\vk|_1} = \sum_{\vk \in \N^\ell} 2^{-k_1} \cdot 2^{-k_2} \cdot \cdots \cdot 2^{-k_\ell} = ( \sum_{i \in \N} 2^{-i} )^\ell = 2^\ell$. Then, plugging~\eqref{eq:total-count} and~\eqref{eq:bad-count} into~\eqref{eqn:etotalnew} and~\eqref{eqn:ebadnew}, we would have $E_{\sf total}  =\Omega(2^{\ell} / n)$, and $E_{\sf bad}=O(2^{3\ell} /n^2)$. Setting $\ell = \frac{1}{2}\cdot \log(n) - c$ for a large enough constant $c$, we would have \begin{equation}
    E_{\sf total} - E_{\sf bad} = \Omega\left (\frac{1}{2^c \sqrt{n}} \right ) - O \left (\frac{1}{2^{3c} \sqrt{n}} \right )\ge \Omega(1/\sqrt{n}).
    \label{eq:good}
\end{equation}

Now we can explain why we chose such an indexing scheme: the existence of $\mu^{\vk}$ and the value of $\nex(\mu^{\vk})$ only depends on the ancestors of $\mu^{\vk}$ in the dependency tree. Since typically there are at most $\polylog(n)$ many ancestors, we can use the $\tau$-wise independence of $\bg_i$ and $\br_i$ to analyze the event $\calF^{\vk} \wedge \nex(\vidx^{\vk}) = u$.

\subsection{Designing the Extended Random Walk}
\label{sec:2.6}
Finally we explain how to design the extended random walk $\bbw$, by constructing an extended tree $\bbT$. We first aim to ensure Condition~\eqref{eq:total-count} holds, leading to a desired lower bound on $E_{\sf total}$. Handling $E_{\sf bad}$ is more challenging; we will discuss that later.

Specifically, we will ensure that~\eqref{eq:total-count} holds for all ``short'' vectors $\vk \in [\tau/4]^\ell$ and $u \in [n]$, where $\tau = O(\log n \log\log n)$ is the independence parameter of our pseudorandom hash function.\footnote{This is already enough for lower bounding $E_{\sf total}$, as the contribution of ``long'' (non-short) $\vk$ is negligible. Intuitively this is true because for a ``long'' $\vk$, we have $|\vk|_1 \ge \max_{i \in[\ell]} k_i > \tau/4$, the probability that $\mu^{\vk}$ exists in the tree is quite small ($2^{-|\vk|_1}$) assuming~\eqref{eq:total-count}. See Lemma~\ref{lem:elong-small} for a formal proof.}

\highlight{Establishing~\eqref{eq:total-count} by induction.} To show~\eqref{eq:total-count}, we wish to prove the following claim.

\begin{claim}\label{claim:induction-intro}
Fix an index $\vk$ corresponding to a level-$i$ node ($\vk = (0,\dots,k_i,k_{i+1},\dots,k_{\ell})$ and $k_i > 0$). Conditioned on the event $\calF^{\vk}$, with $1/2$ probability $\vidx^{\vk}$ has a level-$i$ child $\nu$ (\ie, for $\vk' = (0,\dots,k_i + 1,k_{i+1},\dots,k_{\ell})$, $\calF^{\vk'}$ holds) and $\nex(\nu)$ is distributed uniformly in $[n]$.
\end{claim}

Assuming that~\Cref{claim:induction-intro} holds, then~\eqref{eq:total-count} follows by a simple induction.\footnote{One also needs to show that with probability $1/2$, $\mu$ has a level-$j$ child with a uniformly random $\nex$-value, for all $j < i$. We ignore this part in the technical overview.} However, it is not hard to see that~\Cref{claim:induction-intro} does not hold for the original tree $\bT$. To understand the issue, let $\vk,\vk'$ be as in~\Cref{claim:induction-intro} and assume $\vidx^{\vk}$ exists (\ie, $\calF^{\vk}$ holds). We wish to better understand the conditions under which $\vidx^{\vk'}$ exists. Letting $r_{<i}$ and $g_{<i}$ denote $(r_1,\dotsc,r_{i-1})$ and $(g_1,\dotsc,g_{i-1})$ respectively, we additionally fix $(\br_{<i},\bg_{<i}) = (r_{<i},g_{<i})$ (we use $r_{<i} \wedge g_{<i}$ to denote this event for simplicity). 

\highlight{The existence condition of $\mu^{\vk'}$ in $\bT$.} Let $\alpha$ be the smallest-numbered node such that $\alpha >\mu^{\vk}$ and the level of $\alpha$ is greater than $i - 1$. Then $\mu^{\vk'}$ exists if and only if $\alpha$ exists and $\level(\alpha) = i$. Hence, our goal is to determine $\alpha$. By definition, to move from $\mu^{\vk}$ to $\alpha$ in the random walk $\bw$, one first move to the node corresponding to vertex $\nex(\mu^{\vk})$, and then keep going to the next node, until reaching a node with level at least $i$. 
The following algorithm implements this procedure and returns the simulated random walk, and we observe that it only uses the values of $(r_{\le i},g_{\le i})$. Note that we use $(\cdots)$ to denote a sequence of vertices, and use $\circ$ to denote the concatenation of two sequences.

\newcommand{\simalgo}{\mathsf{sim}}
\newcommand{\findnode}{\mathsf{Find}}
\newcommand{\extfindnode}{\mathsf{ExtFind}}

\begin{algorithm}[H]\label{algo:find}
	\DontPrintSemicolon
	\caption{Simulating the random walk from $s'$ until reaching a level greater than $i$}
	\SetKwProg{Fn}{Function}{}{}
	\Fn{$\simalgo(s',i)$}{
		\If{$i = 0$}{
			\Return $(s')$\tcp*{stop here since all nodes have levels at least $1$}
		}
		$s_0 \leftarrow s', j \leftarrow 0, w \leftarrow ()$ \tcp*{start from $s_0=s'$}
		\Repeat{$g_i(a_{x_j}) = 0$ }{
			$w \leftarrow w \circ \simalgo(s_{j},i-1)$ \tcp*{simulate from $s_{j}$ until hitting a node with level at least $i$}
			$x_{j+1} \leftarrow w_{|w|}$ \tcp*{vertex $x_{j+1}$ corresponds to the next node after $s_{j}$ with level $\ge i$}
			\If{$g_i(a_{x_{j+1}}) = 1$}{
				$s_{j+1} \leftarrow r_i(a_{x_{j+1}})$ \tcp{move to the next node since the node corresponding to $x_{j+1}$ has level $i$}
				$j\leftarrow j + 1$\\
			}
		}
		\Return $x_j$\tcp{stop here since the node corresponding to $x_j$ has level $> i$}
	}
	\Fn{$\findnode(s',i)$}{\Return the last vertex in the sequence returned by  $\simalgo(s',i)$}
\end{algorithm}

One can see that $\simalgo(\nex(\mu),i-1)$ generates the entire \emph{sub-walk} after $\mu$ until reaching the next node with level at least $i$. 
Now, the hope is to argue that, conditioning on $\calF^{\vk} \wedge r_{<i} \wedge g_{<i}$, we have \[\bg_i(\findnode(\nex(\mu),i-1)) = 1\] with probability $1/2$. 

\highlight{Two issues with the original random walk $\bw$.} There are two important issues with the argument above:
\begin{enumerate}
	\item We need to argue $\bg_i(\findnode(\nex(\mu),i-1))$ is independent from the event $\calF^{\vk} \wedge r_{<i} \wedge g_{<i}$.
	
	\item Even if $\bg_i(\findnode(\nex(\mu),i-1)) = 1$, it could be the case that $\bw$ stops during the simulation of $\simalgo(\nex(\mu),i-1)$ due to a collision\footnote{Indeed, if the simulation $\simalgo(\nex(\mu),i-1)$ detects a pair of collision (two nodes $\alpha,\beta$ such that $a_{\bw_\alpha} = a_{\bw_\beta}$), it would loop forever.}, and in that case $\vidx^{\vk'}$ also does not exist. 
\end{enumerate}

The second issue is fundamental, as it reveals the ``global dependency nature'' of the original random walk $\bw$: the event that $\bw$ stops depends on \emph{all} entries in $\bw$. 

\highlight{A locally simulatable extended random walk.} To circumvent the second issue, we wish for our extended random walk $\bbw$ to be \emph{locally simulatable}. That is, knowing that node $\mu$ exists and knowing the value of $\nex(\mu)$, together with fixed $r_{<i}$ and $g_{<i}$, one should be able to simulate the extended random walk $\bbw$ after $\mu$ until reaching a node with level at least $i$.
The second issue above amounts to the fact that $\simalgo(\mu,i)$ fails to locally simulate the walk $\bw$, since it does not have enough information to determine whether $\bw$ has already terminated during its simulation (it cannot determine whether there is a collision between the encountered node and the nodes before in $\bw$). 

Similar to the basic extended random walk in Section~\ref{sec:tech-alternative-analysis}, for each $i \in [\ell]$, we extend the domain of $\bg_i$ and $\br_i$ from $[m]$ to $[m] \cup \{\star_0, \star_1, \dots\}$ as follows: for each $t \in \N$, we sample $\bg_i(\star_t) \getsR \{0,1\}$ and $\br_i(\star_t) \getsR [n]$, where all samples are independent. 

Since the ``local'' simulation with respect to node $0$, $\nex(0) = \bs$ and fixed $r_{\le \ell}$ and $g_{\le \ell}$ is just the entire random walk, we will define our extended random walk by giving its local simulation in Algorithm~\ref{algo:extwalk-intro}, and we set $\bbw \leftarrow \walk(\bs,\ell,0)$.\footnote{see Section~\ref{sec:extendedwalk} for a detailed explanation of Algorithm~\ref{algo:extwalk-intro}.} Note that $\walk(\bs,\ell,0)$ also gives the extended tree $\bbT$ by specifying $\level$ and $\nex$.

\begin{algorithm}[h]\label{algo:extwalk-intro}
	\DontPrintSemicolon
	\caption{Algorithm for extended walk}
	\SetKwProg{Fn}{Function}{}{}
	\Fn{$\walk(s',i,\mu_0)$ \hspace{0.5cm}(where $s'\in [n], 0\le i \le \ell$)}{
		\lIf{$i = 0$}{\Return  $(s')$  \label{line-intro:corner}}
		$C_0 \gets \emptyset$, $\mathsf{star} \gets \text{false}$\\
		$j \gets 0, s_0 \gets s', w \gets ()$ 
		
		\Repeat{$g_i(y) = 0$}{
			$w \gets w \circ \walk(s_j, i - 1, \mu_0 + |w|)$ \label{walk-intro:w}   \\
			$x_{j + 1} \gets w_{|w|}$ \label{walk-intro:xj}\\
			$y,\mathsf{star} \gets \begin{cases} a_{x_{j + 1}}, \text{false} & \text{ if } a_{x_{j + 1}} \not\in C_j \land \lnot \mathsf{star} \\ \star_t, \text{true} & \text{ otherwise (where $t := \min \{t \in \N  \ \vert \ \star_t \not\in C_j\}$)}\end{cases}$\label{walk-intro:y}\\
			$\mu_{j + 1} \gets \mu_0 + |w|$ \label{walk-intro:a}\\
			\If{$g_i(y) = 1$ \label{walk-intro:if}}{
				$C_{j + 1} \gets C_j \cup \{y\}$, $s_{j + 1} \gets r_i(y)$ \label{walk-intro:Cj}\\
				$\level(\mu_{j + 1}) \gets i, \nex(\mu_{j + 1}) \gets r_i(y)$ \label{walk-intro:level}\\
				$j \gets j + 1$
			}
		}
		\Return $w$ \label{line-intro:return}
	}
	\Fn{$\extfindnode(s',i)$}{\Return the last vertex in the sequence returned by  $\walk(s',i,0)$}
\end{algorithm}

\highlight{Establishing Claim~\ref{claim:induction-intro} for $\bbT$.} One can inspect that the algorithm $\walk$ behaves the same as $\simalgo$ until a collision occurs at Line~\ref{walk-intro:y} (that is, there is a collision in $\{a_{x_1},a_{x_2},\dotsc,a_{x_{j+1}}\}$). That is, $\simalgo(\bs,\ell)$ and $\walk(\bs,\ell,0)$ behave the same until reaching a collision $a_{\bw_{j}} = a_{\bw_{k}}$ for $j\ne k$. This implies that~\eqref{eq:bbw-same-cond} holds.

To show Claim~\ref{claim:induction-intro} holds for $\bbw$ and $\bbT$, we still have to argue that $\bg_i(\extfindnode(\nex(\mu),i-1))$ is independent from the event $\calF^{\vk} \wedge r_{<i} \wedge g_{<i}$. Formally proving this requires a delicate induction, but the intuition is that $\calF^{\vk}$ depends on at most $k_{i}$ values in $\bg_{i}$ and $\br_{i}$, and the procedure $\walk$ carefully ensures that $\bg_i(\extfindnode(\nex(\mu),i-1))$ is never one of them. Hence, since $k_{i} \le \tau / 4$ and $\bg_{i}$ is $\tau$-wise independent, we have the desired independence.

\paragraph*{Handling $E_{\sf bad}$ and the two-vertex case.} We have just established Condition~\eqref{eq:total-count} which gives a lower bound for $E_{\sf total}$; now we briefly discuss how to obtain an upper bound on $E_{\sf bad}$ sufficient for proving the desired lower bound on $\Pr[u \in f^*_{a,\bh}(\bs)]$ using \eqref{eq:good}. One can first observe that~\eqref{eq:bad-count} cannot hold for all possible $\vk,\vk^1,\vk^2$, as there could be a collision between these three paths. In fact, let $K$ be the total number of nodes in the union of the paths corresponding to $\vk,\vk^1,\vk^2$. Then a revised estimate for $\Pr[\calF^{\vk} \wedge \nex(\vidx^{\vk}) = u \wedge \calF^{\vk^1} \wedge \calF^{\vk^2} \wedge a_{\nex(\vidx^{\vk^1})} = a_{\nex(\vidx^{\vk^2})}]$ should be $\left(2^{K} \cdot n^2\right)^{-1}$. By a careful calculation, one can show that this revised estimate is still enough to show $E_{\sf bad}$ is upper bounded by $O(2^{3\ell}/n^2)$, which is good enough for our purposes.

However, even establishing this revised estimate is quite challenging. Recall that $\calF^{\vk} \land \calF^{\vk^1} \land \calF^{\vk^2}$ is equivalent to the condition that, for every level-$i$ node $\beta$ on the paths from root to $\mu^{\vk},\mu^{\vk^1}$ or $\mu^{\vk^2}$, it holds that $\bg_i(a_{\bw_{\beta}}) = 1$. This amounts to $K$
events and we hope to show they are all independent. However, this is not true in general, as there can be a collision of $a_{\bw_{\beta}}$ between two different paths among these three paths. 
We overcome this issue by showing that for each ``bad node'' $\mu^{\vk}$, there must exist a 
``bad'' collision pair $\vk^1$ and $\vk^2$ on the extended walk without this issue. In such case one can establish a revised estimate; subtracting all these revised estimates from $E_{\sf good}$ would still yield a good lower bound on $\Pr[u \in f^*_{a,\bh}(\bs)]$.

Our proof for lower-bounding  $\Pr[u,v \in f^*_{a,\bh}(\bs)]$ follows the same template above, while using a more involved analysis to handle the dependency issues across the paths (we have to consider four paths now: two corresponding to $u$ and $v$, and the other two corresponding to the ``bad'' collision pair).

\newcommand{\supp}{\mathrm{supp}}
\newcommand{\zeroTon}[1]{\{0,1,\dotsc,#1\}}
\newcommand{\posN}{\mathbb{N}_{\ge 1}}
\newcommand{\vh}{\vec{h}}

\section{Preliminaries}
\label{sec:prelim}
Let $[n]$ denote $\{1,2,\dots,n\}$. We use $\N$ to denote the set of non-negative integers. We use $\tilde O(f)$ to denote $O(f\cdot \polylog f)$ in the usual way; $\tilde \Omega, \tilde \Theta$ are defined similarly.

We measure the space complexity of an algorithm by the maximum number of bits in its working memory: the read-only input is not counted. We measure the time complexity by the number of word operations (with word length $\Theta(\log n)$) in the word RAM model. 

For {\sc Element Distinctness} and {\sc List Disjointness}, we always assume the input arrays of length $n$ consist of positive integers bounded from above by $m=n^c+c$, where $c$ is a fixed constant independent of $n$. (We often abbrievate this by saying $m=\poly(n)$.) For an array $a\in [m]^n$, define the second frequency moment $F_2(a) = \sum_{i=1}^n \sum_{j=1}^n \mathbf{1}[a_i=a_j]$ as the number of colliding pairs $(i,j)$ (including the case where $i=j$). Note that $n\le F_2(a)\le n^2$.

We will use the following standard pseudorandomness construction.
	\begin{theorem}[Explicit $k$-wise independent hash family, \cite{carter1979universal}; see also {\cite[Corollary 3.34]{DBLP:journals/fttcs/Vadhan12}}]
		\label{thm:k-wise-indep}
		For $n,m,k$, there is a family of $k$-wise independent functions  $\caH \subseteq \{h \mid h\colon \{0,1\}^n\to \{0,1\}^m \}$ such that every function from $\caH$ can be described in $k\cdot \max\{n,m\}$ random bits, and evaluating a function from $\caH$ (given its description, and given an input $x \in \{0,1\}^n$) takes time $\poly(n,m,k)$.
	\end{theorem}

We often use bold font letters (\eg, $\bm{X}$) to denote random variables. We also use $\supp(\bm{X})$ to denote the support of random variable $\bm{X}$. 

For a set $U$, we often use $x \in_{\mathsf{R}} U$ to denote the process of selecting an element $x$ from $U$ uniformly at random.

\section{Properties of the Pseudorandom Family and their Implications}
\label{sec:property}

We will first define our pseudorandom hash family in~\cref{sec:construct-family}, and then give the proofs of our main theorems in \cref{sec:main-proof}, assuming some key technical lemmas that will be proved in subsequent sections.

\subsection{Construction of the Pseudorandom Family}
\label{sec:construct-family}

We first introduce some handy notation. For two functions $a,b\colon [m] \to ([n]\cup \{\star\})$, we naturally view them as ``restrictions'' (where $\star$ means ``unrestricted''), and define their \emph{composition} as
\[(a\bullet b)(x):= \begin{cases} b(x) & \text{ $b(x)	\neq \star$,} \\ a(x) & \text{otherwise.} \end{cases}\]
Observe that $(a\bullet b) \bullet c = a\bullet (b \bullet c)$.

\newcommand{\calH}{\mathcal{H}}
Let $\ell \le  \log n$ and $\tau = O(\log n\log \log n)$ be two positive integer parameters to be determined later. A sample $\bh \colon [m] \to ([n] \cup \{\star\})$ from $\calH_{\ell,m,n}$ is generated by an $\ell$-level iterative restriction process, defined as follows.

\begin{construction}{Drawing a sample $\bh$ from the pseudorandom hash function family $\calH_{\ell,m,n}$}
	
	\begin{enumerate}
		\item For each $i \in [\ell]$, independently draw two random functions $\bg_i \colon [m] \rightarrow \{0,1\}$ and $\br_i \colon [m] \rightarrow [n]$ from $\tau$-wise independent hash families (\Cref{thm:k-wise-indep}). Define $\bh_i\colon [m] \to [n] \cup \{\star\}$ to be 
		$$
		\bh_i(x) \coloneqq \begin{cases} \star &\text{if $\bg_i(x) = 0$,} \\ \br_i(x) &\text{if $\bg_i(x) = 1$.}\end{cases}$$ 
		
		\item Define $\bh$ to be $\bh_\ell \bullet \dots \bullet \bh_2 \bullet \bh_1$.
	\end{enumerate}
\end{construction}
Intuitively, the functions $g_i\colon [m] \to \{0,1\}$ control whether the value of $h(x)$ should be restricted at the $i$-th level, while the functions $r_i \colon [m] \to [n]$ determine the value that $h(x)$ is restricted to, at the $i$-th level. Note that $h(x) = \star$ if $g_1(x)=\cdots =g_\ell(x) = 0$, and $h(x) = r_j(x)$ if $g_1(x) = \cdots = g_{j-1}(x)=0$ and $g_j(x) =1$.

Since $m = \poly(n)$, the seed length for each $i\in [\ell]$ is $O(\log^2 n \log \log n)$ bits (\Cref{thm:k-wise-indep}), and hence the total seed length for describing the hash function $h$ is $O(\ell \log^2 n \log \log n) = O(\log^3 n \log \log n)$. Slightly abusing notation, we also use $\bh \getsR \caH_{\ell,m,n}$ to denote that $\bh$ is a hash function generated as above.

\paragraph{Digraph $G_{a,h}$ and reachable set $f^*_{a,h}(s)$.} Next we set up some notation. Recall that $a\in [m]^n$ is the input array. 
For a hash function $h\colon [m] \to [n]$, we define a mapping $f_{a,h} \colon [n] \to ([n] \cup \{\star\})$ by $f_{a,h}(x):= h(a_x)$.
This mapping naturally defines a $n$-vertex digraph $G_{a,h}$, where each vertex $x\in [n]$  has one outgoing edge $x \mapsto h(a_x)$ if $h(a_x)\neq \star$, and no outgoing edge if $h(a_x) = \star$.

 We use $f^*_{a,h}(s)$ to denote the set of vertices reachable in $G_{a,h}$ from $s$. When $a$ and $h$ are clear from context, we will simply write $f^*_{a,h}(s)$ as $f^*(s)$. Since each vertex in $G_{a,h}$ has at most one outgoing edge, note that the vertices in $f^*(s)$ form either a path or a ``rho-shaped'' component.

\subsection{Proofs of the Main Results}\label{sec:main-proof}

Let $a = (a_1,\dots,a_n)\in [m]^{n} $ be the read-only input array. The BCM Element Distinctness algorithm \cite{beame2013element} uses the following version of Floyd's cycle-finding algorithm performed on the digraph specified by $f_{a,h}$.  

\begin{lemma}[{\cite[Theorem 2.1]{beame2013element}}] \label{lem:cycle-finding}
Assuming oracle access to $f_{a,h}\colon [n] \to ([n] \cup \{\star\})$, there is a deterministic algorithm $\mathsf{COLLIDE}(s)$ which finds the pair $(u, v)\in [n]\times [n]$ (if it exists) such that $u, v \in f_{a,h}^*(s), u \neq v$ and $a_u = a_v$, in $O(|f^*_{a,h}(s)|)$ time and $O(\log n)$ space.\footnote{The original BCM algorithm works for $f_{a,h} \colon [n] \to [n]$. But it works equally well when some vertices $v$ may have no outgoing edges (\ie, $f_{a,h}(v) = \star$).}
\end{lemma}

In the BCM algorithm, $h$ was chosen from a truly random hash family. Our goal is to show that sampling $h$ from our pseudorandom hash family $\caH_{\ell,m,n}$ also suffices. 
To do this, we need the following two properties of our hash family $\caH_{\ell,m,n}$.

\begin{lemma}[Bounding the visit probability for a single vertex]\label{lem:warmup-hit-lower-bound}
	Suppose $\ell = \log n - \frac{\log F_2(a)}{2} - 10$.\footnote{We ignore all floors and ceilings for simplicity.} For every vertex $v \in [n]$, we have 
	\[
	\Pr_{\bh \getsR \calH_{\ell,m,n},\bs \getsR [n]}[v\in f^*_{a,\bh}(\bs)] = \Theta\left (\frac{1}{\sqrt{F_2(a)}}\right ).
	\]
\end{lemma}

\begin{lemma}[Lower bound for collision probability]\label{lem:hit-lower-bound}
Suppose $\ell = \log n - \frac{\log F_2(a)}{2} - 10$. For every $u,v\in [n]$ such that $u \not= v$ and $a_u = a_v$, we have 
\[
\Pr_{\bh \getsR \mathcal{H}_{\ell,m,n}, \bs \getsR [n]} [u, v \in f^*_{a, \bh}(\bs)] \ge \Omega\left (\frac{1}{F_2(a)}\right ).
\]
\end{lemma}

\cref{lem:warmup-hit-lower-bound} is proved in~\cref{sec:one} and~\cref{lem:hit-lower-bound} is proved in~\cref{sec:two}.

\begin{remark}
	In~\cref{lem:warmup-hit-lower-bound}, we obtain both a lower bound and an upper bound for $\Pr_{\bh,\bs}[v\in f^*_{a,\bh}(\bs)]$, and we will see shortly that only the upper bound will be useful in the proof of~\cref{thm:main-ed};  the lower bound part of~\cref{lem:warmup-hit-lower-bound} can be seen as a warm-up for the proof of~\cref{lem:hit-lower-bound}, which requires to prove a lower bound for the more involved two-vertex case (see~\cref{sec:two}).
\end{remark}

Since $\ell \le \log n$, each hash function $h$ from our hash family $\caH_{\ell,m,n}$ can be described with a seed of $O(\log^3 n\log \log n)$ bits and can be evaluated in $\polylog(n)$ time and $O(\log^3 n\log \log n)$ space. Armed with the two lemmas above, we can prove our main theorems.

\begin{reminder}{\Cref{thm:main-ed}}
{\sc Element Distinctness} can be decided by a Monte Carlo algorithm in $\tilde{O}(n^{1.5})$ time, with $O(\log^3 n \log \log n)$ bits of workspace and no random oracle. Moreover, when there is a colliding pair, the algorithm reports one.
\end{reminder}

\begin{proof}
Given input $a\in [m]^n$, we first assume that we know the correct parameter $1\le \ell \le \log n$ required in \cref{lem:warmup-hit-lower-bound} and \cref{lem:hit-lower-bound}, and let $\mathcal{H}$ be the pseudorandom hash family $\mathcal{H}_{\ell,m,n}$.
 We run $O(n\log n)$ trials of the $\textsf{COLLIDE}(s)$ algorithm (\Cref{lem:cycle-finding}) on $f_{a,h}$, where each trial uses a fresh random $h\in \mathcal{H}$.
We return YES if no collisions are found, and return NO otherwise. It is evident that this algorithm only requires one-way access to randomness, and the description of each $h$ can be stored in low space.

We first analyze the running time of this algorithm. By \Cref{lem:cycle-finding}, the running time of each trial is $O(|f^*_{a,h}(s)|)$. By~\cref{lem:warmup-hit-lower-bound}, the expected running time of each trial is 
\[\E_{\bh \in \mathcal{H}, \bs \in [n]}[|f^*_{a,\bh}(\bs)|]\cdot \polylog(n)= \sum_{v\in [n]} \Pr_{\bh\in \caH,\bs\in [n]}[v\in f_{a,\bh}^*(\bs)]\cdot \polylog(n) \le  \frac{n\cdot \polylog n}{\sqrt{F_2(a)}},\]  
where the $\polylog(n)$ factor comes from the time complexity of evaluating $h(\cdot)$.  Hence, the expected total running time of $O(n\log n)$ trials is $\tilde O(n^2/\sqrt{F_2(a)}) \le \tilde O(n^{1.5})$.
By Markov's inequality, with at least $1-o(1)$ probability, the total running time is bounded by $\tilde O(n^{1.5})$.

To analyze the success probability, note that in a ``NO'' instance (\ie, the elements are not distinct) there are $F_2(a)-n>0$ pairs of $u,v\in [n]$ such that $u\neq v$ and $a_u=a_v$. By linearity of expectation,  \cref{lem:hit-lower-bound} implies that the success probability of each trial is 
\[\Omega\left (\frac{F_2(a)-n}{F_2(a)}\right ) \ge \Omega \left ( 1/n\right ).\]
Since the samples of $h\in \mathcal{H}$ are independent across the trials, the probability of not finding any collisions is at most $\left (1-\Omega(1/n)\right )^{n\log n}\le n^{-\Omega(1)}$. 
The proof then follows from a simple union bound.

Recall at the beginning of the proof, we assumed $\ell$ was known. To remove this assumption, our actual algorithm simply tries all possible $\ell \in \{1,2,\dots,\log n\}$ one by one 
	(and terminates a trial if the running time is already too long for a specific $\ell$), which only increases the overall running time by an $O(\log n)$ multiplicative factor.
\end{proof}

Now we similarly prove the performance of the List Disjointness algorithm.

\begin{reminder}{\Cref{thm:listdisj}}
There is a Monte Carlo algorithm for {\sc List Disjointness} such that, given input arrays $a=(a_1,\dots,a_n),b=(b_1,\dots,b_n)$ and an upper bound $p\ge F_2(a)+F_2(b)$, runs in $\tilde O(n\sqrt{p})$ time and uses $O(\log^3 n\log \log n)$ bits of workspace and no random oracle.
\end{reminder}

\begin{proof}
Similar to the proof of \cref{thm:main-ed}, we can assume that the correct $\ell$ required in \cref{lem:warmup-hit-lower-bound} and \cref{lem:hit-lower-bound} is known.
Let array $c$ be the concatenation of $a$ and $b$, which must satisfy $F_2(c) \le 2(F_2(a)+F_2(b))\le 2p$. We run $2p\log n$ trials of the $\textsf{COLLIDE}(s)$ algorithm (\Cref{lem:cycle-finding}) on $f_{c,h}$, each time using a fresh random $h\in \mathcal{H}$.
We return NO if we find a collision in $c$ where the two items come from $a$ and $b$ respectively. We return YES if the total time spent by the algorithm exceeds $\tilde O(n\sqrt{p})$ while no such collisions have been found.

To analyze the running time, we focus on the first $F_2(c)\log n$ trials executed by the algorithm.
By a similar argument in the previous proof, with at least $1-o(1)$ probability, the total running time of these $F_2(c)\log n$ trials is at most \[\tilde O\left (F_2(c) \cdot \frac{n}{\sqrt{F_2(c)}}\right )\le \tilde O(n\cdot \sqrt{F_2(c)}). \]

By \Cref{lem:hit-lower-bound}, the success probability of each trial is $\Omega(1/F_2(c))$ (note that in  the previous proof we had $F_2(a)-n$ pairs of ``good'' collisions $(u,v)$, while here it is possible that we have only one ``good'' pair, along with many ``bad'' pairs coming from the same input array). Then, the probability of finding a collision during the first $F_2(c)\log n$ trials is at least $1-n^{\Omega(1)}$.

By a union bound, we can show that, on a ``NO'' input, with at least $1-o(1)$ probability the algorithm will terminate in one of the first $F_2(c)\log n$ trials, without exceeding the time limit $\tilde O(n\sqrt{p})$.
\end{proof}

Now we similarly give a low-space algorithm for {\sc Set Intersection}, with near-optimal time complexity.

\begin{reminder}{\Cref{thm:setintersection}}
There is a randomized algorithm that, given input arrays $A=(a_1,\dots,a_n),B=(b_1,\dots,b_n)$ where $A$ and $B$ are both YES instances of {\sc Element Distinctness},  prints all elements in $\{a_1,\dots,a_n\}\cap \{b_1,\dots,b_n\}$ in $\tilde O(n^{1.5})$ time, with $O(\log^3 n\log \log n)$ bits of workspace and no random oracle. The algorithm prints elements in no particular order, and the same element may be printed multiple times.\end{reminder}

\begin{proof}
Similar to the proof of \cref{thm:main-ed}, we can assume that the correct $\ell$ required in \cref{lem:warmup-hit-lower-bound} and \cref{lem:hit-lower-bound} is known.

As before, we define $c$ to be the concatenation of $a$ and $b$. We run $n\log^2 n$ trials of the  $\textsf{COLLIDE}(s)$ algorithm (\Cref{lem:cycle-finding}) on $f_{c,h}$, each using a fresh random $h\in \mathcal{H}$. We print all the collisions found. Note these must be elements in $\{a_1,\dots,a_n\}\cap \{b_1,\dots,b_n\}$, by our assumption on the input: since $A$ and $B$ are YES instances of {\sc Element Distinctness}, all colliding pairs must have one element from $A$ and one element from $B$. 

By a similar argument as in the proof of \Cref{thm:main-ed}, with $1-o(1)$ probability the total running time is bounded by $\tilde O(n^{1.5})$. And for every element in the intersection, the probability that it is never printed is at most 
\[ \left (1-\Omega\left (\frac{1}{F_2(c)}\right )\right )^{n\log^2 n} \leq n^{-\omega(1)},\]
where we used $F_2(c) = \Theta(n)$ implied by the input assumption.
 The proof then follows from a simple union bound.
\end{proof}


\section{The Extended Walk and the Dependency Tree}
\label{sec:construction}

In this section, we present the definitions of the extended walk and the dependency tree along with several useful properties of them, which will play an important role in our proof in \cref{sec:one,sec:two}.

\subsection{The Extended Walk}
\label{sec:extendedwalk}

Letting $\bh \getsR \calH_{\ell,m,n}$ and $\bs \getsR [n]$, recall that the reachable set $f^*_{a,\bh}(\bs)\subseteq [n]$ consists of the vertices on the  following pseudorandom walk: starting from a random vertex $\bs$, we repeatedly move from the current vertex $x$ to $\bh(a_x)$, until $\bh(a_x) = \star$, in which case the walk ends. In the case when $f_{a,\bh}^*(\bs)$ contains a cycle, this walk has infinite length, which complicates our analysis.

To facilitate the analysis, we instead define an auxiliary walk $\bw$ that is \emph{jointly distributed} with $f_{a,\bh}^*(\bs)$.
The auxiliary walk $\bw$ starts from $\bs$, terminates with probability 1 (\Cref{lem:terminate}), and has several other nice properties that make it easier to analyze.
We will also see that $\bw$ is related to the reachable set $f^*_{a,\bh}(\bs)$ that we care about. In particular, it includes all the vertices in $f_{a,\bh}^*(\bs)$ (\Cref{lem:subset-of-extwalk}) as a subset, and for this reason we call $\bw$ an \emph{extended walk}. 

Let us formally define the extended walk. The extended walk $\bw$ is a sequence of vertices generated by the recursive process $\walk$ specified by~Algorithm~\ref{algo:extwalk}, which depends on the input array $a$ and the random variables $\bh, \bs$, as well as some additionally sampled random variables.
We summarize them in the following box.

\begin{construction}{The extended walk probability space $\extwalk_{\ell,m,n,a}$}
	\begin{itemize}
	    \item 
		 \textbf{Setup.} We sample the random variables as follows:
		 \begin{itemize}
		     \item Draw the starting vertex $\bs \getsR [n]$.
		     \item  Sample $ \{\bg_i\}_{i\in [\ell]}$ and $ \{\br_i \}_{i\in [\ell]}$, which together determine a sample $\bh \getsR \calH_{\ell,m,n}$ from the pseudorandom hash family, as described in \cref{sec:construct-family}. 
		     \item Then, for each $i \in [\ell]$, we extend the domain of $\bg_i$ and $\br_i$  from $[m]$ to $[m] \cup \{\star_0, \star_1, \dots\}$ as follows: for every $t\in \N$,  we sample $\bg_i(\star_t)\getsR \{0,1\},\br_i(\star_t)\getsR[n]$, where the samples are independent across all $\star_t$ and all levels $i\in [\ell]$.   
		 \end{itemize}

		\item \textbf{Generating the walk.} After fixing $\{g_i\}_{i\in [\ell]},\{r_i\}_{i\in [\ell]}$, we define a function $\walk(s',i, \mu_0)$  (where $s'\in [n]$ and $i \in \zeroTon{\ell}$) by the pseudocode in Algorithm~\ref{algo:extwalk}, which returns a sequence of vertices.\footnote{The sequence returned by the function $\walk(s',i, \mu_0)$ actually depends on the sampled $\{g_i\}_{i\in [\ell]},\{r_i\}_{i\in [\ell]}$ as well, but we choose not to make it explicit in the notation $\walk(s',i, \mu_0)$ for simplicity.} 
		
		Then, the extended walk $\bw$ is defined as $\walk(\bs,\ell,0)$.
		
	\end{itemize}
\end{construction}

\begin{algorithm}[h]\label{algo:extwalk}
\DontPrintSemicolon
	\caption{Algorithm for extended walk}
	\SetKwProg{Fn}{Function}{}{}
	\Fn{$\walk(s',i,\underline{\mu_0})$ \hspace{0.5cm}(where $s'\in [n], 0\le i \le \ell$)}{
	
	\lIf{$i = 0$}{\Return sequence $(s')$ \label{line:corner}}
	$C_0 \gets \emptyset$, $\mathsf{star} \gets \text{false}$\\
	$j \gets 0, s_0 \gets s', w \gets ()$ \tcc{ () means an empty sequence }
	
	\Repeat{$g_i(y) = 0$}{
		$w \gets w \circ \walk(s_j, i - 1, \underline{\mu_0 + |w|})$ \label{walk:w}  \tcc{ $\circ$ means concatenation of two sequences} 
		$x_{j + 1} \gets w_{|w|}$ \label{walk:xj}\tcc{We use 1-based indexing, so $w_{|w|}$ means the last vertex in $w$}
		$y,\mathsf{star} \gets \begin{cases} a_{x_{j + 1}}, \text{false} & \text{ if } a_{x_{j + 1}} \not\in C_j \land \lnot \mathsf{star} \\ \star_t, \text{true} & \text{ otherwise (where $t := \min \{t \in \N  \ \vert \ \star_t \not\in C_j\}$)}\end{cases}$\label{walk:y}\\
		\underline{$\mu_{j + 1} \gets \mu_0 + |w|$, $\a_i(\mu_{j + 1}) \gets y$, $\x_i(\mu_{j + 1}) \gets x_{j + 1}$} \label{walk:a}\\
		\lIf{$j > 0$}{\underline{$\rig(\mu_j) \gets \mu_{j + 1}$}} \label{walk:rig}
		\If{$g_i(y) = 1$ \label{walk:if}}{
			$C_{j + 1} \gets C_j \cup \{y\}$, $s_{j + 1} \gets r_i(y)$ \label{walk:Cj}\\
			\underline{$\level(\mu_{j + 1}) \gets i, \nex(\mu_{j + 1}) \gets r_i(y)$} \label{walk:level}\\
			$j \gets j + 1$
		}
	}
	\Return $w$ \label{line:return}
	}
\end{algorithm}

We remark that in the pseudocode of Algorithm~\ref{algo:extwalk}, all the underlined parts are used for assigning some additional attributes that are helpful for analysis, and have \emph{no effect on the return value} of the function $\walk(s',i, \underline{\mu_0})$. Therefore, when we only need the return value of it, we will simply write $\walk(s', i)$ and ignore all the underlined parts. 
The meanings and properties of these additionally assigned values will be explained in detail later in this section, and they will also be summarized in Table \ref{table:summary-one-vertex} in Section \ref{sec:one}.

Intuitively, in Algorithm~\ref{algo:extwalk}, $\walk(s',i)$ generates a walk starting from vertex $s'$, which travels along the outgoing edges specified by $\{g_{i'}\}_{1\le i'
\le i},\{r_{i'}\}_{1\le i'\le i}$, and stops upon encountering a vertex of level higher than $i$ (\ie, a vertex $x$ with $g_1(x)=g_2(x)= \dots = g_i(x)=0$). 
As depicted in Figure \ref{fig:structure},
the $\walk(s',i)$ process is implemented by recursive calls to $\walk(s_j,i-1)$ generating walks of levels up to $i-1$, which are to be concatenated together using edges $(x_{j+1} \to s_{j+1})$ on level $i$.
More importantly, the extended walk $\walk(s',i)$ uses some mechanism to avoid the infinite cycling that would occur in the actual walk $f_{a,\bh}^*(\bs)$: if a recursive call to $\walk(s_j,i-1)$ ends at some vertex $x_{j+1}$ whose value $a_{x_{j+1}}$ has already appeared for some previous $j'<j$, then we will not reuse this value when generating its outgoing level-$i$ edge (moreover, we will also disregard the $a_{x_{j''+1}}$ values for all future $j''$ during $\walk(s',i)$).

The third parameter $\mu_0$ of $\walk(s',i, \underline{\mu_0})$ simply keeps track of the current position relative to the start of the entire extended walk $\walk(s,\ell,0)$, and is useful for indexing the nodes on the walk.  Note that we stick to the convention of using Greek letters (\eg, $\alpha,\beta,\mu$) for indexing the walk.

\begin{figure}[h]
    \centering
\scalebox{1.8}{
\begin{tikzpicture}[thick]
  \node [draw,circle,fill, minimum size=3,inner sep=0pt, outer sep=3pt] (s0) {};
  \node [above of = s0, yshift = -20] {\tiny$s_0 = s'$}; 

  \node [right of = s0, draw,circle,fill, minimum size=3,inner sep=0pt, outer sep=3pt, xshift = 20] (x1) {} edge [<-, dotted] (s0);
  \node [above of = x1, yshift = -20] {\tiny$x_j$}; 
  \node [below of = x1, yshift = 20] {\tiny$\mu_j$}; 

  \node [right of = x1, draw,circle,fill, minimum size=3,inner sep=0pt, outer sep=3pt, xshift = -5] (s1) {} edge [<-] (x1);
  \node [above of = s1, yshift = -20] {\tiny$s_j$}; 
  \node [right of = s1, draw,circle,fill, minimum size=3,inner sep=0pt, outer sep=3pt, xshift = 10] (x2) {} edge [<-, dashed] (s1);
  \node [above of = x2, yshift = -20] {\tiny$x_{j + 1}$}; 
  \node [below of = x2, yshift = 20] {\tiny$\mu_{j + 1}$}; 
  \node [right of = x2, draw,circle,fill, minimum size=3,inner sep=0pt, outer sep=3pt, xshift = -5] (s2) {} edge [<-] (x2);
  \node [above of = s2, yshift = -20] {\tiny$s_{j + 1}$}; 
  
  \node [right of = s2, draw,circle,fill, minimum size=3,inner sep=0pt, outer sep=3pt, xshift = 20] (xx) {} edge [<-, dotted] (s2);
  
  \draw [blue, draw, dashed] (2.2,-0.6) rectangle (4.25,0.6);
  \node [blue] at (3.5, -0.8)  {\tiny $\mathsf{walk}(s_j, i - 1)$};
\end{tikzpicture}
}
    \caption{The structure of $\walk(s',i)$. Note that $x_{j + 1}$ is the last vertex of $\walk(s_j, i - 1)$. }
    \label{fig:structure}
\end{figure}
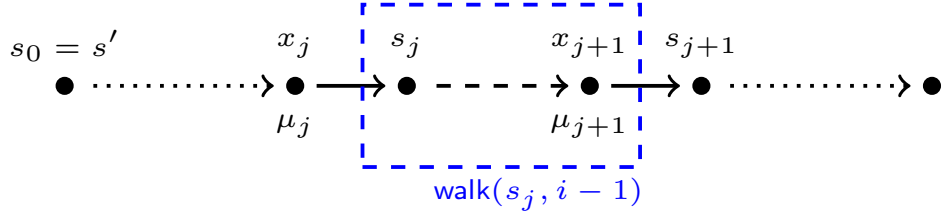

 To better understand Algorithm~\ref{algo:extwalk}, we start with several simple observations.
\begin{observation}
The return value of $\walk(s',i)$ must be a sequence of vertices starting with $s'$.
\label{lem:starts}
\end{observation}
\begin{proof}
This immediately follows from Line~\ref{line:corner}, Line~\ref{walk:w} and Line~\ref{line:return} by a simple induction on $i$.
\end{proof}

\begin{observation}
\label{lem:only-depend}
For every $s'\in [n], i \in \zeroTon{\ell}$, the return value of 
$\walk(s',i, \underline{\mu_0})$
 and the additional values assigned by 
 $\walk(s',i, \underline{\mu_0})$
  only depend on $s',\mu_0$, $\{g_{i'}\}_{1\le i'\le i},\{r_{i'}\}_{1\le i'\le i}$ and the input array $a$.
\end{observation}
\begin{proof}
The observation is trivial when $i = 0$. When $i \ge 1$, in $\walk(s',i)$, the algorithm only examined the values of $g_i(\cdot)$ and $r_i(\cdot)$. The recursive calls $\walk(\cdot,i')$ made by $\walk(s',i)$ can only have lower levels $i'\le i$, and hence only depends on the values of $g_{i'}(\cdot)$ and $r_{i'}(\cdot)$.
\end{proof}

The following lemmas says that with probability $1$, $\walk(\bs,\ell)$ terminates.

\begin{lemma}
\label{lem:terminate}
With probability $1$, $\bw = \walk(\bs,\ell)$ has finite length.
\end{lemma}
\begin{proof}
We will prove a stronger statement that for every $s' \in [n]$ and $i \in \zeroTon{\ell}$, $\walk(s', i)$ has finite length with probability $1$, by an induction on $i$. 

For the base case $i=0$, this clearly holds due to Line~\ref{line:corner}. Now, suppose the inductive hypothesis holds for $i - 1$. We fix an $s' \in [n]$ and consider $\walk(s',i)$, and it follows from the inductive hypothesis that all recursive calls to $\walk(s_j, i - 1)$ terminate with probability $1$.

Next, we consider the following two cases: (1) the repeat loop in $\walk(s', i)$ terminates within $m$ rounds or (2) it executes more than $m$ rounds. In Case (1), $\walk(s',i)$ terminates with probability $1$, so from now on we focus on Case (2). In this case, $y$ eventually becomes $\star_t$ for some $t \in \N$ at Line~\ref{walk:y} since $C_j\subseteq [m]$ when $\textsf{star} = \text{false}$. After that, since $g_i(\star_t)\getsR \{0,1\}$ are independently sampled across all $t\in \N$, with probability $1$ there is $t \in \N$ for which $g_i(\star_t) = 0$. Hence, the \textbf{repeat} loop terminates with probability $1$. Finally, we simply apply a union bound over all starting points $s' \in [n]$, which proves our induction hypothesis for $i$.
\end{proof}

\paragraph*{Assigned values.}  Now, let us elaborate on the values $\a_i(\vidx), \x_i(\vidx), \level(\vidx), \nex(\vidx), \rig(\vidx)$ assigned in the underlined lines in Algorithm~\ref{algo:extwalk}. To begin with, we first explain the role of $\mu_0$. Roughly speaking, $\mu_0$ is the number of vertices before $\walk(s', i)$ in the final extended walk $\bw$. Formally, we have the following lemma. 

\begin{lemma} \label{lem:mu0}
Fix $i \in \zeroTon{\ell}$. Consider all function calls $\walk(\cdot, i, \cdot)$ during the generation of $\bw$. Suppose they are $\walk(s^1, i, \mu^1_0), \walk(s^2, i, \mu^2_0), \dots, \walk(s^t, i, \mu^t_0)$ sorted by increasing order of $\mu^j_0$ for $j \in [t]$. The following hold:

\begin{enumerate}[label=(\arabic*)]
	\item $\bw = \walk(s^1, i, \mu^1_0) \circ \walk(s^2, i, \mu^2_0) \circ \cdots \circ \walk(s^t, i, \mu^t_0)$, and $\mu^j_0 = \sum_{t' = 1}^{j - 1} |\walk(s^{t'}, i)|$.
	
	\item For each $\mu \in  [|\bw|]$, there is a unique function call $\walk(s^j, i, \mu^j_0)$ such that $\mu^j_0 < \mu \leq \mu^j_0 + |\walk(s^j, i)|$.
\end{enumerate}
\end{lemma}
\begin{proof}
We prove Item~(1) by an induction on $i$. By definition, when $i = \ell$, we have $\bw = \walk(\bs, \ell, 0)$. This proves the base case. Now, suppose the statement holds for $i$. We prove it also holds for $i - 1$.

For each $\walk(s^j, i, \mu^j_0)$, by Line \ref{walk:w}, we have
\begin{equation} \label{walk-eq1}
\walk(s^j, i, \mu^j_0) = \walk(s^j_1, i - 1, \mu^{j, 1}_0) \circ \walk(s^j_2, i - 1, \mu^{j, 2}_0) \circ \cdots \circ \walk(s^{j}_{t_j}, i - 1, \mu^{j, t_j}_0),
\end{equation}
and for every $t' \in [t_j]$, it holds that
\begin{equation} \label{walk-eq2}
\mu^{j, t'}_0 = \mu^j_0 + \sum_{q = 1}^{t' - 1} |\walk(s^{j}_{q}, i - 1)|.
\end{equation}

From the induction hypothesis, it follows that $$\bw = \walk(s^1, i, \mu^1_0) \circ \walk(s^2, i, \mu^2_0) \circ \cdots \circ \walk(s^t, i, \mu^t_0),$$ and $\mu^j_0 = \sum_{t' = 1}^{j - 1} |\walk(s^{t'}, i)|$. It also holds for $i - 1$ by expanding each $\walk(s^j, i, \mu^j_0)$ using \eqref{walk-eq1} and \eqref{walk-eq2}. 

Item~(2) then follows directly from the definition of $\mu^j_0$ and Item~(1). 
\end{proof}

The following lemma explains the role of $\nex(\cdot)$. We additionally define $\nex(0)$ to be the starting vertex $w_1 = s$ of the extended walk $w$.
\begin{lemma} \label{obs:w-nex}
For every $w \in \supp(\bw)$, we have $w_{\mu + 1} = \nex(\mu)$ for all $\mu \in [0, |w| - 1]$. 
\end{lemma}
\begin{proof}
Consider the moment when $\nex(\mu)$ is assigned a value (Line~\ref{walk:level} in Algorithm~\ref{algo:extwalk}), which happens inside the if-body of $g_i(y)=1$.
At this point, we have $\mu = \mu_{j+1} = \mu_0 + |w|$,  and we assign $r_i(y)$ to both $\nex(\mu)$ and $s_{j + 1}$. After that, since $g_i(y) = 1$, the repeat loop must execute another round with $j_{\mathsf{new}} \gets j+1$. At the beginning of the new round, we concatenate 
 $w$ with $\walk(s_{j_{\mathsf{new}}}, i - 1)$, which starts with $s_{j_{\mathsf{new}}} = s_{j+1} = \nex(\mu)$ by~\Cref{lem:starts}. Hence, $w_{\mu + 1}$ must equal $\nex(\mu)$. 
\end{proof}

Now, let us look at the properties of $\a_i(\cdot)$ and $\level(\cdot)$. 
From Algorithm~\ref{algo:extwalk} we can see that $\a_i(\mu)$ is the argument we pass to functions $g_i(\cdot)$ and $r_i(\cdot)$ for determining whether (and what) to assign to $\nex(\mu)$ at the current level. 
\begin{observation}
\label{lem:giai}
    Let $(w,g)\in \supp(\bw,\bg)$. For every $\mu \in [|w|-1]$,  we have $g_i(\a_i(\mu))= 0$ for all $i\in [\level(\mu)-1]$, and $g_{\level(\mu)}(\a_{\level(\mu)}(\mu)) = 1$. In addition, for $\mu = |w|$, we have $g_i(\a_i(\mu))= 0$ for all $i\in [\ell]$, and $\level(\mu)$ is undefined\footnote{In Section~\ref{sec:tree} we will specially define its level to be $\ell+1$}.
\end{observation}
\begin{proof}
 Suppose during $\walk(s',i,\mu_0)$ when $\mu_{j+1}=\mu$, the value of $\level(\mu)$ is not yet assigned. If $g_i(\a_i(\mu))= 0$, then the \textbf{if}-test at Line~\ref{walk:if} is not passed and hence Line~\ref{walk:level} is not reached, which means $\level(\mu)$ can only be assigned later at a higher level of the recursion with $\level(\mu)>i$. On the other hand, if $g_i(\a_i(\mu))= 1$ at this point, then we assign $\level(\mu) = i$ at Line~\ref{walk:level}. Hence, we must have $g_{\level(\mu)}(\a_{\level(\mu)}(\mu)) = 1$, and $g_{i'}(\a_{i'}(\mu))= 0$ for all $i'<i=\level(\mu)$.
 
 The ``in addition'' part follows from a similar argument.
\end{proof}
Hence, we introduce the following shorthand.
\begin{definition} \label{def:short-hand-for-a}
We denote $\x(\mu) = \x_{\level(\mu)}(\mu)$ and $\a(\mu) = \a_{\level(\mu)}(\mu)$. 
\end{definition}

Next, we have several simple observations.
\begin{observation}
For every $(w,r)\in \supp(\bw,\br)$ and $\mu \in [|w| - 1]$, we have $\nex(\mu) = r_{\level(\mu)}(\a(\mu))$.
\end{observation}
\begin{proof}
When $\level(\mu)$ and $\nex(\mu)$ are assigned together in Line~\ref{walk:level}, we have $y = \a_i(\mu)$, and hence $\nex(\mu) = r_i(y) = r_i(\a_i(\mu)) = r_{\level(\mu)}(\a_{\level(\mu)}(\mu))$.
\end{proof}

In the following, we write $y = \star_*$ to denote that $y = \star_t$ for some $t \in \N$, and write $y \ne \star_*$ otherwise.

\begin{observation}
\label{lem:amuawmu}
For every $w\in \supp(\bw)$, $\mu \in [|w|-1]$, and $i \in [\level(\mu)]$, if $\a_i(\mu) \neq \star_*$, then $\a_i(\mu) = a_{w_\mu}$.
\end{observation}
\begin{proof}
Note that when we assign $\a_i(\mu)=y$ at Line~\ref{walk:a}, we have $\mu_{j+1} = \mu$, and by Line~\ref{walk:y} we must have $y = a_{x_{j+1}}$ if $y\neq \star_*$. Then we simply note that $x_{j+1} = w_{\mu_{j+1}} = w_\mu$ by Line~\ref{walk:xj} and Line~\ref{walk:a}.
\end{proof}

Then, we examine how the values of $\a_i(\cdot)$ are determined in the \textbf{repeat} loop of  $\walk(s',i)$. 
Observe that, by our definition at Line~\ref{walk:y}, we never assign the same $y$ value to $\a_i(\cdot)$ twice: when the value $a_{x_{j+1}}$ appears for the second time, we will set $\mathsf{star} \gets \text{true}$ and replace this value with $\star_*$.
In more detail, this is formalized in the following lemma.

\begin{lemma} \label{obs:determine-a-from-x}
In $\walk(s',i, \mu_0)$, $\a_i(\mu_j)$ is uniquely determined from $x_1, x_2, \dots, x_j$ as follows:
\begin{enumerate}
	\item Let $j' = \min \{j' \ \vert \ \exists  j'' \text{ s.t. } 1 \leq j'' < j' \leq j,  a_{x_{j''}} = a_{x_{j'}}\}$. 
	\item If no such $j'$ exists, then $\a_i(\mu_j) = a_{x_j}$. Otherwise, $\a_i(\mu_j) = \star_{j - j'}$. 
\end{enumerate}

In particular, $\a_i(\mu_j) \neq \a_i(\mu_{j'})$ holds for all $j\neq j'$.
\end{lemma}
\begin{proof}
By Line \ref{walk:a} and Line \ref{walk:Cj}, we know $C_j = \{\a_i(\mu_1), \a_i(\mu_2), \dots, \a_i(\mu_j)\}$. By Line \ref{walk:y}, we know $\mathsf{star}$ switches from $\text{false}$ to $\text{true}$ when $a_{x_{j + 1}} \in C_j$. For those $j$ before $\mathsf{star}$ switches, $\a_i(\mu_j) = a_{x_j}$, and for those $j$ after the switch, $\a_i(\mu_j) = \star_*$. 

Hence, $\mathsf{star}$ switches at the first $j'$ such that there exists $1 \leq j'' < j'$ with $a_{x_{j''}} = a_{x_{j'}}$. If such $j'$ does not exist, $\mathsf{star}$ is still false at $j$, and we know $\a_i(\mu_j) = a_{x_j}$. Otherwise, $\mathsf{star}$ switches at $j'$, and by Line \ref{walk:y} we have $\a_i(\mu_{j'}) = \star_0, \a_i(\mu_{j' + 1}) = \star_1,\dots$, and $\a_i(\mu_j) = \star_{j - j'}$. 
\end{proof}

Finally, we show the connection between the extended walk $\bw$ and the reachable set $f^*_{a,\bh}(\bs)$.
\begin{lemma}\label{lem:subset-of-extwalk}
Let $(w,h,s) \in \supp(\bw,\bh,\bs)$, where $w=\walk(s,\ell)$ is the extended walk, and $h$ is the hash function. The following hold:
\begin{enumerate}
	\item The reachable set $f^*_{a,h}(s)$ is a subset of the vertices in $w$.
	
	\item For every $\mu \in [|w|]$, if for every two distinct $\alpha,  \beta \in [\mu]$, it holds that $a_{w_{\alpha}}\neq a_{w_{\beta}}$, then $w_{\mu} \in f^*_{a,h}(s)$. 
	In particular, if there are no two distinct $\alpha, \beta \in [|w|]$ such that $a_{w_{\alpha}}=a_{w_{\beta}}$, then $f^*_{a,h}(s)$ contains exactly the same vertices as $w$. 
\end{enumerate}
\end{lemma}
\begin{proof}
We first prove that, for every $\mu \in [|w|]$, if there are no two distinct $\alpha,  \beta \in [\mu]$ such that $a_{w_{\alpha}}=a_{w_{\beta}}$, then $\a_i(\mu) = a_{w_{\mu}}$ for every $i\in [\level(\mu)]$. 

We will use induction on $\mu$.  Suppose the inductive hypothesis holds for $1, 2, \dots, \mu - 1$. Now we show that $\a_i(\mu)\neq \star_*$ for every $i \in [\level(\mu)]$, which immediately implies that $\a_i(\mu) = a_{w_\mu}$ by \cref{lem:amuawmu} and finishes the inductive step.

Suppose for contradiction that we assigned $\a_i(\mu)= \star_*$ at Line~\ref{walk:a} for some $i \in [\level(\mu)]$. Then, by the definition of $y$ at Line~\ref{walk:y}, the only two cases are (1) $a_{w_\mu}\in C_j$, or (2) $\mathsf{star}=\text{true}$ (which implies $\star_0 \in C_j$). In either case, there is an earlier $\eta < \mu$ such that either (1) $a_{w_{\eta}} = a_{w_{\mu}}$ (which follows from the way we update $C_j$ at Line~\ref{walk:Cj}) or (2) $\a_i(\eta) = \star_*$ and $\level(\eta) = i$ (because $\a_i(\mu_{j+1}) = y$ is added to $C_{j + 1}$ at Line~\ref{walk:Cj} only when $\level(\mu_{j+1}) = i$). Case (1) contradicts our assumption that $a_{w_\alpha}\neq a_{w_\beta}$ for every two distinct $\alpha,\beta \in [\mu]$. Case (2) contradicts the inductive hypothesis that $\a_i(\eta) = a_{w_\eta} \neq \star_*$. Therefore we have $\a_i(\mu) = a_{w_\mu} \neq \star_*$.

Hence for every $\mu \in [|w|]$, for all $i \in [\level(\mu)]$, $g_i(\a_i(\mu))$ and $r_i(\a_i(\mu))$ have the same values as the pseudorandom functions $g_i(a_{w_{\mu}})$ and $r_i(a_{w_{\mu}})$ that were used to define $h(a_{w_{\mu}})$ for $h \in \mathcal{H}_{\ell,m,n}$. Then, by \cref{lem:giai} and our definition of $h$, it is evident that $ \nex(\mu) = h(a_{w_\mu})$, and hence $w_{\mu+1} = h(a_{w_\mu})$  by \cref{obs:w-nex}.

The actual reachable set $f^*_{a,h}(s)$ has vertices $\{w'_1,w'_2,\dots\}$ where $w'_1=s$ and $w'_{\mu+1} = h(a_{w'_{\mu}})$ for every $\mu \ge 1$. Note that  $w_1=w'_1 = s$ by \cref{lem:starts}. We have proved that for every $\mu$ such that no two distinct $\alpha, \beta \in [\mu]$ satisfy $a_{w_{\alpha}}=a_{w_{\beta}}$, we have $w_{\mu + 1}= h(a_{w_\mu})$ and $w'_{\mu + 1}=h(a_{w_\mu'})$.
Then, let $\mu_0$ be the smallest $\mu\in [|w|]$ such that there exists $\alpha<\mu$ with $a_{w'_\alpha}=a_{w'_\mu}$. 
If such $\mu_0$ does not exist, then a simple induction shows $w'_{\eta} = w_{\eta}$ holds for all $\eta$. Otherwise, we can similarly show $(w_1, w_2, \dots, w_{\mu_0})= (w'_1, w'_2, \dots, w'_{\mu_0})$. This proves Item~(2).

On the other hand, from $w'_{\mu_0 + 1} = h(a_{w'_{\mu_0}}) = h(a_{w'_{\alpha}}) = w'_{\alpha + 1}$, it follows that $\{w'_{\mu_0 + 1},w'_{\mu_0+2},\dots\} \subseteq \{w'_1,w'_2,\dots,w'_{\mu_0}\}$. Hence, $\{w'_1,w'_2,\dots\}$ is be a subset of $\{w_1,w_2,\dots\}$, which proves Item~(1).
\end{proof}

\subsection{Dependency Tree and Node Indexing} \label{sec:tree}

Using the information of $\level(\cdot)$ determined by the recursive process $w = \walk(s,\ell)$, we can define a natural tree structure that we call the \emph{dependency tree}.
The tree consists of $|w| + 1$ nodes, labeled by integers from $0$ to $|w|$. We will also use Greek letters (\eg, $\alpha,\beta,\mu$) to refer to the nodes on the dependency tree. 
The $\mu$-th vertex $w_\mu$ in the extended walk  corresponds to node $\mu$ on the tree. Node $0$ is the root of the tree, and we define $\level(0)=\ell+1$ and $\nex(0) = s$. Moreover, for the last node $|w|$, its level is not assigned by $\walk(s,\ell)$,  so we define $\level(|w|) = \ell + 1$ as well.

In the rest of the paper,  we will reserve the term ``\emph{node}'' for nodes (referred to using Greek letters) on the dependency tree, and use the term ``\emph{vertex}'' to refer to the vertices in the walks on the digraph $G_{a,h}$, \ie, a ``vertex'' is always in the set $[n]$.

To define the dependency tree, we specify the parent of each node $\mu$ as follows.
\begin{definition}
In the dependency tree, the parent node of node $\mu$ is  defined as
\[
\parent(\mu):= \max_{\mu' < \mu} \{\mu' \ \vert \ \level(\mu') \geq \level(\mu)\},
\]
\ie,the last node before $\mu$ that has level at least $\level(\mu)$. Note that such $\mu'$ always exists as we have $\level(0) = \ell + 1$. 
\end{definition}

\begin{figure}
	\centering
	\scalebox{0.9}{
\begin{tikzpicture}[node distance={14mm}, thick, vertex/.style = {draw, thick, circle,inner sep=0pt,minimum size=15pt,outer sep=3pt}] 
  \node (l1) {level $1$};
  \node [above of = l1, yshift = -0.4cm] (l2) {level $2$};
  \node [above of = l2, yshift = -0.4cm] (l3) {level $3$};
  \node [above of = l3, yshift = -0.4cm] (l4) {level $4$};
  \node [above of = l4, yshift = -0.4cm] (l5) {level $5$};
  \node [above of = l4, yshift = -0.8cm] (ll) {($\ell = 4$)};
  \node [vertex, right of = l5] (n0) {$0$};
  \node [vertex, right of = l1, xshift = 1cm] (n1) {$1$} edge [<-] (n0);
  \node [vertex, right of = l4, xshift = 2cm] (n2) {$2$} edge [<-] (n0);
  \node [vertex, right of = l2, xshift = 3cm] (n3) {$3$} edge [<-] (n2);
  \node [vertex, right of = l3, xshift = 4cm] (n4) {$4$} edge [<-] (n2);
  \node [vertex, right of = l1, xshift = 5cm] (n5) {$5$} edge [<-] (n4);
  \node [vertex, right of = l2, xshift = 6cm] (n6) {$6$} edge [<-] (n4);
  \node [vertex, right of = l3, xshift = 7cm] (n7) {$7$} edge [<-] (n4);
  \node [vertex, right of = l2, xshift = 8cm] (n8) {$8$} edge [<-] (n7);
  \node [vertex, right of = l4, xshift = 9cm] (n9) {$9$} edge [<-] (n2);
  \node [vertex, right of = l3, xshift = 10cm] (n10) {$10$} edge [<-] (n9);
  \node [vertex, right of = l3, xshift = 11cm] (n11) {$11$} edge [<-] (n10);
  \node [vertex, right of = l5, xshift = 12cm] (n12) {$12$} edge [<-] (n0);

\end{tikzpicture}
}
	\caption{The dependency tree $T$. The index of $7$ is $(0, 0, 2, 1)$, since the path $0\to 2 \to 4\to 7$ has two level-$3$ nodes (node $4$ and node $7$), and one level-$4$ node (node $2$). }
	\label{fig:dependency}
\end{figure}
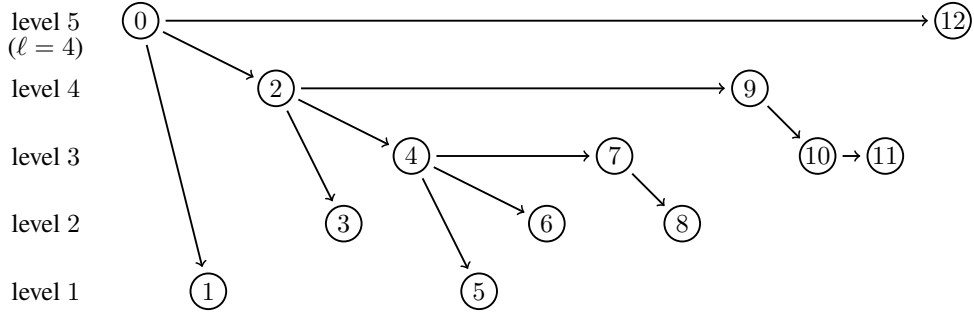

From this definition, we immediately see that the pre-order traversal of the tree is the node sequence $0,1,2,\dots,|w|$ (see example in \cref{fig:dependency}). Since we let tree node $\mu$ correspond to the $\mu$-th vertex $w_\mu$ in $w$, we get a natural correspondence between the extended walk $w$ and the pre-order traversal of the dependency tree.
In \cref{fig:walk-tree}, we illustrate how the  dependency tree is consistent with the recursive structure of  $\walk(s',i)$.

\begin{figure}
    \centering
\begin{tikzpicture}[xscale = 0.7, yscale = 0.5, thick]
  \draw (0,0) -- (1,-1);
  \draw (1,-1) -- (20, -1);
  \draw (1,-1) -- (2, -9);
  \node [draw,circle,fill,minimum size=2.5,inner sep=0pt, outer sep=0pt] at (2,-9) {};
  \node at (2.25,-10) {\small $s_0 = s'$};
  \draw (1.5,-5) -- (3,-5);
  \draw (3,-5) -- (3.5,-9);

  \draw (1,-1) -- (4, -3);
  \node [draw,circle,fill,minimum size=2.5,inner sep=0pt, outer sep=0pt] at (4,-3) {};
  \node at (4.1,-2.15) {\small $x_1$};
  \node at (3.7,-3.65) {\small $\mu_1$};
  \node at (0.4,-1) {\small $\mu_0$};
  \node at (5.3,-2.15) {\small $\cdots$};
  \draw (4,-3) -- (4.75, -9);
  \draw (4.5,-7) -- (6,-7);
  \draw (6,-7) -- (6.25,-9);

  \draw (4,-3) -- (17, -3);
  \draw (6.75,-3) -- (7.5, -9);
  \node [draw,circle,fill,minimum size=2.5,inner sep=0pt, outer sep=0pt] at (6.75,-3) {};
  \node at (6.75,-2.15) {\small $x_j$};
  \node at (6.5,-3.65) {\small $\mu_j$};
  \node [draw,circle,fill,minimum size=2.5,inner sep=0pt, outer sep=0pt] at (7.5,-9) {};
  \node at (7.5,-10) {\small $s_j$};
  \draw (7.25,-7) -- (8.75,-7);
  \draw (8.75,-7) -- (9,-9);

  \draw (6.75,-3) -- (9.75, -5);
  \draw (9.75,-5) -- (12.75,-5);
  \draw (9.75,-5) -- (10.25,-9);
  \draw (10,-7) -- (11.5,-7);
  \draw (11.5,-7) -- (11.75,-9);
  \draw (12.75, -5) -- (13.25, -9);

  \draw (14.25, -3) -- (15,-9);
  \node [draw,circle,fill,minimum size=2.5,inner sep=0pt, outer sep=0pt] at (15,-9) {};
  \node at (15,-10) {\small $s_{j + 1}$};
  \node [draw,circle,fill,minimum size=2.5,inner sep=0pt, outer sep=0pt] at (14.25,-3) {};
  \node at (14.25,-2.15) {\small $x_{j + 1}$};
  \node at (13.6,-3.65) {\small $\mu_{j + 1}$};
  \node at (17.35,-0.2) {\small $x_t$};
  \node at (17,-1.7) {\small $\mu_t$};
  \node at (16,-2.15) {\small $\cdots$};

  \draw [red, draw, dashed] (1.35,-0.5) rectangle (18,-10.75);
  \node [red] at (16, -12) {\small $\mathsf{walk}(s', i, \mu_0)$};
  \node [draw,circle,fill,minimum size=2.5,inner sep=0pt, outer sep=0pt] at (1,-1) {};
  \node [draw,circle,fill,minimum size=2.5,inner sep=0pt, outer sep=0pt] at (17.35,-1) {};

  \draw [blue, draw, dashed] (7.1,-2.6) rectangle (14.55,-9.3);
  \node [blue] at (11, -10) {\small $\mathsf{walk}(s_j, i - 1, \mu_j)$};
\end{tikzpicture}
    \caption{The dependency tree and $\walk(s',i, \mu_0)$. Starting from $s'=s_0$, the function $\walk(s',i,\mu_0)$ first calls $\walk(s_0, i - 1,\mu_0)$ which generates the subtrees of $\mu_0$ of level $\leq i - 1$. Then, $\walk(s_0, i - 1, \mu_0)$ stops at the first vertex $x_1$ of level $\geq i$, where $x_1$ corresponds to node $\mu_1$.
    Then, in $\walk(s',i,\mu_0)$ we find that the level of $\mu_1$ equals $i$, and hence let $s_1 = \nex(\mu_1)$ and explicitly handle the edge $x_1 \rightarrow \nex(\mu_1)$. Then we similarly continue with the recursive call $\walk(s_1, i - 1,\mu_1)$, and so on.
    Finally, $\walk(s',i,\mu_0)$ terminates when it meets a node $\mu_t$ so that $g_i(y) = 0$ at Line~\ref{walk:if} in Algorithm~\ref{algo:extwalk}, implying that $\level(\mu_t)>i$. The vertex $x_t$ will be the last vertex in the walk returned by $\walk(s',i,\mu_0)$.}
    \label{fig:walk-tree}
\end{figure}
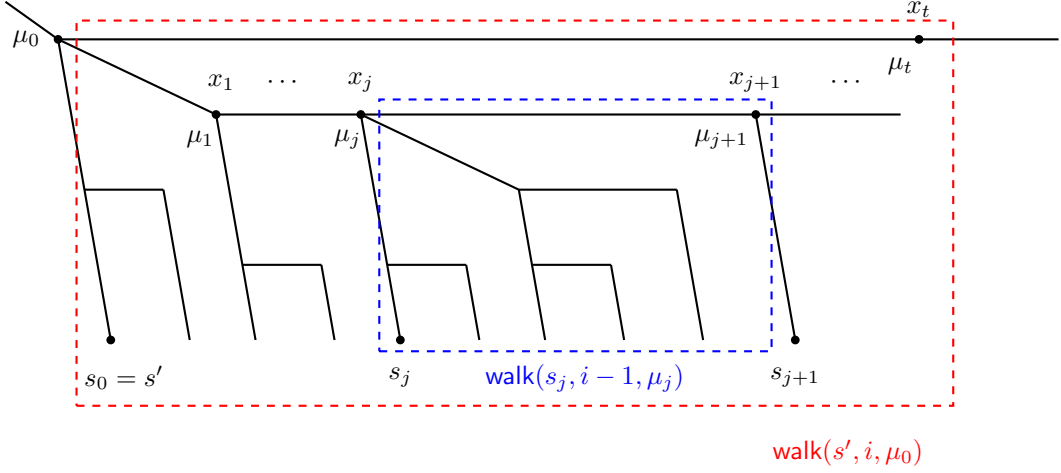

\begin{definition}
We use $p(\mu)$ to denote the \emph{path from the root to the node $\mu$} on the dependency  tree.
\end{definition}
We also observe that the levels of nodes on $p(\mu)$ are non-increasing. Having the tree structure, we introduce a way to \emph{index the nodes}, which will play a crucial role in our proofs in \cref{sec:one,sec:two}.
\begin{definition}[Indexing tree nodes using $\vk$]
We index a node $\vidx$ in the dependency tree by a sequence $\vk = (k_1, k_2, \dots, k_\ell) \in \N^{\ell}$,  where $k_i$ specifies the number of level-$i$ nodes on the path $p(\vidx)$.
We use $\vidx^{\vk}$ or $\vidx[\vk]$ to denote the unique node in the dependency tree indexed by $\vk$. Note that the node $\vidx^{\vk}$ may not exist in the tree.\footnote{One example of node indexing is shown in Figure \ref{fig:dependency}.}
\end{definition}

In our analysis in \cref{sec:one,sec:two}, we will use the strategy of fixing a particular sequence $\vk \in \N^{\ell}$, and letting $\vidx^{\vk}, w_{\vidx^{\vk}}, \a(\vidx^{\vk}),\nex(\vidx^{\vk}),\level(\vidx^{\vk})$ be random variables (provided that $\vidx^{\vk}$ exists in the tree).

We naturally define the ordering of indices as follows.
\begin{definition}
For two sequences $\vk^1, \vk^2 \in \N^\ell$, we say that $\vk^1 < \vk^2$ (or $\vk^1$ is to the left of $\vk^2$), if there is $ i \in [\ell]$ such that, $k^1_i < k^2_i$ and $k^1_t = k^2_t$ for all $t \in \{i+1,\dotsc,\ell\}$. This can be seen as the lexicographical order on the reversed sequences.
\end{definition}

The definition above is justified by the following observation: for two nodes $\mu^{\vk^1},\mu^{\vk^2}$ indexed by $\vk^1,\vk^2$, $\mu^{\vk^1}<\mu^{\vk^2}$ if and only if $\vk^1<\vk^2$.

Finally, we introduce some notation for indexing ancestors in the tree.
\begin{definition}
Given $\vk\in \N^\ell$, $i\in [\ell]$, and $j\in [k_i]$, the $j$-th level-$i$ ancestor of $\vidx^{\vk}$ corresponds to index $\vk^{i,j} = (0, \dots, 0, j, k_{i + 1}, \dots, k_\ell)$, and we use $\vidx^{\vk}_{i,j}$ to denote this ancestor. We also define $\mu^{\vk}_{\ell,0}=0$, namely the root node, and we recursively define $\vidx^{\vk}_{i,0} = \vidx^{\vk}_{i + 1, k_{i + 1}}$,
which is the last ancestor of $\vidx^{\vk}$ with level greater than $i$.
 
 We clarify that, according to our definition, the existence of node $\vidx^{\vk}_{i,j}$ in the dependency tree does not necessarily require the existence of $\vidx^{\vk}$. 
\end{definition}

\subsection{Existence Condition of a Root-to-node Path} \label{sec:path}
As mentioned in \cref{sec:2.5} and \cref{sec:2.6}, our main proof will involve analyzing root-to-node paths on the dependency tree. In this section we will prove several useful lemmas for such analysis.

We first prove a few lemmas on the values of $\rig(\cdot)$ assigned at Line~\ref{walk:rig} in $\walk(s,\ell)$. 
\begin{lemma}
For every node $\mu\in [|w|-1]$, we have $\rig(\mu) = \min_{\alpha>\mu} \{ \alpha \ \vert \ \level(\alpha)\ge \level(\mu) \}$, \ie,
$\rig(\mu)$ is the next node with level at least as high as $\mu$. 
\end{lemma}
\begin{proof}
First, by inspecting the structure of the dependency tree resulted from Algorithm~\ref{algo:extwalk}, we observe that every node $\mu \in [|w|-1]$ must get assigned a value for $\rig(\mu)$ (provided that the walk terminates, which happens with probability $1$ by \cref{lem:terminate}). 

When $\rig(\mu)$ gets assigned at Line~\ref{walk:rig} (where $\mu_j = \mu$), we have $\level(\mu) = i$, and $\rig(\mu) = \mu_{j+1}$. By the definition of $\mu_{j+1}$ at Line~\ref{walk:a}, $\rig(\mu)$ corresponds the last vertex returned by the recursive call $\walk(s_j,i-1)$, and  all nodes $\mu'\in [\mu+1, \rig(\mu)-1]$ must have $\level(\mu')\le i-1$ which were assigned during $\walk(s_j,i-1)$. 

At this point, if $g_i(y)=1$, then we will assign $\level(\rig(\mu)) = \level(\mu_{j+1}) = i$. Otherwise, $g_i(y)=0$, and we will return to the parent level of recursion, where again $\rig(\mu)$ corresponds to the last vertex of the returned walk, and will eventually get assigned a higher level $\level(\rig(\mu))=i'$ during $\walk(s'',i')$ for some $i'>i$. In any case, $\rig(\mu)$ is the first node after $\mu$ that has level at least $\level(\mu)$.
\end{proof}

Moreover, by the definition of $\walk(s, \ell)$, we know $\rig(\mu) = \mu + |\walk(\nex(\mu),\level(\mu)-1)|$. 

For technical reason, we need to extend the definition of $\rig$ as follows. 
\begin{definition}
For $w\in \supp(\bw)$, $\mu \in [|w|-1]$, and $i \in \zeroTon{\ell}$, we define \[\rig_i(\mu) = \min_{\alpha>\mu} \{ \alpha \ \vert \ \level(\alpha)\ge i \},\] namely, the next node with level at least $i$.  
\end{definition}

We also need the following definition to denote the last vertex returned by $\walk(s',i)$.

\newcommand{\last}{\mathsf{last}}
\begin{definition} \label{def:last}
Let $\last(s', i)$ be the last vertex in the sequence returned by $\walk(s',i)$. 
\end{definition}

Now we are ready to prove the following lemma, which determines when $\mu_{i,j}^{\vk}$ exists.

\begin{lemma} \label{lem:existence}
For $(w,g,r) \in \supp(\bw,\bg,\br)$ and $\vk \in \N^{\ell}$, the following hold:
\begin{enumerate}[label=(\alph*)]
	\item Suppose $\mu^{\vk}_{i,j - 1}$ exists. Then, $\mu^{\vk}_{i,j}$ exists if and only if $g_i(\a_i(\rig_i(\mu^{\vk}_{i,j - 1}))) = 1$. \label{item:exist}
	\item $\x_i(\rig_i(\mu^{\vk}_{i,j - 1})) = \last(\nex(\mu^{\vk}_{i,j - 1}), i - 1)$ \label{item:xi}
\end{enumerate}
\end{lemma}
\begin{proof}

We first consider the easier case of $j>1$.
\paragraph*{Case 1: $j>1$.}  
Consider the function call $\walk(s', i, \mu_0)$ such that $\mu_0 < \vidx^{\vk}_{i,j - 1} \leq \mu_0 + |\walk(s', i, \mu_0)|$. It exists and is unique by \Cref{lem:mu0}. Then $\mu_{j - 1} = \vidx^{\vk}_{i,j - 1}$ and $\mu_j = \rig(\mu_{j - 1}) = \rig(\vidx^{\vk}_{i,j - 1})$. The node $\vidx^{\vk}_{i,j}$ exists if and only if $\mu_j$ is of level $i$, which is equivalent to $g_i(\a_i(\mu_j)) = g_i(\a_i(\rig(\vidx^{\vk}_{i,j - 1}))) = 1$ by \cref{lem:giai}. Since in this case $\mu_{i,j - 1}$ is of level $i$, $\rig_i(\vidx^{\vk}_{i,j - 1})$ is the same as $\rig(\vidx^{\vk}_{i,j - 1})$. Thus \ref{item:exist} holds when $j > 1$, 

Moreover, in this case, by Line \ref{walk:a}, $\x_i(\mu_j) = x_j$ where $x_j = \last(s_{j - 1}, i - 1)$ (by Line \ref{walk:w}, \ref{walk:xj}). From Line \ref{walk:Cj}, \ref{walk:level}, we know $s_{j - 1} = \nex(\mu_{j - 1}) = \nex(\vidx^{\vk}_{i,j - 1})$. Together, we get $\x_i(\rig(\vidx^{\vk}_{i,j - 1})) = \x_i(\mu_j) = \last(\nex(\vidx^{\vk}_{i,j - 1}), i - 1)$. This proves \ref{item:xi} when $j > 1$. \\

Now we consider the corner case of $j=1$, which is slightly more technical and makes the $\rig_i$ notation in the lemma statement necessary.
\paragraph*{Case 2: $j= 1$.} In this case, $\vidx^{\vk}_{i,0}$ is of level higher than $ i$. Specifically, by our recursive definition of $\vidx^{\vk}_{i,0}=\vidx^{\vk}_{i + 1, k_{i + 1}}$,
we ultimately have  $\vidx^{\vk}_{i,0} = \vidx^{\vk}_{i', j'}$ where $i' = \min\{i' \in [i + 1, \ell] \ \vert \ k_{i'} > 0\}$ and $j' = k_{i'}$, provided that $i'$ exists; if such $i'$ does not exist, then we set $i'= \ell+1$, and in this case $\mu_{i,0}^{\vk}=0$, namely the root of the dependency tree.

Consider the function call $\walk(\nex(\vidx^{\vk}_{i',j'}), i' - 1, \vidx^{\vk}_{i',j'})$\footnote{In the corner case $i'=\ell+1$, this is $\walk(s,\ell,0)$.}, which  
recursively calls $\walk(\nex(\vidx^{\vk}_{i',j'}), i' - 2, \vidx^{\vk}_{i',j'})$, and so on, until we reach the recursive call $\walk(\nex(\vidx^{\vk}_{i',j'}), i, \vidx^{\vk}_{i',j'})$, in which we have the following:
\begin{enumerate}[label=(\roman*)]
	\item $s_0 = \nex(\vidx^{\vk}_{i',j'}) = \nex(\vidx^{\vk}_{i,0})$. \label{item:s0}
	\item $\mu_1 = \vidx^{\vk}_{i',j'} + |\walk(\nex(\vidx^{\vk}_{i',j'}), i - 1)|$. \label{item:mu1}
\end{enumerate}

Note that $\walk(\nex(\vidx^{\vk}_{i',j'}), i, \vidx^{\vk}_{i',j'})$ first calls $\walk(\nex(\vidx^{\vk}_{i',j'}), i - 1, \vidx^{\vk}_{i',j'})$, which returns a sequence with last vertex corresponding to node $\mu_1$ defined in \ref{item:mu1}. Hence, $\mu_1$ is the first node after $\vidx^{\vk}_{i',j'}$ of level at least $i$.
  Namely, $\mu_1 = \rig_i(\vidx^{\vk}_{i',j'}) = \rig_i(\vidx^{\vk}_{i,0})$.
  Then, note $\vidx^{\vk}_{i,1}$ exists if and only if $\level(\mu_1) = i$, or equivalently, $g_i(\a_i(\mu_1)) = 1$ by \cref{lem:giai}.
Together with $\mu_1 = \rig_i(\vidx^{\vk}_{i,0})$, this proves \ref{item:exist} when $j = 1$.

By \ref{item:s0}, we know $\x_i(\mu_1) = \last(s_0, i - 1) = \last(\nex(\vidx^{\vk}_{i,0}), i - 1)$. Together with $\mu_1 = \rig_i(\vidx^{\vk}_{i,0})$, this proves \ref{item:xi} when $j = 1$.
\end{proof}

Moreover, we remark that the values of $\x_i(\vidx^{\vk}_{i,1}), \x_i(\vidx^{\vk}_{i,2}), \dots, \x_i(\vidx^{\vk}_{i,j - 1}), \x_i(\rig_i(\vidx^{\vk}_{i,j - 1}))$ are enough to uniquely determine $\a_i(\rig_i(\vidx^{\vk}_{i,j - 1}))$. 

\begin{observation} \label{obs:determine-a-from-xi}
Fix $(w, g, r) \in \supp(\bw, \bg, \br)$ and $\vk \in \N^\ell$. Suppose $\vidx^{\vk}_{i,1}, \dots, \vidx^{\vk}_{i,j - 1}$ exist. From $\bar{x}_1 = \x_i(\vidx^{\vk}_{i,1}), \bar{x}_2 = \x_i(\vidx^{\vk}_{i,2}), \dots, \bar{x}_{j - 1} = \ \x_i(\vidx^{\vk}_{i,j - 1}), \bar{x}_j = \x_i(\rig_i(\vidx^{\vk}_{i,j - 1}))$, the value of  $\a_i(\rig_i(\vidx^{\vk}_{i,j - 1}))$ can be uniquely determined as follows:

\begin{enumerate}
	\item Let $j' = \min \{j' \ \vert \ \exists  j'' \text{ s.t. } 1 \leq j'' < j' \leq j,  a_{\bar{x}_{j''}} = a_{\bar{x}_{j'}}\}$. 
	\item If no such $j'$ exists, then $\a_i(\rig_i(\vidx^{\vk}_{i,j - 1})) = a_{\bar{x}_j}$. Otherwise, $\a_i(\rig_i(\vidx^{\vk}_{i,j - 1})) = \star_{j - j'}$.
\end{enumerate}
\end{observation}
\begin{proof}
In the function call $\walk(\nex(\mu_{i,0}^{\vk}), i,\mu_{i,0}^{\vk})$, we have
$\mu_{j - 1} = \vidx^{\vk}_{i, j - 1}$,
and hence $x_{j'}$ in this function call equals $\bar{x}_{j'}$.
Then this observation follows directly from \Cref{obs:determine-a-from-x}. 
\end{proof}


\section{Warm Up Analysis for One Target Vertex}
\label{sec:one}

In this section we prove \cref{lem:warmup-hit-lower-bound}. 

In \cref{sec:one-vertex-notation}, we introduce the important conventions and notation used in this section.
In \cref{sec:one-vertex-main-proof}, we prove \cref{lem:warmup-hit-lower-bound}, assuming the technical lemmas \cref{lem:elong-small}, \cref{lem:good-case}, and  \cref{lem:bad-case}. These technical lemmas will be proved in   \cref{sec:cut-off}, \cref{sec:count_good}, and \cref{sec:count_bad} respectively.

\subsection{Notation}\label{sec:one-vertex-notation}

Throughout this section, we fix $\ell,m,n \in \N$ and $a \in [m]^n$ such that $\ell \le \log n$, and we will always work with (the probability space of) the extended walk $\extwalk_{\ell,m,n,a}$. We use $\bw,\bs,\bg,\br,\bh,\a,\level,\nex$ to denote the corresponding random variables in the extended walk. We also use $\bT$ to denote the dependency tree build on the extended walk.
Note that $\bw,\a,\level,\nex,\bh,\bT$ are all determined by $(\bs,\bg,\br)$. 

For every $i \in \zeroTon{\ell}$, we use $\bg_{\le i}$ to denote the collection $(\bg_{1},\dotsc,\bg_{i})$. Similarly, we use $\br_{\le i}$ to denote the collection $(\br_{1},\dotsc,\br_{i})$. For notational convenience, throughout this section, for $(g_{\le t},r_{\le t}) \in \supp((\bg_{\le t},\br_{\le t}))$, we will always use $g_{\le t} \wedge r_{\le t}$ to denote the event $\left[ \bg_{\le t}=g_{\le t} \wedge \br_{\le t} = r_{\le t} \right]$.

We now set $\tau = 20 \log n \log \log n$. We say $\vk \in \N^\ell$ is \emph{short}, when $k_i \le \tau/4$ for all $i \in [\ell]$. Otherwise, we say $\vk$ is \emph{long}. We use $\Kshort$ to denote the set of all short $\vk \in \N^\ell$, that is, $\Kshort = \zeroTon{\tau/4}^{\ell}$. For $\vk \in \N^\ell$, we let $\calB_{\vk}$ be the collection of all two-dimensional sequences $\vb = \{ b_{i,j} \}_{i \in [\ell], j\in [k_i]}$ with $b_{i,j}\in [n]$ for every $i \in [\ell]$ and $j \in [k_i]$.

Let $\elong$ be the event that the dependency tree has a node whose index is not a short sequence, \ie,
\[
\elong \coloneqq \left[ \text{$\exists \vk \in \N^{\ell}$ s.t. $\max_{i=1}^{\ell} k_i > \tau / 4$ and $\vidx^{\vk}$ exists} \right].
\]
The following lemma, which will be proved in \cref{sec:cut-off}, states that the probability of $\elong$ is small. 
\begin{lemma}\label{lem:elong-small}
	In probability space $(\bw, \bT)$, it holds that
    \[
    \Pr[\elong] \le n\ell/2^{\tau/4}.
    \]
\end{lemma}

Now we formally define the events $\mathcal{F}^{\vk, \vb}_{i, j}$ 
over the probability space of $\extwalk_{\ell,m,n,a}$, which will be used throughout the section.
Let $\vk \in \N^\ell$ be a sequence. For $1\le I\le \ell, 0\le J\le k_I$ and $\vb \in \calB_{\vk}$, we define $\mathcal{F}^{\vk,\vb}_{I,J}$ as the event that the following two hold:
\begin{itemize}
    \item For every $I<i\le \ell$ and every $1\le j\le k_i$, node $\vidx^{\vk}_{i,j}$ exists and $\nex(\vidx^{\vk}_{i,j}) = b_{i,j}$.
    \item For every $1\le j\le J$, node $\vidx^{\vk}_{I,j}$ exists and $\nex(\vidx^{\vk}_{I,j}) = b_{I,j}$.
\end{itemize}
We also use $\mathcal{F}^{\vk, \vb}_i$ as shorthand for $\mathcal{F}^{\vk, \vb}_{i, k_i}$. Specifically, we define $\mathcal{F}^{\vk, \vb}_{\ell + 1}$ to be always true. 
For simplicity, we sometimes use $p(\vk)$ to denote $p(\mu^{\vk})$, \ie, the path from the root to the node $\mu^{\vk}$ on the dependency  tree.

In Table~\ref{table:summary-one-vertex} we provide a summary of all the notations defined and used in this section, as well as the notations for $\textsf{Walk}_{\ell,m,n,a}$ defined in Section~\ref{sec:extendedwalk}. 

\begin{table}[H]\label{table:summary-one-vertex}
	\renewcommand{\arraystretch}{1.2}
	\begin{center}
		\begin{tabular}{|l l|}
			\hline
			\textbf{Notation}  & \textbf{Meaning}  \\ \hline
			$\bw, \bT$ & random variables; the extended walk and the dependency tree\\
			Greek letters ($\alpha,\beta,\gamma$) & tree nodes \\
			$p(\alpha)$ & the path on $T$ from root to $\alpha$\\	
			$\pa(\alpha)$ & the parent of node $\alpha$ on $T$ \\
			$(g_i,r_i)$ & components of hash function in $\calH_{\ell,m,n}$ \\
			$r_{\le i}, g_{\le i}$ & the sequence $(r_{1},\dotsc,r_{i})$ and $(g_{1},\dotsc,g_{i})$ \\
			$r_{\le t} \wedge g_{\le t}$ & the event $\left[ \br_{\le t}=r_{\le t} \wedge \bg_{\le t} = g_{\le t} \right]$ \\
			$\ell$ & number of components (sub-restrictions, levels) in $\calH_{\ell,m,n}$; number of levels; $\ell \le \log n$ \\
			$\tau$ & independence parameter in $\calH_{\ell,m,n}$; $\tau = 20\log n \log\log n$\\
			$\nex(\alpha)$ & $w_{\alpha + 1}$, \ie, next vertex after node $\alpha$\\
			$\a_i(\alpha)$ & the parameter we passed to $g_i, r_i$ trying to determine the next vertex after node $\alpha$\\
			$\level(\alpha)$ & the smallest $i$ such that $g_i(\a_i(\alpha)) = 1$; $\nex(\alpha)$ is determined by $r_{\level(\alpha)}$\\
			$\vidx^{\vk}$ or $\vidx[\vk]$ & the tree node determined by $\vk$ \\
			$\vidx^{\vk}_{i,j}$ & the $j$-th level $i$ ancestor of $\vidx^{\vk}$; equals the parent of $\vidx^{\vk}_{i,1}$ if $j = 0$ (roughtly speaking)\\
			$\Kshort$ & $\zeroTon{\tau/4}^{\ell}$ \\
			$\calB_{\vk}$ & set of two-dimensional sequence $\vb$ with values in $[n]$ and shape $\vk$ \\
			$\rig_i(\alpha)$ & the first node after $\alpha$ of level $\geq i$\\
			$\rig(\alpha)$ & the first node after $\alpha$ of level $\geq \level(\alpha)$ \\
			$\elong$ & the event that $\vidx^{\vk}$ exists for any long $\vk$ \\
			$\mathcal{F}^{\vk, \vb}_{i, j}$ & the event that for all $(i',j')$ before or equal to $(i,j)$, $\mu_{i',j'}^{\vk}$ exists and $\nex(\mu_{i',j'}^{\vk}) = \vb_{i',j'}$\\
			$\star_*$ & $\star_t$ for any $t \in \N$\\
			$p(\vk)$ & the path $p(\mu^{\vk})$ from root to $\mu^{\vk}$ \\
			$\walk(s',i)$ & the level $\leq i$ extended walk from $s'$ \\
			$\last(s',i)$ & the last vertex of $\walk(s',i)$ \\
			\hline
		\end{tabular}
		\caption{Summary of Notation}
	\end{center}
\end{table}

\subsection{Proof of Lemma~\ref{lem:warmup-hit-lower-bound}}\label{sec:one-vertex-main-proof}

\begin{reminder}{\cref{lem:warmup-hit-lower-bound}}
	Suppose $\ell = \log n - \frac{\log F_2(a)}{2} - 10$. For every vertex $v \in [n]$, we have 
	\[
	\Pr_{\bh \getsR \calH_{\ell,m,n},\bs \getsR [n]}[v\in f^*_{a,\bh}(\bs)] = \Theta\left (\frac{1}{\sqrt{F_2(a)}}\right).
	\]
\end{reminder}

Our proof strategy is to utilize Lemma \ref{lem:subset-of-extwalk}. We first count the number of times that vertex $v$ occurs in walk $\bw$, and then subtract the ``bad occurrences'' of $v$, \ie, the occurrences of $v$ in those $\bw$ where there exist $\alpha \not= \beta \in [|\bw|], a_{\bw_{\alpha}} = a_{\bw_{\beta}}$.  

By \Cref{obs:w-nex}, we always have $w_{\mu + 1} = \nex(\mu)$. Thus we have\footnote{We also use $\#A$ to denote the size $|A|$ of a set $A$.} 
$$\#\left\{\mu \ \vert \ \bw_\mu = v \right\} = \#\left\{\vk \in \N^\ell \  \middle\vert \ \nex(\mu^{\vk}) = v\right\}.$$
Note that when we write $\nex(\vidx^{\vk})$, we implicitly require that the node $\vidx^{\vk}$ exists. We will follow this convention in the rest of the paper.  

Since $g_i$ and $r_i$ are only $\tau$-wise independent, our technique can only handle those $\vk \in \Kshort$. Fortunately, the contribution of those $\vk \in \Kshort$ will already be sufficient to provide a good lower bound. Namely, we only count 

$$\#\left\{\vk \in \Kshort \  \middle\vert \ \nex(\mu^{\vk}) = v\right\}.$$

The occurrence $\bw_\mu = v$ is a bad occurence only when there exist\footnote{We use the shorthand $x\neq y\in A$ to mean $x,y\in A$ and $x\neq y$.} $\alpha \not= \beta \in [|\bw|], a_{\bw_{\alpha}} = a_{\bw_{\beta}}$. Hence if we let $\vk'$ be the index of $\alpha - 1$ and $\vk''$ be the index of $\beta - 1$, then $\bw_\mu = v$ is a bad occurrence only when \[\exists \vk' \not= \vk'' \in \N^{\ell}, a_{\nex(\vk')} = a_{\nex(\vk'')}.\]
Note here $\vk'$ and $\vk''$ may not belong in $\Kshort$. But by Lemma \ref{lem:elong-small}, this cannot happen too often. Therefore we can still get our desired bound. 

Formally, we will first prove the following two lemmas.

\begin{lemma}[Counting the number of occurrence of $v$] \label{lem:good-case}
	For every vertex $v \in [n]$, it holds that
	\[
	\frac{2^{\ell}}{n} - \frac{1}{n^3} \leq 
	\E_{\bw, \bT} \Big[\#\left\{\vk \in \Kshort \  \middle\vert \ \nex(\mu^{\vk}) = v\right\} \Big] \leq \frac{2^\ell}{n}.
	\]
\end{lemma}

\begin{lemma}[Upper bounding the bad occurrence of $v$] \label{lem:bad-case}
	For every vertex $v \in [n]$, let $C_v = \#\{i \mid a_i = a_v\}$ be the number of occurrences of $a_v$ in the input $a$. It holds that
	\[
    \E_{\bw, \bT} \left[ \#\left\{\vk \in \Kshort \ \middle\vert \ \nex(\mu^{\vk}) = v \land \exists \vk' \not= \vk'' \in \N^{\ell}, a_{\nex(\mu^{\vk'})} = a_{\nex(\mu^{\vk''})}\right\} \right] \leq 48\frac{8^\ell F_2(a)}{n^3} + 16 \frac{4^\ell C_v}{n^2} + \frac{1}{n^3}.
	\]
\end{lemma}

Based on  \cref{lem:good-case} and \cref{lem:bad-case}, we are ready to prove \cref{lem:warmup-hit-lower-bound}.

\begin{proofof}{Lemma~\ref{lem:warmup-hit-lower-bound}}
	From \cref{obs:w-nex} and \cref{lem:subset-of-extwalk}, we have
	\begin{align*}
	\Pr_{\bh,\bs}[v \in f_{a,\bh}^*(\bs)] 
	\geq&~ \E_{\bw,\bT}\Big[ \#\left\{\vk \in \Kshort \ \middle\vert \ \nex(\mu^{\vk}) = v \land \forall \vk' \not= \vk'' \in \N^{\ell}, a_{\nex(\mu^{\vk'})} \neq a_{\nex(\mu^{\vk''})}\right\} \Big] \\
	=&~ \E_{\bw, \bT} \Big[\#\left\{\vk \in \Kshort \  \middle\vert \ \nex(\mu^{\vk}) = v\right\} \Big] \\
	 & - \E_{\bw, \bT} \left[ \#\left\{\vk \in \Kshort \ \middle\vert \ \nex(\mu^{\vk}) = v \land \exists \vk' \not= \vk'' \in \N^{\ell}, a_{\nex(\mu^{\vk'})} = a_{\nex(\mu^{\vk''})}\right\} \right] \\ 
	\geq&~ \frac{2^\ell}{n} - 48\frac{8^\ell F_2(a)}{n^3} - 16 \frac{4^\ell C_v}{n^2} - \frac{2}{n^3} \tag{Lemma \ref{lem:good-case} and \ref{lem:bad-case}} \\
	=&~ \frac{1}{2^{10}\sqrt{F_2(a)} } - \frac{48}{2^{30}\sqrt{F_2(a)} } - \frac{16 C_v}{2^{20}F_2(a) } - \frac{2}{n^3} \tag{$\ell = \log n - \frac{\log F_2(a)}{2} - 10$}\\
	\ge&~ \Omega\left(\frac{1}{\sqrt{F_2(a)} }\right),
	\end{align*}
	where the last step follows from $C_v= \#\{i \mid a_i = a_v\} \leq \sqrt{F_2(a)} \le n$.

For the other direction, we have 
\begin{align*}
\Pr_{\bh,\bs}[v \in f_{a,\bh}^*(\bs)] 
\leq &~ \Pr_{\bw,\bT}\Big[ \exists \vk \in \Kshort \text{ s.t. }\nex(\mu^{\vk}) = v  \Big] + \Pr[\elong] \\
\leq &~ \E_{\bw,\bT}\Big[ \#\left\{\vk \in \Kshort \ \middle\vert \ \nex(\mu^{\vk}) = v \right\} \Big] + \Pr[\elong] \\
\leq &~\frac{2^\ell}{n} + n\ell/2^{\tau} \tag{Lemma \ref{lem:good-case} and Lemma \ref{lem:elong-small}} \\
\le &~O\left(\frac{1}{\sqrt{F_2(a) }}\right), \tag{$\ell = \log n - \frac{\log F_2(a)}{2} - 10$}
\end{align*}
where the last step follows from the fact that $\ell \leq \log n$, $\tau \geq 20\log n \log \log n$, so that $n \ell / 2^{\tau / 4} \leq \frac{1}{n^3}$. 
\end{proofof}

~\\
The rest of the section is devoted to the proofs of Lemma~\ref{lem:good-case} (Section~\ref{sec:count_good}) and Lemma~\ref{lem:bad-case} (Section~\ref{sec:count_bad}).

\subsection{Counting All Occurrences of a Vertex}\label{sec:count_good}

Now we count the expected number of occurrences of $v$, namely, 
\[
\E \Big[\#\left\{\vk \in \Kshort \  \middle\vert \ \nex(\mu^{\vk}) = v\right\}\Big].
\]
We first enumerate and fix the sequence $\vk$. Then by linearity of expectation, what we want is the summation of the probability
\[
\Pr \Big[\text{$\vidx^{\vk}$ exists} \land \nex(\vidx^{\vk}) = v \Big].
\]
over all $\vk \in \Kshort$. To compute this probability, we will use an  induction between the levels on the dependency tree.

Let us first look at the case within a single level.
Intuitively, when conditioning on $r_{\le i - 1} \land g_{\le i - 1}$ and $\calF^{\vk, \vb}_{i, j - 1}$ (which asserts $\nex(\vidx^{\vk}_{i,j - 1}) = b^{\vk}_{i,j - 1}$) for some fixed $\vb$,  the vertex $x_j$ (in $\walk(s', i)$) is simply the last vertex of $\walk(s_{j - 1}, i -1) = \walk(\nex(\vidx^{\vk}_{i,j-1}),i-1)$, which is determined by $r_{\le i - 1}, g_{\le i-1}$ by \Cref{lem:only-depend}.
Hence $a_{x_j}$ is fixed by $\vb, r_{\le i - 1}, g_{\le i - 1}$, and is independent of $r_i, g_i$. 
Due to the fact $j \leq \tau$ (since $\vk \in \Kshort$) and $g_i, r_i$ are $ \tau$-wise independent, this allow us to argue that $g_i(a_{x_j}) = 1$ with $1/2$ probability and $r_i(a_{x_j}) \getsR [n]$ is \emph{uniformly random}.

Then we can prove $\level(\rig(\mu^{\vk}_{i,j - 1}))=i$ holds with $1/2$ probability, and that $\nex(\mu^{\vk}_{i,j})$ is \emph{uniformly random} (when $\level(\rig(\mu^{\vk}_{i,j - 1}))=i$ and hence $\mu^{\vk}_{i,j}$ exists). This argument is formalized in the following important lemma, which functions as the inductive step in our  whole induction proof.

\begin{lemma} \label{lem:single-level}
	Fix $\vk \in \Kshort$ and $\vb \in \calB_{\vk}$. In probability space $(\bs, \bh, \bw)$, suppose (as induction hypothesis) that the event $\calF^{\vk,\vb}_{i + 1}$ is independent of the joint random variable $(\bg_{\le i},\br_{\le i})$.
	Then, for all $i \in [\ell]$ and $j \in [k_i]$, and all $g_{\le i - 1} \in \supp(\bg_{\le i - 1}),r_{\le i - 1} \in \supp(\br_{\le i - 1})$, it holds that
	\[
	\Pr \left[\calF^{\vk, \vb}_{i,j} \ \middle\vert \ \calF^{\vk, \vb}_{i, j - 1} \wedge g_{\le i - 1} \wedge r_{\le i - 1}\right] = \frac{1}{2n}.
	\]
\end{lemma} 

\begin{proof}
	Fix $i \in [\ell]$ and $j \in [k_i]$ and let $g_{\le i - 1} \in \supp(\bg_{\le i -1})$ and $r_{\le i - 1} \in \supp(\br_{\le i -1})$. We let $\event_{\le i - 1}$ denote the event $\Big[g_{\le i - 1} \wedge r_{\le i - 1}\Big]$ for convenience. Our goal is to show that
	\[
	\Pr\left[\calF^{\vk, \vb}_{i,j} \ \middle\vert \ \calF^{\vk, \vb}_{i, j - 1} \wedge \event_{\le i - 1} \right] = \frac{1}{2n}.
	\]

	By \Cref{lem:existence} \ref{item:exist}, $\vidx^{\vk}_{i,j}$ exists if and only if $g_i(\a_i(\rig_i(\vidx^{\vk}_{i, j - 1}))) = 1$. Then let us inspect how $\a_i(\rig_i(\vidx^{\vk}_{i, j - 1}))$ is determined. Let $\walk(s', i, \mu_0)$ be the function call such that $\mu_0 < \rig_i(\vidx^{\vk}_{i, j - 1}) \leq \mu_0 + |\walk(s', i, \mu_0)|$ which exists and is unique by \Cref{lem:mu0}. In this function call, $\mu_j = \rig_i(\vidx^{\vk}_{i, j - 1})$. 
	
	Conditioning on $\calF^{\vk, \vb}_{i,j - 1}$, we know that $\vidx^{\vk}_{i,1}, \dots, \vidx^{\vk}_{i,j - 1}$ exist.  
	Then by \Cref{obs:determine-a-from-xi}, $\a_i(\mu_j)$ is determined by $\x_i(\vidx^{\vk}_{i,1}), \dots, \x_i(\vidx^{\vk}_{i,j - 1}), \x_i(\rig_i(\vidx^{\vk}_{i,j - 1}))$, namely, by all $\x_i(\rig_i(\vidx^{\vk}_{i,j' - 1}))$ for $j' \in [1, j]$. 

	By \Cref{lem:existence} \ref{item:xi}, $\x_i(\rig_i(\vidx^{\vk}_{i,j' - 1})) = \last(\nex(\vidx^{\vk}_{i,j' - 1}), i - 1)$. Conditioning on $\calF^{\vk, \vb}_{i,j - 1}$, for $j'\le j$ we have $\nex(\vidx^{\vk}_{i,j'-1}) = b_{i, j' - 1}$  (for the corner case of $j' = 1$, we recursively define $b_{i,0} = b_{i - 1, k_{i - 1}}$).

	Moreover, by \Cref{lem:only-depend}, $\last(\cdot, i - 1)$ only depends on $\br_{\leq i - 1}$ and $\bg_{\leq i - 1}$. Therefore, conditioning on $\calF^{\vk, \vb}_{i,j - 1} \land \event_{\le i - 1}$, each $\x_i(\rig_i(\vidx^{\vk}_{i,j' - 1}))$ ($1 \leq j' \leq j$) is uniquely determined from $\vk, \vb, g_{\leq i - 1}, r_{\leq i - 1}$. Hence, they also uniquely determine  $\a_i(\rig_i(\vidx^{\vk}_{i,j' - 1}))$  for $1 \leq j' \leq j$. 

	 Recall that $\mu^{\vk}_{i,j}$ exists if and only if $\bg_i(\a_i(\mu_j)) = 1$.  If $\mu^{\vk}_{i,j}$ exists, we know $\mu^{\vk}_{i,j} = \mu_j$ and $\nex(\mu_j) = \br_i(\a_i(\mu_j))$. So our goal is to show that $\bg_i(\a_i(\mu_j))=1 \land \br_i(\a_i(\mu_j)) = b_{i,j}$ indeed happens with $\frac{1}{2n}$ probability. We prove this using the fact that $\br_i(\cdot)$ and $\bg_i(\cdot)$ are $\tau$-wise independent, and our condition $\calF^{\vk, \vb}_{i, j - 1}$ has only revealed the values of $\br_i(\cdot),\bg_i(\cdot)$ at no more than $j\le \tau$ many points.

	By $\calF^{\vk, \vb}_{i,j - 1}$, for all $j' \in [j - 1]$, we have $\br_i(\a(\mu_{j'})) = \nex(\mu_{j'}) = b_{i,j'}$ and $\bg_i(\a(\mu_{j'}))=1$. For $r_i\in \supp(\br_{i})$ and $g_i \in \supp(\bg_{i})$, define an predicate 
	\[
	P(r_i, g_i) \coloneqq \left[ \forall j' \in [j-1], g_i(\a_i(\mu_{j'})) = 1 \land r_i(\a_i(\mu_{j'})) = b_{i,j'} \right].
	\]

	By definition, $\calF^{\vk, \vb}_{i,j - 1}$ is true if and only if $\calF^{\vk, \vb}_{i + 1}$ is true and for each $j' \leq j - 1$, $\level(\mu_{j'}) = i \land \nex(\mu_{j'}) = b_{i,j'}$, which is equivalent to $\bg_i(\a_i(\mu_{j'})) = 1 \land \br_i(\a_i(\mu_{j'})) = b_{i,j'}$. Thus, 

	\[
	\calF^{\vk, \vb}_{i,j - 1} \land \event_{\leq i - 1}  = \calF^{\vk, \vb}_{i + 1} \land \event_{\leq i - 1} \land  P(\br_i, \bg_i).
	\]
	
	We have shown that each $\a_i(\mu_{j'})$ is uniquely determined by $r_{\le i -1}, g_{\le i -1}$ and $\vk, \vb$, so $P(r_i, g_i)$ is a predicate of $r_i, g_i$ only, and hence $P(\br_i, \bg_i)$ only depends on the randomness of $\br_i, \bg_i$.
	To prevent confusion, we stress that $P$ is defined using the particular $r_{\leq i - 1}, g_{\leq i - 1}$ that we have fixed at the beginning of the proof, and does not depend on the random variables $\br_{\leq i - 1}, \bg_{\leq i - 1}$. 

	From our assumption that $\calF^{\vk,\vb}_{i + 1}$ is independent of the joint random variable $(\br_{\le i}, \bg_{\le i})$, we know that the event $\calF^{\vk,\vb}_{i + 1} \land \event_{\le i - 1}$ is independent of $\br_i, \bg_i$. Since $P(r_i, g_i)$ is a predicate of $r_i, g_i$, we know $\calF^{\vk,\vb}_{i + 1} \land \event_{\le i - 1}$ is still independent of $\br_i, \bg_i$ when conditioning on $P(\br_i, \bg_i)$. Namely, since \[\Pr\left[\br_i = r_i, \bg_i = g_i \ \middle\vert \ \calF^{\vk,\vb}_{i + 1} \land \event_{\le i - 1} \right] = \Pr\left[\br_i = r_i, \bg_i = g_i\right],\]
	and $P(r_i, g_i)$ is a predicate of only $r_i, g_i$, we know that 
	$$\Pr\left[\br_i = r_i, \bg_i = g_i \ \middle\vert \ \calF^{\vk,\vb}_{i + 1} \land \event_{\le i - 1}, P(r_i, g_i) \right] = \Pr\left[\br_i = r_i, \bg_i = g_i \ \middle\vert \ P(r_i, g_i) \right],$$

	Hence, we can derive 
	\begin{align*}
	&\Pr\left[\calF^{\vk, \vb}_{i,j} \ \middle\vert \ \calF^{\vk, \vb}_{i, j - 1} \land \event_{\le i - 1}\right] \\ 
	= \ &\Pr\left[\bg_i(\a_i(\mu_j)) = 1, \br_i(\a_i(\mu_j)) = b_{i,j} \ \middle\vert \ \calF^{\vk, \vb}_{i, j - 1} \land \event_{\le i - 1}\right]\\
	= \ &\Pr\left[\bg_{i}(\a_i(\mu_j)) = 1, \br_{i}(\a_i(\mu_j)) = b_{i,j}\ \middle\vert \ \calF^{\vk, \vb}_{i + 1} \wedge \event_{\le i - 1} \wedge  P(\br_i, \bg_i)\right] \\
	= \ &\Pr\left[\bg_{i}(\a_i(\mu_j)) = 1, \br_{i}(\a_i(\mu_j)) = b_{i,j}\ \middle\vert \ P(\br_i, \bg_i)\right] \\
	= \ &\frac{1}{2n},
	\end{align*}
	where the last step follows from the fact that $\bg_i(\cdot)$ and $\br_i(\cdot)$ are $\tau$-wise independent, $j \le \tau$, and $\a_i(\mu_j)$ is different from all other $\a_i(\mu_{j'})$ by 
	definition (see the ``in particular'' part of \cref{obs:determine-a-from-x})
	 and uniquely determined by $\vb, r_{\leq i - 1}, g_{\leq i - 1}$. 
\end{proof}

Iterative application of Lemma \ref{lem:single-level} leads to the following corollary. 

\begin{cor} \label{cor:single-level-whole}
	Fix $\vk \in \Kshort$ and $\vb \in \calB_{\vk}$. 
	In probability space $(\bs, \bh, \bw)$, suppose (as induction hypothesis) that the event $\calF^{\vk,\vb}_{i + 1}$ is independent of the joint random variable $(\bg_{\le i},\br_{\le i})$. Then, for all $i \in [\ell]$, $g_{\le i - 1} \in \supp(\bg_{\le i - 1})$ and $r_{\le i - 1} \in \supp(\br_{\le i - 1})$, it holds that
	\[
	\Pr\left[\calF^{\vk, \vb}_i \ \middle\vert \ \calF^{\vk, \vb}_{i + 1} \wedge g_{\le i - 1} \wedge r_{\le i - 1} \right] = \frac{2^{-k_i}}{n^{k_i}}.
	\]
\end{cor}
\begin{proof}
	We let $\event_{\le i - 1}$ to denote the event $\Big[g_{\le i - 1} \wedge r_{\le i - 1}\Big]$ for convenience. From the definition of the $\calF^{\vk,\vb}_{i,j}$ and Lemma~\ref{lem:single-level}, we have
	\begin{align*}
	\Pr\left[\calF^{\vk, \vb}_i \ \middle\vert \ \calF^{\vk, \vb}_{i + 1} \wedge \event_{\le i - 1} \right] &= \Pr\left[\calF^{\vk, \vb}_{i,k_i} \ \middle\vert \ \calF^{\vk, \vb}_{i,0} \wedge \event_{\le i - 1} \right]\\
	&=\prod_{j=1}^{k_i} \Pr\left[\calF^{\vk, \vb}_{i,j} \ \middle\vert \ \calF^{\vk, \vb}_{i,j - 1} \wedge \event_{\le i - 1} \right] \\
	&=\frac{2^{-k_i}}{n^{k_i}}. \qedhere
	\end{align*}
\end{proof}

Then we iteratively use Corollary \ref{cor:single-level-whole} to obtain the probability of $\calF_1^{\vk, \vb}$. 

\begin{lemma}\label{lem:induction}
	Fix $\vk \in \Kshort$ and $\vb \in \calB_{\vk}$. For all $i \in [\ell+1]$, letting $g_{\le i - 1} \in \supp(\bg_{\le i - 1})$ and $r_{\le i - 1} \in \supp(\br_{\le i - 1})$. In probability space $(\bs, \bh, \bw)$, we have
	\[
	\Pr\left[\calF^{\vk, \vb}_i \ \middle\vert \ g_{\le i - 1} \wedge r_{\le i - 1}\right] = \frac{2^{-\sum_{j=i}^\ell k_j}}{n^{\sum_{j=i}^\ell k_j}}.
	\]
\end{lemma}
\begin{proof}
	We prove this by induction. For the base case $i = \ell + 1$, $\calF^{\vk,\vb}_i$ is always true by definition, and hence $\Pr\left[\calF^{\vk, \vb}_i \ \middle\vert\ g_{\le i - 1} \wedge r_{\le i - 1}\right] = 1$. 
	
	Suppose the induction hypothesis holds for $i + 1$. Note this implies 
	$\Pr\left[\calF_{i+1}^{\vk, \vb} \middle\vert g_{\le i} \wedge r_{\le i}\right] = \Pr\left[\calF_{i+1}^{\vk, \vb}\right]$ for every possible $g_{\le i}$ and $r_{\le i}$, meaning that $\calF_{i+1}^{\vk, \vb}$ is independent of the joint variable $(\bg_{\le i}, \br_{\le i})$. Hence, it satisfies the premise of Corollary \ref{cor:single-level-whole}. 
	
	From the definition of $\calF^{\vk, \vb}_i$, we have
	\[
	\Pr\left[ \calF^{\vk, \vb}_i \ \middle\vert \ g_{\le i - 1} \wedge r_{\le i - 1}\right] = \Pr\left[ \calF^{\vk, \vb}_i \ \middle\vert \  \calF^{\vk, \vb}_{i + 1} \land g_{\le i - 1} \wedge r_{\le i - 1}\right] \cdot \Pr\left[ \calF^{\vk, \vb}_{i + 1}\  \middle \vert \  g_{\le i - 1} \wedge r_{\le i - 1}\right].
	\]
	
	From induction hypothesis, we have
	\[
	\Pr\left[ \calF^{\vk, \vb}_{i + 1} \ \middle\vert \ g_{\le i - 1} \wedge r_{\le i - 1}\right] = \E_{(g_i,r_i) \getsR (\bg_i,\br_i)}\left[ \Pr\left[ \calF^{\vk, \vb}_{i + 1} \ \middle\vert \ g_{i} \wedge r_{i} \ \land  \ g_{\le i - 1} \wedge r_{\le i - 1} \right] \right]= \frac{2^{-\sum_{j=i + 1}^\ell k_j}}{n^{\sum_{j=i + 1}^\ell k_j}}.
	\]
	
	From \cref{cor:single-level-whole}, it follows that
	\[
	\Pr\left[ \calF^{\vk, \vb}_i \ \middle\vert \  \calF^{\vk, \vb}_{i + 1} \land g_{\le i - 1} \wedge r_{\le i - 1}\right] = \frac{2^{-k_i}}{n^{k_i}}.
	\]
	
	Putting everything together, we have
	\[
	\Pr\left[ \calF^{\vk, \vb}_i \ \middle\vert \ g_{\le i - 1} \wedge r_{\le i - 1}\right] = \frac{2^{-k_i}}{n^{k_i}} \cdot \frac{2^{-\sum_{j=i + 1}^\ell k_j}}{n^{\sum_{j=i + 1}^\ell k_j}} = \frac{2^{-\sum_{j=i}^\ell k_j}}{n^{\sum_{j=i}^\ell k_j}}. \qedhere
	\]
\end{proof}

Finally, we are ready to count the number of occurrences of $v$, and prove Lemma~\ref{lem:good-case}.

\begin{reminder}{Lemma~\ref{lem:good-case}}
	For every vertex $v \in [n]$, it holds that
	\[
	\frac{2^{\ell}}{n} - \frac{1}{n^3} \leq 
	\E_{\bw, \bT} \Big[\#\left\{\vk \in \Kshort \  \middle\vert \ \nex(\mu^{\vk}) = v\right\} \Big] \leq \frac{2^\ell}{n}.
	\]
\end{reminder}
\begin{proof}
	For each $\vk \in \N^\ell$ and $\vb \in \calB_{\vk}$, we say $(\vk,\vb)$ is \emph{good}, if $\nex(\vidx^{\vk}) = v$ holds in the event $\calF^{\vk,\vb}_1$. Recall that $\vk$ is \emph{short}, when $k_i \le \tau/4$ for all $i \in [\ell]$; and otherwise $\vk$ is \emph{long}.
	
	We first break the expectation into the sum of contribution of all possible index $\vk$ and $\vb$. By linearity of expectation, we have
	\[
	\E_{\bw, \bT} \Big[\#\left\{\vk \in \Kshort \  \middle\vert \ \nex(\mu^{\vk}) = v\right\} \Big] = \sum_{\vk \in \Kshort} \sum_{\vb \in \calB_{\vk}} \Pr\left[\calF^{\vk,\vb}_{1} \wedge \text{$(\vk,\vb)$ is good}\right].
	\]

	Then, by Lemma \ref{lem:induction}, for all $\vk\in \Kshort, \vb \in \calB_{\vk}$, it holds that
	\[
	\Pr\left[\calF^{\vk, \vb}_1 \right] = \frac{2^{-\sum_{j=1}^\ell k_j}}{n^{\sum_{j=1}^\ell k_j}}.
	\]
	
	There are $n^{\sum_{j = 1}^\ell k_j}$ many sequences $\vb \in \calB_{\vk}$, and $n^{(\sum_{j = 1}^\ell k_j) - 1}$ of them satisfy $\nex(\vidx^{\vk}) = v$. Thus for all short $\vk$,
	\[
	\sum_{\vb \in \calB_{\vk}} \Pr\left[\calF^{\vk,\vb}_{1} \wedge \text{$(\vk,\vb)$ is good}\right] = \frac{2^{-\sum_{j = 1}^\ell k_j}}{n}.
	\]
	Then, we have
	\begin{align*}
	&\sum_{\vk \in \Kshort} \sum_{\vb \in \calB_{\vk}} \Pr\left[\calF^{\vk,\vb}_{1} \wedge \text{$(\vk,\vb)$ is good}\right]\\
	=&\sum_{\vk \in \Kshort} \frac{2^{-(k_1 + k_2 + \dots + k_\ell)}}{n}\\
	=&\left( \sum_{i=0}^{\tau/4} 2^{-i} \right)^\ell /n\\
	=&\left( 2 - 2^{-\tau/4} \right)^\ell /n.
	\end{align*}
	
	Putting everything together, we have
	\[
	\frac{2^{\ell}}{n} \ge \E_{\bw, \bT} \Big[\#\left\{\vk \in \Kshort \  \middle\vert \ \nex(\mu^{\vk}) = v\right\} \Big] = \left( 2 - 2^{-\tau/4} \right)^\ell /n \ge \frac{2^\ell}{n} - \frac{1}{n^3},
	\]
	which completes the proof.
\end{proof}


\newcommand{\oldlevel}{\textsf{old-level}}
\subsection{Counting All Bad Occurrences of a Vertex}\label{sec:count_bad}

Now we move on to prove Lemma \ref{lem:bad-case} which upper bounds the number of ``bad'' occurrences of $v$ as follows,
\[
    \E_{\bw, \bT} \left[ \#\left\{\vk \in \Kshort \ \middle\vert \ \nex(\mu^{\vk}) = v \land \exists \vk' \not= \vk'' \in \N^{\ell}, a_{\nex(\mu^{\vk'})} = a_{\nex(\mu^{\vk''})}\right\} \right] \leq 12\frac{8^\ell F_2(a)}{n^3} + 8 \frac{4^\ell C_v}{n^2} + \frac{1}{n^3}.
	\]

We first apply a union bound on $\vk'$ and $\vk''$, and 
similar to the proof of Lemma \ref{lem:good-case}, we enumerate three sequences $\vk^1 \in \Kshort, \vk^2 \in \N^{\ell}, \vk^3 \in \N^{\ell}$, and sum up the contribution of \[\Pr_{\bw, \bT}\left[\vidx^{\vk^1}, \vidx^{\vk^2}, \vidx^{\vk^3} \text{ exist } \land \nex(\vidx^{\vk^1}) = v \land a_{\nex(\vidx^{\vk^2})} = a_{\nex(\vidx^{\vk^3})}\right]\]
over all $\vk^1,\vk^2,\vk^3$.
Intuitively, one would expect this to give the desired upper bound.
However there is some subtlety due to possible \emph{collisions} between the   paths from root to the nodes
$\vidx^{\vk^1}, \vidx^{\vk^2}, \vidx^{\vk^3} $, which we formally define below. Recall that $p(\mu)$ denote the path on $T$ from root to node $\mu$.

\begin{definition} \label{def:collision}
 We say a pair of nodes $(\alpha, \beta)$ is a \emph{collision between two paths} $p_1$ and $p_2$, if $\alpha \in p_1 \setminus p_2, \beta \in p_2 \setminus p_1$, and $(\a(\alpha),\level(\alpha)) =(\a(\beta),\level(\beta))$ (where $\a(\cdot)$ is defined in Definition \ref{def:short-hand-for-a}). The \emph{level} of a collision $(\alpha, \beta)$ is defined as $\level(\alpha)$ (which equals $\level(\beta)$). 
\end{definition}

For a collision $(\alpha,\beta)$,  the values of $\nex(\alpha)$ and $\nex(\beta)$ are actually the same random variable $\br_{\level(\alpha)}(\a(\alpha))$.
When collisions appear between $p(\vidx^{\vk^1}), p(\vidx^{\vk^2})$, and $p(\vidx^{\vk^3})$, the correlations caused by these collisions would make our analysis difficult.

\subsubsection{Structure of Bad Cases}\label{sec:bad_structure}

\renewcommand{\i}{\phi}
\newcommand{\bi}{\bm{\phi}}

To get around such difficulty, we need to exploit the combinatorial structure of the dependency trees.
We let $\i = \vidx^{\vk^1} + 1$, $\bar{\alpha} = \vidx^{\vk^2} + 1$ and $\bar{\beta} = \vidx^{\vk^3} + 1$. Then by Observation \ref{obs:w-nex}, $\nex(\vidx^{\vk^1}) = v \land a_{\nex(\vidx^{\vk^2})} = a_{\nex(\vidx^{\vk^3})}$ is equivalent to $w_{\i} = v \land a_{w_{\bar \alpha}} = a_{w_{\bar \beta}}$. 

The following lemma asserts that, fixing any  $(w, T) \in \supp\left((\bw, \bT)\right)$, whenever there exist collisions between $p(\bar \alpha - 1)$ and $p(\bar \beta - 1)$, there always exists another pair of $\alpha$ and $\beta$ satisfying $a_{w_\alpha} = a_{w_{\beta}}$ as well, such that there are no such problematic collisions between paths $p(\alpha - 1)$ and $p(\beta - 1)$.

\begin{lemma} \label{lem:structure}
Fix $(w,T) \in \supp((\bw,\bT))$. Suppose $T$ contains two nodes $ \bar \alpha \not= \bar \beta$ such that $a_{w_{\bar \alpha}} = a_{w_{\bar \beta}}$. Then for every node $\i$ on $T$, there must exist nodes $\alpha$ and $\beta$, such that $\alpha \not= \beta,a_{w_{\alpha}} = a_{w_{\beta}}$, and for any two of the paths $p(\i - 1), p(\alpha - 1), p(\beta - 1)$ there is no collision (as in Definition \ref{def:collision}) between them. 
\end{lemma}

The intuition for Lemma \ref{lem:structure} is as following: We call a pair $(\bar \alpha, \bar \beta)$ a duplicate if and only if $\bar \alpha < \bar \beta$ and $a_{w_{\bar \alpha}} = a_{w_{\bar \beta}}$. Take the duplicate $(\bar \alpha, \bar \beta)$ with the smallest $\bar \beta$. Since it is the first pair of duplicate on $w$ (in the sense that $w_1, w_2, \dots, w_{\bar \beta - 1}$ contains no duplicate), we can prove there is no collision between $p(\bar \alpha - 1)$ and $p(\bar \beta - 1)$. So the only problem left is the possible collisions with $p(\i - 1)$. 

Suppose there is a collision between $p(\bar \alpha - 1)$ and $p(\i - 1)$, namely there are two nodes $\pi_1 \in p(\bar \alpha - 1) \setminus p(\bar \beta - 1)$ and $\pi_2 \in  p(\bar \beta - 1) \setminus p(\bar \alpha - 1)$ that $(\a(\pi_1), \level(\pi_1)) = (\a(\pi_2), \level(\pi_2))$. Then, intuitively, by the way how subtrees of $\pi_1$ and $\pi_2$ are generated, these two subtrees should be the same. Note $\bar \alpha - 1$ is in the subtree of $\pi_1$. We can move it to the corresponding node in the subtree of $\pi_2$. See Figure \ref{fig:moving}. After moving, the original collision $\pi_1, \pi_2$ between path $p(\bar \alpha - 1)$ and $p(\i - 1)$ becomes a common ancestor $\pi_2$. Thus we can eliminate collisions with $p(\i - 1)$ in this way, and find desired $\alpha, \beta$. 

\begin{figure}
    \centering
\scalebox{0.8}{
\begin{tikzpicture}[node distance={23mm}, thick, vertex/.style = {draw, thick, circle,inner sep=0pt,minimum size=26pt,outer sep=0pt}]  
  \node [vertex] (rt1) {};
  \node [vertex, below of = rt1, xshift = -40, yshift = -30] (pi1) {$\pi_1$} edge [-] (rt1);
  \node [vertex, below of = rt1, xshift = 40, yshift = -30] (pi2) {$\pi_2$} edge [-] (rt1);
  \node [vertex, below of = pi1, xshift = -30, yshift = -20] (a) {$\bar \alpha - 1$} edge [-] (pi1);
  \node [vertex, below of = pi2, xshift = 30, yshift = -20] (b) {$\i - 1$} edge [-] (pi2);
  \node [right of = pi2] {\huge $\Longrightarrow$};

  \node [vertex, right of = rt1,xshift = 140] (r2) {};
  \node [dashed, vertex, below of = r2, xshift = -40, yshift = -30] (p1) {$\pi_1$} edge [-, dashed] (r2);
  \node [vertex, below of = r2, xshift = 40, yshift = -30] (p2) {$\pi_2$} edge [-] (r2);
  \node [red, vertex, below of = p2, xshift = -30, yshift = -20] (a) {$\bar \alpha - 1$} edge [-,red] (p2);
  \node [vertex, below of = p2, xshift = 30, yshift = -20] (b) {$\i - 1$} edge [-] (p2);
\end{tikzpicture}
}
    \caption{Moving $\bar \alpha - 1$ to the corresponding node in the subtree of $\pi_2$}
    \label{fig:moving}
\end{figure}
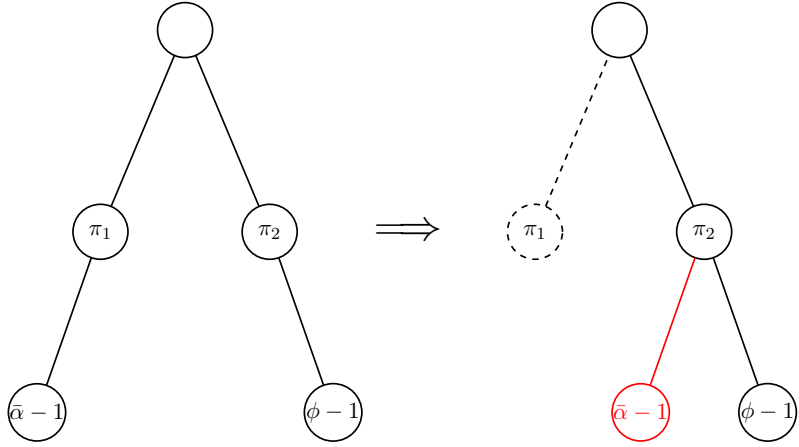

Before proving Lemma \ref{lem:structure}, we first prove a few technical lemmas which formalize the intuition that the subtrees of $\pi_1$ and $\pi_2$ should be the same. We will need a few extra notations.  

\newcommand{\suc}{\mathsf{suc}}

\begin{definition}\label{def:extention}
We define $\suc(\eta)$ to be the set of nodes within the subtree of $\eta$ that has the same level as $\mu$. Then we define $p^*(\mu) = \cup_{\eta \in p(\mu)} \suc(\eta)$ which is an extension of $p(\mu)$. (See Figure \ref{fig:ext})
\end{definition}

\begin{figure}
    \centering
\begin{tikzpicture}[scale = 0.37, very thick] 
  \draw[red] (0,0) -- (2,-3);
  \draw[red] (2,-3) -- (4,-3);
  \draw[red] (4,-3) -- (6, -6);
  \draw[red] (6, -6) -- (10, -6);
  \draw[red] (10, -6) -- (12, -9);
  \draw (12, -9) -- (14, -12);

  \draw[red] (12, -9) -- (14, -9);
  \draw[red] (14, -9) -- (15, -10.5);
  \draw[red] (15, -10.5) -- (16, -10.5);
  \draw[red] (16, -10.5) -- (17, -12);

  \draw[red] (14, -9) -- (17, -9);
  \draw (17, -9) -- (19, -12);

  \draw[red] (10, -6) -- (19, -6);
  \draw (19, -6) -- (23, -12);

  \draw[red] (4,-3) -- (22, -3);
  \draw (22, -3) -- (24, -6);
  \draw (24, -6) -- (28, -12);

  \draw (24, -6) -- (28, -6);
  \draw (28, -6) -- (32, -12);

  \node [draw,circle,red,fill=red,minimum size=4,inner sep=0pt, outer sep=0pt] at (17,-12) {};
  \node at (17,-13) {$\mu$};
\end{tikzpicture}
    \caption{The dependency tree $T$ and extension $p^*(\mu)$ (marked in red)}
    \label{fig:ext}
\end{figure}
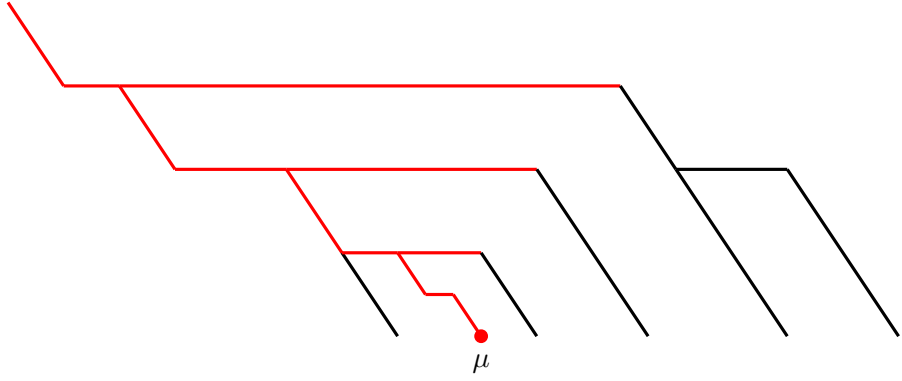

For simplicity, we let $p(\vk)$ denote $p(\mu^{\vk})$, \ie,the path from root to the node with index $\vk$. Similarly, we use $p^*(\vk)$ to denote $p^*(\mu^{\vk})$.

\begin{lemma} \label{lem:level_i_collision}
Fix sequence $\vk^1$ and $w, T \in \supp(\bw, \bT)$. Let $\sigma$ be a node in $T$ and $\vk$ be the index of $\sigma$. Suppose there is a collision between $p(\sigma)$ and $p^*(\vk^1)$. 

Let $i$ be the lowest level that contains such collision. Then there must be a node $\pi_3 \in p^*(\vk^1)$ of level $i$ such that $\a(\pi_3) = \a(\mu^{\vk}_{i, k_i})$. 
\end{lemma}
\begin{proof}
Let $(\pi_1, \pi_2)$ be the collision of level $i$. Formally, $\pi_1 = \mu^{\vk}_{i,j}$ is the node on $p(\sigma)$ with smallest $i$ such that $\exists \pi_2 \in p^*(\vk^1)$ satisfying $(a(\pi_2), \level(\pi_2)) = (a(\pi_1), i)$. If there are multiple such $(\pi_1, \pi_2)$, we choose the one that minimizes $j$. 

Recall the definition of $\rig$ (which can be found in Table \ref{table:summary-one-vertex}). We prove the existence of such $\pi_3$ by the following algorithm. Note that $\level(\pi_1) = \level(\pi_2) = i$. 

\begin{algorithm}[H] \label{algo:find-pi3}
$\alpha \gets \pi_1, \beta \gets \pi_2$

\For{$t \in [j + 1, k_i]$}{
	\eIf{$\a_i(\rig(\alpha)) = \a_i(\rig(\beta))$} {
		$\alpha \gets \rig(\alpha), \beta \gets \rig(\beta)$ \label{code:eq}
	}
	{
		Let $\eta \in p^*(\vk^1)$ be the node such that $\a_i(\eta) = \a_i(\rig(\alpha))$ and $\level(\eta) = i$. \label{code:exist}

		$\alpha \gets \rig(\alpha), \beta \gets \eta$	
	}
}
$\pi_3 \gets \beta$
\caption{Finding $\pi_3$}
\end{algorithm}

To prove the lemma, we need to prove following two facts about the algorithm. 

\begin{enumerate}
	\item Throughout Algorithm \ref{algo:find-pi3}, $\level(\alpha) = \level(\beta) = i$. \label{fact:level-i}
	\item At Line \ref{code:exist}, such node $\eta$ always exists.  \label{fact:exist}
\end{enumerate}

We first show that these two facts are sufficient to prove the lemma. Suppose these facts are true. By Fact~\ref{fact:exist}, the algorithm will not abort by error. When the algorithm terminates, since initially $\alpha = \pi_1 = \mu^{\vk}_{i,j}$ and $\alpha \gets \rig(\alpha)$ is executed for $k_i - j$ steps, we have $\alpha = \mu^{\vk}_{i,k_i}$. Observe that the invariant $\a_i(\alpha) = \a_i(\beta)$ is preserved through the algorithm. Then by Fact \ref{fact:level-i} and the invariant, we know $(\a_i(\beta), \level(\beta)) = (\a_i(\alpha), i)$. Moreover, since $\beta$ only move to $\rig(\beta)$ and $\eta \in p^*(\vk^1)$, we know $\beta \in p^*(\vk^1)$. As a result, we can let $\pi_3 \gets \beta$ and it satisfies the lemma requirements. 

\paragraph*{Proof of Fact \ref{fact:level-i}.} Since initially $\alpha = \pi_1 = \mu^{\vk}_{i,j}$, and for each $t$ it moves to $\rig(\alpha)$, we know that at the beginning of the each loop $\alpha = \mu^{\vk^1}_{i,t - 1}$, and at the end of the each loop $\alpha = \rig(\mu^{\vk^1}_{i,t - 1}) = \mu^{\vk^1}_{i,t}$ (since $t \in [j + 1, k_i]$). We know $\alpha$ is always of level $i$. 

For $\beta$, each time it either move to a level $i$ node $\eta$ or move to $\rig(\beta)$. Since at Line \ref{code:eq}, $g_i(\a_i(\rig(\beta))) = g_i(\a_i(\rig(\alpha))) = 1$, and by our discussion above $\level(\rig(\alpha)) = i$, we know $\rig(\beta)$ must also be of level $i$. 

\paragraph*{Proof of Fact \ref{fact:exist}.} Now we inspect how $\a_i(\rig(\alpha))$ and $\a_i(\rig(\beta))$ are determined. For $\alpha = \mu^{\vk}_{i,t - 1}$, in the function $\walk$, from $s_{t - 1} = \nex(\alpha)$, we first perform $\walk(s_{t - 1}, i - 1)$. $x_t$ is the last vertex of this walk, and therefore $x_t = w_{\rig(\alpha)}$. Then $\a_i(\rig(\alpha))$ is determined by $a_{x_t} = a_{w_{\rig(\alpha)}}$ and $C_{t - 1}$. If $\mathsf{star} = \text{true}$, $\a_i(\rig(\alpha)) = \star_t$ ($t > 0$). Otherwise, either $\a_i(\rig(\alpha)) = a_{w_{\rig(\alpha)}}$, or it equals $\star_0$ because $a_{w_{\rig(\alpha)}} \in C_{t - 1}$. Note $a_{w_{\rig(\alpha)}} \in C_{t - 1}$ if and only if there is a node $\eta \in p(\vk)$ of level $i$ before $\alpha$, such that $\a_i(\eta) = a_{w_{\rig(\alpha)}}$. $\a_i(\beta)$ is determined in the same way. In conclusion, we have following propositions:
\threeprops{
$x_t = w_{\rig(\alpha)}$ is the last vertex of $\walk(\nex(\alpha), i - 1)$.
}{prop:determine-x}{
$\a_i(\rig(\alpha)) = \star_0$ only when $\exists \eta \in p(\vk)$ of level $i$, such that $\a_i(\eta) = a_{w_{\rig(\alpha)}}, \eta < \alpha$.
}{prop:become-star}{
These two also holds for $\beta$ by replacing $\alpha$ with $\beta$, and $p(\vk)$ with $p^*(\vk^1)$. 
}{prop:beta}

Observe that we kept the invariant $\a_i(\alpha) = \a_i(\beta)$ throughout the algorithm. Therefore, since $\nex(\alpha) = r_i(\a_i(\alpha))$ and $\nex(\beta) = r_i(\a_i(\beta))$. We know $\nex(\alpha) = \nex(\beta)$. Thus both $\walk(\nex(\alpha), i - 1)$ and $\walk(\nex(\beta), i - 1)$ are the same walk. Hence, by \eqref{prop:determine-x}, $w_{\rig(\alpha)} = w_{\rig(\beta)}$. If both $\a_i(\rig(\alpha)) \neq \star_*$ and $\a_i(\rig(\beta)) \neq \star_*$ holds, we would have $\a_i(\rig(\alpha)) = a_{w_{\rig(\alpha)}} = a_{w_{\rig(\beta)}} = \a_i(\rig(\beta))$. 

Note if $\a_i(\rig(\alpha)) = \star_{t'}$ for $t' > 0$, this would imply $\a_i(\alpha) = \star_{t' - 1}$. Since $\a_i(\alpha) = \a_i(\beta)$, we must have $\a_i(\rig(\beta)) = \a_i(\rig(\beta))$ then. The same holds for $\beta$. 

Thus the only possibility of entering Line \ref{code:exist} is when exactly one of $\a_i(\rig(\alpha)) = \star_0$ and $\a_i(\rig(\beta)) = \star_0$ happens. Note here $\alpha \in p(\vk)$ and $\beta \in p^*(\vk^1)$, and $w_{\rig(\alpha)} = w_{\rig(\beta)}$ implies $a_{w_{\rig(\alpha)}} = a_{w_{\rig(\beta)}}$. 

\highlight{Case 1: $\a_i(\rig(\alpha)) \not= \star_*$ and $\a_i(\rig(\beta)) = \star_0$.} By \eqref{prop:become-star} and \eqref{prop:beta}, this happens only when there is a node $\eta \in p^*(\vk^1)$ such that $(\a_i(\eta), \level(\eta)) = (a_{w_{\rig(\beta)}}, i)$. On the other side, since $\a_i(\rig(\alpha)) \not= \star_*$, $\a_i(\rig(\alpha)) = a_{w_{\rig(\alpha)}} = a_{w_{\rig(\beta)}} = \a_i(\eta)$. Therefore such $\eta$ exists.
	
\highlight{Case 2: $\a_i(\rig(\alpha)) = \star_0$ and $\a_i(\rig(\beta)) \not= \star_*$.} By \eqref{prop:become-star}, there must be $\eta' \in p(\vk)$ before $\rig(\alpha)$ that $(\a_i(\eta'), \level(\eta')) = (a_{w_{\rig(\alpha)}}, i) = (\a_i(\rig(\beta)), i)$. Note $(\eta', \beta)$ is a collision of level $i$. We then prove $\eta' < \pi_1$ to reach a contradiction with the minimality of $\pi_1$.

Since $\a_i(\rig(\alpha)) = \star_0$, we have not entered Line \ref{code:exist} in Case 2 before since $\alpha$ is always moving to $\rig(\alpha)$, and $\a_i(\rig(\alpha))$ becomes $\star_0$ for at most once. If we have entered Line \ref{code:exist} before in Case 1, there must be a node $\beta' \in p^*(\vk)$ such that $\a_i(\rig(\beta')) = \star_0$. Since $g_i(\rig(\beta')) = g_i(\rig(\alpha)) = 1$, $\rig(\beta')$ is of level $i$. We let $\eta \gets \rig(\beta')$. This proves the existence of $\eta$. 

Otherwise, we have not entered Line \ref{code:exist} before. We know that $(\a_i(\pi_1), \dots, \a_i(\alpha)) = (\a_i(\pi_2), \dots, \a_i(\beta))$, which implies $\eta' < \pi_1$, since otherwise $\a_i(\rig(\beta))$ would also be $\star_0$. However, the way we pick $\pi_1 = \mu^{\vk}_{i,j}$ minimizes $j$. Therefore, there cannot be such $\eta'$, a contradiction. Hence desired $\eta$ always exists.
\end{proof}

Suppose a path $p(\sigma)$ has collision with $p^*(\vk^1)$, and let the lowest such collision be of level $i$. Then Lemma \ref{lem:level_i_collision} states that there exists $\pi_3$ on $p^*(\vk^1)$ such that $\a_i(\pi_3)$ equals that of the last level $i$ node on $p(\vk)$, \ie,$\a_i(\mu^{\vk^1}_{i,k_i})$. This leads to the following corollary saying there must be a node $\sigma'$ within the subtree of $\pi_3$ such that the path from $\pi_3$ to $\sigma'$ is the same as that from $\mu^{\vk}_{i,k_i}$ to $\sigma$. 

We first define what does it mean for two paths to be the same. Roughly speaking, two paths are the same if they have the same shape and same $\a(\mu)$ at each node $\mu$. 

\begin{definition} \label{def:same-below}
We say two paths $p(\vk^1)$ and $p(\vk^2)$ are the same below level $i$ if
\begin{itemize}
	\item $\forall$ $1 \leq j < i$, $k^1_j = k^2_j$. 
	\item $\forall$ $1 \leq j < i, 1 \leq t \leq k^1_j$, $\a(\vidx^{\vk^1}_{j,t}) = \a(\vidx^{\vk^2}_{j,t})$.
\end{itemize}
\end{definition}

\begin{cor} \label{cor:moving-1-path}
Fix sequence $\vk^1$ and $w, T \in \supp(\bw, \bT)$. Let $\sigma$ be a node in $T$ and $\vk$ be the index of $\sigma$. If there is a collision between $p(\sigma)$ and $p^*(\vk^1)$, let $\pi_3$ be defined as Lemma \ref{lem:level_i_collision}. Otherwise, let $\pi_3$ be the lowest common ancestor of them, \ie,the last node on $p(\sigma)$ that is also on $p^*(\vk^1)$. Let $i = \level(\pi_3)$. There must be a descendant $\sigma'$ of $\pi_3$ such that the followings hold:
\fourprops{
$\pi_3$ is the last level $i$ node on $p(\sigma')$. 
}{prop:last-node}{
$p(\sigma')$ is the same as $p(\sigma)$ below level $i$
}{prop:same-below-i}{
$(\a(\sigma'), \level(\sigma')) = (\a(\sigma), \level(\sigma))$	
}{prop:same-a}{
There is no collision between $p(\sigma')$ and $p^*(\vk^1)$.
}{prop:no-collision}
\end{cor}
\begin{proof}
When there is no collision between $p(\sigma)$ and $p^*(\vk^1)$, we let $\sigma' = \sigma$ which satisfies all the requirements. When there is collision between $p(\sigma)$ and $p^*(\vk^1)$, let $\vk$ be the index of $\sigma$ and $\vk^3$ be the index of $\pi_3$. We define $\vk'$ to be \[k'_j = \begin{cases} k^3_j & j \geq i \\ k_j & j < i \end{cases}.\] Let $\sigma'$ be the node $\vidx^{\vk'}$ which a priori may not exist. We will later prove that $\sigma'$ does exist. 

If $\sigma' = \mu^{\vk'}$ exists, since $ k'_j = k^3_j$ for all $j \geq \level(\pi_3)$, it must be a descendant of $\pi_3$, and $\pi_3$ is the last level $i$ node on $p(\sigma')$. 
This proves \eqref{prop:last-node}.

Let $\gamma = \vidx^{\vk}_{i,k_i}$. By the definition of $\walk$, $w_{\gamma + 1}, \dots, w_{\rig(\gamma)}$ is generated by $\walk(\nex(\gamma),i - 1)$. Similarly, $w_{\pi_3 + 1}, \dots, w_{\rig(\pi_3)}$ is generated by the return value of $\walk(\nex(\pi_3), i - 1)$. 

Since $\a_i(\gamma) = \a_i(\pi_3)$ (by \Cref{lem:level_i_collision}), we know $\nex(\gamma) = r_i(\a_i(\gamma)) = r_i(\a_i(\pi_3)) = \nex(\pi_3)$. This implies that $\walk(\nex(\gamma),i - 1)$ is the same walk as $\walk(\nex(\pi_3), i - 1)$. Since $\vk_j$ and $\vk'_j$ are the same for all $j < i$, and their last level $i$ nodes are $\gamma$ and $\pi_3$ respectively. $\vidx^{\vk'}$ exists if and only if $\vidx^{\vk}$ exists. Since $\sigma = \vidx^{\vk}$ exists, we know $\sigma' = \vidx^{\vk'}$ also exists. Besides, below level $i$, these two paths ($p(\sigma)$ and $p(\sigma')$) are generated by the same walk (since $\walk(\nex(\gamma),i - 1)$ is the same as $\walk(\nex(\pi_3), i - 1)$), so they are the same below level $i$. This proves \eqref{prop:same-below-i}. Note if $i > \level(\sigma)$, \eqref{prop:same-below-i} implies  \eqref{prop:same-a}. If $i = \level(\sigma)$, \eqref{prop:last-node} implies $\sigma = \gamma$ and $\pi_3 = \sigma'$. Therefore \eqref{prop:same-a} also holds. 

Finally, we prove \eqref{prop:no-collision}. By definition of $\pi_3$ in \Cref{lem:level_i_collision}, there is no collision between $p(\sigma)$ and $p^*(\vk^1)$. Since $p(\sigma')$ is the same as $p(\sigma)$ below level $i$, there is also no collision between $p(\sigma')$ and $p^*(\vk^1)$ below level $i$. Moreover, by \eqref{prop:last-node}, the last level $i$ node of $p(\sigma')$ is $\pi_3 \in p^*(\vk^1)$. Thus there is also no collision between them above or equal level $i$. This concludes the proof of \eqref{prop:no-collision}. 
\end{proof}

Roughly speaking, we need to apply \Cref{cor:moving-1-path} to $p(\bar \alpha - 1)$ and $p(\bar \beta - 1)$ respectively. Before applying \Cref{cor:moving-1-path}, there is no collision between them. We need to prove that this property is preserved after applying \Cref{cor:moving-1-path}. This gives the following corollary which will be used in the proof of \Cref{lem:structure}. We will later apply it with $\mu = \bar \alpha - 1$ and $\eta = \bar \beta - 1$. 

\begin{cor}\label{cor:moving-2-paths}
Fix sequence $\vk^1$ and $(w,T) \in \supp((\bw,\bT))$. Suppose there are two nodes $\mu$ and $\eta$ such that there is no collision between $p(\mu)$ and $p(\eta)$. Then there must exist two nodes $\mu', \eta'$ such that $(\a(\mu), \level(\mu)) = (\a(\mu'), \level(\mu'))$ and $(\a(\eta), \level(\eta)) = (\a(\eta'), \level(\eta'))$, and there is no collision between $p(\mu'), p(\eta')$ and $p^*(\vk^1)$. 
\end{cor}
\begin{proof}
We apply \Cref{cor:moving-1-path} to $\sigma = \mu$ (resp. $\eta$) and get $\sigma' = \mu'$ (resp. $\eta'$) and $\pi^{\mu}_3$ (resp. $\pi^{\eta}_3$). Let $i_{\mu} = \level(\pi^{\mu}_3)$ and $i_{\eta} = \level(\pi^{\eta}_3)$. 

Our proof is by contradiction. Suppose there are $\alpha' \in p(\mu') \setminus p(\eta')$ and $\beta' \in p(\eta') \setminus p(\mu')$ such that $(\a(\alpha'), \level(\alpha')) = (\a(\beta'), \level(\beta'))$. If $\level(\alpha') \geq i_{\mu}$, we know $\alpha' \in p(\pi^{\mu}_3)$ by \Cref{cor:moving-1-path} \eqref{prop:last-node}. However, since $\pi^{\mu}_3 \in p^*(\vk^1)$, we know $\alpha' \in p^*(\vk^1)$. This contradicts \Cref{cor:moving-1-path} \eqref{prop:no-collision} saying that there is no collision between $p(\eta')$ and $p^*(\vk^1)$. Same contradiction follows if $\level(\beta') \geq i_{\eta}$. 

Therefore, we must have $\level(\alpha') < i_{\mu}$ and $\level(\beta') < i_{\eta}$. By \Cref{cor:moving-1-path} \eqref{prop:same-below-i}, there is a node $\alpha \in p(\mu)$ corresponding to $\alpha'$ such that $(\a_i(\alpha), \level(\alpha)) = (\a_i(\alpha'), \level(\alpha'))$. Similarly, there is also a node $\beta \in p(\eta)$ corresponding to $\beta'$ such that $(\a_i(\beta), \level(\beta)) = (\a_i(\beta'), \level(\beta'))$. 

If $\alpha \not= \beta$, we reach a contradiction with the assumption that there is no collision between $p(\mu)$ and $p(\eta)$. 

If $\alpha = \beta$, let $i$ be the level of the lowest common ancestor of $\mu$ and $\eta$ (\ie,the last node on $p(\mu)$ that is also on $p(\eta)$). Since $\alpha \in p(\mu)$ and $\beta \in p(\eta)$, we know $i \leq \level(\alpha)$. On the other hand, $\level(\alpha) = \level(\alpha') < i_{\mu}$ and $\level(\beta) = \level(\beta') < i_{\eta}$. Thus $i_{\mu}$ and $i_{\eta}$ are strictly higher than $i$. Since $p(\mu), p(\eta)$ overlaps above level $i$, and $i_{\mu}$ is the lowest level that contains a collision (or common ancestor when there is no collision) between $p(\mu)$ and $p^*(\vk^1)$ (resp. $i_{\eta}$ is the lowest level that contains a collision (or common ancestor when there is no collision) between $p(\eta)$ and $p^*(\vk^1)$), we know $i_{\mu} = i_{\eta}$. Moreover, if we let $\gamma$ be the last level $i_{\mu}$ node on $p(\mu)$, it is also the last node on $p(\eta)$ of that level (again because these two paths overlaps above level $i$). 

Hence since $\pi_3^\mu, \pi_3^\eta \in p^*(\vk^1)$ satisfies $(\a(\pi_3^\mu), \level(\pi_3^\mu)) = (\a(\gamma), \level(\gamma)) = (\a(\pi_3^\eta), \level(\pi_3^\eta))$ (by \Cref{lem:level_i_collision}). We must have $\pi_3^\mu = \pi_3^\eta$. Thus from $\alpha = \beta$ and \Cref{cor:moving-1-path} \eqref{prop:same-below-i}, we know $\alpha' = \beta'$. Then it cannot be a collision. 
\end{proof}

We are finally ready to prove \Cref{lem:structure}. 

\begin{reminder}{\Cref{lem:structure}}
Fix $(w,T) \in \supp((\bw,\bT))$. Suppose $T$ contains two nodes $ \bar \alpha \not= \bar \beta$ such that $a_{w_{\bar \alpha}} = a_{w_{\bar \beta}}$. Then for every node $\i$ on $T$, there must exist nodes $\alpha$ and $\beta$, such that $\alpha \not= \beta,a_{w_{\alpha}} = a_{w_{\beta}}$, and for any two of the paths $p(\i - 1), p(\alpha - 1), p(\beta - 1)$ there is no collision (as in Definition \ref{def:collision}) between them. 
\end{reminder}
\vspace{-1.5em}
\begin{proof}[Proof of \Cref{lem:structure}]

Let $(\bar \alpha, \bar \beta) = \arg\min_{(\alpha, \beta)} \{\beta \ \vert \ a_{w_{\alpha}} = a_{w_{\beta}}, \alpha < \beta\}$, \ie, it is the first duplicate in $w$ in the sense that $\bar \beta$ is minimized. We first prove that there is no node $\gamma \in p(\bar \alpha - 1)$ such that $\a(\gamma) = \star_*$. Suppose there is. Take the first such $\gamma$, then we have $\a(\gamma) = \star_0$, and there must be a node $\gamma'$ before $\gamma$ such that $a_{w_{\gamma'}} = a_{w_{\gamma}}$. Since $\gamma < \bar \alpha < \bar \beta$, this contradicts the fact that $(\bar  \alpha, \bar \beta)$ is the first duplicate. 

Suppose there is $\pi_1 \in p(\bar \alpha - 1) \setminus p(\bar \beta - 1)$ and $\pi_2 \in p(\bar \beta - 1) \setminus p(\bar \alpha - 1)$ such that $\a(\pi_1) = \a(\pi_2)$. By the discussion above, we know $\a(\pi_1)$ cannot be $\star_*$. Therefore $a_{w_{\pi_1}} = a_{w_{\pi_2}}$. Again since $\pi_2 < \bar \beta$, this contradicts the fact that $(\bar  \alpha, \bar \beta)$ is the first duplicate. 

Thus there is no collision between $p(\bar \alpha - 1)$ and $p(\bar \beta - 1)$. The remaining problem is that there might be collisions between $p(\i-1)$ and these two paths.

Let $\vk^1$ be the index of node $\i - 1$. Then we apply \Cref{cor:moving-2-paths} with $\mu = \bar \alpha - 1$ and $\eta = \bar \beta - 1$. We get $\mu'$ and $\eta'$ such that there is no collision between $p(\mu'), p(\eta'), p(\vk^1)$. Besides by \Cref{cor:moving-2-paths} \eqref{prop:same-a}, $(\a(\mu'), \level(\mu')) = (\a(\mu), \level(\mu))$ and $(\a(\eta'), \level(\eta')) = (\a(\eta), \level(\eta))$. Therefore $w_{\mu' + 1} = r_{\level_{\mu'}} (\a(\mu')) = r_{\level_{\mu}} (\a(\mu)) =w_{\mu + 1}$ (\ie,$w_{\alpha} = w_{\bar \alpha}$). Similarly, $w_{\eta' + 1} = w_{\eta + 1}$ (\ie,$w_{\beta} = w_{\bar \beta}$).

We can let simply $\alpha = \mu' + 1$ and $\beta = \eta' + 1$. Then we have $a_{w_{\alpha}} = a_{w_{\bar \alpha}} = a_{w_{\bar \beta}} = a_{w_{\beta}}$. Since there is no collision between $p(\bar \alpha - 1), p(\bar \beta - 1)$ and $\bar \alpha < \bar \beta$, we know $(\a(\mu), \level(\mu)) \neq (\a(\eta), \level(\eta))$. Hence, $(\a(\mu'), \level(\mu')) \neq (\a(\eta'), \level(\eta'))$. We get $\mu' \neq \eta'$ and thereby $\alpha \neq \beta$. 
\end{proof}

\subsubsection{Upper Bounding the Bad Occurrences} \label{bad-case}

Now we are finally ready to upper bound those bad occurrences of $v$. In order to handle such case with three paths, we need to extend our definition of $\vidx^{\vk}_{i,j}$ and $\calF^{\vk,\vb}_{i,j}$ to the union of $t$ paths. 

\paragraph{Notation.} We first define the following notion of ancestor.

\begin{definition} \label{def:ancestor}
Given two vectors $\vk^1, \vk^2 \in \N^\ell$, we say that $\vk^1$ is an ancestor of $\vk^2$, if the following holds
\[
\exists i \in [\ell] \text{ s.t. } [\forall t \in [i + 1, \ell], k^1_t = k^2_t] \land [0 < k^1_i \leq k^2_i] \land [\forall t \in [i - 1], k^1_t = 0] \land [\vk^1 \neq \vk^2].
\]
\end{definition}

Note when both $\vidx^{\vk^1}$ and $\vidx^{\vk^2}$ exists, $\vidx^{\vk^1}$ is an ancestor of $\vidx^{\vk^2}$ on $T$ if and only if $\vk^1$ is an ancestor of $\vk^2$. 

In Section \ref{sec:construction}, we defined $\vk^{i,j} = (0,\dots,0,j,k_{i+1},\dots,k_{\ell})$, and we also defined how to compare two indices. We say $\vk^1 < \vk^2$ if and only if 
\[
\exists i \in [\ell] \text{ s.t. } [\forall t \in [i + 1,\ell], k^1_t = k^2_t] \land [k^1_i < k^2_i]
\]

Recall $p(v)$ denotes the path on $T$ from root to node $v$. We use $P^{\vk}$ to denote the set of all indices that are either $\vk$ or an ancestor of $\vk$. Note that if $\vidx^{\vk}$ exists on $T$, then the set of indices of all nodes in $p(v)$ is exactly $P^{\vk}$.

Let $\vK = \{\vk^1, \vk^2, \dots, \vk^t\}$ be a set of $t$ paths. We define $\vK_{i,j}$ to be the $j$-th index in level $i$ in the union of $P^{\vk^1}, P^{\vk^2}, \dots, P^{\vk^t}$. Namely, we take out all the distinct indices in $\{(\vk^1)^{i, j}\}_{j \in [\vk^1_i]} \cup \{(\vk^2)^{i, j}\}_{j \in [\vk^2_i]} \cup \dots \cup \{(\vk^t)^{i, j}\}_{j \in [\vk^t_i]}$ and sort them in increasing order (by the comparsion we defined above). $\vK_{i,j}$ is the $j$-th one among them. Note it is uniquely determined by $\vK, i$, and $j$. 

For an index $\vk$, we use $\pa(\vk)$ be the largest $\vk'$ such that $\vk'$ is an ancestor of $\vk$. We also use $\pa_{i,j}^{\vK}$ to denote $\pa(\vK_{i,j})$. Therefore, suppose $\vk = \vK_{i,j}$, by this definition $\vk' = (0, \dots, 0, k_i - 1, k_{i + 1}, \dots, k_{\ell})$. Note here $k_i \ge 1$ since $\vK_{i,j}$ is of level $i$. 

We make the following observation about existence of $\mu^{\vK}_{i,j}$. 

\begin{observation} \label{obs:existence-K}
Fix $(w,g) \in (\bw, \bg)$ and $\vK$. Recall $\vidx[\vk]$ has the same meaning as $\vidx^{\vk}$. The following holds:
\begin{enumerate}[label=(\alph*)]
	\item Suppose $\vidxb{\pa^{\vK}_{i,j}}$ exists, $\vidx^{\vK}_{i,j}$ exists if and only if $g_i\left(\a_i\left(\rig_i\left(\vidxb{\pa^{\vK}_{i,j}}\right)\right)\right) = 1$. \label{item:exist-K}
	\item $\x_i\left(\rig_i\left(\vidxb{\pa^{\vK}_{i,j}}\right)\right) = \last\left(\nex\left(\vidxb{\pa^{\vK}_{i,j}}\right), i - 1\right)$ \label{item:xi-K}
	\item Suppose $\vidx^{\vK}_{i,1}, \dots, \vidx^{\vK}_{i,j - 1}$ and $\vidxb{\pa^{\vK}_{i,j}}$ exist. From $\tilde{x}_1 = \x_i\left(\vidx^{\vK}_{i,1}\right), \tilde{x}_2 = \x_i\left(\vidx^{\vK}_{i,2}\right), \dots, \tilde{x}_{j - 1} = \x_i\left(\vidx^{\vK}_{i,j - 1}\right), \tilde{x}_j = \x_i\left(\rig_i\left(\vidxb{\pa^{\vK}_{i,j}}\right)\right)$, one can uniquely determine $\a_i\left(\rig_i\left(\vidxb{\pa^{\vK}_{i,j}}\right)\right)$. 
	
	Specifically, let $j_0 = \max\{j_0 \ \vert \ \vK_{i,j_0} \text{ is an ancestor of } \vK_{i,j}\}$. Then let $j' = \min \{j' \ \vert \ \exists  j'' \text{ s.t. }  j_0 \leq j'' < j' \leq j,  a_{\bar{x}_{j''}} = a_{\bar{x}_{j'}}\}$. If such $j'$ does not exist,  then $\a_i\left(\rig_i\left(\vidxb{\pa^{\vK}_{i,j}}\right)\right) = a_{\bar{x}_j}$. Otherwise, $\a_i\left(\rig_i\left(\vidxb{\pa^{\vK}_{i,j}}\right)\right) = \star_{j - j'}$.
	\label{item:determine-a-from-xi-K}
	\item If $\vidx^{\vK}_{i,j}$ exists, then $\vidx^{\vK}_{i,j} = \rig_i\left(\vidxb{\pa^{\vK}_{i,j}}\right)$. \label{item:exist-equal}
\end{enumerate}
\end{observation}
\begin{proof}
Suppose $\vK_{i,j} \in P^{\vk}$ for $\vk \in \vK$ and $\vK_{i,j} = \vk^{i,t} = \vk'$. 

By definition of $\vidx^{\vk}_{i,t - 1}$, its index is $(0, \dots, 0, k'_i - 1, k'_{i + 1}, \dots, k'_{\ell})$. This is exactly $\pa^{\vK}_{i,j}$. Thus \ref{item:exist-K} and \ref{item:xi-K} follows directly from \Cref{lem:existence}. 

Note $j_0 = j - t + 1$. If we let $\bar{x}_{j'} = \tilde{x}_{j' + j_0 - 1} = \x_i(\vidx^{\vk}_{i,j'})$ for $j' \in [1,t-1]$ and $\bar{x}_t = \tilde{x}_j$. We can directly apply \Cref{obs:determine-a-from-xi} and get \ref{item:determine-a-from-xi-K}. For \ref{item:exist-equal}, note if $\vidx^{\vK}_{i,j}$ exists, then $\vidx^{\vK}_{i,j} = \vidx^{\vk}_{i,t}$, and $\rig_i\left(\vidxb{\pa^{\vK}_{i,j}}\right) = \rig_i\left(\vidx^{\vk}_{i,t - 1}\right) = \vidx^{\vk}_{i,t}$.
\end{proof}
For the ease of notation, we let $\zeta^{i,j} = \rig_i\left(\vidxb{\pa^{\vK}_{i,j}}\right)$. 

Let $K_i$ be the number of distinct level $i$ indices in the union of these $t$ paths $P^{\vk^1}, P^{\vk^2}, \dots, P^{\vk^t}$.
We use $\calB_{\vK}$ to denote the collection of all two-dimensional sequence $\vb = \{ b_{i,j} \}_{i \in [\ell], j\in [K_i]}$ with $b_{i,j}\in [n]$ for every $1\le i\le \ell, 1\le j\le K_i$.

Fix $\vK$ and $\vb \in \calB_{\vK}$. For $1 \leq I \leq \ell, 0 \leq J \leq K_I$, we define $\calF^{\vK, \vb}_{I, J}$ to be the following event:

\begin{itemize}
	\item For every $I < i \leq \ell$ and every $1 \leq j \leq K_i$, node $\mu^{\vK}_{i,j}$ exists and $\nex(\mu^{\vK}_{i,j}) = b_{i,j}$.
	\item For every $1 \leq j \leq J$, node $\mu^{\vK}_{I,j}$ exists and $\nex(\mu^{\vK}_{I,j}) = b_{I, j}$. 
\end{itemize}

Same as before, we use $\calF^{\vK, \vb}_i$ to denote $\calF^{\vK, \vb}_{i, K_i}$. Specifically, $\calF^{\vK, \vb}_{\ell+1}$ is always true.

Furthermore, we define the following event capturing collisions between these $t$ paths in $\vK$. Let $\tilde{\calG}^{\vK}_{i,j}$ denote the event that for all $1 \leq t_1 < t_2 \leq j$, $\a_i(\vidx^{\vK}_{i, t_1}) \not= \a_i(\vidx^{\vK}_{i, t_2})$. We also define $\tilde{\calG}^{\vK}_i = \tilde{\calG}^{\vK}_{i, K_i} \land \tilde{\calG}^{\vK}_{i + 1}$ and let $\tilde{\calG}^{\vK}_{\ell + 1}$ to be always true. Note $\tilde{\calG}^{\vK}_i$, $\calF^{\vK, \vb}_i$, and $\calF^{\vK, \vb}_{i, K_i}$ involves different levels while $\tilde{\calG}^{\vK}_{i,j}$ only involves level $i$. 

$\tilde{\calG}^{\vK}_{i,j}$ captures level $i$ collisions between these paths for the following reason: By how our extended walk assign $\a_i(\mu)$ to each node $\mu$, for level $i$ nodes $\alpha, \beta$ on the same path, we always have $\a_i(\alpha) \not= \a_i(\beta)$. Therefore if $\a_i(\vidx^{\vK}_{i, t_1}) \not= \a_i(\vidx^{\vK}_{i, t_2})$ they must belong to different paths. 

We summarize the notation in Table~\ref{table:Summary-2}. 

\begin{table}[H] \label{table:Summary-2}
    \renewcommand{\arraystretch}{1.3}
    \begin{center}
    	\begin{tabular}{|l l|}
        \hline
        \textbf{Notation}  & \textbf{Meaning}  \\ \hline
		$\vidx^{\vk}$ or $\vidx[\vk]$ & the tree node determined by $\vk$ \\
		$\ell$ & number of components (sub-restrictions, levels) in $\calH_{\ell,m,n}$; number of levels; $\ell \le \log n$ \\
		$\tau$ & independence parameter in $\calH_{\ell,m,n}$; $\tau = 20\log n \log\log n$\\
		$(r_i,g_i)$ & components of hash function in $\calH_{\ell,m,n}$ \\
		$r_{\le i}, g_{\le i}$ & the sequence $(r_{1},\dotsc,r_{i})$ and $(g_{1},\dotsc,g_{i})$ \\
		$r_{\le t} \wedge g_{\le t}$ & the event $\left[ \br_{\le t}=r_{\le t} \wedge \bg_{\le t} = g_{\le t} \right]$ \\
    	$\vK$ & a set of indices; subset of $\N^{\ell}$ \\
    	$P^{\vk}$ & the set of indices of all ancestors of $\vk$\\			
    	$\Kshort$ & $\zeroTon{\tau/4}^{\ell}$ \\
    	$\calB_{\vk}$ & set of two-dimensional sequence $\vb$ with values in $[n]$ and shape $\vk$ \\
    	$\vK_{i,j}$ & the $j$-th index among all level $i$ indices in $\cup_{\vk \in \vK} P^{\vk}$\\
    	$\vidx^{\vK}_{i,j}$ & the node $\mu^{\vK_{i,j}}$ \\
    	$K_i$ & the number of distinct level $i$ indices in $\cup_{\vk \in \vK} P^{\vk}$\\
    	$\pa(\mu)$ & a node; the parent of $\mu$\\
    	$\pa(\vk)$ & an index; the parent of $\vk$\\
    	$\pa^{\vK}_{i,j}$ & an index; the parent of $\vK_{i,j}$ \\
    	$\calF^{\vK,\vb}_{i,j}$ & the event that for all $(i',j')$ before or equal to $(i,j)$, $\mu_{i',j'}^{\vK}$ exists and $\nex(\mu_{i',j'}^{\vK}) = \vb_{i',j'}$\\
    	$\calF^{\vK,\vb}_{i}$ & the event $\calF^{\vK,\vb}_{i,K_i}$\\
    	$\tilde{\calG}^{\vK,\vb}_{i,j}$ & the event that for all $\a_i(\mu^{\vK}_{i,j'})$ are distinct for $1 \leq j' \leq j$\\
    	$\calG^{\vK,\vb}_{i}$ & the event $\calG^{\vK,\vb}_{i, K_i} \land \calG^{\vK,\vb}_{i+1, K_{i+1}} \land \cdots \land \calG^{\vK,\vb}_{\ell, K_\ell}$
    	\\$p(\vk)$ & the path $p(\mu^{\vk})$ from root to $\mu^{\vk}$ \\
    	$\zeta^{i,j}$ & the node $\rig_i\left(\vidxb{\pa^{\vK}_{i,j}}\right)$\\
    	\hline
    	\end{tabular}
    \caption{Summary of Notation}
    \end{center}
\end{table}

\paragraph*{Proof idea.}

We have the following lemmas which are extensions of \Cref{lem:single-level} and \Cref{cor:single-level-whole}. Similar as \Cref{lem:single-level}, let us first look at the case within a single level. Recall by our definition of $\vidxb{{\pa}^{\vK}_{i,j}}$ above, it is well-defined even when $\vidx^{\vK}_{i,j}$ does not exist. For some fixed $\vK, \vb$, conditioning on $r_{\le i - 1} \land g_{\le i - 1}$ and $\calF^{\vK, \vb}_{i, j - 1}$, $\vidxb{{\pa}^{\vK}_{i,j}}$ is guaranteed to exist.

The vertex $w_{\zeta^{i,j}}$ corresponding to node $\zeta^{i,j} = \rig_i\left(\vidxb{{\pa}^{\vK}_{i,j}}\right)$ is simply $\last\left(\nex\left(\vidxb{{\pa}^{\vK}_{i,j}}, i - 1\right)\right)$ which is determined by $r_{\le i-1}, g_{\le i - 1}$ by Lemma \ref{lem:only-depend}. Assume $j \leq  \tau$. Then by the $ \tau$-wise independence of $\br_i, \bg_i$, since our condition $\calF^{\vK, \vb}_{i, j - 1}$ only reveal $\br_i(\cdot)$ and $\bg_i(\cdot)$ no more than $j - 1 \leq  \tau - 1$ many points. Then we divide into two cases to upper bound the probability of $\tilde{\calG}^{\vK}_i \land \calF^{\vK, \vb}_i$. When $\a_i(\zeta^{i,j}) \not= \a_i(\vidx^{\vk}_{i,j'})$ for all $j' \leq j - 1$, we know $\level(\zeta^{i,j}) = i$ holds with $1/2$ probability and $\nex(\zeta^{i,j})$ is \emph{uniformly random}. Otherwise, $g_i(\a_i(\zeta^{i,j})) = g_i(\a_i(\vidx^{\vk}_{i,j'})) = 1$ so that $\level(\zeta^{i,j}) = i$ . Then we know $(\vidx^{\vk}_{i,j'}, \zeta^{i,j})$ is a collision of level $i$ and so that $\tilde{\calG}^{\vK}_i = 0$. Therefore, in both cases, we are able to upper bound $\tilde{\calG}^{\vK}_i \land \calF^{\vK, \vb}_i$. 

In Lemma \ref{lem:single-level}, we were able to exactly compute probability of $\calF^{\vk, \vb}_i$ which is independent of $\br_{\le i-1}, \bg_{\le i - 1}$. Therefore we can argue that conditioning on $\calF^{\vk, \vb}_i$, $\br_{\le i - 1}, \bg_{\le i - 1}$ is still uniformly distributed. This is necessary for the induction step in \Cref{cor:single-level-whole} and \Cref{lem:induction}. However, here we can only get an upper bound of $\calF^{\vK, \vb}_i \land \tilde{\calG}^{\vK}_i$. Neither $\calF_i^{\vK, \vb}$ nor $\calF_i^{\vK, \vb} \land \tilde{\calG}^{\vK}_i$ is independent of $\br_{\le i - 1}, \bg_{\le i - 1}$. 

To remedy this, we add auxiliary events $\calA^{\vK, \vb}_i$ so that the probability of event $(\calF_i^{\vK, \vb} \land \tilde{\calG}^{\vK}_i) \lor \calA^{\vK, \vb}_i$ exactly matches the upper bound for $\calF_i^{\vK, \vb} \land \tilde{\calG}^{\vK}_i$. This guarantees that $(\calF_i^{\vK, \vb} \land \tilde{\calG}^{\vK}_i) \lor \calA^{\vK, \vb}_i$ is independent of $\br_{\le i - 1}, \bg_{\le i - 1}$.

Below we explicitly define the sequence $\vc$ to emphasize that it only depends on $\vK, \vb, i,  r_{\leq i - 1}, g_{\leq i - 1})$. One may compare it with \Cref{obs:existence-K} \ref{item:determine-a-from-xi-K}. 

\begin{definition} \label{def:vc}
Recall $\last(s', i)$ is defined in \Cref{def:last}. Let $\vc(\vK, \vb, i, r_{\leq i - 1}, g_{\leq i - 1})$ to be a sequence defined as following:

For each $j \in K_i$, let $i^{\pa}, j^{\pa}$ be two integers such that $\pa(\vK_{i,j}) = \vK_{i^{\pa}, j^{\pa}}$. Note here $i^{\pa}, j^{\pa}$ can be determined from $\vK, i, j$. We let $x_j = \last(b_{i^{\pa},j^{\pa}}, i - 1)$. This is well-defined since by Lemma \ref{lem:only-depend}, $\last(\cdot, i - 1)$ only depends on $r_{\leq i - 1}, g_{\leq i - 1}$.

Then let $j_0 = \max\{j_0 \ \vert \ \vK_{i,j_0} \text{ is an ancestor of } \vK_{i,j}\}$. Then we know the level $i$ ancestors of $\vK_{i,j}$ are exactly $\vK_{i,j_0}, \dots, \vK_{i,j - 1}$. Similar to \Cref{obs:existence-K} \ref{item:determine-a-from-xi-K}, we let $j' = \min\{j' \ \vert \ \exists j'' \text{ s.t. } j_0 \leq j'' < j' \leq j, a_{x_{j''}} = a_{x_{j'}}\}$. Finally, we let $$\left[\vc(\vK, \vb, i, r_{\leq i - 1}, g_{\leq i - 1})\right]_{j} = \begin{cases} a_{x_j} & \text{$j'$ does not exist} \\ \star_{j - j'} & \text{Otherwise} \end{cases}$$
 
When $\vK, \vb, i, r_{\leq i - 1}, g_{\leq i - 1}$ are clear from context, we drop them and simply write $\vc$. 
\end{definition}

\begin{observation} \label{obs:determined-K}
Fix a level $i \in [\ell]$, and fix $\vK = \{\vk^1, \vk^2, \dots, \vk^t\}$, $j \in [K_i]$, $\vb \in \calB_{\vK}$, $r_{\leq i - 1}$ and $g_{\leq i - 1}$. Recall we defined $\zeta^{i,j} = \rig_i\left(\vidxb{\pa^{\vK}_{i,j}}\right)$. Let $\vc$ be defined in \Cref{def:vc}. Assuming $\calF^{\vK, \vb}_{i,j - 1}$ holds, we have $\a_i(\zeta^{i,j}) = c_j$.

\end{observation}
\begin{proof}
Below we use the same definition of $x_j, i^{\pa}, j^{\pa}$ as \Cref{def:vc}. 

By $\calF^{\vK, \vb}_{i,j - 1}$, we know $\nex\left(\vidxb{\vK_{i^{\pa}, j^{\pa}}}\right) = b_{i^{\pa}, j^{\pa}}$. Then from \Cref{obs:existence-K} \ref{item:xi-K}, we know that $x_j = \last\left(\nex\left(\vidxb{\vK_{i^{\pa}, j^{\pa}}}\right) , i - 1\right) = \x_i\left(\rig_i\left(\vidxb{\vK_{i^{\pa}, j^{\pa}}}\right)\right) = \x_i\left(\rig_i\left(\vidxb{\pa^{\vK}_{i,j}}\right)\right)$. This also holds for all $1 \leq j' < j$ for the same reason. 

By $\calF^{\vK, \vb}_{i,j - 1}$, $\vidx^{\vK}_{i,j'}$ exists for all $1 \leq j' < j$. Therefore $\rig_i\left(\vidxb{\pa^{\vK}_{i,j'}}\right) = \vidx^{\vK}_{i,j'}$. Thus for all $1 \leq j' < j$, we have $x_{j'} = \x_i(\vidx^{\vK}_{i,j'})$. Thus from \Cref{obs:existence-K} \ref{item:determine-a-from-xi-K} and the definition of $\vc$, we know that $\a_i(\zeta^{i,j}) = c_j$. 
\end{proof}

\begin{definition} \label{def:defPK}
$P^{\vK, \vb, r_{\leq i - 1}, g_{\leq i - 1}}_{i,j}(r_i, g_i)$ is a predicate of $r_i, g_i$ defined as following:

Let $\vc$ be the sequence defined in \Cref{def:vc}. For $r_i \in \supp(\br_i)$ and $g_i \in \supp(\bg_i)$, 
\[
P^{\vK, \vb, r_{\leq i - 1}, g_{\leq i - 1}}_{i,j}(r_i, g_i) \coloneqq \left[ \forall j' \in [j], g_i(c_{j'}) = 1 \land r_i(c_{j'}) = b_{i,j'} \right].
\]

When $\vK, \vb, r_{\leq i - 1}, g_{\leq i - 1}$ are clear from the context, we simply write $P_{i,j}(r_i, g_i)$. We also view $P_{i,j}(\br_i, \bg_i)$ as an event in probability space $(\br_i, \bg_i)$. 
\end{definition}

Let $\event_{\leq i - 1}$ denote the event $[r_{\leq i - 1} \land g_{\leq i - 1}]$. We have the following observation. 

\begin{observation} \label{obs:F-and-P}
Fix $\vK, \vb$.

$$\calF^{\vK, \vb}_{i,j} \land \event_{\leq i - 1} = \calF^{\vK, \vb}_{i + 1} \land \event_{\leq i - 1} \land P_{i,j}(r_i, g_i)$$
Specifically, we have
$$\calF^{\vK, \vb}_i \land \event_{\leq i - 1} = \calF^{\vK, \vb}_{i + 1} \land \event_{\leq i - 1} \land P_i(r_i, g_i)$$
\end{observation}
\begin{proof}
We prove by induction. The base case is when $j = 0$, both sides of the equation are exactly the same. Assume this holds for $j - 1$. 

Let $\vc$ be define as \Cref{def:vc}. By definition,
\begin{align*}
	\calF^{\vK, \vb}_{i,j} \land \event_{\leq i - 1} &= \calF^{\vK, \vb}_{i,j - 1} \land \event_{\leq i - 1} \land [\vidx^{\vK}_{i,j} \text{ exists } \land \nex(\vidx^{\vK}_{i,j}) = b_{i,j}] \\
	&= \calF^{\vK, \vb}_{i,j - 1} \land \event_{\leq i - 1} \land [\bg_i(\a_i(\zeta^{i,j}) = 1 \land \br_i(\a_i(\zeta^{i,j})) =b_{i,j}] \tag{\Cref{obs:existence-K} \ref{item:exist-K}} \\
	&= \calF^{\vK, \vb}_{i,j - 1} \land \event_{\leq i - 1} \land [\bg_i(c_j) = 1 \land \br_i(c_j) =b_{i,j}] \tag{\Cref{obs:determined-K}} \\
	&= \calF^{\vK, \vb}_{i + 1} \land \event_{\leq i - 1} \land P_{i,j - 1}(r_i, g_i)\land [\bg_i(c_j) = 1 \land \br_i(c_j) =b_{i,j}] \tag{\text{Inductive hypothesis}}\\
	&= \calF^{\vK, \vb}_{i + 1} \land \event_{\leq i - 1} \land P_{i,j}(r_i, g_i)
\end{align*}
\end{proof}

Recall the definition of $\tilde{\calG}^{\vK}_{i,j}$ depends on $\vidx^{\vK}_{i,t_1}, \vidx^{\vK}_{i,t_2}$ for $1 \leq t_1 < t_2 \leq j$. Therefore it is not well-defined when $\calF^{\vK, \vb}_{i,j}$ is not true since these nodes may not exist. For this technical reason, we extend it to the following definition. 

\begin{definition} \label{def:defGK}
$\calG^{\vK, \vb}_{i,j}$ is the event defined as following. Let $\vc$ be the sequence $\vc(\vK, \vb, \br_{\leq i - 1}, \bg_{i - 1})$ defined in \Cref{def:vc}. We let

\[
\calG^{\vK, \vb}_{i,j} \coloneqq \left[ \forall 1 \leq t_1 < t_2 \leq j, c_{t_1} \not= c_{t_2}\right].
\]

We also define $\calG^{\vK}_i = \calG^{\vK}_{i, K_i} \land\calG^{\vK}_{i + 1}$. 
\end{definition}

When $\calF^{\vK, \vb}_{i,j}$ holds, $\calG^{\vK}_{i,j}$ is the same as $\tilde{\calG}^{\vK}_{i,j}$. So it also captures the collision between paths. But it has the nice property that even when $\calF^{\vK, \vb}_{i,j}$ does not hold, $\calG^{\vK}_{i,j}$ is still well-defined.

Now we are ready to state the following lemma which is an extension of \Cref{lem:single-level}.

\begin{lemma}\label{lem:single-level-K}
Fix $\vK = \{\vk^1, \vk^2, \dots, \vk^t\}\subseteq \Kshort$  ($t\le 4$) and $\vb \in \calB_{\vK}$. In probability space $(\bs, \bh, \bw)$, for any event $\calA^{\vK, \vb}_{i + 1}$ such that $\left(\calF^{\vK, \vb}_{i + 1} \land \calG^{\vK}_{i + 1}\right) \lor \calA^{\vK, \vb}_{i + 1}$ is independent of $\br_{\leq i}, \bg_{\leq i}$. Fix $r_{\leq i-1}, g_{\leq i-1} \in \supp(\br_{\leq i-1}, \bg_{\leq i-1})$. For $j \leq \tau$, we have 

$$\Pr\left[P_{i,j}(r_i, g_i) \land \calG^{\vK}_{i,j} \ \middle\vert \ \left(\left(\calF^{\vK, \vb}_{i + 1} \land \calG^{\vK}_{i + 1}\right) \lor \calA^{\vK, \vb}_{i + 1}\right) \land \event_{\leq i - 1} \land P_{i, j - 1}(r_i, g_i) \land \calG^{\vK}_{i, j - 1}\right] \leq \frac{1}{2n}.$$
\end{lemma}
\begin{proof}

Let $\vc$ be defined as in \Cref{def:vc}. Since $\vc$ is fixed by $(\vK, \vb, i,  r_{\leq i - 1}, g_{\leq i - 1})$, the value of $\calG^{\vK}_{i,j - 1}$ is also uniquely determined. Therefore we can drop $\calG^{\vK}_{i,j - 1}$ in the condition. Namely,
\begin{align*}
&\Pr\left[P_{i,j}(r_i, g_i) \land \calG^{\vK}_{i,j} \ \middle\vert \ \left(\left(\calF^{\vK, \vb}_{i + 1} \land \calG^{\vK}_{i + 1}\right) \lor \calA^{\vK, \vb}_{i + 1}\right) \land \event_{\leq i - 1} \land P_{i,j - 1}(r_i, g_i) \land \calG^{\vK}_{i, j - 1}\right] \\ = &\Pr\left[P_{i,j}(r_i, g_i) \land \calG^{\vK}_{i,j} \ \middle\vert \ \left(\left(\calF^{\vK, \vb}_{i + 1} \land \calG^{\vK}_{i + 1}\right) \lor \calA^{\vK, \vb}_{i + 1}\right) \land \event_{\leq i - 1} \land P_{i,j - 1}(r_i, g_i)\right]
\end{align*}
By definition, 
\begin{align*}
&\Pr\left[P_{i,j}(r_i, g_i)  \land \calG^{\vK}_{i,j} \ \middle\vert \  \left(\left(\calF^{\vK, \vb}_{i + 1} \land \calG^{\vK}_{i + 1}\right) \lor \calA^{\vK, \vb}_{i + 1}\right) \land \event_{\leq i - 1} \land P_{i,j - 1}(r_i, g_i)\right] \\ 
= & \Pr\left[g_i(c_j) = 1 \land r_i(c_j) = b_{i,j} \land \calG^{\vK}_{i,j} \ \middle\vert \ \left(\left(\calF^{\vK, \vb}_{i + 1} \land \calG^{\vK}_{i + 1}\right) \lor \calA^{\vK, \vb}_{i + 1}\right) \land \event_{\leq i - 1} \land P_{i,j - 1}(r_i, g_i)\right]
\end{align*}

The main difference with Lemma \ref{lem:single-level} is that $c_j$ may not be different from all other $c_{j'}$ ($j' < j$) now. However we know either $c_j$ is different from all other $c_{j'}$ or we have $\calG^{\vK}_{i,j} = 0$. 

Then we have
\begin{align*}
&\Pr\left[g_i(c_j) = 1 \land r_i(c_j) = b_{i,j}\land \calG^{\vK}_{i,j} \ \middle\vert \ \left(\left(\calF^{\vK, \vb}_{i + 1} \land \calG^{\vK}_{i + 1}\right) \lor \calA^{\vK, \vb}_{i + 1}\right) \land \event_{\leq i - 1} \land P_{i, j - 1}(r_i, g_i)\right] \\
\leq&\Pr\left[g_i(c_j) = 1 \land r_i(c_j) = b_{i,j}\land [\forall j' < j, c_{j'} \not= c_j] \ \middle\vert \ \left(\left(\calF^{\vK, \vb}_{i + 1} \land \calG^{\vK}_{i + 1}\right) \lor \calA^{\vK, \vb}_{i + 1}\right) \land \event_{\leq i - 1} \land P_{i, j - 1}(r_i, g_i)\right] \\
= &\Pr\left[g_i(c_j) = 1 \land r_i(c_j) = b_{i,j}\land[\forall j' < j, c_{j'} \not= c_j] \ \middle\vert \ P_{i, j - 1}(r_i, g_i)\right]\\
\leq &\frac{1}{2n}.
\end{align*}

The second last step follows from our assumption that $\left(\calF^{\vK, \vb}_{i + 1} \land \calG^{\vK}_{i + 1}\right) \lor \calA^{\vK, \vb}_{i + 1}$ is independent of $\event_{\leq i}$, which implies its joint event with $\event_{\leq i - 1}$ is independent of $r_i, g_i$. Since $P_{i, j - 1}(r_i, g_i)$ is merely a predicate of $r_i$ and $g_i$, we know such independence is still true conditioning on $P_{i, j - 1}(r_i, g_i)$. 

The last step follows from the fact that $j \leq \tau$ and $\tau$-wise independence of $r_i$ and $g_i$.
\end{proof}

\begin{lemma}\label{lem:single-level-whole-K}
Fix $\vK = \{\vk^1, \vk^2, \dots, \vk^t\}$ such that $\vk^j \in \Kshort$ for every $j \in [t]$ and $\vb \in \calB_{\vK}$. Suppose $\calA_{i+1}^{\vK, \vb}$ is an event such that $(\calF^{\vK, \vb}_{i + 1} \land \calG^{\vK}_{i + 1}) \lor \calA^{\vK, \vb}_{i + 1}$ is independent of $\br_{\leq i}, \bg_{\leq i}$. Let $r_{\leq i-1}, g_{\leq i-1} \in \supp(\br_{\leq i-1}, \bg_{\leq i-1})$.

There is an event $\calA^{\vK, \vb}_i$ such that 
\[
\Pr\left[(\calF^{\vK, \vb}_i \land \calG^{\vK}_i) \lor \calA^{\vK, \vb}_i \ \middle\vert \ \left(\left(\calF^{\vK, \vb}_{i + 1} \land \calG^{\vK}_{i + 1}\right) \lor \calA^{\vK, \vb}_{i + 1}\right) \land r_{\leq i-1} \land g_{\leq i-1}\right] = \frac{2^{-K_i}}{n^{K_i}}.
\]

Moreover, $\calA^{\vK, \vb}_i$ is true only when $\left(\calF^{\vK, \vb}_{i + 1} \land \calG^{\vK}_{i + 1}\right) \lor \calA^{\vK, \vb}_{i + 1}$ is true. 
\end{lemma}
\begin{proof}

Let $\event_{\leq i -1}$ be the event $r_{\leq i-1} \land g_{\leq i-1}$. Since by \Cref{obs:F-and-P}, $$\calF^{\vK, \vb}_i \land \event_{\leq i - 1} = \calF^{\vK, \vb}_{i + 1} \land \event_{\leq i - 1} \land P_i(r_i, g_i)$$
We know that 

\begin{align*}
&\Pr\left[\calF^{\vK, \vb}_i \land \calG^{\vK}_i \ \middle\vert \ \left(\left(\calF^{\vK, \vb}_{i + 1} \land \calG^{\vK}_{i + 1}\right) \lor \calA^{\vK, \vb}_{i + 1}\right) \land \event_{\leq i - 1}\right] \\
= &\Pr\left[\calF^{\vK, \vb}_{i + 1} \land P^{\vb, \vc}_i \land \calG^{\vK}_i \ \middle\vert \ \left(\left(\calF^{\vK, \vb}_{i + 1} \land \calG^{\vK}_{i + 1}\right) \lor \calA^{\vK, \vb}_{i + 1}\right) \land \event_{\leq i - 1}\right] \\
\leq &\Pr\left[P^{\vb, \vc}_i \land \calG^{\vK}_i \ \middle\vert \ \left(\left(\calF^{\vK, \vb}_{i + 1} \land \calG^{\vK}_{i + 1}\right) \lor \calA^{\vK, \vb}_{i + 1}\right) \land \event_{\leq i - 1}\right]
\end{align*}

 By Lemma~\ref{lem:single-level-K}, we have
\begin{align*}
&\Pr\left[\calF^{\vK, \vb}_i \land \calG^{\vK}_i \ \middle\vert \ \left(\left(\calF^{\vK, \vb}_{i + 1} \land \calG^{\vK}_{i + 1}\right) \lor \calA^{\vK, \vb}_{i + 1}\right) \land \event_{\leq i - 1}\right] \\
\leq &\Pr\left[P^{\vb, \vc}_i \land \calG^{\vK}_i \ \middle\vert \ \left(\left(\calF^{\vK, \vb}_{i + 1} \land \calG^{\vK}_{i + 1}\right) \lor \calA^{\vK, \vb}_{i + 1}\right) \land \event_{\leq i - 1}\right] \\
=&\prod_{j=1}^{k_i} \Pr\left[P^{\vb, \vc}_{i,j} \land \calG^{\vK}_{i,j} \ \middle\vert \ \left(\left(\calF^{\vK, \vb}_{i + 1} \land \calG^{\vK}_{i + 1}\right) \lor \calA^{\vK, \vb}_{i + 1}\right) \land \event_{\leq i - 1} \land P^{\vb, \vc}_{i, j - 1} \land \calG^{\vK}_{i, j - 1}\right] \\
\leq&\frac{2^{-K_i}}{n^{K_i}}.
\end{align*}

For each $r_{\leq i - 1}, g_{\leq i - 1} \in \supp(\br_{\leq i - 1}, \bg_{\leq i - 1})$, we choose an arbitrary event $$\calA^{\vK, \vb, r_{\leq i}, g_{\leq i - 1}}_{i} \subset \left(\left(\calF^{\vK, \vb}_{i + 1} \land \calG^{\vK}_{i + 1}\right) \lor \calA^{\vK, \vb}_{i + 1}\right) \land \event_{\leq i - 1}$$ to increase the probability and make 

$$\Pr\left[\left(\calF^{\vK, \vb}_i \land \calG^{\vK}_i\right) \lor \calA^{\vK, \vb, r_{\leq i - 1}, g_{\leq i - 1}}_{i} \ \middle\vert \ \left(\left(\calF^{\vK, \vb}_{i + 1} \land \calG^{\vK}_{i + 1}\right) \lor \calA^{\vK, \vb}_{i + 1}\right) \land \event_{\leq i - 1}\right] = \frac{2^{-K_i}}{n^{K_i}}$$

Then we take the disjoint union of them and let

$$\calA^{\vK, \vb}_{i} = \bigsqcup_{r_{\leq i - 1}, g_{\leq i - 1}} \calA^{\vK, \vb, r_{\leq i - 1}, g_{\leq i - 1}}_{i}$$

Therefore for all $r_{\leq i - 1}, g_{\leq i - 1} \in \supp(\br_{\leq i - 1}, \bg_{\leq i - 1})$, 
\[
\Pr\left[\left(\calF^{\vK, \vb}_i \land \calG^{\vK}_i\right) \lor \calA^{\vK, \vb}_i \ \middle\vert \ \left(\left(\calF^{\vK, \vb}_{i + 1} \land \calG^{\vK}_{i + 1}\right) \lor \calA^{\vK, \vb}_{i + 1}\right) \land r_{\leq i-1} \land g_{\leq i-1}\right] = \frac{2^{-K_i}}{n^{K_i}}.
\]

At the same time, $\calA^{\vK, \vb}_{i} \subset (\calF^{\vK, \vb}_{i + 1} \land \calG^{\vK}_{i + 1}) \lor \calA^{\vK, \vb}_{i + 1}$. 
\end{proof}
The last piece is the lemma extending Lemma \ref{lem:induction}. We need the following proposition first.

\begin{prop}\label{prop:cond-fact}
    Let $\calX,\calY,\calZ$ be three events. We have
    \[
    \Pr[\calX|\calZ] = \Pr[\calX|\calY \land \calZ] \cdot \Pr[\calY|\calZ] + \Pr[\calX|\neg\calY \land \calZ] \cdot \Pr[\neg\calY|\calZ].
    \]
    
    In particular, when $\calX$ is a subset event of $\calY$, we have $\Pr[\calX|\neg\calY \land \calZ] = 0$ and hence
    \[
    \Pr[\calX|\calZ] = \Pr[\calX|\calY \land \calZ] \cdot \Pr[\calY|\calZ].
    \]
\end{prop}

\begin{lemma}\label{induction-K}
Fix $\vK = \{\vk^1, \vk^2, \dots, \vk^t\} (t \leq 4)$ such that $\vk^j \in \Kshort$ for every $j \in [t]$, and fix $\vb \in \calB_{\vK}$. For all $i  \in [\ell]$, let $r_{\leq i-1}, g_{\leq i-1} \in \supp(\br_{\leq i-1}, \bg_{\leq i-1})$. Then there is a sequence of events $\left\{\calA^{\vK, \vb}_i\right\}_{i \in [\ell]}$ such that:
\[
\Pr\left[\left(\calF_i^{\vK, \vb} \land \calG^{\vK}_i\right) \lor \calA^{\vK, \vb}_i \ \middle\vert \ r_{\leq i - 1} \land g_{\leq i - 1}\right] = \frac{2^{-\sum_{j = i}^\ell K_j}}{n^{\sum_{j = i}^\ell K_j}}, \quad \forall i \in [\ell]
\]

In particular,
\[
\Pr\left[\calF_1^{\vK, \vb} \land \calG^{\vK}_1 \right] \leq \frac{2^{-\sum_{j = 1}^\ell K_j}}{n^{\sum_{j = 1}^\ell K_j}}.
\]
\end{lemma}
\begin{proof}
We prove this by induction. The base case is when $i = \ell + 1$. $\calF^{\vK, \vb}_i$ and $\calG^{\vK}_i$ are always true by definition. Therefore $\Pr\left[\calF_i^{\vK, \vb} \land \calG^{\vK}_i \ \middle\vert \ \event_{i - 1}\right] = 1$. The event $\calA^{\vK, \vb}_i$ is set to be an empty event.

Otherwise, suppose the induction hypothesis holds for $i + 1$. Note this implies for all $r_{\leq i}, g_{\leq i} \in \supp(\br_{\leq i}, \bg_{\leq i})$, it holds that $\Pr\left[\left(\calF_{i + 1}^{\vK, \vb} \land \calG^{\vK}_{i + 1}\right) \lor \calA^{\vK, \vb}_{i + 1}  \ \middle\vert \ r_{\leq i} \land g_{\leq i}\right] = \Pr\left[\left(\calF_{i + 1}^{\vK, \vb} \land \calG^{\vK}_{i + 1}\right) \lor \calA^{\vK, \vb}_{i + 1} \right]$. This shows that $(\calF^{\vK, \vb}_{i + 1} \land \calG^{\vK}_{i + 1}) \lor \calA^{\vK, \vb}_{i + 1}$ is independent of $\br_{\leq i}, \bg_{\leq i}$, meaning that it satisfies the premise of Lemma~\ref{lem:single-level-whole-K}. 

Let $\event_{\leq i-1}$ be the event that $\br_{\leq i - 1} = r_{\leq i-1} \land \bg_{\leq i - 1} = g_{\leq i-1}$. We have
\begin{align*}
& \Pr\left[\left(\calF_i^{\vK, \vb} \land \calG^{\vK}_i\right) \lor \calA^{\vK, \vb}_i \ \middle\vert \ \event_{i - 1}\right] \\
= &\Pr\left[\left(\calF_i^{\vK, \vb} \land \calG^{\vK}_i\right) \lor \calA^{\vK, \vb}_i\ \middle\vert \ \left(\left(\calF_{i + 1}^{\vK, \vb} \land \calG^{\vK}_{i + 1}\right) \lor \calA^{\vK, \vb}_{i + 1}\right) \land \event_{i - 1}\right] \Pr\left[\left(\calF_{i + 1}^{\vK, \vb} \land \calG^{\vK}_{i + 1}\right) \lor \calA^{\vK, \vb}_{i + 1} \ \middle\vert \ \event_{i - 1}\right].
\end{align*}

The last equality follows from Proposition~\ref{prop:cond-fact}. To check the premise of Proposition~\ref{prop:cond-fact}, we need to prove that 
$$\Pr\left[(\calF_i^{\vK, \vb} \land \calG^{\vK}_i) \lor \calA^{\vK, \vb}_i \ \middle\vert \ \neg \left(\left(\calF_{i + 1}^{\vK, \vb} \land \calG^{\vK}_{i + 1}\right) \lor \calA^{\vK, \vb}_{i + 1}\right) \land \event_{i - 1}\right] = 0$$

This follows from the fact that $\calF_i^{\vK, \vb} \land \calG^{\vK}_i \subset \calF_{i + 1}^{\vK, \vb} \land \calG^{\vK}_{i + 1}$ and $\calA^{\vK, \vb}_i \subset (\calF_{i + 1}^{\vK, \vb} \land \calG^{\vK}_{i + 1}) \lor \calA^{\vK, \vb}_{i + 1}$.

Then from induction hypothesis, we have
\begin{align*}
\Pr\left[\left(\calF_{i + 1}^{\vK, \vb} \land \calG^{\vK}_{i + 1}\right) \lor \calA^{\vK, \vb}_{i + 1}  \ \middle\vert \ \event_{i - 1}\right] &= \E_{(r_i, g_i) \getsR (\br_i,\bg_i)}\left[ \Pr\left[\left(\calF_{i + 1}^{\vK, \vb} \land \calG^{\vK}_{i + 1}\right) \lor \calA^{\vK, \vb}_{i + 1} \ \middle\vert \ r_{\leq i} \land g_{\leq i}\right] \ \middle\vert \ \event_{i - 1} \right] \\
&= \frac{2^{-\sum_{j=i + 1}^\ell K_j}}{n^{\sum_{j=i + 1}^\ell K_j}}.
\end{align*}

From Lemma~\ref{lem:single-level-whole-K},
\[
\Pr\left[\left(\calF_i^{\vK, \vb} \land \calG^{\vK}_i\right) \lor \calA^{\vK, \vb}_i \ \middle\vert \ \left(\left(\calF_{i + 1}^{\vK, \vb} \land \calG^{\vK}_{i + 1}\right) \lor \calA^{\vK, \vb}_{i + 1}\right) \land \event_{i - 1}\right] = \frac{2^{-K_i}}{n^{K_i}}.
\]

Putting everything together,
\[
\Pr\left[\left(\calF_i^{\vK, \vb} \land \calG^{\vK}_i\right) \lor \calA^{\vK, \vb}_i \ \middle\vert \ \event_{i - 1}\right] = \frac{2^{-K_i}}{n^{K_i}} \cdot \frac{2^{-\sum_{j=i + 1}^\ell K_j}}{n^{\sum_{j=i + 1}^\ell K_j}} = \frac{2^{-\sum_{j=i}^\ell K_j}}{n^{\sum_{j=i}^\ell K_j}}. \qedhere
\]
\end{proof}

Since $K_i$ only counts the number of level $i$ nodes in the union of $P^{\vk^1},P^{\vk^2}, \dots, P^{\vk^t}$, it maybe smaller than $\vk^1_i + \vk^2_i + \dots + \vk_i$. Therefore, we need a lemma to account for that. 

\begin{lemma} \label{lem:count-path}
We have
\[
2^{t \ell} \le \sum_{\substack{(\vk^1, \vk^2, \dots, \vk^t) \in (\N^\ell)^t\\ K = \{\vk^1, \vk^2, \dots, \vk^t\}}}\prod_{i = 1}^\ell 2^{-K_i} \leq t! \cdot 2^{t(\ell + 1)}
\]
\end{lemma}

\begin{proof}

We first show the first inequality. Note that for $K = \{\vk^1, \vk^2, \dots, \vk^t\}$, we have
\[
\sum_{i=1}^{\ell} K_i \le \sum_{i=1}^{\ell} \sum_{j=1}^{t} k^{j}_i.
\]

Hence we have
\[
\sum_{\substack{(\vk^1, \vk^2, \dots, \vk^t) \in (\N^\ell)^t\\ K = \{\vk^1, \vk^2, \dots, \vk^t\}}}\prod_{i = 1}^\ell 2^{-K_i} \ge \sum_{(\vk^1, \vk^2, \dots, \vk^t) \in (\N^\ell)^t} \prod_{j=1}^{t} \prod_{i=1}^{\ell} 2^{-k^j_i} \ge 2^{t\ell}.
\]

Next we show the second inequality, and we will prove it by induction. For the base case $t = 1$, the expression simplifies to
\[
\sum_{\vk \in \N^\ell} \prod_{i=1}^{\ell} 2^{-k_i} = 2^{\ell},
\]
which proves the base case.

Then suppose the statement holds for $t - 1$. And our goal is to upper bound
\[
\sum_{\substack{(\vk^1, \vk^2, \dots, \vk^t) \in (\N^\ell)^t\\ K = \{\vk^1, \vk^2, \dots, \vk^t\}}}\prod_{i = 1}^\ell 2^{-K_i}.
\]
Let $\vK' = \{\vk^1, \dots, \vk^{t-1} \}$ be the union of first $t - 1$ indices, and let $\vK = \{\vk^1, \vk^2, \dots, \vk^t\}$. We have that
\[
\sum_{i=1}^{\ell} K_i = \sum_{i=1}^{\ell} K'_i + \sum_{i=1}^{\ell} k^t_i - |\vk^t \cap \vK'|,
\]
where $|\vk^t \cap \vK'|$ denote the common length of path $\vk^t$ and the union $\vK'$.

Now we fix $\vK'$ and try to calculate its contribution together with all possible $\vk^t \in \N^\ell$ such that $\vk^t$ is (one of) the left-most vertices among $\vk^1, \dots, \vk^{t-1},\vk^{t}$.

Suppose $|\vk^t \cap \vK'| = j$, then we know the $j$-length prefix of $\vk^t$ has at most one possibility (the left-most depth-$j$ node on the sub-tree formed by $\vK'$, if depth-$j$ nodes exist in $\vK'$). And we can bound the contribution of this case by
\[
2^{\ell - j} \cdot \prod_{i = 1}^\ell 2^{-K'_i}.
\]
By a union bound, the contribution of this $\vK'$ together all possible $\vk^t \in \N^\ell$ such that $\vk^t$ is the left-most vertex can be bounded by
\[
\sum_{j=0}^{\ell} 2^{\ell - j} \cdot \prod_{i = 1}^\ell 2^{-K'_i} \le 2^{\ell + 1} \cdot \prod_{i = 1}^\ell 2^{-K'_i}.
\]

Summing up for all possible $\vK'$, we can bound the contribution when $\vk^t$ is the left-most vertex by
\[
\sum_{\substack{\vk^1, \vk^2, \dots, \vk^{t-1} \\ \vK' = \{\vk^1, \vk^2, \dots, \vk^{t-1}\} }} \prod_{i = 1}^\ell 2^{-K'_i} \cdot 2^{\ell + 1}.
\]
By the induction hypothesis, this can be further bounded by
\[
(t-1)! \cdot 2^{t \cdot (\ell + 1)}.
\]

By symmetry, for each $i \in [t]$, the contribution when $\vk^i$ is (one of) the left-most vertex can also be bounded by the above quantity. The lemma then follows from a union bound over the left-most vertex.
\end{proof}

Finally, we can obtain the desired upper bound, and prove Lemma~\ref{lem:bad-case}.

\begin{reminder}{Lemma~\ref{lem:bad-case}}
Let $C_v = \#\{i \mid a_i = a_v\}$ be the number of occurrences of $a_v$ in the input $a$. It holds that
\[
    \E_{\bw, \bT} \left[ \#\left\{\vk \in \Kshort \ \middle\vert \ \nex(\mu^{\vk}) = v \land \exists \vk' \not= \vk'' \in \N^{\ell}, a_{\nex(\mu^{\vk'})} = a_{\nex(\mu^{\vk''})}\right\} \right] \leq 48\frac{8^\ell F_2(a)}{n^3} + 16 \frac{4^\ell C_v}{n^2} + \frac{1}{n^3}.
	\]
\end{reminder}
\begin{proof}

We apply Lemma \ref{lem:structure} with $\phi = \vidx^{\vk} + 1$, $\alpha = \vidx^{\vk'} + 1$, and $\beta = \vidx^{\vk''} + 1$. Note by \Cref{obs:w-nex}, we know that $w_{\phi} = \nex(\vidx^{\vk})$. Similarly, $w_{\alpha} = \nex(\vidx^{\vk'})$ and $w_{\beta} = \nex(\vidx^{\vk''})$. 

\begin{align*}
& \E_{\bw, \bT} \left[ \#\left\{\vk \in \Kshort \ \middle\vert \ \nex(\mu^{\vk}) = v \land \exists \vk' \not= \vk'' \in \N^{\ell}, a_{\nex(\mu^{\vk'})} = a_{\nex(\mu^{\vk''})}\right\} \right]  \\
= &\E_{\bw, \bT} \Big[ \#\Big\{\vk^1 \in \Kshort \ \Big\vert \ \nex(\mu^{\vk^1}) = v \land \exists \vk^2 \not= \vk^3 \in \N^{\ell}, a_{\nex(\mu^{\vk^2})} = a_{\nex(\mu^{\vk^3})}, \\ 
&\hspace{2cm}\text{no collision between $p(\vk^1), p(\vk^2), p(\vk^3)$}\Big\} \Big] \\
\leq &\E_{\bw, \bT} \Big[ \#\Big\{\vk^1 \in \Kshort \ \Big\vert \ \nex(\mu^{\vk^1}) = v \land \exists \vk^2 \not= \vk^3 \in \Kshort, a_{\nex(\mu^{\vk^2})} = a_{\nex(\mu^{\vk^3})}, \\ 
&\hspace{2cm}\text{no collision between $p(\vk^1), p(\vk^2), p(\vk^3)$}\Big\} \Big] + \E[|\Kshort| \cdot \elong] \\
\leq &\sum_{\substack{(\vk^1, \vk^2,\vk^3) \in (\Kshort)^3\\ \text{$\vk^2 \ne \vk^3$}}} \Pr \left[ \calG^{\{\vk^1, \vk^2,\vk^3\}}_1 \land  \nex(\vidx^{\vk^1}) = v \land a_{\nex(\vidx^{\vk^2})} = a_{\nex(\vidx^{\vk^3})}\right] + \E[|\Kshort| \cdot \elong]. 
\end{align*}

First, by Lemma~\ref{lem:elong-small}, $\tau \geq 20\log n \log \log n$, and $\ell \leq \log n$, we have
\[
\E\left[|\Kshort| \cdot \elong\right] \le \tau^{\ell} \cdot n \ell / 2^{\tau/4} \leq n \ell \cdot 2^{\ell \log \tau - \tau/4} \leq n \log n \cdot 2^{2\log n \log \log n - 5 \log n \log \log n} \leq \frac{1}{n^3}
\]

Next we bound
\[
\sum_{(\vk^1, \vk^2,\vk^3) \in (\Kshort)^3 \text{s.t. $\vk^2 \ne \vk^3$}} \Pr \left[ \calG^{\{\vk^1, \vk^2,\vk^3\}}_1 \land  \nex(\vidx^{\vk^1}) = v \land a_{\nex(\vidx^{\vk^2})} = a_{\nex(\vidx^{\vk^3})}\right].
\]

There are two cases, first case is that $\vk^1 = \vk^2$ or $\vk^1 = \vk^3$. By symmetry, we only consider the case when $\vk^1 = \vk^3$ here. Let $\vK = \{\vk^1, \vk^2\}$. 

\paragraph{When $\vk^1 = \vk^3$.} By Lemma \ref{induction-K}, for any sequence $\vb \in \calB_{\vK}$, we have 
\[
\Pr\left[\calF_1^{\vK, \vb} \land \calG^{\vK}_1\right] = \frac{2^{-\sum_{j = 1}^\ell K_j}}{n^{\sum_{j = 1}^\ell K_j}}.
\]

Note there are $n^{\sum_{j = 1}^\ell K_j}$ many such sequence $\vb \in \calB_{\vK}$, and $n^{\sum_{j = 1}^\ell K_j - 2} \cdot C_v$ of them satisfy that $\nex(\vidx^{\vk^1}) = v$ and $a_{\nex(\vidx^{\vk^2})} = a_{\nex(\vidx^{\vk^3})} = a_{\nex(\vidx^{\vk^1})} = a_v$. 

We have
\begin{align*}
&\sum_{(\vk^1, \vk^2,\vk^3) \in (\Kshort)^3 \text{s.t. $\vk^2 \ne \vk^3 \land \vk^1 = \vk^3$}} \Pr \left[ \calG^{\{\vk^1, \vk^2,\vk^3\}}_1 \land  \nex(\vidx^{\vk^1}) = v \land a_{\nex(\vidx^{\vk^2})} = a_{\nex(\vidx^{\vk^3})}\right]\\
=& \sum_{(\vk^1, \vk^2) \in (\Kshort)^2 \text{s.t.~$\vk^1 \ne \vk^2$}} \Pr \left[ \calG^{ \{\vk^1, \vk^2\} }_1 \land \nex(\vidx^{\vk^1}) = v \land a_{\nex(\vidx^{\vk^2})} = a_v\right] \\
\leq& \sum_{\substack{(\vk^1, \vk^2) \in (\N^\ell)^2\\ K = \{\vk^1,\vk^2\}}} \frac{2^{-\sum_{j = 1}^\ell K_j}}{n^{\sum_{j = 1}^\ell K_j}} \cdot n^{\sum_{j = 1}^\ell K_j - 2} \cdot C_v \\ 
\leq& \frac{8C_v \cdot 4^{\ell} }{n^2} \tag{Lemma~\ref{lem:count-path}}.
\end{align*}

\paragraph{When $\vk^1, \vk^2, \vk^3$ are distinct.} Now we consider the other case when $\vk^1, \vk^2, \vk^3$ are distinct. Let $\vK’ = \{\vk^1, \vk^2, \vk^3\}$. 

Same as before, by Lemma \ref{induction-K}, for any sequence $\vb \in \calB_{\vK'}$, we have
\[
\Pr\left[\calF_1^{\vK', \vb} \land \calG^{\vK'}_1\right] = \frac{2^{-\sum_{j = 1}^\ell K'_j}}{n^{\sum_{j = 1}^\ell K'_j}}
\]

Note there are $n^{\sum_{j = 1}^\ell K_j}$ many such sequence $\vb \in \calB_{\vK'}$, and $n^{\sum_{j = 1}^\ell K_j - 3} \cdot F_2(a)$ of those satisfies that $\nex(\vidx^{\vk^1}) = v$ and $a_{\nex(\vidx^{\vk^2})} = a_{\nex(\vidx^{\vk^3})}$. 

Hence, we have
\begin{align*}
&\sum_{\substack{(\vk^1, \vk^2,\vk^3) \in (\Kshort)^3\\\text{s.t. $\vk^1,\vk^2,\vk^3$ are distinct}}} \Pr \left[ \calG^{\{\vk^1,\vk^2,\vk^3\}}_1 \land \nex(\vidx^{\vk^1}) = v \land a_{\nex(\vidx^{\vk^2})} = a_{\nex(\vidx^{\vk^3})}\right] \\
\leq &\sum_{\substack{(\vk^1, \vk^2,\vk^3) \in (\N^\ell)^3\\ K' = \{\vk^1,\vk^2,\vk^3\}}} \frac{2^{-\sum_{j = 1}^\ell K'_j}}{n^{\sum_{j = 1}^\ell K'_j}} \cdot n^{\sum_{j = 1}^\ell K'_j - 3} \cdot F_2(a)\\ 
\leq &\frac{48 \cdot 8^{\ell} F_2(a)}{n^3}. \tag{Lemma~\ref{lem:count-path}}
\end{align*}
Summing up everything proves the theorem.
\end{proof}

\subsection{Proof of Lemma~\ref{lem:elong-small}} \label{sec:cut-off}

\begin{reminder}{Lemma~\ref{lem:elong-small}}
	It holds that
	\[
	\Pr[\elong] \le \ell n/2^{\tau/4}.
	\]
\end{reminder}
\vspace{-1.5em}
\begin{proof}
	For every $i \in [\ell]$, we define event $\elong^i$ as
	\[
	\elong^i \coloneqq \left[ \text{$\exists k \in \N^{\ell}$ s.t. $ k_i > \tau/4$ and $\vidx^{\vk}$ exists} \right].
	\]
	Then we can see $\elong = \bigcup_{i=1}^{\ell} \elong^{i}$.
	
	In the following, we will show that for each $i \in [\ell]$, $\Pr[ \elong^i]$ is small. Now we fix $i \in [\ell]$, suppose there exists $\vk \in \N^{\ell}$ such that $\vidx^{\vk}$ exists and $k_i > \tau/4$. We are going to fix $(r_{\le i -1},g_{\le i -1}) \in \supp((\br_{\le i -1},\bg_{\le i -1}))$ and conditioning on the event $r_{\le i -1} \land g_{\le i -1}$. Moreover here we also fix $r_ \in \supp(\br_i)$.
	
	Now, $\vidx^{\vk}$ exists and $k_i > \tau/4$ in particular implies there exists a starting point $s_0 = w_{\vidx^{\vk}_{i,0}} \in [n]$ such that the walk $\walk(s_0,i)$ visits at least other $\tau/4$ level-$i$ nodes $\vidx^{\vk}_{i,1}, \vidx^{\vk}_{i,2} \dots, \vidx^{\vk}_{i,\tau/4}$. For these nodes, by \Cref{obs:determine-a-from-xi}, we know that $\a_i(\vidx^{\vk}_{i,j})$ is uniquely determined by $\x_i(\vidx^{\vk}_{i,1}), \x_i(\vidx^{\vk}_{i,2}), \dots, \x_i(\vidx^{\vk}_{i,j})$. On the other hand, each $\x_i(\vidx^{\vk}_{i,j}) = \last(\nex(\vidx^{\vk}_{i,j - 1}), i - 1)$ by \Cref{lem:existence} \ref{item:xi}. Here $\nex(\vidx^{\vk}_{i,j - 1}) = r_i(\a_i(\vidx^{\vk}_{i,j - 1}))$ (by \Cref{obs:w-nex}), and $\last(\cdot, i - 1)$ only depends on $r_{\leq i - 1}$ and $g_{\leq i - 1}$ (by \Cref{lem:only-depend}). Hence, by simple induction, each $\a_i(\vidx^{\vk}_{i,j})$ is independent of $\bg_i$. Moreover, by the definition of extended walk, they are all distinct. 
	
	Therefore for a fixed $s_0 \in [n]$, we have that 
	\[
	\Pr[\bg_{i}(\a_i(\mu^{\vk}_{i,j})) = 1 \text{ for all $j \in [\tau/4]$}] \le 2^{-\tau/4}.
	\]
	
	By a union bound over different $s_0 \in [n]$, we have that
	\[
	\Pr[\elong^i] \le n/2^{\tau/4}.
	\]
	The lemma follows from another union bound over $i \in [\ell]$.
\end{proof}


\section{The Case of Two Target Vertices}
\label{sec:two}

This section is devoted to proving~\cref{lem:hit-lower-bound}, which is restated below.

\begin{reminder}{\cref{lem:hit-lower-bound}}
Suppose $\ell = \log n - \frac{\log F_2(a)}{2} - 10$. For every $u,v\in [n]$ such that $u \not= v$ and $a_u = a_v$, we have 
\[
\Pr_{\bh \getsR \mathcal{H}_{\ell,m,n}, \bs \getsR [n]} [u, v \in f^*_{a, \bh}(\bs)] \ge \Omega\left (\frac{1}{F_2(a)}\right).
\]
\end{reminder}

We recall that the second frequency moment $F_2(a) = \sum_{i=1}^n \sum_{j=1}^n \mathbf{1}[a_i=a_j]$ (including the case where $i=j$). The main difficulty in extending the previous proof to the case of two vertices is handling the collisions between paths. Suppose we enumerate $\vk^1, \vk^2$ and compute $\Pr[\nex(\vidx^{\vk^1}) = u, \nex(\vidx^{\vk^2}) = v]$.
There may be collisions between paths $p(\vidx^{\vk^1})$ and $p(\vidx^{\vk^2})$, which cause the same problem we encountered in Section \ref{bad-case}.
However it is hard to exploit the structure of such two paths as we did in Section \ref{bad-case}, since now even estimating the total counts involves more than one path.
Note that we need a lower bound on the total counts, while our previous approach in~\cref{sec:one} that exploits the combinatorial structure (\ie, \cref{lem:structure}) only gives us an upper bound. Therefore we will take different approach. 

We will define a different walk $\bw^{\vk^1, \vk^2}$ (called a \emph{relaxed extended walk}) for each pair of $(\vk^1, \vk^2)$ separately, and bound the contribution of $(\vk^1,\vk^2)$ by analyzing this walk $\bw^{\vk^1,\vk^1}$. 
Roughly speaking, $\bw^{\vk^1, \vk^2}$ is obtained by adapting the extended walk $\bw$ so that whenever a collision happens between $p(\vk^1)$ and $p(\vk^2)$, we replace the later one with true randomness. 
In this way, before we visit any vertex twice, the walk $\bw^{\vk^1, \vk^2}$ behaves exactly as the original walk $f^*_{a,\bh}$. 
Therefore, we can sum up $\Pr\left[\bw^{\vk^1, \vk^2}_{\vidx^{\vk^1}} = u, \bw^{\vk^1, \vk^2}_{\vidx^{\vk^2}} = v\right]$, and subtract $\Pr\left[\bw^{\vk^1, \vk^2}_{\vidx^{\vk^1}} = u, \bw^{\vk^1, \vk^2}_{\vidx^{\vk^2}} = v \ \land \ \exists \alpha \not= \beta, a_{\bw^{\vk^1, \vk^2}_\alpha} = a_{\bw^{\vk^1, \vk^2}_\beta}\right]$. This naturally lower-bounds the contribution of $(\vk^1, \vk^2)$ to $\Pr[u,v \in f^*_{a,\bh}]$. Finally we conclude the proof by summing up the contribution over all $\vk^1, \vk^2$.

\subsection{The Relaxed Extended Walk $\extwalk^{\vk^1,\vk^2}_{\ell,m,n,a}$}

Now we define the relaxed extended walk $\extwalk^{\vk^1,\vk^2}_{\ell,m,n,a}$ for each $\vk^1, \vk^2 \in \N^\ell$ such that  $\vk^1 < \vk^2$.

\renewcommand{\index}{\mathsf{index}}

The main (and only) difference between $\extwalk^{\vk^1,\vk^2}_{\ell,m,n,a}$ and $\extwalk_{\ell,m,n,a}$ is how the nodes on $p^*(\vk^2)$ are handled. (Recall that the extension $p^*(\mu)$ of a node $\mu$ is defined in \Cref{def:extention}.)
For a node $\mu \in p^*(\vk^2)$ with $\level(\mu) = i$, $\extwalk^{\vk^1,\vk^2}_{\ell,m,n,a}$ replaces $\a_i(\mu)$ with $\star_*$ if it would otherwise become a collision with $p(\vk^1)$.
We do this by letting the initial $C_0$ on $p^*(\vk^2)$ ``inherit'' the set $C_{k^1_i}$ from $p(\vk^1)$. See Figure \ref{fig:inherit}. To implement this, here we pass an extra parameter $\vk$ to $\walk^{\vk^1, \vk^2}(s', i, \mu_0, \vk)$, where $\vk$ is the index of $\mu_0$. The index of node $\mu_j$ in $\walk^{\vk^1, \vk^2}(s', i, \mu_0, \vk)$ is just $(0, \dots, 0, j, k_{i + 1}, \dots, k_{\ell})$ assuming it is of level $i$. Hence we can tell whether node $\mu_j$ is on $p^*(\vk^2)$ by looking at $\vk$. 

\begin{figure}
	\centering
	\begin{tikzpicture}[scale = 0.45, thick]
  \draw (0,0) -- (1, -3);
  \draw (1, -3) -- (8, -3);
  \draw (5, -3) -- (7, -9);
  \node at (7,-9.5) {\small $p^*({\vec{k}}^1)$};

  \node [draw,circle,fill,minimum size=2.5,inner sep=0pt, outer sep=0pt, blue] at (5,-3) {};
  \node [blue] at (5,-2) {\small $C^{i,\vec{k}^1} \gets C_{k^1_i}$};

  \draw (9,0) -- (10, -3);
  \node [draw,circle,fill,minimum size=2.5,inner sep=0pt, outer sep=0pt, blue] at (10,-3) {};
  \node [blue] at (12.5,-2) {\small $C_0 \gets C^{i,\vec{k}^1}$};

  \draw [->, blue] (5.3,-3.2) to [out = 330, in = 200] (9.8, -3.2);

  \draw (10,-3) -- (17, -3);
  \draw (14, -3) --  (16, -9);
  \node at (16,-9.5) {\small $p^*({\vec{k}}^2)$};

\end{tikzpicture}
	\caption{The initial $C_0$ on $p^*(\vk^2)$ ``inherit'' the set $C_{k^1_i}$ from $p(\vk^1)$.}
	\label{fig:inherit}
\end{figure}
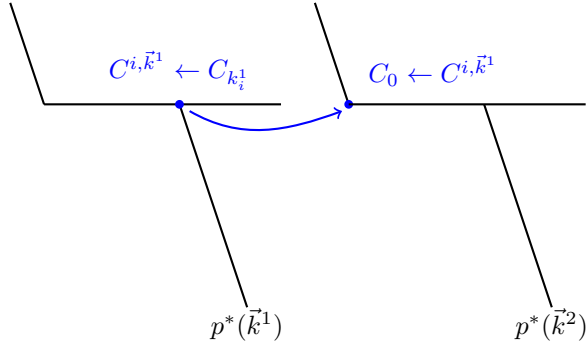

To minimize the effect of this change, instead of letting $x_{j + 1} \gets w_{|w|}$, we invoke $\last(s_j,i - 1)$ (defined in \cref{def:last}) and let $x_{j+1}$ be the vertex it returns. 
In this way, $x_{j + 1}$ is uniquely determined by $s_j, r_{\leq i - 1}, g_{\leq i - 1}$,  independent of the extra parameter $\vk$. 
We will explain the benefit of this later, after giving the formal definition of $\walk^{\vk^1, \vk^2}$. 

\renewcommand{\j}{\psi}

\newcommand{\bj}{\bm{\psi}}

\begin{construction}{The Relaxed Extended Random Walk Probability space $\extwalk^{\vk^1,\vk^2}_{\ell,m,n,a}$}
	\begin{itemize}
		\item \textbf{Setup.} We sample the random variables as follows:
		 \begin{itemize}
		     \item Draw the starting vertex $\bs \getsR [n]$.
		     \item  Sample $\{\bg_i\}_{i\in [\ell]}$ and $\{\br_i\}_{i\in [\ell]}$, which together determine a sample $\bh \getsR \calH_{\ell,m,n}$ from the pseudorandom hash family. 
		     \item Then, for each $i \in [\ell]$, we extend the domain of $\bg_i$ and $\br_i$  from $[m]$ to $[m] \cup \{\star_0, \star_1, \dots\}$ as follows: for every $\star_t$ we sample $\bg_i(\star_t)\getsR \{0,1\},\br_i(\star_t)\getsR[n]$, where the samples are independent across all $\star_t$.
		 \end{itemize}

		\item \textbf{Generating the walk.} The sampled $\{g_i\}_{i\in [\ell]},\{r_i\}_{i\in [\ell]}$ and $s$ uniquely determine a sequence $w^{\vk^1, \vk^2} = (w^{\vk^1, \vk^2}_1,w^{\vk^1, \vk^2}_2,\dots)$ of vertices returned by the function $\walk^{\vk^1, \vk^2}(s,\ell,0,(0,0,\dots,0))$ defined in Algorithm~\ref{algo:relaxedextwalk}. 

		\item \textbf{Assigned values.} Here, in additional to $\a_i(\mu), \nex(\mu), \rig(\mu), \level(\mu)$, we also we explicitly assign $\index(\mu)$ just to emphasize the index of $\mu$.
		We underline the parts where we assign these values, and note that they have no effect on the returned value of the function. We also write $\walk^{\vk^1, \vk^2}(s',i, \vk)$ and drop $\mu_0$ when we only need its return value. 

	\end{itemize}
\end{construction} 

\begin{algorithm}[H]\label{algo:relaxedextwalk}
\DontPrintSemicolon
	\caption{Generating a relaxed extended walk}
	
	\SetKwProg{Fn}{Function}{}{}
	\Fn{$\walk^{\vk^1, \vk^2}(s',i,\underline{\mu_0},\vk)$  (where $s'\in [n], 0\le i \le \ell$)}{
	\lIf{$i = 0$}{\Return sequence $(s')$ which contains a single vertex.}
	\eIf{$[\forall t \in [i + 1, \ell], k_t = k^2_t] \land [\exists t \in [i + 1, \ell], k^1_t < k^2_t] \land [k^2_i > 0]$ \label{line:when-inherit} }{
	\tcc{This condition says that $\mu_0 \not\in p(\vk^1)$ and is the last node of $p(\vk^2)$ above level $i$ and $p^*(\vk^2)$ is non-empty on level $i$. It is equivalent to $\mu_1, \mu_2, \dots \in p^*(\vk^2) \setminus p^*(\vk^1)$.} 
    	$C_0 \gets C^{i, \vk^1}$ \label{line:C_0-1}
	}{ 
    	$C_0 \gets \emptyset$ \label{line:C_0-2}
	}
	$\mathsf{star} \gets \text{false}, j \gets 0, s_0 \gets s', w = ()$.\\
	\Repeat{$g_i(y) = 0$}{
	$\vk' \gets (0, 0, \dots, 0, j, k_{i+1}, \dots, k_{\ell})$ \label{line:setjinkprime} \tcc{\underline{Here $\vk'$ equals $\index(\mu_j)$.}}
	$w = w \circ \walk^{\vk^1, \vk^2}(s_j, i - 1, \underline{\mu_0 + |w|}, \vk')$ \tcc{\underline{Here $\mu_0 + |w|$ equals $\mu_j$.}}
	$x_{j + 1} \gets \last(s_j,i - 1)$\\
	 $y, \mathsf{star} \gets \begin{cases} a_{x_{j + 1}}, \text{false} & \text{ if } a_{x_{j + 1}} \not\in C_j \land \lnot \mathsf{star} \\ \star_t, \text{true} & \text{ otherwise (where $t = \min \{t \in \N \ \vert \ \star_t \not\in C\}$)}\end{cases}$ \label{line:assign-y}\\
	 \underline{Let $\mu_{j + 1} = \mu_0 + |w|$, $\x_i(\mu_{j + 1}) \gets x_{j + 1}, \a_i(\mu_{j + 1}) \gets y$.} \label{line:assign-tag}\\
	 \lIf{$j > 0$}{\underline{$\rig(\mu_j) \gets \mu_{j + 1}$}}
		\If{$g_i(y) = 1$}{
		$C_{j + 1} \gets C_j \cup \{y\}$, $s_{j + 1} \gets r_i(y)$\\
		\underline{$\level(\mu_{j+1}) \gets i, \nex(\mu_{j+1}) \gets r_i(y)$}\\
		\underline{$\index(\mu_{j+1}) \gets (0,\cdots,0, j + 1,  k_{i+1}, \dots, k_{\ell})$}\label{line:assign-index-mu}\\
			$j \gets j + 1$.
		}
	}
\If {$[\forall t \in [i + 1, \ell], k_t = k^1_t] \land [\exists t \in [i + 1, \ell], k^1_t < k^2_t]$ \label{line:whenCk1}}{
\tcc{This condition says that $\mu_0$ is the last node of $p(\vk^1)$ above level $i$ and is not the last node of $p(\vk^2)$. It is equivalent to $\mu_1, \mu_2, \dots \in p^*(\vk^1) \setminus p^*(\vk^2)$ when $k^1_i > 0$.}
		    	 $C^{i, \vk^1} \gets C_{\min(j, k^1_i)}$} \label{line:assing-C1}
	\Return $w$.
	}
\end{algorithm}

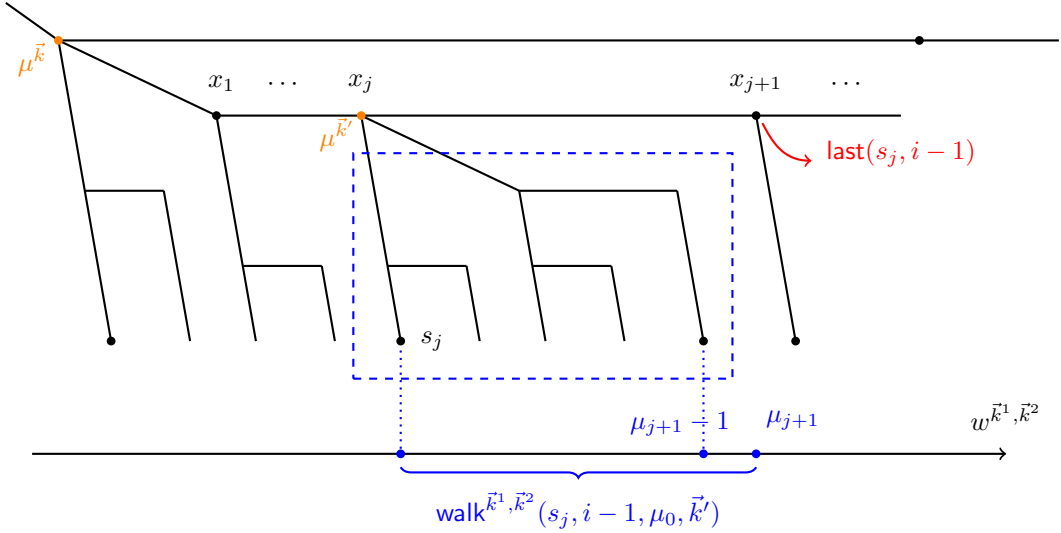
\begin{figure}
	\centering
\begin{tikzpicture}[xscale = 0.7, yscale = 0.5, thick]
  \draw (0,0) -- (1,-1);
  \draw (1,-1) -- (20, -1);
  \draw (1,-1) -- (2, -9);
  \node [draw,circle,fill,minimum size=2.5,inner sep=0pt, outer sep=0pt] at (2,-9) {};
  \draw (1.5,-5) -- (3,-5);
  \draw (3,-5) -- (3.5,-9);

  \draw (1,-1) -- (4, -3);
  \node [draw,circle,fill,minimum size=2.5,inner sep=0pt, outer sep=0pt] at (4,-3) {};
  \node at (4.1,-2.15) {\small $x_1$};
  \node at (5.3,-2.15) {\small $\cdots$};
  \draw (4,-3) -- (4.75, -9);
  \draw (4.5,-7) -- (6,-7);
  \draw (6,-7) -- (6.25,-9);

  \draw (4,-3) -- (17, -3);
  \draw (6.75,-3) -- (7.5, -9);
  \node at (6.75,-2.15) {\small $x_j$};
  \node [draw,circle,fill,minimum size=2.5,inner sep=0pt, outer sep=0pt] at (7.5,-9) {};
  \node at (8.1,-9) {\small $s_j$};
  \draw (7.25,-7) -- (8.75,-7);
  \draw (8.75,-7) -- (9,-9);

  \draw (6.75,-3) -- (9.75, -5);
  \draw (9.75,-5) -- (12.75,-5);
  \draw (9.75,-5) -- (10.25,-9);
  \draw (10,-7) -- (11.5,-7);
  \draw (11.5,-7) -- (11.75,-9);
  \draw (12.75, -5) -- (13.25, -9);
  \node [draw,circle,fill,minimum size=2.5,inner sep=0pt, outer sep=0pt] at (13.25,-9) {};

  \draw (14.25, -3) -- (15,-9);
  \node [draw,circle,fill,minimum size=2.5,inner sep=0pt, outer sep=0pt] at (15,-9) {};
  \node [draw,circle,fill,minimum size=2.5,inner sep=0pt, outer sep=0pt] at (14.25,-3) {};
  \node at (14.25,-2.15) {\small $x_{j + 1}$};
  \draw [red, ->] (14.37,-3.2) to [out=290,in=180] (15.3, -4.2);
  \node [red] at (17, -4) {\small $\mathsf{last}(s_j, i - 1)$};

  \node at (16,-2.15) {\small $\cdots$};

  \node [draw,circle,fill,minimum size=2.5,inner sep=0pt, outer sep=0pt] at (17.35,-1) {};

  \draw [->] (0.5, -12) -- (19, -12);
  \node at (19, -11) {\small $w^{\vec{k}^1, \vec{k}^2}$};
  \draw [dotted, blue] (7.5,-9.25) -- (7.5, -12);
  \node [draw,circle,blue,fill=blue,minimum size=2.5,inner sep=0pt, outer sep=0pt] at (7.5,-12) {};

  \draw [dotted, blue] (13.25,-9.25) -- (13.25, -12);
  \node [draw,circle,blue,fill=blue,minimum size=2.5,inner sep=0pt, outer sep=0pt] at (13.25,-12) {};
  \node [blue] at (12.8,-11.2) {\small $\mu_{j + 1} - 1$};
  \node [draw,circle,blue,fill=blue,minimum size=2.5,inner sep=0pt, outer sep=0pt] at (14.25,-12) {};
  \node [blue] at (14.95,-11.1) {\small $\mu_{j + 1}$};
  \draw [decorate,decoration={brace,amplitude=4pt},xshift=0pt,yshift=7pt, blue] (14.25,-12.6)--(7.5,-12.6);
  \draw [blue, draw, dashed] (6.6,-4) rectangle (13.8,-10);
  \node [blue] at (10.875, -13.5) {\small $\mathsf{walk}^{\vec{k}^1, \vec{k}^2}(s_j, i - 1, \mu_0, \vec{k}')$};

  \node [draw,circle,fill,minimum size=2.5,inner sep=0pt, outer sep=0pt, orange] at (6.75,-3) {};
  \node [orange] at (6.3, -3.5) {\small $\mu^{\vec{k}'}$};

  \node [draw,circle,fill,minimum size=2.5,inner sep=0pt, outer sep=0pt, orange] at (1,-1) {};
  \node [orange] at (0.5, -1.5) {\small $\mu^{\vec{k}}$};
\end{tikzpicture}
	\caption{$\walk^{\vk^1, \vk^2}(s, i, \mu_0, \vk)$. $\vk$ is the index of $\mu_0$. For each $j$, it recursively call $\walk^{\vk^1, \vk^2}(s_j, i - 1, \mu_j, \vk')$ to generate a subwalk of level $\leq i - 1$. $w^{\vk^1, \vk^2}_{\mu_{j + 1}}$ is also generated by this subwalk, and it is the last vertex of this subwalk. However, here $x_{j + 1}$ may not equal $w^{\vk^1, \vk^2}_{\mu_{j + 1}}$. This is due to the fact that it is now determined by $\last(s_j, i - 1)$ which is the last vertex of the original (not relaxed) subwalk. } 
	\label{fig:determine}
\end{figure}

The structure of $\walk^{\vk^1, \vk^2}(s, i, \mu_0, \vk)$ is illustrated in \Cref{fig:determine}. We first set up some notations. 

\paragraph*{Notation.} 
Throughout this section, we fix $\ell,m,n \in \N$, $a \in [m]^n$ and $\vk^1, \vk^2 \in \N^\ell, \vk^1<\vk^2$, and we will always work with the (relaxed) extended walk $\extwalk^{\vk^1, \vk^2}_{\ell,m,n,a}$. We use $\bw^{\vk^1, \vk^2},\bs,\bg,\br,\level,\nex,\bh$ to denote the corresponding random variables in the extended walk. We also use $\bT$ to denote the dependency tree built on the extended walk $\bw^{\vk^1, \vk^2}$.

Note that $\bw^{\vk^1, \vk^2},\level,\nex,\bh,\bT$ are all determined by $(\bs,\bg,\br)$. Then we recall the following shorthand. 

\begin{reminder} {\Cref{def:short-hand-for-a}}
We denote $\x(\mu) = \x_{\level(\mu)}(\mu)$ and $\a(\mu) = \a_{\level(\mu)}(\mu)$. 
\end{reminder}

Recall the definition of collision is as following.

\begin{reminder}{\Cref{def:collision}}
Let $p(\mu)$ denote the path on $T$ from root to node $\mu$. We say $(\alpha, \beta)$ is a \emph{collision between two paths} $p_1$ and $p_2$ if $\alpha \in p_1 \setminus p_2, \beta \in p_2 \setminus p_1$, and $(\a(\alpha),\level(\alpha)) =(\a(\beta),\level(\beta))$ (where $\a(\cdot)$ is defined in \Cref{def:short-hand-for-a}). The level of a collision $(\alpha, \beta)$ is the same as $\level(\alpha)$ (which equals $\level(\beta)$). 
\end{reminder}

We make the following definition according to the condition at Line \ref{line:when-inherit} and Line \ref{line:whenCk1}.

\begin{definition} \label{def:consistent_vk2}
We say an index $\vk$ is consistent with $\vk^2$ but not $\vk^1$ above level $i$ if and only if $[\forall t \in [i + 1, \ell], k_t = k^2_t] \land [\exists t \in [i + 1, \ell], k^1_t < k^2_t] \land [k^2_i > 0]$. This means $\vk$ is the same as $\vk^2$ above level $i$ and different from $\vk^1$ above level $i$. Besides, it also requires $k^2_i$ to be nonzero. 

Similarly, we say an index $\vk$ is consistent with $\vk^1$ but not $\vk^2$ above level $i$ if and only if $[\forall t \in [i + 1, \ell], k_t = k^1_t] \land [\exists t \in [i + 1, \ell], k^1_t < k^2_t]$. This means $\vk$ is the same as $\vk^1$ above level $i$ and different from $\vk^2$ above level $i$. Note here we do not require $k^1_i$ to be nonzero. 
\end{definition}

We have the following lemma about these conditions. 

\begin{lemma} \label{obs:conditions}
Consider a function call $\walk^{\vk^1, \vk^2}(s', i, \mu_0, \vk)$. Suppose it generates nodes $\mu_1, \mu_2, \dots, \mu_t$. 
\begin{enumerate}
    \item If $\vk$ is consistent with $\vk^2$ but not with $\vk^1$ above level $i$, then $\mu_1, \mu_2, \dots, \mu_t \in p^*(\vk^2) \setminus p^*(\vk^1)$. Moreover, $\mu_j = \vidx^{\vk^2}_{i,j}$ for $j \in [k^2_i]$ if $\vidx^{\vk^2}$ exists. \label{item:cond-k2}
    \item If $\vk$ is consistent with $\vk^1$ but not with $\vk^2$ above level $i$, then $\mu_1, \mu_2, \dots, \mu_t \in p^*(\vk^1) \setminus p^*(\vk^2)$ when $k^1_i > 0$. Moreover, $\mu_j = \vidx^{\vk^1}_{i,j}$ for $j \in [k^1_i]$ if $\vidx^{\vk^1}$ exists. \label{item:cond-k1}
\end{enumerate}
\end{lemma}
\begin{proof}
By Line \ref{line:assign-index-mu}, we know $\index(\mu_j) = (0,\cdots,0, j,  k_{i+1}, \dots, k_{\ell})$. 

Suppose $\vk$ is consistent with $\vk^2$ but not with $\vk^1$ above level $i$. From (1) $k^2_i \geq 1 = [\index(\mu_1)]_i$ since $k^2_i > 0$, and (2) $[\index(\mu_1)]_t = k^2_t$ for every $t \in [i + 1, \ell]$, it follows that $\index(\mu_1)$ is an ancestor of $\vk^2$ according to \Cref{def:ancestor}. Hence we have $\mu_1 \in p(\mu^{\vk^2})$. Since by \Cref{def:extention}, $\suc(\mu_1) \subset p^*(\mu^{\vk^2})$ where $\suc(\mu_1)$ contains the level $i$ descendants of $\mu_1$. Consequently, $\mu_1, \mu_2, \dots, \mu_t \in p^*(\mu^{\vk^2})$. 

On the other hand, since there exists ${t'} \in [i + 1, \ell], k^1_{t'} < k^2_{t'} = k_{t'}$, the same argument shows that $\mu_1 \not\in p(\mu^{\vk^1})$. Hence $\mu_1, \mu_2, \dots, \mu_t \not\in p^*(\mu^{\vk^1})$. Together, they imply $\mu_1, \mu_2, \dots, \mu_t \in p^*(\vk^2) \setminus p^*(\vk^1)$. Moreover, when $\vidx^{\vk^2}$ exists, we have $t \geq k^2_i$. Since $[\index(\mu_j)]_i = j$ and $\mu_j \in p(\vk^2)$, we know that $\mu_j = \vidx^{\vk^2}_{i,j}$. 

The case when $\vk$ is consistent with $\vk^1$ but not with $\vk^2$ above level $i$ and $k^1_i > 0$ follows from the same argument.
\end{proof}

The following lemma is the main purpose of this relaxation. 
\begin{lemma}\label{obs:no-colision-k1-k2}
Fix two sequences $\vk^1, \vk^2$ such that $\vk^1 < \vk^2$. For all $(w^{\vk^1, \vk^2}, T) \in \supp(\bw^{\vk^1, \vk^2}, \bT)$ such that $\vidx^{\vk^1}$ and $\vidx^{\vk^2}$ exist, we know there is no collision between $p(\vk^1)$ and $p(\vk^2)$.
\end{lemma}
\begin{proof}
For each level $i \in [\ell]$, we prove there is no collision of level $i$ between these two paths. If $p(\vk^1)$ and $p(\vk^2)$ intersects on level $i$, we know there is no collision between them since one must contain the other on level $i$. 

If they do not intersect on level $i$, since $\vk^1 < \vk^2$, there must be $t \in [i + 1, \ell]$ such that $k^1_t < k^2_t$. The corner case when $k^1_i = 0$ or $k^2_i = 0$ is evident since then there cannot be level $i$ collisions. Hence, without loss of generality, we can assume that $k^1_i > 0$ and $k^2_i > 0$. 

Since $\mu^{\vk^1}$ exists, consider the function call $\walk^{\vk^1, \vk^2}(s,i,\mu,\vk)$ in which $\vk$ is consistent with $\vk^1$ but not with $\vk^2$ above level $i$. By $k^1_i > 0$ and \Cref{obs:conditions}~\eqref{item:cond-k1}, $\mu_j = \vidx^{\vk^1}_{i,j}$ for $j \in [k^1_i]$ (\ie, the function call generates the level $i$ nodes on $p(\vk^1)$). From Line \ref{line:assing-C1} of Algorithm \ref{algo:relaxedextwalk}, we know $C^{i,k^1} = C_{k^1_i} = \{\a_i(\alpha) \ \vert \ \alpha \in p(\vk^1) \land \level(\alpha) = i\}$.

Then since $\vidx^{\vk^2}$ exists, consider the function call $\walk^{\vk^1, \vk^2}(s, i, \mu, \vk)$ in which $\vk$ is consistent with $\vk^2$ but not with $\vk^1$ above level $i$. By \Cref{obs:conditions}~\eqref{item:cond-k2}, we know $\mu_j = \vidx^{\vk^2}_{i,j}$ for $j \in [k^2_i]$ (\ie, the function call generates the level $i$ nodes on $p(\vidx^{\vk^2})$). In this function call, from Line \ref{line:C_0-1} of Algorithm \ref{algo:relaxedextwalk}, we have $C_0 = C^{i,k^1}$.  Therefore, by how $\a_i(\mu_j)$ is assigned at Line \ref{line:assign-tag} and how $y$ is assigned at Line \ref{line:assign-y}. We know $\a_i(\mu_j) \not\in C_0 = C^{i,k^1}$. Thus $\a_i(\mu_j) \not= \a_i(\alpha)$ for all level $i$ nodes $\alpha$ on $p(\vk^1)$. 
\end{proof}

We have the following lemma similar to \cref{obs:w-nex} whose proof is also the same as that of \cref{obs:w-nex}.
\begin{lemma} \label{obs:w-nex-2}
Fix $(w^{\vk^1, \vk^2}, T) \in \supp(\bw^{\vk^1, \vk^2}, \bT)$. We have $w^{\vk^1, \vk^2}_{\mu + 1} = \nex(\mu)$ for every $\mu \in [|w^{\vk^1, \vk^2}| - 1]$. 
\end{lemma}

We also prove an analogue of \Cref{obs:determine-a-from-x}.

\begin{lemma}\label{obs:determine-a-from-x-2}
In $\walk^{\vk^1, \vk^2}(s',i, \mu_0)$, one can uniquely determine $\a_i(\mu_j)$ from $C_0, x_1, x_2, \dots, x_j$ as follows:
\begin{enumerate}
	\item Let $j' = \min \{j' \ \vert \ [\exists  j'' \text{ s.t. } 1 \leq j'' < j' \leq j,  a_{x_{j''}} = a_{x_{j'}}] \lor [a_{x_{j'}} \in C_0]\}$. 
	\item If no such $j'$ exists, then $\a_i(\mu_j) = a_{x_j}$. Otherwise, let $t_0 = \#\{t \geq 0 \ \vert \ \star_t \in C_0\}$. We have $\a_i(\mu_j) = \star_{t_0 + j - j'}$. 
\end{enumerate}
\end{lemma}
\begin{proof}
By Line \ref{walk:a} and Line \ref{walk:Cj}, we know $C_j = C_0 \cup \{\a_i(x_1), \a_i(x_2), \dots, \a_i(x_j)\}$. By Line \ref{walk:y}, we know $\mathsf{star}$ switches from $\text{false}$ to $\text{true}$ when $a_{x_{j + 1}} \in C_j$. For those $j$ before $\mathsf{star}$ switches, $\a_i(x_j) = a_{x_j}$, and for those $j$ after switch, $\a_i(x_j) = \star_*$. 

Hence, $\mathsf{star}$ switches at the first $j'$ such that there either exists $1 \leq j'' < j'$ with $a_{x_{j''}} = a_{x_{j'}}$, or $a_{x_{j'}} \in C_0$. If no such $j'$ exists, then $\mathsf{star}$ is still false at $j$, and hence $\a_i(\mu_j) = a_{x_j}$. Otherwise, $\mathsf{star}$ switches at $j'$, and by Line \ref{walk:y} we have $\a_i(\mu_{j'}) = \star_{t_0}, \a_i(\mu_{j' + 1}) = \star_{t_0 + 1},\dots$, and $\a_i(\mu_j) = \star_{t_0 + j - j'}$. 
\end{proof}

The following lemma follows from essentially the same proof of \Cref{obs:determine-a-from-xi} by replacing \Cref{obs:determine-a-from-x} with \Cref{obs:determine-a-from-x-2} in the proof. 

\begin{lemma}\label{obs:determine-a-from-xi-2}
Fix $\vk^1, \vk^2 \in \N^\ell$ such that $\vk^1 < \vk^2$, $(w^{\vk^1, \vk^2}, T) \in (\bw^{\vk^1, \vk^2}, \bT)$, and $\vk \in \N^\ell$. Suppose $\vidx^{\vk}_{i,1}, \dots, \vidx^{\vk}_{i,j - 1}$ exist. If $\vk$ is consistent with $\vk^2$ but not $\vk^1$ above level $i$ (as defined in \Cref{def:consistent_vk2}), let $\bar{C}_0 = C^{i,\vk^1}$. Otherwise, let $\bar{C}_0 = \emptyset$. 

From $\bar{C}_0, \bar{x}_1 = \x_i(\vidx^{\vk}_{i,1}), \bar{x}_2 = \x_i(\vidx^{\vk}_{i,2}), \dots, \bar{x}_{j - 1} = \x_i(\vidx^{\vk}_{i,j - 1}), \bar{x}_j = \x_i(\rig_i(\vidx^{\vk}_{i,j - 1}))$, one can uniquely determine $\a_i(\rig_i(\vidx^{\vk}_{i,j - 1}))$ as follows:
\begin{enumerate}
	\item Let $j' = \min \{j' \ \vert \ [\exists  j'' \text{ s.t. } 1 \leq j'' < j' \leq j,  a_{\bar{x}_{j''}} = a_{\bar{x}_{j'}}] \lor [a_{\bar{x}_{j'}} \in \bar{C}_0]\}$. 
	
	\item If no such $j'$ exists, then $\a_i(\rig_i(\vidx^{\vk}_{i,j - 1})) = a_{\bar{x}_j}$. Otherwise, let $t_0 = \#\{t \geq 0 \ \vert \ \star_t \in \bar{C}_0\}$. We have $\a_i(\rig_i(\vidx^{\vk}_{i,j - 1})) = \star_{t_0 + j - j'}$. 
\end{enumerate}
\end{lemma}
\begin{proof}
Consider the function call $\walk^{\vk^1, \vk^2}(s, i, \mu_0, \vk^0)$ in which we assign $a_i(\rig_i(\vidx_{i,j - 1}^{\vk}))$. 

In such function call, by Line \ref{line:setjinkprime} of Algorithm \ref{algo:relaxedextwalk}, we know $k'_t = k^0_t$ for all $t \in [i + 1, \ell]$. Since $\vk'$ is simply the index of $\mu_j$ (by Line  \ref{line:setjinkprime}, \ref{line:assign-index-mu}), when we assign $\a_i(\rig_i(\vidx^{\vk}_{i,j - 1}))$ (at Line \ref{line:assign-tag}), $\vk' = \vK_{i,j}$. Together, we know $k_t = k^0_t$ for all $t \in [i + 1, \ell]$. 
Noticing in this function call, Line \ref{line:when-inherit} only depends on $k^0_t$ for those $t \in [i + 1, \ell]$, therefore by definition of $\bar{C}_0$, we know $C_0 = \bar{C}_0$. Hence we can apply \Cref{obs:determine-a-from-x-2}, and this concludes the proof. 
\end{proof}

\begin{lemma}\label{lem:out-k2}
	Fix $\vk^1, \vk^2 \in \N^\ell$ such that $\vk^1 < \vk^2$ and $(w^{\vk^1, \vk^2}, T) \in (\bw^{\vk^1, \vk^2}, \bT)$. For every node $\mu \leq \vidx^{\vk^2}$, we have $w^{\vk^1,\vk^2}_{\mu} = \x_i(\mu)$ for every $i \in [\level(\mu)]$.  
\end{lemma}
\begin{proof}
Consider the function call which assign $\x_i(\mu)$ at Line \ref{line:assign-tag}. Note that $w^{\vk^1,\vk^2}_{\mu}$ is the last vertex of $\walk^{\vk^1, \vk^2}(s_j, i - 1, \vk')$, while $\x_i(\mu)$ is the last vertex of  ${\walk}(s_j, i - 1)$. Here we have $j = k_i'$, since  in Algorithm~\ref{algo:relaxedextwalk} we let $k'_i \gets j$ at Line~\ref{line:setjinkprime} before calling $\walk^{\vk^1, \vk^2}(s_j, i - 1, \vk')$. Then, we know that the index of $\mu$ is $(0, 0, \dots 0, k'_{\level(\mu)} + 1,  k'_{\level(\mu) + 1}, \dots, k'_{\ell})$ with $\level(\mu) \geq i$. 

Suppose for contradiction that  $w^{\vk^1,\vk^2}_{\mu} \not= \x_i(\mu)$. Then it must be the case that, at the beginning of $\walk^{\vk^1, \vk^2}(s_j, i - 1, \vk')$ (or one of its recursive calls on lower levels), the set $C_0$ is initialized to an non-empty set
 (since otherwise the behavior of $\walk^{\vk^1, \vk^2}(s_j, i - 1, \vk')$ and ${\walk}(s_j, i - 1)$ would be exactly the same).
 
Let $\walk^{\vk^1, \vk^2}(\cdot, i', \vk'')$ ($i' \le i - 1$) be the recursive call where $C_0$ is not initialized empty.
Since it is an recursive call made by $\walk^{\vk^1, \vk^2}(s_j, i - 1, \vk')$, we have  $\forall t \geq i, k''_t = k'_t$. 
We also have
 $\forall t \geq i'+1, k''_t = k^2_t$, which follows from $C_0\neq \emptyset$ and the condition at Line~\ref{line:when-inherit} in Algorithm~\ref{algo:relaxedextwalk}.
 Together they imply $k'_t = k_t^2$ for all $t\ge i$.
 
Hence, the index of $\mu$ can be alternatively written as $(0, 0, \dots 0, k^2_{\level(\mu)} + 1,  k^2_{\level(\mu) + 1}, \dots, k^2_{\ell})$ where $\level(\mu) \geq i$. 
This contradicts $\mu \le \mu^{\vk^2}$. 
\end{proof}

The following lemma relates our relaxed extended walk $w^{\vk^1,\vk^2}$ to the actual reachable set $f^*_{a,h}(s)$ in the original walk, provided that they are defined using the same $\{g_i\}_{i\in [\ell]}$, $\{r_i\}_{i \in [\ell]}$ and $s\in [n]$.
 
\newcommand{\bbi}{\phi}
\newcommand{\bbj}{\psi}

\begin{lemma}\label{lem:subset-of-extwalk-2}
Fix $\vk^1, \vk^2 \in \N^\ell$ such that $\vk^1 < \vk^2$ and $(w^{\vk^1, \vk^2}, T,h,s) \in (\bw^{\vk^1, \vk^2},  \bT,\bh,\bs)$.

Suppose $\vidx^{\vk^1}$ and $\vidx^{\vk^2}$ exist. Let $\bbi = \vidx^{\vk^1} + 1, \bbj = \vidx^{\vk^2} + 1$. If there are no two distinct $\alpha,  \beta \in [\bbj - 1]$ such that $a_{w^{\vk^1, \vk^2}_{\alpha}}=a_{w^{\vk^1, \vk^2}_{\beta}}$, then $w^{\vk^1, \vk^2}_{\bbi}, w^{\vk^1, \vk^2}_{\bbj} \in f^*_{a,h}(s)$. 
\end{lemma}
\begin{proof}
To simplify notation, in this proof we drop the superscript $\vk^1,\vk^2$ on the variable $w^{\vk^1,\vk^2}$ representing the relaxed extended walk, and simply write $w$ instead.

The proof is similar to that of \cref{lem:subset-of-extwalk}. 
We will first prove that, for every $\mu \leq \bbj - 1$, if there are no $\alpha,  \beta \in [\mu], \alpha \neq \beta$ such that $a_{w_{\alpha}}=a_{w_{\beta}}$, then we must have $\a_i(\mu) = a_{w_{\mu}}$ for all $i \in [\level(\mu)]$. 

To prove the statement above, we again use induction on $\mu$. Suppose the hypothesis holds for $1, 2, \dots, \mu - 1$. Since $\mu \leq \bbj - 1$, by \cref{lem:out-k2}, we know $w_{\mu} = \x_i(\mu)$ for $i \in [\level(\mu)]$. Now we will show that $\a_i(\mu) \not= \star_*$ for $i \in [\level(\mu)]$, which will immediately imply $\a_i(\mu) = a_{\x_i(\mu)} = a_{w_{\mu}}$ and finish the inductive step. 

Suppose for contradiction that we assigned $\a_i(\mu) = \star_*$ in Algorithm~\ref{algo:relaxedextwalk}. Then, at this point, the only two cases are (1) $a_{w_{\mu}} \in C_j$, or (2) $\star_* \in C_j$.
The main difference with \cref{lem:subset-of-extwalk} is that now the initial value of $C_0$ may be a non-empty set $C^{i, \vk^1}$. 
But we can still see that there must be a node $\eta < \mu$ such that (1) $a_{w_{\eta}} = a_{w_{\mu}}$, or (2) $\a_i(\eta) = \star_*$ while $\level(\eta) = i$. Case (1) contradicts our assumption of $a_{w_\alpha} \neq a_{w_\beta}$ for all $\alpha,\beta\in [\psi - 1], \alpha \neq \beta$. Case (2) contradicts the inductive hypothesis of $\a_i(\eta) = a_{w_{\eta}} \not= \star_*$.
Therefore we must have $\a_i(\mu) = a_{w_{\mu}} \not= \star_*$ for all $i \in [\level(\mu)]$. 

Again, similar to \cref{lem:subset-of-extwalk}, for such $\mu$, for all $i \in \level(\mu)$, $g_i(\a_i(\mu)),r_i(\a_i(\mu))$ will have the same values as the pseudorandom functions $g_i(a_{w_{\mu}}),r_i(a_{w_{\mu}})$ that were used to define $h(a_{w_{\mu}})$ for $h \in \mathcal{H}_{\ell,m,n}$. Then, it is evident that $w_{\mu+1} = \nex(\mu) = h(a_{w_\mu})$ (where the first equality follows from  \Cref{obs:w-nex-2}).

Suppose the actual reachable set $f^*_{a,h}(s)$ has vertices $\{w'_1,w'_2,\dots\}$ where $w'_1=s$ and $w'_{\mu+1} = h(a_{w'_{\mu}})$. By our induction before, for every $\mu \leq \bbj - 1$ such that no $\alpha, \beta \in [\mu]$ satisfy $a_{w_{\alpha}}=a_{w_{\beta}}$, we must have $w'_{\mu + 1}=w_{\mu + 1}$. Hence, we know $w_{\eta} = w'_{\eta}$ for all $\eta \in [\bbj]$. In particular, we have $w_{\bbi}, w_{\bbj} \in f^*_{a,h}(s)$.
\end{proof}

We also observe that the following lemma holds with the same proof as \Cref{lem:existence}. 

\begin{lemma}\label{lem:existence-2}
	Fix $\vk^1, \vk^2 \in \N^\ell$ such that $\vk^1 < \vk^2$, $(w^{\vk^1, \vk^2}, g, r) \in (\bw^{\vk^1, \vk^2}, \bg, \br)$, and $\vk \in \N^\ell$, the following hold:
\begin{enumerate}[label=(\alph*)]
	\item Suppose $\mu^{\vk}_{i,j - 1}$ exists, $\mu^{\vk}_{i,j}$ exists if and only if $g_i(\a_i(\rig_i(\mu^{\vk}_{i,j - 1}))) = 1$. \label{item:exist-2}
	\item $\x_i(\rig_i(\mu^{\vk}_{i,j - 1})) = \last(\nex(\mu^{\vk}_{i,j - 1}), i - 1)$ \label{item:xi-2}
\end{enumerate}
\end{lemma}

Recall $\vidx[\vk]$ has the same meaning as $\vidx^{\vk}$. Its following generalization to multiple paths also holds. 

\begin{lemma} \label{obs:existence-K-2}
Fix $\vk^1, \vk^2 \in \N^\ell$ such that $\vk^1 < \vk^2$, $(w^{\vk^1, \vk^2}, g) \in \supp(\bw^{\vk^1, \vk^2}, \bg)$, and $\vK \subseteq \N^\ell$ such that $\vk^1, \vk^2 \in \vK$. The following holds:
\begin{enumerate}[label=(\alph*)]
	\item Suppose $\vidxb{\pa^{\vK}_{i,j}}$ exists, $\vidx^{\vK}_{i,j}$ exists if and only if $g_i\left(\a_i\left(\rig_i\left(\vidxb{\pa^{\vK}_{i,j}}\right)\right)\right) = 1$. \label{item:exist-K-2}
	\item $\x_i\left(\rig_i\left(\vidxb{\pa^{\vK}_{i,j}}\right)\right) = \last\left(\nex\left(\vidxb{\pa^{\vK}_{i,j}}\right), i - 1\right)$ \label{item:xi-K-2}
	\item Suppose $\vidx^{\vK}_{i,1}, \dots, \vidx^{\vK}_{i,j - 1}$ and $\vidxb{\pa^{\vK}_{i,j}}$ exist. From $\tilde{x}_1 = \x_i\left(\vidx^{\vK}_{i,1}\right), \tilde{x}_2 = \x_i\left(\vidx^{\vK}_{i,2}\right), \dots, \tilde{x}_{j - 1} = \x_i\left(\vidx^{\vK}_{i,j - 1}\right), \tilde{x}_j = \x_i\left(\rig_i\left(\vidxb{\pa^{\vK}_{i,j}}\right)\right)$, one can uniquely determine $\a_i\left(\rig_i\left(\vidxb{\pa^{\vK}_{i,j}}\right)\right)$. \label{item:determine-a-from-xi-K-2}
	\item If $\vidx^{\vK}_{i,j}$ exists, then $\vidx^{\vK}_{i,j} = \rig_i\left(\vidxb{\pa^{\vK}_{i,j}}\right)$. \label{item:exist-equal-2}
 \end{enumerate}
\end{lemma}
\begin{proof}
\ref{item:exist-K-2} and \ref{item:xi-K-2} follows from the same proof as \Cref{obs:existence-K} by replacing \Cref{lem:existence} with \Cref{lem:existence-2}. \ref{item:exist-equal-2} follows from exact the same proof of  \Cref{obs:existence-K}~\ref{item:exist-equal}.

For \ref{item:determine-a-from-xi-K-2}, note there are two cases. From the index $\vK_{i,j}$, we can tell when we assign $\a_i\left(\rig_i\left(\vidxb{\pa^{\vK}_{i,j}}\right)\right)$, whether $C_0 = \emptyset$ or $C^{i, \vk^1}$. 

If $C_0 = \emptyset$, it follows the same proof as \Cref{obs:existence-K} \ref{item:determine-a-from-xi-K} by replacing \Cref{obs:determine-a-from-xi} with \Cref{obs:determine-a-from-xi-2}. 

If $C_0 = C^{i, \vk^1}$. Since $\vk^1 \in \vK$ and $\vk^1 < \vk^2$, we know there exists $1 \leq j_0 \leq j_1 < j$ such that level $i$ nodes on $p(\vk^1)$ are exactly $\vidx^{\vK}_{i,j_0}, \vidx^{\vK}_{i,j_0 + 1}, \dots, \vidx^{\vK}_{i,j_1}$. (Note when there is no level $i$ node on $p(\vk^1)$, we have $C_0 = C^{i, \vk^1} = \emptyset$ which belongs to the previous case.) When determine $\a_i$ of these nodes, we have $C_0 = \emptyset$. Thus by applying the previous case, we can determine $\a_i(\alpha)$ for all $\alpha \in p(\vk^1), \level(\alpha) = i$. Then $C^{i, \vk^1}$ contains exactly $\a_i(\alpha)$ for all such $\alpha$. Thus as $C^{i, \vk^1}$ is determined, we can apply \Cref{obs:determine-a-from-xi-2} to conclude the proof. 
\end{proof}

\begin{definition} \label{def:func-A}
Fix $\vk^1, \vk^2 \in \N^\ell$ such that $\vk^1 < \vk^2, \vK$, $\vK \subseteq \N^\ell$ such that $\vk^1, \vk^2 \in \vK$, and level $i \in [\ell]$. There exists a function family $\{\mathsf{A}^{\vk^1, \vk^2, \vK,i,j}\}_{j \in K_i}$ which maps $(x_1, x_2, \dots, x_j) \in [n]^j$ to $[m] \cup \{\star_t\}_{t \in \N}$ satisfying the following.

For any $(w^{\vk^1, \vk^2}, T) \in \supp(\bw^{\vk^1, \vk^2}, \bT)$, let $\tilde{x}_1 = \x_i\left(\vidx^{\vK}_{i,1}\right), \tilde{x}_2 = \x_i\left(\vidx^{\vK}_{i,2}\right), \dots, \tilde{x}_{j - 1} = \x_i\left(\vidx^{\vK}_{i,j - 1}\right), \tilde{x}_j = \x_i\left(\rig_i\left(\vidxb{\pa^{\vK}_{i,j}}\right)\right)$. 

If $\vidx^{\vK}_{i,1}, \dots, \vidx^{\vK}_{i,j - 1}$ and $\vidxb{\pa^{\vK}_{i,j}}$ exist, we always have $\a_i\left(\rig_i\left(\vidxb{\pa^{\vK}_{i,j}}\right)\right) = \mathsf{A}^{\vk^1, \vk^2, \vK,i,j}(\tilde{x}_1, \tilde{x}_2, \dots, \tilde{x}_j)$. 
\end{definition}

Note the existence of such function family is guaranteed by \Cref{obs:existence-K-2} \ref{item:determine-a-from-xi-K-2}. 

\subsection{Proof of~\cref{lem:hit-lower-bound}}
Similar to the case for one vertex, we will then prove our result using the following two lemmas. 
But note here that we separately consider the contribution of each pair $(\vk^1, \vk^2)$. Therefore in the rest of the paper, we will use the following notations. 

\paragraph{Notation.}
Recall $\Kshort = \{0, 1, \dots, \tau / 4\}^\ell$ and $\tau = 20 \log n \log \log n$.

We will always use $\mathcal{F}^{\vK, \vb}_{i, j}$ to denote the corresponding event under $\extwalk^{\vk^1, \vk^2}_{\ell,m,n,a}$. And recall that we use $\mathcal{F}^{\vK, \vb}_i$ to denote $\mathcal{F}^{\vK, \vb}_{i, K_i}$, and we define $\mathcal{F}^{\vK, \vb}_{\ell + 1}$ to be always true. 

For every $i \in \zeroTon{\ell}$, we use $\bg_{\le i}$ to denote the collection $(\bg_{1},\dotsc,\bg_{i})$. Similarly, we use $\br_{\le i}$ to denote the collection $(\br_{1},\dotsc,\br_{i})$.

For notation convenience, throughout this section, for $(g_{\le t},r_{\le t}) \in \supp((\bg_{\le t},\br_{\le t}))$, we will always use $g_{\le t} \wedge r_{\le t}$ to denote the event $\left[ \bg_{\le t}=g_{\le t} \wedge \br_{\le t} = r_{\le t} \right]$.

\begin{table}[h]
    \renewcommand{\arraystretch}{1.3}
    \begin{center}
    	\begin{tabular}{|l l|}
        \hline
        \textbf{Notation}  & \textbf{Meaning}  \\ \hline
		$\vidx^{\vk}$ or $\vidx[\vk]$ & the tree node determined by $\vk$ \\
		$\ell$ & number of components (sub-restrictions, levels) in $\calH_{\ell,m,n}$; number of levels; $\ell \le \log n$ \\
		$\tau$ & independence parameter in $\calH_{\ell,m,n}$; $\tau = 20\log n \log\log n$\\
		$(g_i,r_i)$ & components of hash function in $\calH_{\ell,m,n}$ \\
		$r_{\le i}, g_{\le i}$ & the sequence $(r_{1},\dotsc,r_{i})$ and $(g_{1},\dotsc,g_{i})$ \\
    	$\vK$ & a set of indices; subset of $\N^{\ell}$ \\
    	$P^{\vk}$ & the set of indices of all ancestors of $\vk$\\			
    	$\Kshort$ & $\zeroTon{\tau/4}^{\ell}$ \\
    	$\calB_{\vk}$ & set of two-dimensional sequence $\vb$ with values in $[n]$ and shape $\vk$ \\
    	$\vK_{i,j}$ & the $j$-th index among all level $i$ indices in $\cup_{\vk \in \vK} P^{\vk}$\\
    	$\vidx^{\vK}_{i,j}$ & the node $\mu^{\vK_{i,j}}$ \\
    	$K_i$ & the number of distinct level $i$ indices in $\cup_{\vk \in \vK} P^{\vk}$\\
    	$\pa^{\vK}_{i,j}$ & an index; the parent of $\vK_{i,j}$ \\
    	$\calF^{\vK,\vb}_{i,j}$ & the event that for all $(i',j')$ before or equal to $(i,j)$, $\mu_{i',j'}^{\vK}$ exists and $\nex(\mu_{i',j'}^{\vK}) = \vb_{i',j'}$\\
    	$\calF^{\vK,\vb}_{i}$ & the event $\calF^{\vK,\vb}_{i,K_i}$\\
    	$\tilde{\calG}^{\vK,\vb}_{i,j}$ & the event that for all $\a_i(\mu^{\vK}_{i,j'})$ are distinct for $1 \leq j' \leq j$\\
    	$\calG^{\vK,\vb}_{i}$ & the event $\calG^{\vK,\vb}_{i, K_i} \land \calG^{\vK,\vb}_{i+1, K_{i+1}} \land \cdots \land \calG^{\vK,\vb}_{\ell, K_\ell}$\\$p(\vk)$ & the path $p(\mu^{\vk})$ from root to $\mu^{\vk}$ \\
    	\hline
    	\end{tabular}
    \caption{Summary of Notation}
    \end{center}
\end{table}

We will need the following two lemmas that handle the total occurrences and the bad occurrences, respectively.

\begin{lemma}[Lower bounding the number of occurrences of $u, v$] \label{lem:good-case-2}
	Fix $\vk^1, \vk^2 \in \Kshort$ such that $\vk^1 < \vk^2$. Let $K = \{\vk^1, \vk^2\}$, $\bi = \vidx^{\vk^1} + 1$, and $\bj = \vidx^{\vk^2} + 1$. 

    For every two distinct vertices $u, v \in [n]$, it holds that
	\[
	\Pr\left[\bw^{\vk^1, \vk^2}_{\bi} = u \land \bw^{\vk^1, \vk^2}_{\bj} = v \right] =  \frac{2^{- \sum_{j = 1}^\ell K_j}}{n^2}.
	\]
\end{lemma}

\begin{lemma}[Upper bounding the bad occurrences of $u, v$] \label{lem:bad-case-2}
   Fix $\vk^1, \vk^2 \in \Kshort$ such that $\vk^1 < \vk^2$ and two distinct vertices $u, v \in [n]$ such that $a_u = a_v$. Let $C_{u} = \#\{i \in [n] \mid a_i = a_u\}$ denote the number of occurrences of $a_u$ in the input array $a$, $\bi = \vidx^{\vk^1} + 1$, and $\bj = \vidx^{\vk^2} + 1$.

	It holds that
	\begin{align*}
	&\Pr\left[\bw^{\vk^1, \vk^2}_{\bi} = u \land \bw^{\vk^1, \vk^2}_{\bj} = v \land \left[\exists (\alpha,\beta), 0 < \alpha < \beta < \bj \land a_{\bw^{\vk^1, \vk^2}_\alpha} = a_{\bw^{\vk^1, \vk^2}_\beta} \right] \right] \\ \le &\sum_{\substack{(\vk^3, \vk^4) \in (\N^\ell)^2\\ K = \{\vk^1, \vk^2, \vk^3, \vk^4\}}} \frac{2^{- \sum_{j=1}^\ell K_j} F_2(a)}{n^4} + 4\sum_{\substack{\vk^3 \in \N^\ell \\ K = \{\vk^1, \vk^2, \vk^3\}}} \frac{2^{- \sum_{j=1}^\ell K_j}  C_u}{n^3} + n^2\ell/2^{\tau/4}.
	\end{align*}
\end{lemma}

We also need the following simple lemma to accompany the use of Lemma \ref{lem:good-case-2}. 

\begin{lemma}\label{lem:count-path-2}
If $2 \leq \ell < \log n$, we have
\[
\sum_{\substack{(\vk^1, \vk^2) \in (\Kshort)^2 \\ \vk^1 \neq \vk^2, K = \{\vk^1, \vk^2\}}}\prod_{i = 1}^\ell 2^{-K_i} \geq 2^{2 \ell - 1}.
\]
\end{lemma}
\begin{proof}

For $K = \{\vk^1, \vk^2\}$, we have

\[
\sum_{i=1}^{\ell} K_i \le \sum_{i=1}^{\ell} \sum_{j=1}^2 k^{j}_i.
\]

Hence we have
\begin{align*}
\sum_{\substack{(\vk^1, \vk^2) \in (\Kshort)^2 \\ \vk^1 \neq \vk^2, K = \{\vk^1, \vk^2\}}}\prod_{i = 1}^\ell 2^{-K_i} \geq &\sum_{\substack{(\vk^1, \vk^2) \in (\Kshort)^2 \\ K = \{\vk^1, \vk^2\}}}\prod_{i = 1}^\ell 2^{-K_i} - \sum_{\substack{\vk^1 \in \Kshort}} \prod_{i = 1}^\ell 2^{-k^1_i}\\
\geq &\sum_{\substack{(\vk^1, \vk^2) \in (\Kshort)^2 \\ K = \{\vk^1, \vk^2\}}}\prod_{j = 1}^{2} \prod_{i = 1}^\ell 2^{-k^j_i} - 2^{\ell} \\
= & ~ (2 - 2^{-\tau/4})^{2\ell} - 2^{\ell} \\
\geq & ~ 2^{2 \ell} (1 - \ell 2^{-\tau/4} - 2^{-\ell}) \tag{by $(1 - x)^n \geq 1 - nx$}\\
\geq & ~ 2^{2 \ell - 1},
\end{align*}
where the last step follows from the fact that $\ell 2^{-\tau/4} \leq \frac{\log n}{(\log n)^{\log n}} \leq \frac{1}{4}$ and $2^{-\ell} \leq \frac{1}{4}$. 
\end{proof}

Now we are ready to prove~\cref{lem:hit-lower-bound}.

\begin{proofof}{\cref{lem:hit-lower-bound}}

We let $\bi = \vidx^{\vk^1} + 1$ and $\bj = \vidx^{\vk^2} + 1$. By~\cref{lem:subset-of-extwalk-2}, we have 
\begin{align}
\Pr_{\bh,\bs}[u, v \in f_{a,\bh}^*(\bs)] &\geq \sum_{\substack{\vk^1 < \vk^2 \\ (\vk^1, \vk^2) \in (\Kshort)^2}} \Pr\left[\bw^{\vk^1, \vk^2}_{\bi} = u \land \bw^{\vk^1, \vk^2}_{\bj} = v \land \left[\forall \alpha \neq \beta \in [\bj - 1], a_{\bw^{\vk^1, \vk^2}_{\alpha}} \not= a_{\bw^{\vk^1, \vk^2}_{\beta}}\right]\right] \nonumber\\
&= \sum_{\substack{\vk^1 < \vk^2 \\ (\vk^1, \vk^2) \in (\Kshort)^2}} \bigg( \Pr\left[\bw^{\vk^1, \vk^2}_{\bi} = u \land \bw^{\vk^1, \vk^2}_{\bj} = v \right] \nonumber\\ &\hspace{0.2cm} - \Pr\left[\bw^{\vk^1, \vk^2}_{\bi} = u \land \bw^{\vk^1, \vk^2}_{\bj} = v \land \left[\exists (\alpha,\beta), 0 < \alpha < \beta < \bj \land a_{\bw^{\vk^1, \vk^2}_\alpha} = a_{\bw^{\vk^1, \vk^2}_\beta} \right] \right] \bigg) \label{equ:two-vertex-eq}
\end{align}

Plugging Lemma \ref{lem:good-case-2} and \ref{lem:bad-case-2} in \eqref{equ:two-vertex-eq}, we have
\begin{align}
\eqref{equ:two-vertex-eq} &\geq \frac{1}{2} \sum_{\substack{(\vk^1, \vk^2) \in (\Kshort)^2 \\ \vk^1 \neq \vk^2, K = \{\vk^1, \vk^2\}}} \frac{2^{-\sum_{j=1}^{\ell} K_j}}{n^2} -  \frac{1}{2}\sum_{\substack{(\vk^1, \vk^2, \vk^3, \vk^4) \in (\N^\ell)^4\\ K = \{\vk^1, \vk^2, \vk^3, \vk^4\}}} \frac{2^{- \sum_{j=1}^\ell K_j} F_2(a)}{n^4} \\ & \hspace*{0.5cm} - \frac{1}{2} \cdot 4 \cdot \sum_{\substack{(\vk^1, \vk^2, \vk^3) \in (\N^\ell)^3 \\ \vK = \{\vk^1, \vk^2, \vk^3\}}} \frac{2^{- \sum_{j=1}^\ell K_j}  C_u}{n^3} - \frac{1}{2}|\Kshort|^2 n^2\ell/2^{\tau/4} \label{equ:two-vertex-sum}
\end{align}

Here the first summation can be lower bounded using Lemma \ref{lem:count-path-2}, while the other two summations can be upper bounded by Lemma \ref{lem:count-path}. So we have,

\begin{align*}
\eqref{equ:two-vertex-sum} &\geq \frac{1}{2}\left(\frac{4^\ell}{4 n^2} -  4! \cdot 2^4 \cdot \frac{16^\ell F_2(a)}{n^4} - 4 \cdot 3! \cdot 2^3 \cdot \frac{8^\ell C_u}{n^3} - \tau^{2 \ell} n^2 \ell 2^{-\tau/4} \right) \\
&= \frac{1}{2}\left( \frac{1}{2^{22}F_2(a) } -  \frac{384}{2^{40}F_2(a) } - \frac{192  C_u}{2^{30}F_2(a)^{1.5} } - n^2 \ell 2^{2\ell \log \tau - \tau/4} \tag{$\ell = \log n - \frac{\log F_2(a)}{2} - 10$ and $C_u \leq \sqrt{F_2(a)}$} \right)\\
&\ge \Omega\left(\frac{1}{F_2(a)}\right).
\end{align*}

In the last step, we bound $n^2 \ell 2^{2\ell \log \tau - \tau/4} $ using the fact that $\tau \geq 20 \log n \log \log n, \ell \leq \log n$. Hence $\log \tau \leq 2 \log \log n$ and 
\[
n^2 \ell 2^{2\ell \log \tau - \tau/4} \leq n^2 \log n 2^{4 \log n \log \log n - 5\log n \log \log n} \le O\left(\frac{1}{n^{\log n - 3}}\right).
\]
\end{proofof}

\subsection{Counting Total Occurrences}

The following lemma is analogous to Lemma~\ref{lem:single-level}. We essentially mimic the proof of Lemma~\ref{lem:single-level}, and remark one place where we crucially use the fact that we are working with $\extwalk_{\ell,m,n,a}^{\vk^1,\vk^2}$ instead of $\extwalk_{\ell,m,n,a}$.

\begin{lemma} \label{lem:single-level-2}
	Fix $\vk^1, \vk^2 \in (\Kshort)^2$ such that $\vk^1 < \vk^2$. Let $\vK = \{\vk^1, \vk^2\}$ and fix $\vb \in \calB_{\vK}$. 
	
	\underline{In the joint probability space $(\bw^{\vk^1, \vk^2}, \bh, \bT)$}, suppose (as induction hypothesis) that the event $\calF^{\vK,\vb}_{i + 1}$ is independent of the joint random variable $(\bg_{\le i},\br_{\le i})$.

    Then, for all $i \in [\ell]$ and $j \in [K_i]$, letting $g_{\le i - 1} \in \supp(\bg_{\le i - 1})$ and $r_{\le i - 1} \in \supp(\br_{\le i - 1})$, it holds that
	\[
	\Pr \left[\calF^{\vK, \vb}_{i,j} \ \middle\vert \ \calF^{\vK, \vb}_{i, j - 1} \wedge g_{\le i - 1} \wedge r_{\le i - 1}\right] = \frac{1}{2n}.
	\]
\end{lemma}

\begin{proof}
	Fix $i \in [\ell]$ and $j \in [K_i]$ and let $g_{\le i - 1} \in \supp(\bg_{\le i -1})$ and $r_{\le i - 1} \in \supp(\br_{\le i -1})$. We let $\event_{\le i - 1}$ denote the event $\Big[g_{\le i - 1} \land r_{\le i - 1}\Big]$ for convenience. Our goal is to show that
	\[
	\Pr\left[\calF^{\vK, \vb}_{i,j} \ \middle\vert \ \calF^{\vK, \vb}_{i, j - 1} \land \event_{\le i - 1} \right] = \frac{1}{2n}.
	\]

	By \Cref{obs:existence-K-2} \ref{item:exist-K-2}, $\vidx^{\vK}_{i,j}$ exists if and only if $g_i\left(\a_i\left(\rig_i\left(\vidxb{\pa^{\vK}_{i,j}}\right)\right)\right) = 1$. Then let us inspect how $\a_i\left(\rig_i\left(\vidxb{\pa^{\vK}_{i,j}}\right)\right)$ is determined. 

	By \Cref{obs:existence-K-2}~\ref{item:determine-a-from-xi-K-2}, $\a_i\left(\rig_i\left(\vidxb{\pa^{\vK}_{i,j}}\right)\right)$ is determined by $\x_i(\vidx^{\vK}_{i,1}), \dots, \x_i(\vidx^{\vK}_{i,j - 1}), \x_i\left(\rig_i\left(\vidxb{\pa^{\vK}_{i,j}}\right)\right)$. Note since $\calF^{\vK, \vb}_{i,j - 1}$ holds, by \Cref{obs:existence-K-2}~\ref{item:exist-equal-2}, $\rig_i\left(\vidxb{\pa^{\vK}_{i,j'}}\right) = \vidx^{\vK}_{i,j'}$ for $j' \in [j - 1]$. 
	
	Namely, $\a_i\left(\rig_i\left(\vidxb{\pa^{\vK}_{i,j}}\right)\right)$ is determined by all $\x_i\left(\rig_i\left(\vidxb{\pa^{\vK}_{i,j'}}\right)\right)$ for $j' \in [j]$.

	By \Cref{obs:existence-K-2} \ref{item:xi-K-2}, $\x_i\left(\rig_i\left(\vidxb{\pa^{\vK}_{i,j'}}\right)\right) = \last\left(\nex\left(\vidxb{\pa^{\vK}_{i,j'}}\right), i - 1\right)$. Conditioning on $\calF^{\vK, \vb}_{i,j' - 1}$ is true, let $i^{\pa}$ and $j^{\pa}$ be such that $\pa^{\vK}_{i,j'} = \vK_{i^{\pa}, j^{\pa}}$, we know $\nex\left(\vidxb{\pa^{\vK}_{i,j'}}\right) = b_{i^{\pa}, j^{\pa}}$. Thus $\nex\left(\vidxb{\pa^{\vK}_{i,j'}}\right)$ can be determined from $\vK, \vb, i, j'$.

	Moreover, by \Cref{lem:only-depend}, $\last(\cdot, i - 1)$ only depends on $\br_{\leq i - 1}$ and $\bg_{\leq i - 1}$.\footnote{Note $\walk^{\vk^1, \vk^2}(\cdot, i - 1, \cdot, \cdot)$ does not have such nice property. It not only depends on $\br_{\leq i - 1}$ and $\bg_{\leq i - 1}$ but also on $C^{i', \vk^1}$ for $i' \in [i - 1]$. This is why we are using $\last(\cdot, i - 1)$ instead of $\walk^{\vk^1, \vk^2}(\cdot, i - 1, \cdot, \cdot)$ in the relaxed extended walk (Algorithm~\ref{algo:relaxedextwalk}).} Therefore, conditioning on $\calF^{\vK, \vb}_{i,j - 1} \land \event_{\le i - 1}$, each $\x_i\left(\rig_i\left(\vidxb{\pa^{\vK}_{i,j'}}\right)\right)$ ($1 \leq j' \leq j$) is uniquely determined from $\vK, \vb, g_{\leq i - 1}, r_{\leq i - 1},i,j'$. Hence $\a_i\left(\rig_i\left(\vidxb{\pa^{\vK}_{i,j'}}\right)\right)$ is also uniquely determined for $1 \leq j' \leq j$. 

	Formally, for every $j' \in [j]$, let $\mu_{j'} = \rig_i\left(\vidxb{\pa^{\vK}_{i,j'}}\right)$. Note that $\vidx^{\vK}_{i,j}$ exists if and only if $\bg_i(\a_i(\mu_j)) = 1$. If $\mu^{\vK}_{i,j}$ exists, we know $\mu^{\vK}_{i,j} = \mu_j$ and $\nex(\mu_j) = \br_i(\a_i(\mu_j))$. So our goal is to show that $\bg_i(\a_i(\mu_j))=1 \land \br_i(\a_i(\mu_j)) = b_{i,j}$ indeed happens with probability $\frac{1}{2n}$. 
	
	For $r_i\in \supp(\br_{i})$ and $g_i \in \supp(\bg_{i})$, define a predicate
	\[
	P(r_i, g_i) \coloneqq \left[ \forall j' \in [j-1], g_i(\a(\mu_{j'})) = 1 \land r_i(\a(\mu_{j'})) = b_{i,j'} \right].
	\]
	By the discussion above, we have
	\[
	\calF^{\vK,\vb}_{i,j - 1} \land \event_{\leq i - 1}  = \calF^{\vK, \vb}_{i + 1} \land \event_{\leq i - 1} \land  P(r_i, g_i).
	\]
	
	As we have shown, for every $j' \in [j-1]$, $\a(\mu_{j'})$ is determined by $r_{\le i -1}, g_{\le i -1}$ and $\vK, \vb, i, j'$, so  $P(\br_i, \bg_i)$ only depends on the randomness of $(\br_i, \bg_i)$. (Note since it is defined using $r_{\le i - 1}, g_{\le i - 1}$ which we have fixed, it does not depend on the randomness of $\br_{\le i - 1}, \bg_{\le i - 1}$.)
	
	Together with our assumption, we know that the event $\calF^{\vK, \vb}_{i + 1} \land \event_{\le i - 1}$ is still independent of $\br_i, \bg_i$ when conditioning on $P(\br_i, \bg_i)$. 

	Hence,  \begin{align*}
	&\Pr\left[\calF^{\vK, \vb}_{i,j} \ \middle\vert \ \calF^{\vK, \vb}_{i, j - 1} \land \event_{\le i - 1}\right] \\ 
	= &\Pr\left[\bg_i(\a(\mu_j)) = 1 \wedge \br_i(\a(\mu_j)) = b_{i,j} \ \middle\vert \ \calF^{\vK, \vb}_{i, j - 1} \land \event_{\le i - 1}\right]\\
	= &\Pr\left[\bg_{i}(\a_i(\mu_j)) = 1 \wedge \br_{i}(\a_i(\mu_j)) = b_{i,j}\ \middle\vert \ \calF^{\vK, \vb}_{i + 1} \wedge \event_{\le i - 1} \wedge  P(\br_i, \bg_i)\right] \\
	= &\Pr\left[\bg_{i}(\a_i(\mu_j)) = 1 \wedge \br_{i}(\a_i(\mu_j)) = b_{i,j}\ \middle\vert \ P(\br_i, \bg_i)\right] \\
	= & ~\frac{1}{2n}.
	\end{align*}

	The last step follows from the fact that $\bg_i(\cdot)$ and $\br_i(\cdot)$ are $ \tau$-wise independent, $j \le K_i \le  \tau$, and since $\vK = \{\vk^1, \vk^2\}$, by \Cref{obs:no-colision-k1-k2}, $\a_i(\mu_j) \ne \a_i(\mu_{j'})$ for every $j' \in [j-1]$.

	\paragraph*{Remark.} In the last step, we crucially used the fact that we are working with $\extwalk_{\ell,m,n,a}^{\vk^1,\vk^2}$ instead of $\extwalk_{\ell,m,n,a}$, since for $\extwalk_{\ell,m,n,a}$, there may be collisions between $p(\vk^1)$ and $p(\vk^2)$, but for $\extwalk^{\vk^1, \vk^2}_{\ell,m,n,a}$, \Cref{obs:no-colision-k1-k2} guarantees that there is no such collision. 
\end{proof}

Repeated applications of Lemma \ref{lem:single-level-2} lead to the following corollary. We omit the proof of Corollary~\ref{cor:single-level-whole-2} since it is identical to that of Corollary~\ref{cor:single-level-whole}.

\begin{cor} \label{cor:single-level-whole-2}
    Fix $\vk^1, \vk^2 \in \N^\ell$ such that $\vk^1 < \vk^2$. Let $\vK = \{\vk^1, \vk^2\}$ and fix $\vb \in \calB_{\vK}$. 
	
	\underline{In the joint probability space $(\bw^{\vk^1, \vk^2}, \bh, \bT)$}, suppose (as induction hypothesis) that the event $\calF^{\vK,\vb}_{i + 1}$ is independent of the joint random variable $(\bg_{\le i},\br_{\le i})$.
	Then, for all $i \in [\ell]$, letting $g_{\le i - 1} \in \supp(\bg_{\le i - 1})$ and $r_{\le i - 1} \in \supp(\br_{\le i - 1})$, it holds that
	\[
	\Pr\left[\calF^{\vK, \vb}_i \ \middle\vert \ \calF^{\vK, \vb}_{i + 1} \wedge g_{\le i - 1} \wedge r_{\le i - 1} \right] = \frac{2^{-K_i}}{n^{K_i}}.
	\]
\end{cor}

Iteratively applying \cref{cor:single-level-whole-2}, we can obtain the following lemma. We omit its proof since it can be proved in the same way as 
\cref{lem:induction}.

\begin{lemma}\label{lem:induction-2}
    Fix $\vk^1, \vk^2 \in \N^\ell$ such that $\vk^1 < \vk^2$. Let $\vK = \{\vk^1, \vk^2\}$ and fix $\vb \in \calB_{\vK}$. 
    
    \underline{In the joint probability space $(\bw^{\vk^1, \vk^2}, \bh, \bT)$}, for all $i \in [\ell+1]$, letting $g_{\le i - 1} \in \supp(\bg_{\le i - 1})$ and $r_{\le i - 1} \in \supp(\br_{\le i - 1})$, it holds that
	\[
	\Pr\left[\calF^{\vK, \vb}_i \ \middle\vert \ g_{\le i - 1} \wedge r_{\le i - 1}\right] = \frac{2^{-\sum_{j=i}^\ell K_j}}{n^{\sum_{j=i}^\ell K_j}}.
	\]
\end{lemma}

Finally we can prove \cref{lem:good-case-2}, which give a lower bound on the number of good occurrences. 

\begin{reminder}{Lemma~\ref{lem:good-case-2}}
	Fix $\vk^1, \vk^2 \in \Kshort$. Let $K = \{\vk^1, \vk^2\}$, $\bi = \vidx^{\vk^1} + 1$, and $\bj = \vidx^{\vk^2} + 1$. 

    For every two distinct vertices $u, v \in [n]$, it holds that
	\[ \Pr\left[\bw^{\vk^1, \vk^2}_{\bi} = u \land \bw^{\vk^1, \vk^2}_{\bj} = v \right] =  \frac{2^{- \sum_{j = 1}^\ell K_j}}{n^2}. \]
\end{reminder}

\begin{proofof}{\Cref{lem:good-case-2}}
The proof is similar to the first half of the proof of \cref{lem:good-case}. 
	For $\vb \in \calB_{\vK}$, we say that $\vb$ is \emph{good}, if $\nex(\vidx^{\vk^1}) = u$ and $\nex(\vidx^{\vk^2}) = v$ according to $\vb$. 
	
	We first break the probability into the contributions of all possible $\vb \in \calB_{\vK}$,
	\[
	\Pr\left[\bw^{\vk^1, \vk^2}_{\vidx^{\vk^1} + 1} = u \land \bw^{\vk^1, \vk^2}_{\vidx^{\vk^2} + 1} = v \right] = \sum_{\vb \in \calB_{\vK}} \Pr[\calF^{\vK,\vb}_{1} \wedge \text{$\vb$ is good}].
	\]
	By Lemma \ref{lem:induction-2}, for every sequence $\vb \in \calB_{\vK}$, it holds that
	\[
	\Pr\left[\calF^{\vK, \vb}_1 \right] = \frac{2^{-\sum_{j=1}^\ell K_j}}{n^{\sum_{j=1}^\ell K_j}}.
	\]
	
	There are $n^{\sum_{j = 1}^\ell K_j}$ many sequences $\vb \in \calB_{\vK}$, and $n^{\sum_{j = 1}^\ell K_j - 2}$ of them satisfy $\nex(\vidx^{\vk^1}) = u$ and $\nex(\vidx^{\vk^2}) = v$ (\ie, they are good). Thus, we have
	\[
	\sum_{\vb \in \calB_{\vK}} \Pr\left[\calF^{\vk,\vb}_{1} \wedge \text{$\vb$ is good}\right] = \frac{2^{-\sum_{j = 1}^\ell K_j}}{n^2}. \qedhere
	\]
\end{proofof}

\subsection{Upper Bounding the Bad Occurrences} 

Now we generalize Lemma \ref{lem:structure} to two vertices. Its proof is defered to Appendix \ref{sec:appendix-structure}. We first recall the definition of collisions between two paths.

\begin{reminder}{Definition~\ref{def:collision}}
Let $p(\mu)$ denote the path on $T$ from root to node $\mu$. We say there is a \emph{collision between two paths} $p_1$ and $p_2$ if there are two nodes $\alpha$ and $\beta$ such that $\alpha \in p_1 \setminus p_2, \beta \in p_2 \setminus p_1$, and $(\a_{\level(\alpha)}(\alpha),\level(\alpha)) = (\a_{\level(\beta)}(\beta),\level(\beta))$.
\end{reminder}

We also say that there is no collision between a set of paths if there is no collision between any two of the paths in the set.

\begin{lemma}\label{lem:structure-2}
Let $u,v\in [n]$ be such that $u \not= v$ and $a_u = a_v$. Fix $(w^{\vk^1, \vk^2},T) \in \supp((\bw^{\vk^1, \vk^2},\bT))$. Let $\i = \vidx^{\vk^1} + 1$, $\j = \vidx^{\vk^2} + 1$ and assume that $w_\i = u$ and $w_\j = v$.  If there is a pair of nodes $(\tilde{\alpha},\tilde{\beta})$ such that $0 < \tilde{\alpha} < \tilde{\beta} < \j$ and $a_{w_{\tilde{\alpha}}} = a_{w_{\tilde{\beta}}}$, then the following holds:
\begin{itemize}
	\item There are two nodes $\alpha$ and $\beta$ such that $\alpha \not= \beta$, $a_{w_\alpha} = a_{w_\beta}$, $\{\alpha, \beta \} \neq \{\i, \j\}$, and there is no collision between $p(\i - 1), p(\j - 1), p(\alpha - 1)$, and $p(\beta - 1)$.
\end{itemize} 
\end{lemma}

Now we extend \Cref{def:vc} to our relaxed walk $\walk^{\vk^1, \vk^2}$.

\begin{definition} \label{def:vc-2}
Fix $\vk^1, \vk^2 \in \N^\ell$ such that $\vk^1 < \vk^2$ and $\vK \subseteq \N^\ell$ such that $\vk^1, \vk^2 \in \vK$. Fix $i \in [\ell]$, $\vb \in \calB_{\vK}$, and $(r_{\leq i - 1}, g_{\leq i - 1}) \in \supp(\br_{\leq i - 1}, \bg_{\leq i - 1})$. Let $\{\mathsf{A}^{\vk^1, \vk^2, \vK,i,j}\}_{j \in [K_i]}$ be the family of functions defined in \Cref{def:func-A}. $\vc(\vK, \vb, i, r_{\leq i - 1}, g_{\leq i - 1})$ is a vector of length $K_i$ defined as follows:

For each $j \in [K_i]$, let $i^{\pa}$ and $j^{\pa}$ be such that $\pa(\vK_{i,j}) = \vK_{i^{\pa}, j^{\pa}}$. Note here $i^{\pa}, j^{\pa}$ can be determined from $\vK, i$, and $j$. We let $x_j = \last(b_{i^{\pa},j^{\pa}}, i - 1)$. This is well-defined since by Lemma \ref{lem:only-depend}, $\last(\cdot, i - 1)$ only depends on $r_{\leq i - 1}$ and $g_{\leq i - 1}$.

Then for each $j \in [K_i]$, we let 

$$\left[\vc(\vK, \vb, r_{\leq i - 1}, g_{\leq i - 1})\right]_{j} = \mathsf{A}^{\vk^1, \vk^2, \vK,i,j}(x_1, x_2, \dots, x_j).$$
 
When $\vK, \vb, i, r_{\leq i - 1}$, and $g_{\leq i - 1}$ are clear from context, we drop them and simply write $\vc$. 
\end{definition}

\begin{lemma}\label{obs:determined-K-2}
Fix $(w^{\vk^1, \vk^2}, T) \in \supp(\bw^{\vk^1, \vk^2}, \bT)$ and level $i \in [\ell]$, and fix $\vK = \{\vk^1, \vk^2, \dots, \vk^t\}$, $j \in [K_i]$, $\vb \in \calB_{\vK}$, $r_{\leq i - 1}$ and $g_{\leq i - 1}$. We defined $\zeta^{i,j} = \rig_i\left(\vidxb{\pa^{\vK}_{i,j}}\right)$. Let $\vc$ be defined in \Cref{def:vc-2}. Assuming $\calF^{\vK, \vb}_{i,j - 1}$ holds, we have $\a_i(\zeta^{i,j}) = c_j$.

\end{lemma}
\begin{proof}
We use the same notation as \Cref{def:vc}. For each $j \in [K_i]$, we let $i^{\pa}$ and $j^{\pa}$ be such that $\pa(\vK_{i,j}) = \vK_{i^{\pa}, j^{\pa}}$. Note here $i^{\pa}, j^{\pa}$ can be determined from $\vK, i$, and $j$. We let $x_j = \last(b_{i^{\pa},j^{\pa}}, i - 1)$. This is well-defined since by Lemma \ref{lem:only-depend}, $\last(\cdot, i - 1)$ only depends on $r_{\leq i - 1}$ and $g_{\leq i - 1}$.

By $\calF^{\vK, \vb}_{i,j - 1}$, we know $\nex\left(\vidxb{\vK_{i^{\pa}, j^{\pa}}}\right) = b_{i^{\pa}, j^{\pa}}$. Then from \Cref{obs:existence-K-2} \ref{item:xi-K-2}, we know that $x_j = \last\left(\nex\left(\vidxb{\vK_{i^{\pa}, j^{\pa}}}\right) , i - 1\right) = \x_i\left(\rig_i\left(\vidxb{\vK_{i^{\pa}, j^{\pa}}}\right)\right) = \x_i\left(\rig_i\left(\vidxb{\pa^{\vK}_{i,j}}\right)\right)$. For the same reason, for all $1 \leq j' < j$, we have $x_{j'} = \x_i\left(\rig_i\left(\vidxb{\pa^{\vK}_{i,j'}}\right)\right)$.

Since $\calF^{\vK, \vb}_{i,j - 1}$ holds, for every $j' \in [j-1]$, we know $\vidx^{\vK}_{i,j'}$ exists and therefore $\rig_i\left(\vidxb{\pa^{\vK}_{i,j'}}\right) = \vidx^{\vK}_{i,j'}$. Thus for all $1 \leq j' < j$, we have $x_{j'} = \x_i\left(\rig_i\left(\vidxb{\pa^{\vK}_{i,j'}}\right)\right) = \x_i(\vidx^{\vK}_{i,j'})$. Thus from \Cref{def:func-A} and the definition of $\vc$, we know that $\a_i(\zeta^{i,j}) = c_j$. 
\end{proof}

We then define $P^{\vK, \vb, r_{\leq i - 1}, g_{\leq i - 1}}_{i,j}(r_i, g_i)$ and $\calG^{\vK, \vb}_{i,j}$ in the same way as \Cref{def:defPK} and \Cref{def:defGK}. Except now we use $\vc$ defined in \Cref{def:vc-2}.

\begin{definition} \label{def:P}
    $P^{\vK, \vb, r_{\leq i - 1}, g_{\leq i - 1}}_{i,j}(r_i, g_i)$ is a predicate of $r_i, g_i$ defined as following:
Let $\vc$ be the sequence defined in \Cref{def:vc-2}. For $r_i \in \supp(\br_i)$ and $g_i \in \supp(\bg_i)$, 
\[
P^{\vK, \vb, r_{\leq i - 1}, g_{\leq i - 1}}_{i,j}(r_i, g_i) \coloneqq \left[ \forall j' \in [j], g_i(c_{j'}) = 1 \land r_i(c_{j'}) = b_{i,j'} \right].
\]

We also define $P^{\vK, \vb, r_{\leq i - 1}, g_{\leq i - 1}}_{i} (r_i, g_i) = P^{\vK, \vb, r_{\leq i - 1}, g_{\leq i - 1}}_{i,K_i}(r_i, g_i)$. When $\vK, \vb, r_{\leq i - 1}, g_{\leq i - 1}$ are clear from the context, we simply write $P_{i,j}(r_i, g_i)$ and $P_i(r_i, g_i)$. 
\end{definition}

\begin{definition}
$\calG^{\vK, \vb}_{i,j}$ is the event defined as following. Let $\vc$ be the sequence $\vc(\vK, \vb, \br_{\leq i - 1}, \bg_{\leq i - 1})$ defined in \Cref{def:vc-2}. We let

\[
\calG^{\vK, \vb}_{i,j} \coloneqq \left[ \forall 1 \leq t_1 < t_2 \leq j, c_{t_1} \not= c_{t_2}\right].
\]

We also define $\calG^{\vK}_i = \calG^{\vK}_{i, K_i} \land\calG^{\vK}_{i + 1}$. 
\end{definition}

The following observation holds with the same proof as \Cref{obs:F-and-P} except replacing \Cref{obs:existence-K} with \Cref{obs:existence-K-2} and \Cref{obs:determined-K} with \Cref{obs:determined-K-2} in the proof. 

\begin{observation} \label{obs:F-and-P-2}
Fix $\vK, \vb$. For $P$ defined in \Cref{def:P}. 

$$\calF^{\vK, \vb}_{i,j} \land \event_{\leq i - 1} = \calF^{\vK, \vb}_{i + 1} \land \event_{\leq i - 1} \land P_{i,j}(r_i, g_i)$$
Specifically, we have
$$\calF^{\vK, \vb}_i \land \event_{\leq i - 1} = \calF^{\vK, \vb}_{i + 1} \land \event_{\leq i - 1} \land P_i(r_i, g_i)$$
\end{observation}

Thus the following lemmas follows. 

\begin{lemma} \label{lem:single-level-K-2}
Fix $\vk^1, \vk^2 \in \Kshort$ such that $\vk^1 < \vk^2$,  $\vK = \{\vk^1, \vk^2, \dots, \vk^t\} \subseteq \Kshort$ such that $t \le 4$, and $\vb \in \calB_{\vK}$. \underline{In the joint probability space $(\bw^{\vk^1, \vk^2}, \bh, \bT)$}, fixing $r_{\leq i-1}, g_{\leq i-1} \in \supp(\br_{\leq i-1}, \bg_{\leq i-1})$, for any event $\calA^{\vK, \vb}_{i + 1}$ such that $(\calF^{\vK, \vb}_{i + 1} \land \calG^{\vK}_{i + 1}) \lor \calA^{\vK, \vb}_{i + 1}$ is independent of $\br_{\leq i}, \bg_{\leq i}$, we have 
$$\Pr\left[P_{i,j}(r_i, g_i) \land \calG^{\vK}_{i,j} \ \middle \vert \ \left(\left(\calF^{\vK, \vb}_{i + 1} \land \calG^{\vK}_{i + 1}\right) \lor \calA^{\vK, \vb}_{i + 1}\right) \land \event_{\leq i - 1} \land P_{i, j - 1}(r_i, g_i) \land \calG^{\vK}_{i, j - 1}\right] \leq \frac{1}{2n}.$$
\end{lemma}
\begin{proof}
The proof is the same as that of Lemma \ref{lem:single-level-K}. 
\end{proof}

We also generalize Lemma \ref{induction-K} as follows.
\begin{lemma} \label{induction-K-2}
	Fix $\vk^1, \vk^2 \in \Kshort$ such that $\vk^1 < \vk^2$,  $\vK = \{\vk^1, \vk^2, \dots, \vk^t\} \subseteq \Kshort$ such that $t \le 4$, and $\vb \in \calB_{\vK}$. Let $i \in [\ell]$, and $r_{\leq i-1}, g_{\leq i-1} \in \supp(\br_{\leq i-1}, \bg_{\leq i-1})$. \underline{In the joint probability space $(\bw^{\vk^1, \vk^2}, \bh, \bT)$}, there is a sequence of events $\calA^{\vK, \vb}_i$ such that:
\[
\Pr\left[\left(\calF_i^{\vK, \vb} \land \calG^{\vK}_i\right) \lor \calA^{\vK, \vb}_i \ \middle\vert \ r_{\leq i - 1} \land g_{\leq i - 1}\right] = \frac{2^{-\sum_{j = i}^\ell K_j}}{n^{\sum_{j = i}^\ell K_j}}.
\]

In particular,
\[
\Pr\left[\calF_1^{\vK, \vb} \land \calG^{\vK}_1 \right] \leq \frac{2^{-\sum_{j = 1}^\ell K_j}}{n^{\sum_{j = 1}^\ell K_j}}.
\]
\end{lemma}
\begin{proof}
This follows the same proof as Lemma \ref{lem:single-level-K} with the only difference is that the use of Lemma \ref{lem:single-level-K} is replaced by Lemma \ref{lem:single-level-K-2}, and the use of \Cref{obs:F-and-P} is replaced by \Cref{obs:F-and-P-2}. 
\end{proof}

Then Lemma \ref{lem:bad-case-2} follows the similar proof strategy as that of single vertex case. 

Recall we define $\elong$
\[
\elong \coloneqq \left[ \text{$\exists \vk \in \N^{\ell}$ s.t. $\max_{i=1}^{\ell} k_i > \tau/4$ and $\vidx^{\vk}$ exists} \right].
\]

We observe that a similar conclusion as Lemma \ref{lem:elong-small} also holds for $\walk^{\vk^1, \vk^2}$ whose proof is deferred to Appendix \ref{sec:appendix-elong}.

\begin{lemma}\label{lem:elong-small-2}
	Fix $\vk^1, \vk^2 \in \Kshort$ such that $\vk^1 < \vk^2$. \underline{In the probability space $(\bw^{\vk^1, \vk^2}, \bT)$}, it holds that
    \[
    \Pr[\elong] \le n^2 \ell/2^{\tau/4}.
    \]
\end{lemma}

Now we are ready to prove~\cref{lem:bad-case-2}, which is restated below.

\begin{reminder}{\cref{lem:bad-case-2}}
	Fix $\vk^1, \vk^2 \in \Kshort$ and two distinct vertices $u, v \in [n]$ such that $a_u = a_v$. Let $C_{u} = \#\{i \in [n] \mid a_i = a_u\}$ denote the number of occurrences of $a_u$ in the input array $a$, $\bi = \vidx^{\vk^1} + 1$, and $\bj = \vidx^{\vk^2} + 1$.
	
	It holds that
	\begin{align*}
	&\Pr\left[\bw^{\vk^1, \vk^2}_{\bi} = u, \bw^{\vk^1, \vk^2}_{\bj} = v \land \exists (\alpha,\beta), 0 < \alpha < \beta < \bj, a_{\bw^{\vk^1, \vk^2}_\alpha} = a_{\bw^{\vk^1, \vk^2}_\beta} \right] \\ \le &\sum_{\substack{(\vk^3, \vk^4) \in (\N^\ell)^2\\ K = \{\vk^1, \vk^2, \vk^3, \vk^4\}}} \frac{2^{- \sum_{j=1}^\ell K_j} F_2(a)}{n^4} + 4\sum_{\substack{\vk^3 \in \N^\ell \\ K = \{\vk^1, \vk^2, \vk^3\}}} \frac{2^{- \sum_{j=1}^\ell K_j}  C_u}{n^3} + n^2\ell/2^{\tau/4}.
	\end{align*}
\end{reminder}
\vspace{-1.5em}
\begin{proofof}{\Cref{lem:bad-case-2}}
By Lemma \ref{lem:structure-2}, we have
\begin{align*}
& \Pr\left[\bw^{\vk^1, \vk^2}_{\bi} = u, \bw^{\vk^1, \vk^2}_{\bj} = v \land \exists 0 < \alpha < \beta < \bj, a_{\bw^{\vk^1, \vk^2}_\alpha} = a_{\bw^{\vk^1, \vk^2}_\beta} \right] \\
\leq &\Pr\big[\bw^{\vk^1, \vk^2}_{\bi} = u, \bw^{\vk^1, \vk^2}_{\bj} = v \land \exists (\alpha,\beta), \alpha < \beta, a_{\bw^{\vk^1, \vk^2}_\alpha} = a_{\bw^{\vk^1, \vk^2}_\beta} \land \\ &\hspace*{5cm}\text{no collision between $p(\bi - 1)$, $p(\bj - 1)$, $p(\alpha - 1)$, $p(\beta - 1)$} \big] \\
\leq &\Pr\big[\bw^{\vk^1, \vk^2}_{\bi} = u, \bw^{\vk^1, \vk^2}_{\bj} = v \land \neg\elong \land \exists (\alpha,\beta), \alpha < \beta, a_{\bw^{\vk^1, \vk^2}_\alpha} = a_{\bw^{\vk^1, \vk^2}_\beta} \land \\ &\hspace*{5cm}\text{no collision between $p(\bi - 1)$, $p(\bj - 1)$, $p(\alpha - 1)$, $p(\beta - 1)$} \big] + \Pr[\elong]\\
\leq & \sum_{\substack{\vk^3, \vk^4 \in \Kshort \\\vk^3 \not= \vk^4}} \Pr\big[\nex(\vidx^{\vk^1}) = u \wedge \nex(\vidx^{\vk^2}) = v \land a_{\nex(\vidx^{\vk^3})} = a_{\nex(\vidx^{\vk^4})} \land \{\vk^3, \vk^4\} \neq \{\vk^1, \vk^2\} \land \\ &\hspace*{5cm}\text{no collision between $p(\vk^1)$, $p(\vk^2)$, $p(\vk^3)$, $p(\vk^4)$} \big] + \Pr[\elong]
\end{align*}

First, from Lemma \ref{lem:elong-small-2}, we have $\Pr[\elong] \le n^2 \ell/2^{\tau/4}$. 

Next, we will prove an upper bound on

\begin{align*}
&\sum_{\substack{\vk^3, \vk^4 \in \Kshort \\\vk^3 \not= \vk^4}} \Pr\big[\nex(\vidx^{\vk^1}) = u \wedge \nex(\vidx^{\vk^2}) = v \land a_{\nex(\vidx^{\vk^3})} = a_{\nex(\vidx^{\vk^4})} \land \{\vk^3, \vk^4\} \neq \{\vk^1, \vk^2\} \land \\ &\hspace*{5cm}\text{no collision between $p(\vk^1)$, $p(\vk^2)$, $p(\vk^3)$, $p(\vk^4)$} \big].
\end{align*}

There are two cases, the first case is that at least one of the four following equalities holds (1) $\vk^1 = \vk^3$, (2) $\vk^1 = \vk^4$, (3) $\vk^2 = \vk^3$, or (4) $\vk^2 = \vk^4$, and the second case is that $\vk^1,\vk^2,\vk^3,\vk^4$ are distinct. Now we consider the first case, and by symmetry, we only need to consider the case of $\vk^1 = \vk^4$. 

\paragraph{When $\vk^1 = \vk^4$.} 

Since $\vk^1 = \vk^4$, $\{\vk^3, \vk^4\} \neq \{\vk^1, \vk^2\}$ implies $\vk^3 \not\in \{\vk^1, \vk^2\}$. Therefore $\vk^1, \vk^2, \vk^3$ are distinct. Let $\vK = \{\vk^1, \vk^2, \vk^3\}$. Since $\vK \subseteq \Kshort$, by Lemma \ref{induction-K-2}, for every sequence $\vb \in \calB_{\vK}$, we have 
\[
\Pr\left[\calF_1^{\vK, \vb} \land \calG^{\vK}_1\right] = \frac{2^{-\sum_{j = 1}^\ell K_j}}{n^{\sum_{j = 1}^\ell K_j}}.
\]

There are $n^{\sum_{j = 1}^\ell K_j}$ many sequences $\vb \in \calB_{\vK}$, and $n^{\sum_{j = 1}^\ell K_j - 3} \cdot C_u$ of them satisfy that $\nex(\vidx^{\vk^1}) = u$, $\nex(\vidx^{\vk^2}) = v$, and $a_{\nex(\vidx^{\vk^3})} = a_u$. 

We have
\begin{align*}
&\sum_{\substack{\vk^3 \in \Kshort\\ \vk^3 \not\in \{\vk^1, \vk^2\}}} \Pr\big[\nex(\vidx^{\vk^1}) = u, \nex(\vidx^{\vk^2}) = v \land a_{\nex(\vidx^{\vk^3})} = a_{\nex(\vidx^{\vk^4})} \land \\ &\hspace*{5cm}  \text{no collision between $p(\vk^1)$, $p(\vk^2)$, $p(\vk^3)$} \big] \\
= &\sum_{\substack{\vk^3 \in \Kshort\\ \vk^3 \not\in \{\vk^1, \vk^2\}}} \Pr \left[ \calG^{\{\vk^1, \vk^2,\vk^3\}}_1 \land \nex(\vidx^{\vk^1}) = u \land \nex(\vidx^{\vk^2}) = v \land a_{\nex(\vidx^{\vk^3})} = a_u \right] \\
\leq &\sum_{\substack{\vk^3 \in \N^\ell \\ \vK = \{\vk^1, \vk^2, \vk^3\}}} \frac{2^{- \sum_{t=1}^\ell K_t}  C_u}{n^3}. 
\end{align*}

\paragraph{When $\vk^1, \vk^2, \vk^3, \vk^4$ are distinct.} 

Now we consider the other case when $\vk^1, \vk^2, \vk^3, \vk^4$ are distinct. Let $\vK = \{\vk^1, \vk^2, \vk^3, \vk^4\}$. 

Similar to the previous case, by Lemma \ref{induction-K}, for any sequence $\vb \in \calB_{\vK}$, we have
\[
\Pr\left[\calF_1^{\vK, \vb} \land \calG^{\vK}_1\right] = \frac{2^{-\sum_{j = 1}^\ell K_j}}{n^{\sum_{j = 1}^\ell K_j}}.
\]

Note there are $n^{\sum_{j = 1}^\ell K_j}$ many sequences $\vb \in \calB_{\vK'}$, and $n^{\sum_{j = 1}^\ell K_j - 4} \cdot F_2(a)$ of them satisfy that $\nex(\vidx^{\vk^1}) = u$, $\nex(\vidx^{\vk^2}) = v$ and $a_{\nex(\vidx^{\vk^3})} = a_{\nex(\vidx^{\vk^4})}$. 

Hence, we have 

\begin{align*}
&\sum_{\substack{\vk^3, \vk^4 \in \Kshort \\ |\{\vk^1, \vk^2, \vk^3, \vk^4\}| = 4 }} \Pr\big[\bw^{\vk^1, \vk^2}_{\bi} = u \land \bw^{\vk^1, \vk^2}_{\bj} = v \land a_{\nex(\vidx^{\vk^3})} = a_{\nex(\vidx^{\vk^4})} \land \\ &\hspace*{5cm} \text{no collision between $p(\vk^1)$, $p(\vk^2)$, $p(\vk^3)$, $p(\vk^4)$} \big] \\
= &\sum_{\substack{\vk^3, \vk^4 \in \Kshort \\ |\{\vk^1, \vk^2, \vk^3, \vk^4\}| = 4 }} \Pr \left[ \calG^{\{\vk^1, \vk^2,\vk^3, \vk^4\}}_1 \land \nex(\vidx^{\vk^1}) = u \land \nex(\vidx^{\vk^2}) = v \land a_{\nex(\vidx^{\vk^3})} = a_{\nex(\vidx^{\vk^4})} \right] \\
\leq &\sum_{\substack{\vk^3, \vk^4 \in \N^\ell \\ \vK = \{\vk^1, \vk^2, \vk^3, \vk^4\}}} \frac{2^{- \sum_{j=1}^\ell K_j}  F_2(a)}{n^4}.
\end{align*}

Summing up these two cases proves the theorem. 
\end{proofof}

\bibliographystyle{alphaurl} 
\bibliography{main}

\appendix
\section{Proof of~\cref{lem:structure-2}} \label{sec:appendix-structure}

In this appendix we prove~\cref{lem:structure-2}. In~\cref{sec:useful-facts}, we prove many useful facts about the relaxed extended random walk $\walk^{\vk^1,\vk^2}$, which will be useful for later proofs. In~\cref{sec:important-lemmas}, we prove several lemmas that are crucial for our proof of~\cref{lem:structure-2}. Finally, we prove~\cref{lem:structure-2} in~\cref{sec:bad-structure-final-proof}.

\highlight{Notation.} In this appendix, since we always refer to $w^{\vk^1, \vk^2}$, we drop the superscript and simply write $w$. 

\subsection{Useful Facts about the Relaxed Extended Walk $\walk^{\vk^1, \vk^2}$}\label{sec:useful-facts}

\begin{algorithm}[H]\label{algo:relaxedextwalk-rem}
	\DontPrintSemicolon
	\caption{Generating relaxed extended walk $\walk^{\vk^1, \vk^2}(s',i, \underline{\mu_0}, \vk)$:  (where $s'\in [n], 0\le i \le \ell$)}
	\lIf{$i = 0$}{\Return sequence $(s')$ which contains a single vertex.}
	\eIf{$[\forall t \in [i + 1, \ell], k_t = k^2_t] \land [\exists t \in [i + 1, \ell], k^1_t < k^2_t] \land [k^2_i > 0]$ \label{line:when-inherit-rem}}{
	\tcc{Here the condition says that $\mu_0 \not\in p(\vk^1)$ and is the last node of $p(\vk^2)$ above level $i$ and $p^*(\vk^2)$ is non-empty on level $i$. It is equivalent to $\mu_1, \mu_2, \dots \in p^*(\vk^2) \setminus p^*(\vk^1)$.} 
    	$C_0 \gets C^{i, \vk^1}$ \label{line:C_0-1-rem}
	}{ 
    	$C_0 \gets \emptyset$ \label{line:C_0-2-rem}
	}
	$\mathsf{star} \gets \text{false}, j \gets 0, s_0 \gets s', w = ()$.\\
	\Repeat{$g_i(y) = 0$}{
	$\vk' \gets (0, 0, \dots, 0, j, k_{i+1}, \dots, k_{\ell})$ \label{line:setjinkprime-rem} \tcc{\underline{Here $\vk'$ equals $\index(\mu_j)$.}}
	$w = w \circ \walk^{\vk^1, \vk^2}(s_j, i - 1, \underline{\mu_0 + |w|}, \vk')$ \tcc{\underline{Here $\mu_0 + |w|$ equals $\mu_j$.}}
	$x_{j + 1} \gets \last(s_j,i - 1)$\label{line:assign-x-rem}\\
	 $y, \mathsf{star} \gets \begin{cases} a_{x_{j + 1}}, \text{false} & \text{ if } a_{x_{j + 1}} \not\in C_j \land \lnot \mathsf{star} \\ \star_t, \text{true} & \text{ otherwise (where $t = \min \{t \in \N \ \vert \ \star_t \not\in C\}$)}\end{cases}$ \label{line:assign y-rem}\\
	 \underline{Let $\mu_{j + 1} = \mu_0 + |w|$, $\x_i(\mu_{j + 1}) \gets x_{j + 1}, \a_i(\mu_{j + 1}) \gets y$.}\label{line:assign xa-rem}\\
	 \lIf{$j > 0$}{\underline{$\rig(\mu_j) \gets \mu_{j + 1}$}\label{line:assign rig-rem}}
		\If{$g_i(y) = 1$\label{line:if-rem}}{
		$C_{j + 1} \gets C_j \cup \{y\}$, $s_{j + 1} \gets r_i(y)$\label{line:assign cs-rem}\\
		\underline{$\level(\mu_{j+1}) \gets i, \nex(\mu_{j+1}) \gets r_i(y)$}\\
		\underline{$\index(\mu_{j+1}) \gets (0,\cdots,0, j + 1,  k_{i+1}, \dots, k_{\ell})$}\label{line:assign index}\\
			$j \gets j + 1$.
		}
	}
\If {$[\forall t \in [i + 1, \ell], k_t = k^1_t] \land [\exists t \in [i + 1, \ell], k^1_t < k^2_t]$}{
\tcc{Here the condition says that $\mu_0$ is the last node of $p(\vk^1)$ above level $i$ and is not the last node of $p(\vk^2)$. It is equivalent to $\mu_1, \mu_2, \dots \in p^*(\vk^1) \setminus p^*(\vk^2)$.}
		    	 $C^{i, \vk^1} \gets C_{\min(j, k^1_i)}$} \label{line:C-k1-rem}
	\Return $w$.
\end{algorithm}

We will need the following facts about $\walk^{\vk^1, \vk^2}$. For convenience, we also recall the code of function $\walk^{\vk^1,\vk^2}$ in Algorithm~\ref{algo:relaxedextwalk-rem}.

\begin{fact} \label{observations}
Fix $(w, T) \in (\bw^{\vk^1, \vk^2}, \bT)$ and assume that $\vidx^{\vk^1}, \vidx^{\vk^2}$ exist.
\begin{enumerate}[label=(\alph*)]

	\item Within $\walk^{\vk^1, \vk^2}(s', i, \vk)$, $\forall \vk^* \in \N^{\ell}$, if $\level(\mu_j) = i$, $\mu_j \in p^*(\vk^*)$ if and only if $\forall t \in [i + 1, \ell], [k_t = k^*_t] \land [k^*_i > 0]$. Moreover $\mu_j \in p(\vk^*)$ if and only if $\mu_j \in p^*(\vk^*)$ and $j \leq k^*_i$. \label{obs:on-path}
	\item Let $\alpha, \beta \leq \vidx^{\vk^2}$ be two nodes. If $a_{\x(\alpha)} = a_{\x(\beta)}$ or $\a(\alpha) = \a(\beta) \not= \star_*$, we have $a_{w_{\alpha}} = a_{w_{\beta}}$. \label{obs:collision-imply-duplicate}

	\item Suppose $\x_i(\alpha) = \x_i(\beta)$. $\a_i(\alpha) \not= \a_i(\beta)$ only when $\a_i(\alpha) = \star_*$ or $\a_i(\beta) = \star_*$. \label{obs:same-x-same-a}
	
	\item For any two nodes $\alpha, \beta$ both of level $i$, if $\a_i(\alpha) = \a_i(\beta)$, we have $\x_i(\rig(\alpha)) = \x_i(\rig(\beta))$. \label{obs:same-x}
	
	\item There is no collision between $p(\vk^1)$ and $p^*(\vk^2)$. \label{obs:no-collision-k1-k2}
	
	\item For a node $\alpha$ on level $i$, if $g_i(\rig(\alpha)) = 1$, then $\rig(\alpha)$ is also on level $i$. \label{obs:right-alpha-level}
	
	\item For a node $\alpha$ on level $i$ with index $\vk$, let $j \in [k_i - 1]$ and $\beta = \mu^{\vk}_{i,j}$. If $\a_i(\beta) = \star_*$, then $\a_i(\alpha) = \star_*$. Consequently, if $\a_i(\alpha) \ne \star_*$, then $\a_i(\beta) \ne \star_*$. 
	 \label{obs:if-star-then-all-stars}

	 \item In $\walk^{\vk^1, \vk^2}(s', i, \vk)$, suppose the condition at Line \ref{line:when-inherit-rem} is met. we know $\mu_j \in p^*(\vk^2)$ if $\level(\mu_j) = i$. \label{obs:if-inherit-in-p2} 

	 \item For a node $\alpha \not\in p^*(\vk^2)$ on level $i$, if $\a_i(\alpha) \not= \star_*$ and $\a_i(\rig(\alpha)) = \star_t$, there must be a node $\eta \in p(\alpha)$ such that $\a_i(\eta) = a_{\x_i(\rig(\alpha))}$ and $\level(\eta) = i$. Moreover, we must have $t = 0$. 
	 \label{obs:first-star}

	 \item For a node $\beta \in \tilde{p}(\vk^1, \vk^2)$ on level $i$, if $\a_i(\beta) \neq \star_*$ and $\a_i(\rig(\beta)) = \star_*$, there must be a node $\eta \in \tilde{p}(\vk^1, \vk^2)$ such that $a_{\x_i(\eta)} = a_{\x_i(\rig(\beta))}$ and $\level(\eta) = i$. 

	  \label{obs:first-star-beta}

	 \item For a node $\alpha \not\in p^*(\vk^2)$ on level $i$, if $\a_i(\alpha) = \star_*$ or $\a_i(\rig(\alpha)) = \star_t$ for $t > 0$, there must be a level $i$ node $\eta$ on $p(\alpha)$ such that $\a_i(\eta) = \star_0$. 
	 \label{obs:if-star-must-first}

	 \item For a node $\beta \in \tilde{p}(\vk^1, \vk^2)$ on level $i$, if $\a_i(\beta) = \star_*$ or $\a_i(\rig(\beta)) = \star_t$ for $t > 0$, there must be a level $i$ node $\eta$ on $\tilde{p}(\vk^1, \vk^2)$ such that $\a_i(\eta) = \star_0$. 
	 \label{obs:if-star-must-first-2}

\end{enumerate}
\end{fact}
\begin{proof}\item
\highlight{Proof of \ref{obs:on-path}.}
If $k^*_i > 0$, since the first level $i$ ancestors of $\mu^{\vk^*}$ has index $(\vk^*)^{i,1} = (0, \dots, 0, 1, k^*_{i + 1}, k^*_{\ell})$, we know level $i$ nodes of $p^*(\vk^*)$ are $(0, \dots, 0, j, k^*_{i + 1}, k^*_{\ell})$ with $j \geq 1$. If $\level(\mu_j) = i$, by how its index is assigned (Line \ref{line:assign index}), we know it is in $p^*(\vk^*)$ when $\forall t \in [i + 1, \ell], [k_t = k^*_t]$. 

Specifically, if $k^*_i = 0$, there is no level $i$ node on $p(\vk^*)$. By definition of $p^*(\vk^*)$, there is also no level $i$ node on it. The moreover part follows from the fact that all level $i$ ancestors of $\mu^{\vk^*}$ have indices $(\vk^*)^{i,j} = (0, \dots, 0, j, k^*_{i + 1}, k^*_{\ell})$ for $1 \leq j \leq k^*_i$. 

\highlight{Proof of \ref{obs:collision-imply-duplicate}.} Since $\alpha, \beta \leq \vidx^{\vk^2}$, by Lemma \ref{lem:out-k2}, we have $\x(\alpha) = w_{\alpha}$ and $\x(\beta) = w_{\beta}$. Therefore $a_{\x(\alpha)} = a_{\x(\beta)}$ implies $a_{w_{\alpha}} = a_{w_{\beta}}$. For any node $\mu$, if $\a(\mu) \neq \star_*$, then $\a(\mu) = a_{\x(\mu)}$ holds by Line \ref{line:assign y-rem}, \ref{line:assign xa-rem}. Thus $\a(\alpha) = \a(\beta) \not= \star_*$ also suffice. 

\highlight{Proof of \ref{obs:same-x-same-a}.}
By line \ref{line:assign y-rem}, \ref{line:assign xa-rem}, we know $\a_i(\alpha) \not= a_{\x_i(\alpha)}$ only when $\a_i(\alpha) = \star_*$. The same also holds for $\beta$. Therefore since $a_{\x_i(\alpha)} = a_{\x_i(\beta)}$, the observation holds.

\highlight{Proof of \ref{obs:same-x}.} For each node $\mu_j$ of level $i$ at Line \ref{line:setjinkprime-rem} - \ref{line:assign rig-rem}, we know $\rig(\mu_j) = \mu_{j + 1}$ and $\x_i(\mu_{j + 1}) = x_{j + 1} = \last(s_j, i - 1)$(Line \ref{line:assign xa-rem}, \ref{line:assign-x-rem}). Here $s_j = r_i(\a_i(\mu_j))$ (Line \ref{line:assign xa-rem}, \ref{line:assign cs-rem}). 

Hence for each node $\mu$ of level $i$, $\x_i(\rig(\mu))$ is the return value of $\last(r_i(\a_i(\mu)), i - 1)$. From $\a_i(\alpha) = \a_i(\beta)$, we get $\x_i(\rig(\alpha)) = \x_i(\rig(\beta))$ directly. 

\highlight{Proof of \ref{obs:no-collision-k1-k2}.} By Line \ref{line:assign xa-rem}, \ref{line:assign cs-rem}, we know $C_j = \{\a_i(\mu_1), \a_i(\mu_2), \dots, \a_i(\mu_j)\}$. Therefore by Line \ref{line:C-k1-rem} and \ref{obs:on-path}, we know $C^{i, \vk^1} =\{\a_i(\alpha) \vert \alpha \in p(\vk^1) \land \level(\alpha) = i\}$. 

Let the level $i$ nodes on $p^*(\vk^2)$ be $\beta_1, \beta_2, \dots, \beta_j$. Initially, $C_0 = C^{i,\vk^1}$. For each $t \in [j]$, $\a_i(\beta_t)$ is chosen in a way so that $\a_i(\beta_t) \not\in C_{t - 1}$. Then $C_t = C_{t - 1} \cup \{\a_i(\beta_t)\}$. Since $C^{i, \vk^1} \subseteq C_t$ holds for all $t \in [j]$, we know $\a_i(\alpha) \not= \a_i(\beta_t), \forall t \in [j]$. 

Hence there is no collision between $p(\vk^1)$ and $p^*(\vk^2)$. 

\highlight{Proof of \ref{obs:right-alpha-level}.} 
When $\mu_j = \alpha$, by Line \ref{line:assign rig-rem}, $\rig(\alpha) = \mu_{j + 1}$. By Line \ref{line:assign xa-rem} and Line \ref{line:if-rem}, $\level(\mu_{j + 1}) = i$ if and only if $g_i(\a_i(\mu_{j + 1})) = 1$. Thus the statement holds. 

\highlight{Proof of \ref{obs:if-star-then-all-stars}.} By Line \ref{line:assign cs-rem}, once $\mathsf{star}$ switches from $\text{false}$ to true, $y$ is always $\star_*$, and $\mathsf{star}$ is always true. Therefore, since $\a_i(\mu_j) = y$ (Line \ref{line:assign xa-rem}), if $\a_i(\mu_j) = \star_*$ for $j \leq k_i$, then we must have $\a_i(\mu_{k_i}) = \star_*$. This proves that once $\a_i(\beta) = \a_i(\vidx^{\vk}_{i,j}) = \star_*$, we must have $\a_i(\alpha) = \a_i(\vidx^{\vk}_{i,k_i}) = \star_*$. 

\highlight{Proof of \ref{obs:if-inherit-in-p2}.} If the condition at Line \ref{line:when-inherit-rem} met, and $\level(\mu_j) = i$, we know $\index(\mu_j) = (0, \dots, 0,j + 1, k_{i +1}, \dots, k_{\ell})$ by Line \ref{line:assign index} where $k_j = k^2_j$ for $j \in [i + 1, \ell]$ and $k^2_i > 0$. Thus we can see that $\index(\mu_1)$ is an ancestor of $\vk^2$. Therefore, by definition of $p^*(\vk^2)$, $\mu_j \in p^*(\vk^2)$. 

\highlight{Proof of \ref{obs:first-star}.} Consider the function call that assigns $\a_i(\alpha)$. Since $\alpha \not\in p^*(\vk^2)$, by the contrapositive of \ref{obs:if-inherit-in-p2}, we know that $C_0 = \emptyset$. Suppose $\alpha = \mu_j$ and $\rig(\alpha) = \mu_{j + 1}$. Then by Line \ref{line:assign cs-rem}, $C_j = \{\a_i(\mu_1), \a_i(\mu_2), \dots, \a_i(\mu_j)\}$. Since $\a_i(\alpha) \not= \star_*$, $\mathsf{star} = \text{false}$, and there is no $\star_*$ in $C_j$. Hence if $\a_i(\rig(\alpha)) = \star_t$, we must have $t = 0$. Moreover, by Line \ref{line:assign y-rem}, this happens only when $a_{\x_i(\rig(\alpha))} \in C_j$ which means there is $1 \leq j' \leq j$ such that $\a_i(\mu_{j'}) = a_{\x_i(\rig(\alpha))}$. Let $\eta = \mu_{j'}$. This concludes the proof. 

\highlight{Proof of \ref{obs:first-star-beta}.} Also consider the function call that assigns $\a_i(\beta)$. The difference with \ref{obs:first-star} is that now $C_0$ may not be $\emptyset$. By Line \ref{line:assign cs-rem}, $C_j = C_0 \cup \{\a_i(\mu_1), \a_i(\mu_2), \dots, \a_i(\mu_j)\}$. Suppose $\beta = \mu_j$, $\rig(\beta) = \mu_{j + 1}$. Since $\a_i(\beta) \not= \star_*$, $\mathsf{star} = \text{false}$. By Line \ref{line:assign y-rem}, $\a_i(\mu_{j + 1}) = \star_*$ only when $a_{\x_i(\rig(\beta))} \in C_j$. If $a_{\x_i(\rig(\beta))} \in C_j \setminus C_0$, there is $1 \leq j' \leq j$ such that $\a_i(\mu_{j'}) = a_{\x_i(\rig(\beta))}$. We simply let $\eta = \mu_{j'}$. $\eta \in \tilde{p}(\vk^1, \vk^2)$ since $\beta = \mu_j \in \tilde{p}(\vk^1, \vk^2)$ and $j' < j$. 

If $a_{\x_i(\rig(\beta))} \in C_0$, by Line \ref{line:C_0-1-rem}, $C_0 = C^{i,\vk^1}$ which equals $\{\a_i(\alpha) \ \vert \ \alpha \in p(\vk^1) \land \level(\alpha) = i\}$ by Line \ref{line:C-k1-rem}. Hence there is $\eta \in p(\vk^1) \subseteq \tilde{p}(\vk^1, \vk^2)$ that satisfies the requirement. In either case, we can find such $\eta$. This concludes the proof. 

\highlight{Proof of \ref{obs:if-star-must-first}.} Similar as \ref{obs:first-star}. Also consider the function call that assigns $\a_i(\alpha)$. Since $\alpha \not\in p^*(\vk^2)$, by contrapositive of \ref{obs:if-inherit-in-p2}, we know that $C_0 = \emptyset$. Suppose $\alpha = \mu_j$ and $\rig(\alpha) = \mu_{j + 1}$. 

By $C_0 = \emptyset$ and Line \ref{line:assign y-rem}, we know $\a_i(\mu_1) \not= \star_*$. Then because $\a_i(\alpha) = \star_*$ or $\a_i(\rig(\alpha)) = \star_t$ for $t > 0$, there must exist $1 \leq j' \leq j + 1$ such that $\a_i(\mu_{j'}) = \star_*$ and $\a_i(\mu_{j' - 1}) \not= \star_*$. By \ref{obs:first-star}, we know $\a_i(\mu_j') = \star_0$. Since $t > 0$, we know $j' \not= j + 1$. Thus $1 \leq j' \leq j$, and we let $\eta = \mu_j$ which is of level $i$. $\eta \in p(\alpha)$ since $j' \leq j$. 

\highlight{Proof of \ref{obs:if-star-must-first-2}.} If there is a node $\alpha \in p(\vk^1) \setminus p^*(\vk^2)$ such that $\a_i(\alpha) = \star_*$ and $\level(\alpha) = i$. Then by \ref{obs:if-star-must-first} such node $\eta \in p(\alpha) \subseteq \tilde{p}(\vk^1, \vk^2)$ exists. 

Then suppose there is no such $\alpha$. Consider the function call that assigns $\a_i(\beta)$. We will prove there is no $\star_*$ in $C_0$. Let $S_1$ be the set of level $i$ nodes on $p(\vk^1)$, and $S_2$ be the set of level $i$ nodes on $p^*(\vk^2)$. By the tree structure, $S_1$ is either a subset of $S_2$ or disjoint with $S_2$. 

If $S_1$ is a subset of $S_2$, we know $C_0 = \emptyset$ by Line \ref{line:when-inherit-rem}, since there is no $t \in [i + 1, \ell]$ that $k^1_t < k^2_t$ (otherwise $S_1 \cap S_2 = \emptyset$). If $S_1$ and $S_2$ are disjoint, $C_0 = C^{i,\vk^1}$, and by Line \ref{line:C-k1-rem}, $C^{i,\vk^1} = \{\a_i(\mu) \vert \mu \in S_1\}$. We prove by contradiction. Suppose there is $\mu \in S_1$ such thhat $\a_i(\mu) = \star_*$, namely $\star_* \in C_0$. We let $\alpha = \mu$. Then $\alpha \in p(\vk^1) \setminus p^*(\vk^2) = S_1 \setminus S_2 = S_1$ and $\level(\alpha) = i$. This contradicts with the fact there is no such $\alpha$. 

In either case, $\star_* \not\in C_0$. Thus let $\mu_{j'}$ be the first node among $\mu_1, \dots, \mu_j, \mu_{j + 1}$ such that $\a_i(\mu_{j'}) = \star_*$. By $C_{j' - 1} = C_0 \cup \{\a_i(\mu_1), \dots, \a_i(\mu_{j' - 1})\}$, we know $\star_* \not\in C_{j' - 1}$. Hene by Line \ref{line:assign y-rem}, we have $\a_i(\mu_{j'}) = \star_0$. Thus we know $j' \leq j$ and we simply let $\eta = \mu_j$. In such case, $\eta \in p^*(\vk^2) \subseteq \tilde{p}(\vk^1, \vk^2)$ and $\level(\eta) = i$. 
\end{proof}

\subsection{Some Structure Lemmas}\label{sec:important-lemmas}

Similar as before, we shall first extend Lemma \ref{lem:level_i_collision}, Corollary \ref{cor:moving-1-path} and \ref{cor:moving-2-paths}. To do so, we first need an extra definition. See Figure \ref{fig:p_tilde}. 

\begin{figure}
    \centering
	\begin{tikzpicture}[scale = 0.6, thick]
  \draw (0,0) -- (1, -3);
  \draw (1, -3) -- (8, -3);
  \draw (5, -3) -- (6, -6);
  \node at (6,-6.5) {\small $p^*({\vec{k}}^1)$};

  \node [draw,circle,fill,minimum size=2.5,inner sep=0pt, outer sep=0pt] at (6,-3) {};
  \node at (6,-3.7) {\small $\gamma_i$};

  \node [draw,circle,fill=red,red,minimum size=2.5,inner sep=0pt, outer sep=0pt] at (7,-3) {};
  \node [red] at (7,-3.7) {\small $\alpha$};

  \draw (9,0) -- (10, -3);
  \draw (10,-3) -- (17, -3);
  \draw (14, -3) --  (15, -6);
  \node at (15,-6.5) {\small $p^*({\vec{k}}^2)$};

  \node [draw,circle,fill=red,red,minimum size=2.5,inner sep=0pt, outer sep=0pt] at (16,-3) {};
  \node [red] at (16,-3.7) {\small $\beta$};

  \node [red] at (10,-5) {\small $\mathsf{a}_i(\alpha) = \mathsf{a}_i(\beta)$};
  
  \draw [decorate,decoration={brace,amplitude=4pt},xshift=0pt,yshift=7pt, thick, blue] (1.2,-3) -- (5.8,-3);
  \node [blue] at (3.5,-1.7) {\small $\tilde{p}(\vec{k}^1, \vec{k}^2)$};

    \draw [decorate,decoration={brace,amplitude=4pt},xshift=0pt,yshift=7pt, thick, blue] (10.2,-3) -- (16.8,-3);
  \node [blue] at (13.5,-1.7) {\small $\tilde{p}(\vec{k}^1, \vec{k}^2)$};

\end{tikzpicture}\textbf{}
    \caption{The definition of $\gamma_i$ and the level $i$ nodes in $\tilde{p}(\vk^1, \vk^2)$}
    \label{fig:p_tilde}
\end{figure}
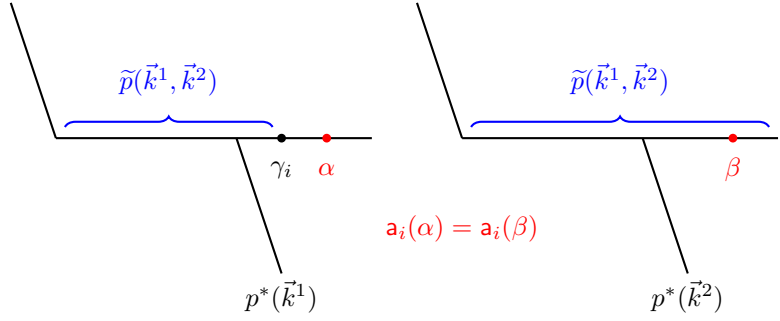

\begin{definition} \label{def:p_tilde}
Assuming that $\vidx^{\vk^1}$ and $\vidx^{\vk^2}$ exist and $\vk^1 < \vk^2$. 

Define $\gamma_i:=\max \{\gamma \in p^*(\vk^1) \ \vert \ \level(\gamma) = i \ \land \ p(\gamma) \text{ has no collision wtih } p^*(\vk^2) \}$, namely, the last level $i$ node on $p^*(\vk^1)$ such that its path has no collision with $p^*(\vk^2)$. We define $\tilde{p}(\vk^1, \vk^2)$ to be $$\tilde{p}(\vk^1, \vk^2) = \left(\bigcup_{\substack{i=1\\\gamma_i \text{exists}}}^{\ell} p(\gamma_i)\right) \bigcup p^*(\vk^2).$$

By  \cref{observations}\ref{obs:no-collision-k1-k2}, if $k^1_i > 0$, $\gamma_i$ must exist, and $p(\vk^1) \subseteq \tilde{p}(\vk^1, \vk^2)$.
\end{definition}

\newcommand{\Findmainalgo}{\textsf{Find}}
\newcommand{\FindalgoI}{\Findmainalgo\textsf{-Case-I}}
\newcommand{\FindalgoII}{\Findmainalgo\textsf{-Case-II}}

\begin{algorithm} \label{algo:find-pi3-2}
	$\vk = \index(\sigma)$\\
	\eIf{there is a collision between $p(\sigma)$ and $\tilde{p}(\vk^1, \vk^2)$}{
		Let $i$ be the level of the lowest such collision\; \label{line:def-i}
		Let $\pi_1, \pi_2$ be such a collision on level $i$; \tcc{break ties by picking the lexicographically first pair} \label{line:defi-pi12}
		\eIf{$\a_i(\vidx^{\vk}_{i,k_i}) \not= \star_*$}{
			\Return $\FindalgoI(\pi_1,\pi_2,\vk,T)$; \tcc{see Algorithm~\ref{algo:find-pi3-2-1}}
		}{
			\Return $\FindalgoII(\pi_1,\pi_2,\vk,T)$;  \tcc{see Algorithm~\ref{algo:find-pi3-2-2}}
		}
	}{
		Let  $\pi_3$ be the last node on $p(\sigma)$ that is also on $\tilde{p}(\vk^1, \vk^2)$\;
		\Return $\pi_3$\;
	}
\caption{$\Findmainalgo(\sigma,T)$}
\end{algorithm}

\begin{lemma}\label{lem:level_i_collision-2}
Fix $\vk^1, \vk^2 \in \N^\ell$ such that $\vk^1 < \vk^2$ and $(w, T) \in \supp(\bw^{\vk^1, \vk^2}, \bT)$. Assuming that $\vidx^{\vk^1}$ and $\vidx^{\vk^2}$ exist. Suppose $\sigma = \vidx^{\vk}$ is a node in $T$. Let $\pi_3 = \Findmainalgo(\sigma,T)$. The following hold:
\begin{enumerate}
	\item $\pi_3 \in \tilde{p}(\vk^1, \vk^2)$.
	
	\item If there is a collision between $p(\sigma)$ and $\tilde{p}(\vk^1, \vk^2)$, then letting $i$ be the level of the lowest such collision, it holds that $(\a_i(\pi_3), \level(\pi_3)) = (\a_i(\vidx^{\vk}_{i, k_i}), i)$.
	
	\item There is no collision between $p(\sigma)$ and $\tilde{p}(\vk^1, \vk^2)$ below $\level(\pi_3)$.
\end{enumerate}

\end{lemma}
\begin{proof}

Let $\pi_3 = \Findmainalgo(\sigma,T)$. First, if there is no collision between $p(\sigma)$ and $\tilde{p}(\vk^1, \vk^2)$, then one can straightforwardly verify $\pi_3$ satisfy all the required conditions. From now on we assume that there is a collision between $p(\sigma)$ and $\tilde{p}(\vk^1, \vk^2)$.

Let $i,\pi_1,\pi_2$ be as defined in Line~\ref{line:def-i} and Line~\ref{line:defi-pi12} in Algorithm~\ref{algo:find-pi3-2}. Formally, $\pi_1, \pi_2, i$ are defined as following. $\pi_1 = \vidx^{\vk}_{i,j}$ is the node on $p(\sigma) \setminus \tilde{p}(\vk^1, \vk^2)$ with the smallest $i$ such that $\exists \pi_2 \in \tilde{p}(\vk^1, \vk^2) \setminus p(\sigma)$ satisfying $(\a_i(\pi_2), \level(\pi_2)) = (\a_i(\pi_1), i)$. If there are multiple such pairs, we pick the lexicographically first pair. 

Depending on whether $\a_i(\vidx^{\vk}_{i,k_i}) \not= \star_*$, we divide the proof into two cases.

\highlight{The case when $\a_i(\vidx^{\vk}_{i,k_i}) \not= \star_*$.} In this case, $\pi_3$ is found by Algorithm~\ref{algo:find-pi3-2-1}.
\begin{algorithm}[h] \label{algo:find-pi3-2-1}
\tcc{Finding $\pi_3$ when $\a_i(\mu^{\vk}_{i,k_i}) \not= \star_*$}
$\alpha \gets \pi_1, \beta \gets \pi_2$\;

\While{$\alpha \not= \vidx^{\vk}_{i,k_i}$}{
	\eIf{$\a_i(\rig(\alpha)) = \a_i(\rig(\beta))$} {
		$\alpha \gets \rig(\alpha), \beta \gets \rig(\beta)$\; \label{code:eq-2}

		\If{$\beta \not\in \tilde{p}(\vk^1, \vk^2)$}
		{
			\text{Let $\eta \in p^*(\vk^2)$ be the node such that $\a_i(\eta) = \a_i(\beta)$ and $\level(\eta) = i$}\; \label{code:exist-2-2}
			$\beta \gets \eta$; \tcc{If $\eta$ does not exist, the algorithm aborts.}
		}
	}
	{
	  	Let $\eta \in \tilde{p}(\vk^1, \vk^2)$ be the node such that $\a_i(\eta) = \a_i(\rig(\alpha))$ and $\level(\eta) = i$\; \label{code:exist-2-1}	  	
		$\alpha \gets \rig(\alpha), \beta \gets \eta$;	 \tcc{If $\eta$ does not exist, the algorithm aborts.}
	}
}
$\pi_3 \gets \beta$\;
\Return $\pi_3$\;
\caption{$\FindalgoI(\pi_1,\pi_2,\vk,T)$}
\end{algorithm}
We will prove the following three facts about Algorithm~\ref{algo:find-pi3-2-1}:

\begin{enumerate}
	\item Throughout Algorithm \ref{algo:find-pi3-2-1}, $\level(\alpha) = \level(\beta) = i$.  \label{fact1-1}
	\item Algorithm \ref{algo:find-pi3-2-1} terminates. \label{fact1-2}
	\item At Line \ref{code:exist-2-2} and \ref{code:exist-2-1}, the node $\eta$ always exists.\label{fact1-3}
\end{enumerate}

We first show that these three facts are sufficient. First observe that the invariant $\a_i(\alpha) = \a_i(\beta)$ is preserved during algorithm. Suppose these facts are true. When the algorithm terminates, we have $\alpha = \vidx^{\vk}_{i,k_i}$, $(\a(\beta), \level(\beta)) = \left(\a(\alpha), i)\right)$. Moreover by Line \ref{code:exist-2-2} and Line~\ref{code:exist-2-1}, we know that it always holds that $\beta \in \tilde{p}(\vk^1, \vk^2)$. Hence, $\pi_3$ satisfies the requirements of the lemma by Fact~1, the invariant $\a_i(\alpha) = \a_i(\beta)$, the definition of $i$, and Fact~\ref{fact1-2}.

Now we prove these three facts. 

\highlight{Proof of Fact \ref{fact1-1}.} Observe that initially $\alpha = \pi_1$ which is on $p(\vk)$ and of level $i$. Then in each iteration, $\alpha$ always moves to $\rig(\alpha)$ until $\alpha = \vidx^{\vk}_{i,k_i}$. From the assumption that $\sigma = \vidx^{\vk}$ exists, we know $\vidx^{\vk}_{i,k_i}$ exists. Thus, $\level(\alpha) = i$ holds throughout Algorithm~\ref{algo:find-pi3-2-1}. 

For $\beta$, each time it either (1) moves to a level-$i$ node $\eta$, or (2) moves to $\rig(\beta)$. For Case~(1) clearly we still have $\level(\beta) = i$. For Case(2), from $\a_i(\rig(\alpha)) = \a_i(\rig(\beta))$, we have $g_i(\a_i(\rig(\beta))) = g_i(\a_i(\rig(\alpha))) = 1$, implying that $\rig(\beta)$ must also be of level $i$ by \cref{observations}\ref{obs:right-alpha-level}. Hence, $\level(\beta) = i$ always holds. This finishes the proof of Fact~\ref{fact1-1}.

\highlight{Proof of Fact~\ref{fact1-2}.} Fact \ref{fact1-2} follows from the observation that after each iteration we have $\alpha \gets \rig(\alpha)$ and $\alpha$ never moves to its left. So eventually the algorithm must stop. 

\highlight{Proof of Fact~\ref{fact1-3}.} We first prove that $\eta$ always exists at Line \ref{code:exist-2-1}.

Base on Fact \ref{fact1-1} and $\a_i(\alpha) = \a_i(\beta)$, it follows from \cref{observations}\ref{obs:same-x} that $\x_i(\rig(\alpha)) = \x_i(\rig(\beta))$. Therefore by \cref{observations}\ref{obs:same-x-same-a}, the only possibility of entering Line \ref{code:exist-2-1} is when at least one of $\a_i(\rig(\alpha)) = \star_*$ and $\a_i(\rig(\beta)) = \star_*$ is true. Since $\a_i(\mu^{\vk}_{i,k_i}) \not= \star_*$, from \cref{observations}\ref{obs:if-star-then-all-stars}, we know $\a_i(\rig(\alpha)) \not= \star_*$. Thus here $\a_i(\rig(\beta)) = \star_*$. 

If $\a_i(\beta) = \star_*$, by $\a_i(\alpha) = \a_i(\beta)$, we would have $\a_i(\rig(\alpha)) = \a_i(\rig(\beta))$ and would not enter Line \ref{code:exist-2-1}. Therefore $\a_i(\beta) \not= \star_*$. Since by Line \ref{code:exist-2-2}, we always have $\beta \in \tilde{p}(\vk^1, \vk^2)$, by \cref{observations}\ref{obs:first-star-beta}, there must exist $\eta \in \tilde{p}(\vk^1, \vk^2)$ such that $(\a_i(\eta), \level(\eta)) = (a_{\x_i(\rig(\beta))}, i)$. On the other side, since $\a(\rig(\alpha)) \not= \star_0$, $\a_i(\rig(\alpha)) = a_{x_i(\rig(\alpha))} = a_{x_i(\rig(\beta))} = \a_i(\eta)$. Therefore such $\eta$ exists. 

Then we prove that $\eta$ always exists at Line \ref{code:exist-2-2}. Since each time $\beta$ either move to $\rig(\beta)$ or to $\eta \in \tilde{p}(\vk^1, \vk^2)$ (note $p^*(\vk^2) \subseteq \tilde{p}(\vk^1, \vk^2)$), the only possibility of $\beta \not\in \tilde{p}(\vk^1, \vk^2)$ is when $\beta = \rig(\gamma_i)$ where $\gamma_i$ is defined as in \Cref{def:p_tilde}. Note by Fact \ref{fact1-1}, $\beta = \rig(\gamma_i)$ is of level $i$. Then suppose there is no $\eta \in p^*(\vk^2)$ such that $\a_i(\eta) = \a_i(\rig(\gamma_i))$. Instead of $\gamma_i$, $\rig(\gamma_i)$ should be the last level $i$ node on $p^*(\vk^1)$ whose path $p(\rig(\gamma_i))$ has no collision with $p^*(\vk^2)$. This contradicts the definition of $\gamma_i$. Thus such node $\eta$ must exist. 

\highlight{The case when $\a_i(\vidx^{\vk}_{i,k_i}) = \star_*$.} In this case, $\pi_3$ is found by Algorithm~\ref{algo:find-pi3-2-2}.

\begin{algorithm}[h] \label{algo:find-pi3-2-2}
\tcc{Finding $\pi_3$ when $\a_i(\mu^{\vk}_{i,k_i}) = \star_*$}
$\alpha \gets \pi_1, \beta \gets \pi_2$\;
\While{$\alpha \not= \vidx^{\vk}_{i,k_i}$}{
	\eIf{$\a_i(\rig(\alpha)) = \a_i(\rig(\beta))$} {
		$\alpha \gets \rig(\alpha), \beta \gets \rig(\beta)$\; \label{code:eq-3}
	}
	{
		\eIf{$\a_i(\rig(\beta)) = \star_*$ \label{code:neq-3}} {
		\text{Let $\eta_\alpha \in p(\sigma)$ be the node that $\a_i(\eta_\alpha) = \star_0$ and $\level(\eta_\alpha) = i$}\;
		\text{Let $\eta_\beta \in \tilde{p}(\vk^1, \vk^2)$ be the node that $\a_i(\eta_\beta) = \star_0$ and $\level(\eta_\beta) = i$}\; \label{code:exist-3-3}
		$\alpha \gets \eta_\alpha, \beta \gets \eta_\beta$\;
		}{
		\tcc{now it must hold that $\a_i(\rig(\alpha))=\star_*$}
		\text{Let $\eta_\alpha \in p(\sigma)$ be the node that $\a_i(\eta_\alpha) = \a_i(\rig(\beta))$ and $\level(\eta_\alpha) = i$}\;	\label{code:exist-3-1}
		$\alpha \gets \eta_\alpha, \beta \gets \rig(\beta)$\;
		}
	}
	\If{$\beta \not\in \tilde{p}(\vk^1, \vk^2)$}
	{
		\text{Let $\eta_\beta \in p^*(\vk^2)$ be the node such that $\a_i(\eta_\beta) = \a_i(\beta)$ and $\level(\eta_\beta) = i$}\; \label{code:exist-3-2}
		$\beta \gets \eta_\beta$\;
	}
}
$\pi_3 \leftarrow \beta$\;
\Return $\beta$\;
\caption{$\FindalgoII(\pi_1,\pi_2,\vk,T)$}
\end{algorithm}

Similar to the previous case, we will prove the following three facts about Algorithm \ref{algo:find-pi3-2-2}:
\begin{enumerate}
	\item Throughout Algorithm \ref{algo:find-pi3-2-2}, $\level(\alpha) = \level(\beta) = i$.  \label{fact2-1}
	\item Algorithm \ref{algo:find-pi3-2-2} terminates. \label{fact2-2}
	\item At Line \ref{code:exist-3-3} both $\eta_\alpha$ and $\eta_\beta$ always exist; At Line \ref{code:exist-3-1}, $\eta_\alpha$ always exists; at Line~\ref{code:exist-3-2}, $\eta_\beta$ always exists. \label{fact2-3}
\end{enumerate}

These three facts above are enough to imply that the found $\pi_3$ satisfies the requirements of the lemma, by the same argument as that of the previous case.

\highlight{Proof of Fact \ref{fact2-1}.} Initially $\alpha = \pi_1 \in p(\sigma)$. Then in each iteration, $\alpha$ may move to either $\rig(\alpha)$ or a node $\eta_\alpha \in p(\sigma)$ of level $i$ until we reach $\alpha = \vidx^{\vk}_{i,k_i}$. Note $\alpha$ moves to $\eta_{\alpha}$ either at Line \ref{code:exist-3-3} where $(\a_i(\eta_{\alpha}), \level(\eta_{\alpha})) = (\star_0, i)$ or at Line \ref{code:exist-3-1} where $\a_i(\eta_{\alpha}) \neq \star_*$ and $\level(\eta_{\alpha}) = i$. Since $\a_i(\vidx^{\vk}_{i,k_i}) = \star_*$, in both cases, such $\eta_{\alpha} \in p(\sigma) = p(\vidx^{\vk})$ is always a node before or equal $\vidx^{\vk}_{i,k_i}$. Therefore, $\alpha$ is always on path $p(\vidx^{\vk}_{i,k_i})$ and of level $i$. 

For $\beta$, same as that of Algorithm \ref{algo:find-pi3-2-1}, in each iteration it move to either (1) a level $i$ node $\eta_\beta$, or (2) move to $\rig(\beta)$. The fact clearly holds in Case~(1). For Case~(2), it moves to $\rig(\beta)$ either at Line \ref{code:eq-3} or Line \ref{code:exist-3-1}. For Line \ref{code:eq-3}, since $\a_i(\rig(\alpha)) = \a_i(\rig(\beta))$ holds, we know $g_i(\a_i(\rig(\beta))) = g_i(\a_i(\rig(\alpha))) = 1$, and $\rig(\beta)$ must also be of level $i$ by \cref{observations}\ref{obs:right-alpha-level}. For Line \ref{code:exist-3-1}, since $\a_i(\eta_{\alpha}) = \a_i(\rig(\beta))$ holds, and $\eta_{\alpha}$ is of level $i$, we know $g_i(\a_i(\rig(\beta))) = g_i(\a_i(\eta_{\alpha})) = 1$, and $\rig(\beta)$ must also be of level $i$ by \cref{observations}\ref{obs:right-alpha-level}.

\highlight{Proof of Fact~\ref{fact2-3}.} Here we prove Fact \ref{fact2-3} before Fact \ref{fact2-2}. For Line \ref{code:exist-3-2}, the analysis is the same as that of Line \ref{code:exist-2-2} of Algorithm \ref{algo:find-pi3-2-1}.

For Line \ref{code:exist-3-3} and \ref{code:exist-3-1}, similar as before, given Fact \ref{fact2-1} and $\a_i(\alpha) = \a_i(\beta)$, it follows from \cref{observations}\ref{obs:same-x} that $\x_i(\rig(\alpha)) = \x_i(\rig(\beta))$. Therefore, by \cref{observations}\ref{obs:same-x-same-a}, the only possibility of entering Line \ref{code:neq-3} is when at least one of $\a_i(\rig(\alpha)) = \star_*$ and $\a_i(\rig(\beta)) = \star_*$ happens.

Note since $\pi_1 \not\in \tilde{p}(\vk^1, \vk^2)$, we know $\alpha \not\in p^*(\vk^2)$. Similarly, $\vidx^{\vk}_{i,k_i} \not\in p^*(\vk^2)$. We will need this fact in following case analysis. 

Suppose $\a_i(\rig(\alpha)) = \star_{t_1}$ and $\a_i(\rig(\beta)) = \star_{t_2}$. Note $t_1 \neq t_2$, or otherwise we would not have entered Line \ref{code:neq-3}. If $t_1 > t_2$, from $\alpha \not\in p^*(\vk^2)$, we have $\a_i(\alpha) = \star_{t_1 - 1}$. Then by $\a_i(\alpha) = \a_i(\beta) = \star_{t_1 - 1}$, we must have $\a_i(\rig(\beta)) = \star_{t_1}$. This contradicts with the assumption $t_1 > t_2$. Since same thing holds for $t_2 > t_1$, we know exactly one of $\a_i(\rig(\alpha)) = \star_*$ and $\a_i(\rig(\beta)) = \star_*$ happens. Also we must have $\a_i(\alpha) = \a_i(\beta) \neq \star_*$.

\begin{itemize}
	\item If $\a_i(\rig(\alpha)) = \star_*$ and $\a_i(\rig(\beta)) \not= \star_*$, since $\a_i(\alpha) \neq \star_*$ and $\alpha \not\in p^*(\vk^2)$, by \cref{observations}\ref{obs:first-star}, there must be a node $\eta \in p(\sigma)$ such that $(\a_i(\eta), \level(\eta)) = (a_{\x_i(\rig(\alpha))}, i)$.

	Since $\a_i(\rig(\beta)) \not= \star_*$, $\a_i(\rig(\beta)) = a_{\x_i(\rig(\beta))} = a_{\x_i(\rig(\alpha))} = \a_i(\eta)$. Here $\x_i(\rig(\alpha)) = \x_i(\rig(\beta))$ follows from Fact \ref{fact2-1}, $\a_i(\alpha) = \a_i(\beta)$, and \cref{observations}\ref{obs:same-x}. This proves the existence of $\eta_{\alpha}$ at Line \ref{code:exist-3-1}.
	
	\item If $\a_i(\rig(\alpha)) \not= \star_*$ and $\a_i(\rig(\beta)) = \star_t$, by $\a_i(\vidx^{\vk}_{i,k_i}) = \star_*$, $\vidx^{\vk}_{i,k_i} \not\in p^*(\vk^2)$ and \cref{observations}\ref{obs:if-star-must-first}, we know there must exist such $\eta_\alpha = \vidx^{\vk}_{i,j}$ such that $j \leq k$ and $\a_i(\eta_{\alpha}) = \star_0$. This proves the existence of $\eta_{\alpha}$ at Line \ref{code:exist-3-3}. 

	Similarly, since $\a_i(\rig(\beta)) = \star_t$, by \cref{observations}\ref{obs:if-star-must-first-2}, if $t > 0$, such $\eta_\beta$ must also exist. If $t = 0$, since $g_i(\rig(\beta)) = g_i(\eta_\alpha) = 1$, by \cref{observations}\ref{obs:right-alpha-level}, we know $\level(\rig(\beta)) = i$. Then we can simply let $\eta_\beta = \rig(\beta)$. This proves the existence of $\eta_{\beta}$ at Line \ref{code:exist-3-3}.
\end{itemize}

\highlight{Proof of Fact~\ref{fact2-2}.} We will consider the following two cases depending on whether Algorithm~\ref{algo:find-pi3-2-2} enters Line~\ref{code:exist-3-3} during the execution.

We first observe that Algorithm~\ref{algo:find-pi3-2-2} enters Line~\ref{code:exist-3-3} at most once. Once Algorithm~\ref{algo:find-pi3-2-2} enters Line~\ref{code:exist-3-3} during the execution, after $\alpha \gets \eta_{\alpha}, \beta \gets \eta_{\beta}$, we have $\a_i(\alpha) = \a_i(\beta) = \star_0$. 

We prove that the algorithm then stops within $T$ steps without entering Line~\ref{code:exist-3-3} again, where $T$ is the integer such that $\a_i(\vidx^{\vk}_{i,k_i}) = \star_T$. This follows from a simple induction. Suppose after $t$ ($t \geq 0$) steps, $\a_i(\alpha) = \a_i(\beta) = \star_t$, we always have $\a_i(\rig(\alpha)) = \a_i(\rig(\beta)) = \star_{t + 1}$. Thus in the $t + 1$ step, we would enter Line~\ref{code:eq-3} and have $\alpha \gets \rig(\alpha), \beta \gets \rig(\beta)$. Noticing Line~\ref{code:exist-3-2} preserves $\a_i(\alpha)$ and $\a_i(\beta)$, this finishes the inductive step. The base case follows from the fact that $\a_i(\alpha) = \a_i(\beta) = \star_0$ before the first step. 

Otherwise, Algorithm~\ref{algo:find-pi3-2-2} never enters Line~\ref{code:exist-3-3}. When it enters Line~\ref{code:eq-3} and Line~\ref{code:exist-3-1}, $\beta$ always moves to $\rig(\beta)$. For Line \ref{code:exist-3-2}, $\beta \not\in \tilde{p}(\vk^1, \vk^2)$ can only happen if it equals $\rig(\gamma_i)$, and $\beta$ then moves to a node $\eta \in p^*(\vk^2)$ which is after $\rig(\gamma_i) \in p^*(\vk^1) \setminus \tilde{p}(\vk^1, \vk^2)$. Putting these together, since $\beta$ always moves to $\rig(\beta)$  or a node $\eta$ after it, the algorithm must eventually stop. 

\end{proof}

The following remark will be useful for later proofs.

\begin{remark}\label{rem:same-pi-3}
	Fix $\vk^1, \vk^2 \in \N^\ell$ such that $\vk^1 < \vk^2$ and $(w, T) \in \supp(\bw^{\vk^1, \vk^2}, \bT)$. Assuming that $\vidx^{\vk^1}$ and $\vidx^{\vk^2}$ exist. Let $\sigma,\eta$ be two nodes in $T$, and let $\pi_3^\sigma = \textsf{Find($\sigma,T$)}$ and $\pi_3^\eta = \textsf{Find($\eta,T$)}$. Let $i$ be the level of the lowest common ancestor of $\sigma$ and $\eta$. If $\level(\pi_3^\sigma) > i$ and $\level(\pi_3^\eta) > i$, then $\pi_3^\sigma = \pi_3^\eta$.
\end{remark}
\begin{proof}
Let $\pi_1^{\sigma}, \pi_2^{\sigma}$ be the nodes $\pi_1, \pi_2$ in $\textsf{Find($\sigma,T$)}$. Let $\pi_1^{\eta}, \pi_2^{\eta}$ be the nodes $\pi_1, \pi_2$ in $\textsf{Find($\eta,T$)}$. Let $i_0 = \min(\level(\pi_1^\sigma), \level(\pi_1^\eta))$.  Since $\level(\pi_1^\sigma) = \level(\pi_3^\sigma) > i$ and $\level(\pi_1^\eta) = \level(\pi_3^\eta) > i$, we know $i_0 > i$. Therefore, $p(\sigma)$ and $p(\eta)$ contains exactly the same level $i_0$ nodes. Let $\mu$ be the last level $i_0$ node on $p(\sigma)$. It is also the last level $i_0$ node on $p(\eta)$. 

Then by Line \ref{line:defi-pi12}, we know $\level(\pi^{\sigma}_1) = i_0$ (resp. $\level(\pi^{\eta}_1) = i_0$) since it is the lowest level that contains collision between $p(\sigma)$ (resp. $p(\eta)$) and $\tilde{p}(\vk^1, \vk^2)$.

By \Cref{lem:level_i_collision-2}, we know $(\a(\pi_3^\sigma), \level(\pi_3^\sigma)) = (\a(\mu), i_0) = (\a(\pi_3^\eta), \level(\pi_3^\eta))$ and $\pi_3^\sigma, \pi_3^\eta \in \tilde{p}(\vk^1, \vk^2)$. By \Cref{observations} \ref{obs:no-collision-k1-k2}, there is no collision between $p(\vk^1)$ and $p^*(\vk^2)$. Together with \Cref{def:p_tilde}, this implies that all $\alpha \not= \beta \in \tilde{p}(\vk^1, \vk^2)$ such that $\level(\alpha) = \level(\beta)$ must have $\a_i(\alpha) \not= \a_i(\beta)$. 

Hence $(\a(\pi_3^\sigma), \level(\pi_3^\sigma)) = (\a(\pi_3^\eta), \level(\pi_3^\eta))$ implies $\pi_3^\sigma = \pi_3^\eta$. 
\end{proof}

Next we recall the definition of two paths being the same below level $i$.

\begin{reminder}{Definition~\ref{def:same-below}}
We say two paths $p(\vk^1)$ and $p(\vk^2)$ are the same below level $i$ if
\begin{itemize}
	\item $\forall$ $1 \leq j < i$, $k^1_j = k^2_j$. 
	\item $\forall$ $1 \leq j < i, 1 \leq t \leq k^1_j$, $\a_j(\vidx^{\vk^1}_{j,t}) = \a_j(\vidx^{\vk^2}_{j,t})$.
\end{itemize}
\end{reminder}

The follow lemma is similar to \Cref{cor:moving-1-path}. However, since the initial value of $C_0$ in $\walk^{\vk^1, \vk^2}$ may be non-empty. It is more complicated. 

\begin{lemma} \label{lem:moving-1-path-2}
Fix $\vk^1, \vk^2 \in \N^\ell$ such that $\vk^1 < \vk^2$ and $(w, T) \in \supp(\bw^{\vk^1, \vk^2}, \bT)$. Assuming that $\vidx^{\vk^1}$ and $\vidx^{\vk^2}$ exist. Fix a node $\sigma \leq \vidx^{\vk^2}$, and let $\pi_3 = \Findmainalgo(\sigma,T)$.

Then, there must be a descendant $\sigma'$ of $\pi_3$ such that the following hold:
\begin{enumerate}
	\item $(\a(\sigma'), \level(\sigma')) = (\a(\sigma), \level(\sigma))$.
	\item There is no collision between $p(\sigma')$ and $\tilde{p}(\vk^1, \vk^2)$. 
	\item $p(\sigma')$ is the same as $p(\sigma)$ below level $\level(\pi_3)$. 
\end{enumerate}
\end{lemma}

\begin{proof}
	
When there is no collision between $p(\sigma)$ and $\tilde{p}(\vk^1, \vk^2)$, we simply take $\sigma' = \sigma$. One can verify that $\sigma'$ satisfies all the required conditions.

In the rest of the proof we assume there there is a collision between $p(\sigma)$ and $\tilde{p}(\vk^1, \vk^2)$. Let $i$ be the level of the lowest such collision. By Lemma~\ref{lem:level_i_collision-2}, we have that $\pi_3 \in \tilde{p}(\vk^1, \vk^2)$ and $(\a_i(\pi_3), \level(\pi_3)) = (\a_i(\vidx^{\vk}_{i, k_i}), i)$. We also let $\vk$ be the index of $\sigma$ and $\vk^3$ be the index of $\pi_3$. 

We decompose the proof into two claims. \Cref{claim:out-induction} is an outer induction step between levels, and \Cref{claim:inner-induction} is an inner induction step within a single level. The lemma is proved by an outer induction that repeatedly applies \Cref{claim:out-induction}, while \Cref{claim:out-induction} itself is proved by an inner induction that repeatedly applies \Cref{claim:inner-induction}.   

\highlight{The node $\sigma'$.} We define $\vk'$ and $\sigma'$ as follows:
\[
k'_j = \begin{cases} k^3_j & j \geq \level(\pi_3) \\ k_j & j < \level(\pi_3) \end{cases} \quad\text{for $j \in \zeroTon{\ell}$}, \quad\text{and}\quad \sigma' = \vidx^{\vk'}.
\] Note that here a priori the node $\sigma'$ may not exist. If it exists, since $\forall j \geq \level(\pi_3), k'_j = k^3_j$, it must be a descendant of $\pi_3$. 

In the rest of the proof we will prove that the node $\sigma'$ always exists, and it satisfies the requirement of the lemma.

Recall that we use $\index(\mu)$ to denote the index of a node $\mu$. Let $i = \level(\pi_3)$. The node $\sigma$ is generated by $\walk^{\vk^1, \vk^2}(\nex(\vidx^{\vk}_{i,k_i}),i - 1, \index(\vidx^{\vk}_{i,k_i}))$ while $\sigma'$ (if exists) is generated by $w^{\vk^1, \vk^2}$ by $\walk^{\vk^1, \vk^2}(\nex(\pi_3),i - 1, \index(\pi_3))$. The difference with Corollary \ref{cor:moving-1-path} is that now although $\nex(\vidx^{\vk}_{i,k_i}) = \nex(\pi_3)$, these two walks could still be different since $\index(\vidx^{\vk}_{i,k_i}) \not= \index(\pi_3)$.\footnote{As they could affect the initial value of $C_0$ at Line~\ref{line:C_0-1-rem} and~\ref{line:C_0-2-rem} in Algorithm~\ref{algo:relaxedextwalk-rem}.} But still, we are going to prove that since there is no collision strictly below level $i$ between $p(\sigma)$ and $\tilde{p}(\vk^1, \vk^2)$, we have that $p(\sigma')$ is the same as $p(\sigma)$ below level $i$. 

\newcommand{\barj}{J}

We need the following claim.

\begin{claim}\label{claim:out-induction}
	For $j \in [i-1]$, suppose the following holds
	\vadj
	\begin{flalign}\label{prop:out-ind-hypo}
	\qquad\bullet\qquad\text{$\vidx^{\vk'}_{\barj, k'_\barj}$ exists and $\a(\vidx^{\vk}_{\barj, k_\barj}) = \a(\vidx^{\vk'}_{\barj, k'_\barj})$ for every $\barj \in \{j+1,\dotsc,i\}$.}&&
	\end{flalign}
	\vadj
	Then,
	\vadj
	\begin{flalign}\label{prop:out-ind-step}
	\qquad\bullet\qquad\text{$\vidx^{\vk'}_{j, t}$ exists and $\a(\vidx^{\vk}_{j, t}) = \a(\vidx^{\vk'}_{j, t})$ for every $t \in \zeroTon{k_j}$.}&&
	\end{flalign}
\end{claim}

Before proving Claim~\ref{claim:out-induction}, we show that it implies our lemma.

\highlight{The Outer Induction.} Note that $\pi_3 = \vidx^{\vk'}_{i, k'_i}$ (by definition of $\vk'$) exists and $\a_i(\vidx^{\vk}_{i, k_i}) = \a_i(\pi_3)$ (by the guarantee on $\pi_3$).~\eqref{prop:out-ind-hypo} holds for $j = i - 1$. From Claim~\ref{claim:out-induction}, it further implies that \eqref{prop:out-ind-step} holds for $j = i - 1$ as well. Then we can apply Claim~\ref{claim:out-induction} repetitively, to show that~\eqref{prop:out-ind-step} holds for every $j \in [\level(\sigma), i - 1]$. By Definition~\ref{def:same-below}, it follows that $p(\sigma')$ is the same as $p(\sigma)$ below level $i$.
\end{proof}

Note \eqref{prop:out-ind-hypo} implies the special case of \eqref{prop:out-ind-step} when $t = 0$. By our definition, $\vidx^{\vk}_{j, 0}$ and $\vidx^{\vk'}_{j, 0}$ equals $\vidx^{\vk}_{p, k_p}$ and $\vidx^{\vk'}_{q, k_q}$ for $p = \min\{p \in [j + 1, \ell] \ \vert \ k_p > 0 \}$ and $q = \min\{q \in [j + 1, \ell] \ \vert \ k'_q > 0 \}$ respectively. As $j \in [i - 1]$, we have $p \leq i$ and $q \leq i$ because node $\pi_3 \in p(\vk')$ and $\mu^{\vk}_{i, k_i} \in p(\vk)$ are both of level $i$. Furthermore, since $k_{j'} = k'_{j'}$ for $j' \in [i - 1]$, we know $p = q$ always holds. Hence by \eqref{prop:out-ind-hypo}, $\a(\vidx^{\vk}_{j,0}) = \a(\vidx^{\vk}_{p, k_p}) = \a(\vidx^{\vk}_{q, k_q}) = \a(\vidx^{\vk'}_{j,0})$. This will be the base case of the inner induction. 

Specifically, \eqref{prop:out-ind-hypo} implies the following:
\begin{flalign}\label{prop:base-case}
\bullet\qquad\text{$\vidx^{\vk'}_{j, 0}$ exists and $\a(\vidx^{\vk}_{j, 0}) = \a(\vidx^{\vk'}_{j, 0})$. Moreover, we also know that $\level(\vidx^{\vk}_{j, 0}) = \level(\vidx^{\vk'}_{j, 0})$}.&&
\end{flalign}

Now we prove Claim~\ref{claim:out-induction}.

\begin{proofof}{Claim~\ref{claim:out-induction}}
	
Fix $j \in [i-1]$. Note that if $k_j = 0$ the claim holds immediately. Thus from now on we assume $k_j > 0$.

Assuming~\eqref{prop:out-ind-hypo} holds, we will establish~\eqref{prop:out-ind-step} by proving the following claim. For ease of notation, for $t \in \zeroTon{k_j}$, we let $\zeta_t = \vidx^{\vk}_{j, t}$ and $\zeta'_t = \vidx^{\vk'}_{j, t}$. 

\begin{claim}\label{claim:inner-induction}
	For $t \in [k_j]$, suppose 
	\vadj
	\begin{flalign}\label{prop:inner-ind-hypo}
	\qquad\bullet\qquad\text{$\zeta'_{t'}$ exists and $\a(\zeta_{t'}) = \a(\zeta'_{t'})$ for $t' \in \zeroTon{t-1}$. Specifically, $\level(\zeta_0) = \level(\zeta'_0)$. }&&
	\end{flalign}
	\vadj
	Then
	\vadj
	\begin{flalign}\label{prop:inner-ind-step}
	\qquad\bullet\qquad\text{$\zeta'_{t}$ exists and $\a(\zeta_{t}) = \a(\zeta'_{t})$}.&&
	\end{flalign}
\end{claim}

Clearly~\eqref{prop:out-ind-step} follows from Claim~\ref{claim:inner-induction} by a simple induction. Here the base case (\ie,\eqref{prop:inner-ind-hypo} with $t = 1$)  of the induction follows from \eqref{prop:base-case}. 

Before proving Claim~\ref{claim:inner-induction}, we first inspect how the existence of $\zeta'_t$ and the values of $\a_j(\zeta_t)$ and $\a_j(\zeta'_t)$ are determined in Algorithm~\ref{algo:relaxedextwalk}.

\highlight{How $\a_j(\zeta_t)$ and $\a_j(\zeta'_t)$ are determined.} For $\a_j(\zeta_t)$ ($t \in [k_j]$), it is determined by function call $$\walk^{\vk^1, \vk^2}(\nex(\vidx^{\vk}_{j, 0}), j, \index(\vidx^{\vk}_{j, 0})).$$ Initially, $C_0 = \emptyset$. This is because by definition, $i = \level(\pi_3)$ is the lowest level such that there is a collision between $p(\sigma)$ and $\tilde{p}(\vk^1, \vk^2)$. Therefore, $i$ must be lower than the level of least common ancestor of $p(\sigma)$ and $p^*(\vk^2)$. Together with $j < i$, we know $\zeta_t \not\in p^*(\vk^2)$. Hence in such function call, the initial value of $C_0$ has to be the empty set. The the function call determines each $\x_i(\mu_j)$ and $\a_i(\mu_j)$ in order. 

Note here we have proved
\begin{flalign}\label{prop:zeta_t}
\bullet\qquad\text{$\zeta_{t} \not\in p^*(\vk^2)$}.&&
\end{flalign}

For $\a_j(\zeta'_t)$, it is determined by function call $$\walk^{\vk^1, \vk^2}(\nex(\vidx^{\vk'}_{j, 0}), j, \index(\vidx^{\vk'}_{j, 0})).$$ 

If $\index(\vidx^{\vk'}_{j, 0})$ satisfies the condition at Line \ref{line:when-inherit-rem}, Algorithm \ref{algo:relaxedextwalk-rem}, $C'_0$ will be $(C')^{j,\vk^1}$. Otherwise, $C'_0 \gets \emptyset$. Then the function call determined each $\x_i(\mu'_j)$ and $\a_i(\mu'_j)$ in order, and $C'_{j + 1}$ will be $C'_j \cup \{\a_i(\mu'_{j + 1})\}$. \\

Now we are ready to prove Claim~\ref{claim:inner-induction}.
\begin{proofof}{Claim~\ref{claim:inner-induction}}
We will first show that it suffices to prove $\a_j(\zeta_t) = \a_j(\rig(\zeta'_{t - 1}))$, and then prove $\a_j(\zeta_t) = \a_j(\rig(\zeta'_{t - 1}))$ via a proof by contradiction.

Assuming $\a_j(\zeta_{t}) = \a_j(\rig(\zeta'_{t - 1}))$, we have $g_j(\a_j(\rig(\zeta'_{t - 1}))) = g_j(\a_j(\zeta_{t})) = 1$ and $\zeta'_t$ must exist.

From~\eqref{prop:inner-ind-hypo} we have that for $t' \in \zeroTon{t-1}$, $\zeta'_{t'}$ exists and $\a(\zeta_{t'}) = \a(\zeta'_{t'})$. Moreover, for $t' \in [t-1]$, by definition, we know $\level(\zeta_{t'}) = \level(\zeta'_{t'}) = j$. For $t' = 0$, we also know $\level(\zeta_0) = \level(\zeta'_0)$. 

Since $\nex(\zeta_{t - 1}) = r_{\level(\zeta_{t - 1})}(\a(\zeta_{t - 1})) =  r_{\level(\zeta'_{t - 1})}(\a(\zeta'_{t - 1})) = \nex(\zeta'_{t - 1})$, we have 
\[
\x_j(\zeta_t) = \last(\nex(\zeta_{t - 1}), j - 1) = \last(\nex(\zeta'_{t - 1}), j - 1) = \x_j(\zeta'_t).
\]

As $\zeta'_{t} = \rig(\zeta'_{t - 1})$, we have $\a_j(\zeta_{t}) = \a_j(\zeta'_{t})$, which proves the claim.

Now it remains to prove $\a_j(\zeta_{t}) = \a_j(\rig(\zeta'_{t - 1}))$. Suppose that $\a_j(\zeta_t) \not= \a_j(\rig(\zeta'_{t - 1}))$ for the sake of contradiction. By~\Cref{observations}\ref{obs:same-x-same-a}, we know that 
\begin{flalign}\label{prop:at-least-one-true}
\bullet\qquad\text{at least one of $\a_j(\zeta_t) = \star_*$ and $\a_j(\rig(\zeta'_{t - 1})) = \star_*$ is true.}&&
\end{flalign}
Below we first prove under our assumption $\a_j(\zeta_t) \neq \a_j(\rig(\zeta'_{t - 1}))$, the following hold: 

\twoprops{$\a_j(\rig(\zeta'_{t - 1})) = \star_*$, and }{a-j-star}{either $\a_j(\zeta_t) \neq \star_*$ or $\a_j(\zeta_t) = \star_0$.}{either-nstar-or-starz}

\highlight{Proving~\eqref{a-j-star} and~\eqref{either-nstar-or-starz}.} We first consider the case $t = 1$. Since $C_0 = \emptyset$, it follows that $\a_j(\zeta_t) \not= \star_*$. Hence  $\a_j(\rig(\zeta'_{t - 1})) = \star_*$ by~\eqref{prop:at-least-one-true}. Therefore, both of~\eqref{a-j-star} and~\eqref{either-nstar-or-starz} hold when $t = 1$.

Now consider the case when $t > 1$. Suppose $\a_j(\zeta_{t - 1}) = \star_x$, by $\a_j(\zeta_{t - 1}) = \a_j(\zeta'_{t - 1})$, we must have $\a_j(\zeta_t) = \a_j(\rig(\zeta'_{t-1}) = \star_{x + 1}$. This contradicts with our assumption that $\a_j(\zeta_t) \neq \a_j(\rig(\zeta'_{t-1})$. The same thing holds for $\zeta'_{t - 1}$. Suppose $\a(\zeta'_{t-1}) = \star_{x}$. By $\a_j(\zeta_{t - 1}) = a_j(\zeta'_{t - 1})$, we reach the same contradiction. So we have
\begin{flalign}\label{prop:previous_no_star}
\bullet\qquad\text{$\a_j(\zeta_{t - 1}) \ne \star_*$ and $\a_j(\zeta'_{t - 1}) \not= \star_*$.}&&
\end{flalign}

In addition, by \eqref{prop:zeta_t}, $\zeta_t \not\in p^*(\vk^2)$. Thus by \cref{observations}~\ref{obs:first-star}, we know either $\a_j(\zeta_t) = \star_0$ or $\a_j(\zeta_t) \not= \star_*$. 

If $\a_j(\zeta_t) = \star_0$, by \cref{observations}~\ref{obs:first-star}, there must be a node $\zeta_{t'}$ with $t' < t$ such that $\a_j(\zeta_{t'}) = a_{\x_j(\zeta_{t})}$. Since $\a_j(\zeta'_{t'}) = \a_j(\zeta_{t'}) = a_{\x_j(\zeta_{t})} = a_{\x_j(\zeta'_{t})}$, in such case, we must have $a_{\x_j(\zeta'_{t'})} \in C'_{t - 1}$ and $\a_j(\rig(\zeta'_{t - 1})) = \star_*$. Besides, by \eqref{prop:at-least-one-true}, if $\a_j(\zeta_t) \neq \star_*$, we must have $\a_j(\rig(\zeta'_{t - 1})) = \star_*$. 

Now, given~\eqref{a-j-star} and~\eqref{either-nstar-or-starz}, we consider the following two cases, and show that both of them lead to contradictions.

\highlight{Case 1: $\a_j(\rig(\zeta'_{t - 1})) = \star_*$ and $\a_j(\zeta_t) \not= \star_*$.} Since $\a_j(\zeta'_{t - 1}) \not= \star_*$ \eqref{prop:previous_no_star}, by \Cref{observations} \ref{obs:first-star-beta}, there must be $\eta \in \tilde{p}(\vk^1, \vk^2)$ with $(\a_j(\eta), \level(\eta)) = (\a_j(\zeta_t), j)$. Since $j < i$, this contradicts the fact that there is no collision between $p(\sigma)$ and $\tilde{p}(\vk^1, \vk^2)$ below level $i$. 

\highlight{Case 2: $\a_j(\rig(\zeta'_{t - 1})) = \star_d$ and $\a_j(\zeta_t) = \star_*$.} By \eqref{either-nstar-or-starz}, we have $\a_j(\zeta_t) = \star_0$ here. Moreover, since we assumed $\a_j(\rig(\zeta'_{t - 1})) \neq \a_j(\zeta_t)$ for contradiction, we know $d > 0$ in this case. By \Cref{observations} \ref{obs:if-star-must-first-2}, there must be $\eta \in \tilde{p}(\vk^1, \vk^2)$ with $(\a_j(\eta), \level(\eta)) = (\star_0, j) = (\a_j(\zeta_t), \level(\zeta_t))$. This leads to the same contradiction. 

\end{proofof}

\end{proofof}

Next we need the following corollary of~\cref{lem:moving-1-path-2}.

\begin{cor} \label{cor:moving-2-paths-2}
	Fix $\vk^1, \vk^2 \in \N^\ell$ such that $\vk^1 < \vk^2$ and $(w, T) \in \supp(\bw^{\vk^1, \vk^2}, \bT)$. Assuming that $\vidx^{\vk^1}$ and $\vidx^{\vk^2}$ exist. Suppose there are two nodes $\mu$ and $\eta$ such that
\twoprops{$\mu \ne \eta$ and $\mu,\eta \leq \vidx^{\vk^2}$, and}{asp-1}{there is no collision between $p(\mu)$ and $p(\eta)$.}{asp-2}
Then there exist two nodes $\mu', \eta'$ such that 
\twoprops{$(\a(\mu), \level(\mu)) = (\a(\mu'), \level(\mu'))$ and $(\a(\eta), \level(\eta)) = (\a(\eta'), \level(\eta'))$, and}{conc-1}{there is no collision between $p(\mu'), p(\eta')$, $p(\vk^1)$ and $p(\vk^2)$.}{conc-2}
\end{cor}
\begin{proof}
We first apply Lemma~\ref{lem:moving-1-path-2} twice: (1) with node $\sigma = \mu$ to get $\pi_3^{\mu} = \Findmainalgo(\mu,T)$ and its descendant $\mu'$ and (2) with node $\sigma = \eta$ to get $\pi_3^{\eta} = \Findmainalgo(\eta,T)$ and its descendant $\eta'$. 

For convenience we let $i_{\mu} = \level(\pi_3^{\mu})$ and $i_{\eta} = \level(\pi_3^{\eta})$. By Lemma~\ref{lem:moving-1-path-2}, we have
\threeprops{$(\a(\mu'), \level(\mu')) = (\a(\mu), \level(\mu))$ and $(\a(\eta'), \level(\eta')) = (\a(\eta), \level(\eta))$.}{cpp-2}
{For $\sigma' \in \{ \mu',\eta' \}$, there is no collision $p(\sigma')$ and $\tilde{p}(\vk^1, \vk^2)$.}{cpp-3}
{$p(\mu')$ is the same as $p(\mu)$ below level $i_\mu$; $p(\eta')$ is the same as $p(\eta)$ below level $i_\eta$.}{cpp-4}

Note that~\eqref{conc-1} follows immediately from~\eqref{cpp-2}. So we only need to show there is no collision between $p(\mu'), p(\eta')$, $p(\vk^1)$ and $p(\vk^2)$ (\ie,~\eqref{conc-2}). Note that $p(\vk^1), p(\vk^2) \subseteq \tilde{p}(\vk^1,\vk^2)$,~\eqref{cpp-3} and \cref{observations}\ref{obs:no-collision-k1-k2} further imply that we only need to show there is no collision between $p(\mu')$ and $p(\eta')$.

Let $i$ be the level of the lowest common ancestor of $\mu$ and $\eta$. We will consider two cases below.

\highlight{Case I: $i_\mu > i$ and $i_\eta > i$.} By Remark~\ref{rem:same-pi-3}, it follows that $\pi_3^\eta = \pi_3^\mu$.

From~\eqref{cpp-4}, $\min(i_\mu,i_\eta) > i$ and the fact that $\mu$ and $\eta$ have lowest common ancestor at level $i$, it follows that $p(\mu')$ and $p(\eta')$ also have lowest common ancestor at level $i$. Therefore, it suffices to check $p(\mu')$ and $p(\eta')$ have no collision below level $i$.

Applying~\eqref{cpp-4} together with~\eqref{asp-2}, it follows that there is no collision between $p(\mu')$ and $p(\eta')$. 

\highlight{Case II: $i_{\mu} \leq i$ or $i_{\eta} \leq i$.} Without loss of generality, we assume that $i_{\mu} \leq i_{\eta}$. Consequently, it follows that $i_{\mu} \leq i$. In this case, we have that $p(\mu')$ and $p(\eta')$ are the same with $p(\mu)$ and $p(\eta)$ below level $i_{\mu}$ respectively (from~\eqref{cpp-4} and note $i_{\mu} \leq i_{\eta}$). From~\eqref{asp-2}, it follows that $p(\mu')$ and $p(\eta')$ has no collision below $i_{\mu}$. 

The part of $p(\mu')$ with level above or equal to $i_{\mu}$ is exactly $p(\pi^{\mu}_3) \subseteq \tilde{p}(\vk^1, \vk^2)$. Therefore since $p(\eta')$ has no collision with $\tilde{p}(\vk^1, \vk^2)$ (from~\eqref{cpp-3}), there is no collision between $p(\mu')$ and $p(\eta')$.
\end{proof}

\subsection{Proof of~\cref{lem:structure-2}}\label{sec:bad-structure-final-proof}

\highlight{Notation.} To prove~\cref{lem:structure-2}, we first recall the following notations. We use $\a(\mu)$ as a shorthand for $\a_{\level(\mu)}(\mu)$ and $\x(\mu)$ as a shorthand for $\x_{\level(\mu)}(\mu)$. We always fix $\vk^1, \vk^2 \in \N^\ell$ such that $\vk^1 < \vk^2$ and $(w^{\vk^1, \vk^2}, T) \in \supp(\bw^{\vk^1, \vk^2}, \bT)$. We let $\i = \vidx^{\vk^1} + 1$ and $\j = \vidx^{\vk^2} + 1$. Recall that we denote $w^{\vk^1, \vk^2}$ by $w$ to simplify the notation. We also need the definition of pairs of good duplicates.

\begin{definition} \label{def:earliest-duplicate}
We call $(\bar \alpha, \bar \beta)$ a pair of good duplicates if $\bar \alpha < \bar \beta < \j$ and $a_{w_{\bar \alpha}} = a_{w_{\bar \beta}}$. A pair of good duplicates is said to be non-dominated if there is no other good duplicate $(\alpha', \beta') \not= (\alpha, \beta)$ such that $\alpha' \leq \alpha, \beta' \leq \beta$.
\end{definition}

\begin{lemma} \label{lem:no-star}
The following hold:
\begin{itemize}
    \item Suppose $(\alpha, \beta)$ is a non-dominated pair of good duplicates. For every $\mu \in [\alpha]$, $\a(\mu) \not= \star_*$. 
    \item If $(\alpha, \beta)$ is the pair of good duplicates with the minimum $\beta$, then for every $\mu \in [\beta - 1]$, $\a(\mu) \not= \star_*$. 
\end{itemize}
\end{lemma}
\begin{proof}
Let $(\alpha, \beta)$ be the pair of good duplicates with the minimum $\beta$. We first show that Item~(2) implies Item~(1). To see it, note that for every non-dominated pair of good duplicates $(\alpha', \beta')$, since $\beta \leq \beta'$, by its non-dominated property, we must have $\mu \leq \alpha' \leq \alpha \leq \beta - 1$. 

So now it suffices to prove Item~(2). It suffices to prove that there is no node $\gamma \in [\beta - 1]$ such that $\a(\gamma) = \star_0$. Suppose there is such $\gamma$, then there must be a node $\gamma' < \gamma$ such that $a_{\x(\gamma')} = a_{\x(\gamma)}$. Since we also know $\gamma, \gamma' < \j$ by \Cref{def:earliest-duplicate}, from \cref{observations}~\ref{obs:collision-imply-duplicate}, we know $a_{w_{\gamma}} = a_{w_{\gamma'}}$. This contradicts the assumption that $(\alpha, \beta)$ is the pair of good duplicates  with minimum $\beta$, since $\gamma < \beta$.
\end{proof}

We also need the following corollary.

\begin{cor} \label{cor:min-alpha-1} 
Suppose $(\alpha, \beta)$ is the pair of good duplicates with minimum $\alpha$. For every $\mu \in [\j] \setminus \{\alpha\}$, it holds that $\a(\alpha - 1) \not= \a(\mu - 1)$. 
\end{cor}
\begin{proof}
Suppose $\a(\alpha - 1) = \a(\mu - 1)$. By Lemma \ref{lem:no-star}, we know $\a(\alpha - 1) \not= \star_*$. Since $\alpha - 1, \mu - 1 < \j$, by \ref{obs:collision-imply-duplicate} of \cref{observations}, we have $a_{w_{\alpha - 1}} = a_{w_{\mu - 1}}$. Since $\mu \leq \j$ and $\mu \not= \alpha$, $(\alpha - 1, \mu - 1)$ is a pair of good duplicates. This contradicts the minimality of $\alpha$.  
\end{proof}

The following lemma is crucial for the proof of~\cref{lem:structure-2}.

\begin{lemma} \label{lem:any-earliest-dup} 
Let $(\bar \alpha, \bar \beta)$ be any non-dominated pair of good duplicates. There is no collision between $p(\bar \alpha)$ and $p(\bar \beta)$. 

Moreover, if $a_{w_{\bar \alpha}} = a_{w_{\bar \beta}} \not= a_{w_\i} = a_{w_\j}$, we can always find $\alpha, \beta$ such that $\alpha \not= \beta$, $a_{w_\alpha} = a_{w_\beta}$, $\{\alpha, \beta \} \not\subseteq \{\i, \j\}$, and there is no collision between $p(\i - 1), p(\j - 1), p(\alpha - 1), p(\beta - 1)$.
\end{lemma}
\begin{proof}
Suppose there is $\pi_1 \in p(\bar \alpha - 1) \setminus p(\bar \beta - 1)$ and $\pi_2 \in p(\bar \beta - 1) \setminus p(\alpha - 1)$ such that $\a(\pi_1) = \a(\pi_2)$. By \cref{lem:no-star}, we know $\a(\pi_1)$ and $\a(\pi_2)$ cannot be $\star_*$. Moreover, $\pi_1, \pi_2 \leq \bar \beta - 1 < \j$ (\ie,$\pi_1, \pi_2 \leq \vidx^{\vk^2}$ as $\j = \vidx^{\vk^2} + 1$), by \cref{observations}~\ref{obs:collision-imply-duplicate}, $a_{w_{\pi_1}} = a_{w_{\pi_2}}$. Again since $\pi_1 < \bar \alpha, \pi_2 < \bar \beta$, this contradicts the fact that $(\bar \alpha, \bar \beta)$ is a non-dominated pair of good duplicates. Hence, there is no collision between $p(\bar \alpha - 1)$ and $p(\bar \beta - 1)$.

Applying Corollary \ref{cor:moving-2-paths-2} to $\bar \alpha - 1$ and $\bar \beta - 1$, we can get $\alpha - 1$ and $\beta - 1$ such that there is no collision between $p(\i - 1), p(\j - 1), p(\alpha - 1)$, and $p(\beta - 1)$. However, this does not guarantee that $\{\alpha, \beta\} \not\subseteq \{\i, \j\}$. 

We now show $\{\alpha, \beta\} \not\subseteq \{\i, \j\}$ via a proof by contradiction. Assume $\{\alpha, \beta\} \not\subseteq \{\i, \j\}$. Then we know $a_{w_{\alpha}} = a_{w_{\beta}} = a_{w_\i} = a_{w_\j}$.

Since $(\a(\alpha - 1), \level(\alpha - 1)) = (\a(\bar \alpha - 1), \level(\bar \alpha - 1))$, we know $\nex(\alpha - 1) = r_{\level(\alpha - 1)}(\a(\alpha - 1)) = r_{\level(\bar \alpha - 1)}(\a(\bar \alpha - 1)) = \nex(\bar \alpha - 1)$. Therefore, $w_{\alpha} = \nex(\alpha - 1) = \nex(\bar \alpha - 1) = w_{\bar \alpha}$. The same also holds for $\beta$, and we know $w_{\beta} = w_{\bar \beta}$. 

Then $a_{w_{\alpha}} = a_{w_{\bar \alpha}}$, $a_{w_{\beta}} = a_{w_{\bar \beta}}$. Recall we also know $a_{w_{\alpha}} = a_{w_{\beta}} = a_{w_\i} = a_{w_\j}$. This contradicts with the assumption that $a_{w_{\bar \alpha}} = a_{w_{\bar \beta}} \not= a_{w_\i} = a_{w_\j}$. Hence $\{\alpha, \beta \} \not\subseteq \{\i, \j\}$. 
\end{proof}

Finally, we are ready to prove~\cref{lem:structure-2}, which is restated below.

\begin{reminder}{\cref{lem:structure-2}}
For every $u,v\in [n]$ such that $u \not= v$ and $a_u = a_v$, fix $(w,T) \in \supp(\bw^{\vk^1, \vk^2},\bT)$. Let $(\i - 1, \j - 1)$ be the pair of nodes $(\vidx^{\vk^1}, \vidx^{\vk^2})$ and assume $w_\i = u, w_\j = v$. Let $(\tilde{\alpha},\tilde{\beta})$ be any pair such that $0 < \tilde{\alpha} < \tilde{\beta} < \j, a_{w_{\tilde{\alpha}}} = a_{w_{\tilde{\beta}}}$. 

If such pair $(\tilde{\alpha},\tilde{\beta})$ exists, the following must hold:

\begin{itemize}
	\item There are two nodes $\alpha$ and $\beta$ such that $\alpha \not= \beta$, $a_{w_\alpha} = a_{w_\beta}$, $\{\alpha, \beta \} \neq \{\i, \j\}$, and there is no collision between $p(\i - 1), p(\j - 1), p(\alpha - 1), p(\beta - 1)$.
\end{itemize} 
\end{reminder}

\begin{proofof}{Lemma~\ref{lem:structure-2}}

Since $(\tilde{\alpha},\tilde{\beta})$ is a pair of good duplicates, we know there exists pairs of good duplicates. By Lemma \ref{lem:any-earliest-dup}, the only case left is when $a_{w_{\bar \alpha'}} = a_{w_{\bar \beta'}} = a_{w_\i} = a_{w_\j}$ holds for all non-dominated pairs of good duplicates $(\bar \alpha', \bar \beta')$. Below we take $\bar \alpha$ to be the first $\bar \alpha$ such that $a_{w_{\bar \alpha}} = a_{w_\i}$ and take $\bar \beta$ to be the second one. This implies $(\bar \alpha, \bar \beta)$ dominates all other non-dominated pairs of good duplicates (if there are other such pairs). Thus it must be the unique non-dominated pair of good duplicates. Therefore, $(\bar \alpha, \bar \beta)$ is not only the pair with the minimum $\bar \alpha$, but also the pair with the minimum $\bar \beta$.  

By Lemma \ref{lem:any-earliest-dup}, there is no collision between $p(\bar \alpha - 1)$ and $p(\bar \beta - 1)$. Thus we can apply Corollary \ref{cor:moving-2-paths-2} to $\bar \alpha - 1$ and $\bar \beta - 1$ and get $\alpha$ and $\beta$ such that:
\twopropsshort{there is no collision between $p(\i - 1), p(\j - 1), p(\alpha - 1), p(\beta - 1)$, and}{no-col-ijab}{$(\a(\alpha-1), \level(\alpha-1)) = (\a(\bar\alpha - 1), \level(\bar\alpha - 1))$ and $(\a(\beta-1), \level(\beta-1)) = (\a(\bar \beta - 1), \level(\bar \beta - 1))$.}{same-alpha-beta}

We now consider the following three cases.

\highlight{Case 1: $\bar \alpha \not= \i$.} Since $\bar \alpha < \bar \beta \leq \j$, we also know that $\bar \alpha \not= \j$. By Corollary \ref{cor:min-alpha-1} and $\bar \alpha \notin \{\i,\j\}$, we have $\a(\bar \alpha - 1) \not= \a(\i - 1)$ and $\a(\bar \alpha - 1) \not= \a(\j - 1)$, meaning that $\a(\bar\alpha - 1) \notin \{\a(\i-1),\a(\j-1)\}$. Since $\a(\alpha - 1) = \a(\bar \alpha - 1)$ from~\eqref{same-alpha-beta}, it also follows that $\a(\alpha - 1) \notin \{\a(\i-1),\a(\j-1)\}$, and consequently $\alpha \not\in \{\i, \j\}$. Thus, we have $\{\alpha, \beta\} \not= \{\i,\j\}$.

\highlight{Case 2: $\bar \alpha = \i$ and $\beta \not= \j$.} In this case, we will prove that $\beta \not\in \{\i, \j\}$. 

Since $\bar \alpha \not= \bar \beta$ and there is no collision between $p(\bar \alpha - 1)$ and $p(\bar \beta - 1)$ from~\eqref{no-col-ijab}, we know that $(\a(\bar \alpha - 1),\level(\bar\alpha - 1)) \not= (\a(\bar \beta - 1),\level(\bar \beta - 1))$.

Then we also know $\bar \alpha \not= \beta$ since $(\a(\bar \alpha - 1),\level(\bar \alpha - 1)) \not= (\a(\bar \beta - 1),\level(\bar \beta -1)) = (\a(\beta - 1),\level(\beta - 1))$, where the last equality follows from~\eqref{same-alpha-beta}. Therefore, $\beta \not= \i = \bar\alpha$. By our assumption of such case, $\beta \not= \j$. Thus $\{\alpha, \beta\} \neq \{\i,\j\}$. 

\highlight{Case 3: $\bar \alpha = \i$ and $\beta = \j$.} This is the trickiest case. We will prove that we can still find two nodes $\gamma'_0$ and $\gamma'_1$ to satisfy the requirements of this lemma. 

Let $\tau$ be the node $\arg \max_\tau \{(\level(\tau), -\tau) \vert \tau \in [\bar \beta, \j - 2]\}$. Intuitively, $\tau$ is the node separating $\bar \beta - 1$ from $p(\j - 1)$. Note here $\bar \beta < \j$ by the definition of a pair of good duplicates.

We first show the existence of $\tau$ and it has higher level than $\j - 1$.
\begin{claim}
Node $\tau$ exists and $\level(\tau) > \level(\j - 1)$.
\end{claim}
\begin{proof}
Since by \eqref{same-alpha-beta} and the assumption of this case, $\level(\bar \beta - 1) = \level(\beta - 1)$   and $\beta = \j$, we know $\level(\bar \beta - 1) = \level(\j - 1)$. For the sake of contradiction, suppose that $\tau$ does not exist or $\level(\tau) \leq \level(\j - 1)$, we would have $\bar \beta - 1 \in p(\j - 1)$. Then by the definition of our extended walk, $(\a(\bar \beta - 1), \level(\bar \beta - 1)) \not= (\a(\j - 1), \level(\j - 1))$. 

On the other hand, the assumption $\beta = \j$ and~\eqref{same-alpha-beta} imply that $(\a(\bar \beta - 1), \level(\bar \beta - 1)) = (\a(\beta - 1), \level(\beta - 1)) = (\a(\j - 1), \level(\j - 1))$, a contradiction. This proves the claim.
\end{proof}

By $\beta = \j$ and~\eqref{same-alpha-beta}, we have $\level(\bar \beta - 1) = \level(\beta - 1) = \level(\j - 1)$. Since $\level(\tau) > \level(\j - 1) = \level(\bar \beta - 1)$, the parent of $\tau$ must be before $\bar \beta - 1$, namely $\pa(\tau) < \bar \beta - 1$. 

Let $\gamma_0 \in [\pa(\tau) + 1, \bar \beta - 1]$ and $\gamma_1 \in [\tau + 1, \j - 1]$ be the pair of good duplicates (\ie, $a_{w_{\gamma_0}} = a_{w_{\gamma_1}}$) within such range  that minimizes $\gamma_1$. See Figure \ref{fig:tau_and_gamma}. 

\begin{figure}
    \centering
\scalebox{0.9}{
\begin{tikzpicture}[very thick]
  \node [draw,circle,minimum size=32,inner sep=0pt, outer sep=0pt] (pt) { $\mathsf{par}(\tau)$};

  \node [left of = pt, yshift = 0,inner sep=0pt, outer sep=-0.3pt] (dd) {} edge [-] (pt);

  \node [right of = pt, yshift = 0,xshift=90,inner sep=0pt, outer sep=-0.3pt] (d1) {};

  \node [below of = d1, yshift = -20,inner sep=0pt, outer sep=-0.7pt] (d2) {} edge [-] (d1);

  \node [draw, circle, minimum size=32,inner sep=1pt, right of = d2] (d3) {$\bar{\beta}-1$} edge [-] (d2);

  \node [draw,circle,minimum size=32,inner sep=1pt, outer sep=0pt, right of = pt, xshift = 160] (t) { $\tau$}edge [-] (pt);

  \node [right of = pt, yshift = 0,xshift=18,inner sep=0pt, outer sep=-0.3pt] (t1) {};

  \node [below of = t1, yshift = -20,inner sep=0pt, outer sep=-0.7pt] (t2) {} edge [-, blue] (t1);

  \node [draw, circle, minimum size=32,inner sep=1pt, right of = t2, blue] (t3) {$\gamma_0-1$} edge [-, blue] (t2);

  \node [right of = pt, yshift = 0,xshift=90,inner sep=0pt, outer sep=-0.3pt] (d1) {};

  \node [below of = d1, yshift = -20,inner sep=0pt, outer sep=-0.7pt] (d2) {} edge [-] (d1);

  \node [draw, circle, minimum size=32,inner sep=1pt, right of = d2] (d3) {$\bar{\beta}-1$} edge [-] (d2);

  \node [right of = t, yshift = 0,xshift=90,inner sep=0pt, outer sep=-0.3pt] (tt) {} edge [-] (t);

  \node [right of = pt, yshift = 0,xshift=190,inner sep=0pt, outer sep=-0.3pt] (c1) {};

  \node [below of = c1, yshift = -20,inner sep=0pt, outer sep=-0.7pt] (c2) {} edge [-, blue] (c1);

  \node [draw, circle, minimum size=32,inner sep=1pt, right of = c2, blue] (c3) {$\gamma_1-1$} edge [-, blue] (c2);

  \node [right of = pt, yshift = 0,xshift=265,inner sep=0pt, outer sep=-0.3pt] (d1) {};

  \node [below of = d1, yshift = -20,inner sep=0pt, outer sep=-0.7pt] (d2) {} edge [-] (d1);

  \node [draw, circle, minimum size=32,inner sep=1pt, right of = d2] (d3) {$\j-1$} edge [-] (d2);

\end{tikzpicture}
}
    \caption{The structure of $\tau$, $\pa(\tau)$, and $\gamma_0$, $\gamma_1$}
    \label{fig:tau_and_gamma}
\end{figure}
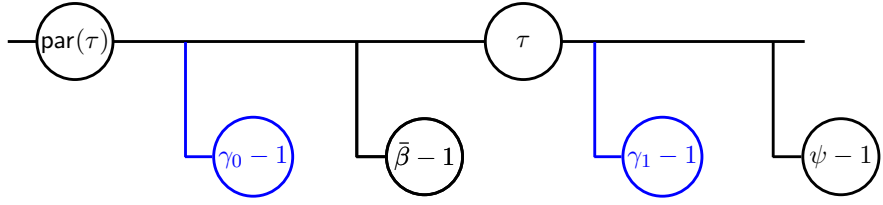

We prove the following two claims about $\gamma_0$ and $\gamma_1$.

\begin{claim} \label{claim:gamma-exist}
$\gamma_0$ and $\gamma_1$ exist.
\end{claim}
\begin{proof}
From $\beta = \j$ and~\eqref{same-alpha-beta}, we know $\a(\bar \beta - 1) = \a(\beta - 1) = \a(\j - 1)$. By Lemma \ref{lem:no-star} and our choice of $(\bar\alpha,\bar\beta)$, we know $\a(\bar \beta - 1) \not= \star_*$. Together with $\bar \beta - 1, \j - 1 < \j$, we can apply \cref{observations}~\ref{obs:collision-imply-duplicate}, and get $a_{w_{\bar \beta - 1}} = a_{w_{\j - 1}}$. Therefore, $\gamma_0$ and $\gamma_1$ must exists
\end{proof}

\begin{claim}
There is no collision between $p(\gamma_0 - 1)$ and $p(\gamma_1 - 1)$. 
\end{claim}
\begin{proof}
Here the proof idea is similar to that of Lemma \ref{lem:any-earliest-dup}. 

Suppose there is $\pi_0 \in p(\gamma_0 - 1) \setminus p(\gamma_1 - 1)$ and $\pi_1 \in p(\gamma_1 - 1) \setminus p(\gamma_0 - 1)$ such that $(\a(\pi_0), \level(\pi_0)) = (\a(\pi_1), \level(\pi_1))$. Since $\pi_0 \leq \gamma_0 - 1 \leq \bar \beta - 1$, by Lemma \ref{lem:no-star}, we know $\a(\pi_0) \not= \star_*$. Then together with $\pi_0, \pi_1 \leq \gamma_1 \leq \j - 1 = \vidx^{\vk^2}$, we can apply \cref{observations}~\ref{obs:collision-imply-duplicate} to get $a_{w_{\pi_0}} = a_{w_{\pi_1}}$. 

We then prove $\pi_0 \in [\pa(\tau) + 1, \bar \beta - 1]$ and $\pi_1 \in [\tau + 1, \j - 1]$. Together with $a_{w_{\pi_0}} = a_{w_{\pi_1}}$ and $\pi_1 \leq \gamma_1 - 1$, this contradicts the minimality of $\gamma_1$. 

Since $\gamma_0 - 1 \in [\pa(\tau), \tau]$, we know $\pa(\tau) \in p(\gamma_0 - 1)$. Also we have $\pa(\tau) \in p(\tau) \subseteq p(\gamma_1 - 1)$ because by the definition of $\tau$ it has the maximum level among nodes in $[\tau, \j - 1]$. (Note $\level(\tau) > \level(\j - 1)$ by \Cref{claim:gamma-exist}.) Thus we have $\pi_0 \geq \pa(\tau) + 1$ by $\pi_0 \in p(\gamma_0 - 1) \setminus p(\gamma_1 - 1)$. This implies $\level(\pi_0) < \level(\tau)$ since otherwise $\pa(\tau)$ would have been $\pi_0$. Hence by $\level(\pi_1) = \level(\pi_0)$ and $\pi_1 \in p(\gamma_1 - 1) \setminus p(\gamma_0 - 1)$, we know $\pi_1 \geq \tau + 1$.
\end{proof}

Since there is no collision between $p(\gamma_0 - 1)$ and $p(\gamma_1 - 1)$, we can apply \cref{cor:moving-2-paths-2} to $\gamma_0 - 1$ and $\gamma_1 - 1$ to get $\gamma'_0$ and $\gamma'_1$ such that:
\twoprops{there is no collision between $p(\gamma_0' - 1), p(\gamma_1' - 1), p(\i - 1), p(\j - 1)$, and}{no-col-gamma}{$(\a(\gamma_j-1), \level(\gamma_j-1)) = (\a(\gamma'_j - 1), \level(\gamma'_j - 1))$ for every $j \in \{0,1\}$.}{same-gamma} 
From~\cref{obs:w-nex-2} and~\eqref{same-gamma}, for $j \in \{0,1\}$, we have
\begin{equation}\label{eq:same-w-gamma}
w_{\gamma'_j} = r_{\level(\gamma'_j - 1)}(\a(\gamma'_j - 1)) = r_{\level(\gamma_j - 1)}(\a(\gamma_j - 1)) = w_{\gamma_j}.
\end{equation}

So from $a_{w_{\gamma_0}} = a_{w_{\gamma_1}}$ and~\eqref{eq:same-w-gamma}, we also have $a_{w_{\gamma'_0}} = a_{w_{\gamma'_1}}$. 

Finally, we show that the pair $(\gamma_0',\gamma_1')$ satisfies the requirements of the lemma.

\begin{claim}
$\gamma'_1 \not\in \{\i, \j\}$.
\end{claim}
\begin{proof}
Since $\gamma_1 - 1 \geq \tau \geq \bar \beta$, and $\bar \beta > \bar \alpha$. We know $\gamma_1 - 1 \not= \bar \alpha - 1$. By~\cref{cor:min-alpha-1} and our choice of $(\bar\alpha,\bar\beta)$, we have $\a(\bar \alpha - 1) \not= \a(\gamma_1 - 1)$. From our assumption $\i = \bar \alpha$ and~\eqref{same-gamma}, it follows that $\a(\i - 1) \not= \a(\gamma_1 - 1) = \a(\gamma'_1 - 1)$. Therefore, $\gamma'_1 \not= \i$. 

If $\gamma'_1 = \j$, we would have $(\a(\gamma_1 - 1), \level(\gamma_1 - 1)) = (\a(\gamma'_1 - 1), \level(\gamma'_1 - 1)) = (\a(\j - 1), \level(\j - 1))$ from~\eqref{same-gamma}. On the other hand, from our assumption $\beta = \j$ and~\eqref{same-alpha-beta}, we also know $(\a(\j - 1), \level(\j - 1)) = (\a(\beta - 1), \level(\beta - 1)) = (\a(\bar \beta - 1), \level(\bar \beta - 1))$. 

Thus $\a(\gamma_1 - 1) = \a(\bar \beta - 1)$. By Lemma \ref{lem:no-star}, we have $\a(\bar \beta - 1) \not= \star_*$. Together with $\gamma_1 - 1, \bar \beta - 1 \leq \j$ and \cref{observations}~\ref{obs:collision-imply-duplicate}, it follows that $a_{w_{\gamma_1 - 1}} = a_{w_{\bar \beta - 1}}$. 

Moreover, in this case, $\gamma_1 - 1\not= \tau$ since $\level(\tau) > \level(\j - 1)$ while $\level(\gamma_1 - 1) = \level(\j - 1)$. Thus, $\gamma_1 - 1 \in [\tau + 1, \j - 1]$. Together with $a_{w_{\gamma_1 - 1}} = a_{w_{\bar \beta - 1}}$, this contradicts the minimality of $\gamma_1$. Hence $\gamma_1' \ne \j$ and consequently $\gamma'_1 \not\in \{\i, \j\}$.
\end{proof}

Therefore, we found $(\gamma'_0, \gamma'_1) \not= (\i,\j)$ such that (1) there is no collision between $p(\gamma'_0 - 1), p(\gamma'_1 - 1), p(\i - 1), p(\j - 1)$, (2) $a_{w_{\gamma'_0}} = a_{w_{\gamma'_1}}$, and (3) $\{\gamma'_0, \gamma'_1\} \neq \{\i,\j\}$. This completes the whole proof.
\end{proofof}
\section{Proof of Lemma~\ref{lem:elong-small-2}}  \label{sec:appendix-elong}
\begin{reminder}{Lemma~\ref{lem:elong-small-2}}
	\underline{In probability space $(\bw^{\vk^1, \vk^2}, \bT)$}, it holds that
    \[
    \Pr[\elong] \le n^2 \ell/2^{\tau/4}.
    \]
\end{reminder}
\vspace{-1.5em}
\begin{proof}
	For every $i \in [\ell]$, we define event $\elong^i$ as
	\[
	\elong^i \coloneqq \left[ \text{$\exists k \in \N^{\ell}$ s.t. $ k_i > \tau/4$ and $\vidx^{\vk}$ exists} \right].
	\]
	Then we can see $\elong = \bigcup_{i=1}^{\ell} \elong^{i}$.
	
	In the following, we will show that for each $i \in [\ell]$, $\Pr[ \elong^i]$ is small. Now we fix $i \in [\ell]$, suppose there exists $\vk \in \N^{\ell}$ such that $\vidx^{\vk}$ exists and $k_i > \tau / 4$. We are going to fix $(r_{\le i -1},g_{\le i -1}) \in \supp((\br_{\le i -1},\bg_{\le i -1}))$ and conditioning on the event $r_{\le i -1} \land g_{\le i -1}$.
	
	Now, $\vidx^{\vk}$ exists and $k_i > \tau / 4$. Let $\vk' = (0,\dots,0,0,k_{i + 1},\dots,k_{\ell})$. This implies, there exists a starting point $s_0 \in [n]$, such that consider the walk $\walk^{\vk^1, \vk^2}(s_0,i,\vk')$, it visits at least $\tau/4$ level-$i$ nodes with corresponding vertices $\{x_j\}_{j \in [\tau/4]}$, which can be determined by $x_{j + 1} = \walk(s_j,i-1)$ based on $r_{\le i -1},g_{\le i -1}$. 

	\highlight{Case 1: $[\exists j \in [i + 1, \ell], k_j \not= k^2_j] \lor [\forall j \in [i + 1, \ell], k^1_j = k^2_j] \lor [k^2_i = 0]$.} In this case, we know $C_0 \gets \emptyset$. Then $\walk^{\vk^1, \vk^2}(s_0, i, \vk')$ visits the same vertices $x_j$ as $\walk(s_0, i - 1)$. With the same argument as Lemma~\ref{lem:elong-small}, for each $i \in [\ell]$, we know probability of existing such $\vidx^{\vk}$ is bounded by $n /2^{\tau/4}$. 

	\highlight{Case 2: $[\forall j \in [i + 1, \ell], k_j = k^2_j] \land [\exists j \in [i + 1, \ell], k^1_j \not= k^2_j] \land [k^2_i > 0]$.} In this case, $C_0 = C^{i, \vk^1}$ where $C^{i, \vk^1}$ is $\{\a_i(\mu)\}$ for all level $i$ nodes $\mu^{\vk^1}_{i, j}$ ($1 \leq j \leq k^1_i$) visited by $\walk(\bar s, i)$ for some starting point $\bar s$. Note here we use $\walk(\bar s, i)$ since these nodes are not on $p^*(\vk^2)$ and therefore falls into the previous case where $\walk^{\vk^1, \vk^2}$ visits the same vertices as $\walk$. 

	Suppose the vertices corresponds to the first $k^1_i$ level $i$ nodes $\walk(\bar s, i)$ visits are $\bar \mu_1, \bar \mu_2, \dots, \bar \mu_{k^1_i}$. Each $\x(\bar \mu_j)$ is determined by $\x(\bar \mu_j) = \last(\nex(\bar \mu_{j - 1}), i - 1)$ where $\nex(\bar \mu_0) = \bar s$. Fixing $r_{\leq i -1}, g_{\leq i - 1}$, from $\bar s$, we can uniquely determine $C^{i,\vk^1}$ by making $k^1_i$ adaptive queries to $g_i, r_i$.

	Then after determine $C^{i, \vk^1}$, for $\walk^{\vk^1, \vk^2}(s_0, i, \vk')$, we can determine each $\a_i(\mu^{\vk}_{i, j})$, $j \in [\tau/4]$ similarly by making $\tau / 4$ adaptive queries to $g_i, r_i$, and by definition they are distinct from those in $C^{i, \vk^1}$. 

	Note that $\bg_i,\br_i$ is independent of $(\bg_{\le i -1},\br_{\le i-1})$. Hence fixing $\bar s_0 \in [n]$ and $s_0 \in [n]$, we have that
	\[
	\Pr[g_i(a_i(\mu^{\vk}_{i,j})) = 1 \text{ for all $j \in [\tau/4]$}] \le 2^{-\tau/4}.
	\]
	
	Hence, by a union bound, we have that
	\[
	\Pr[\elong^i] \le n^2/2^{\tau/4}.
	\]	
	The lemma follows from another union bound over $i \in [\ell]$.
\end{proof}
\end{document}